\documentclass[10pt,a4paper]{scrreprt}
\usepackage[utf8]{inputenc} 

\usepackage[dvipsnames]{xcolor}

\usepackage{scrhack} 

\usepackage{helvet} 
\usepackage{palatino} 

\setkomafont{disposition}{\normalfont\sffamily\normalcolor\bfseries} 

\usepackage[onehalfspacing]{setspace} 

\usepackage{amsmath,amsfonts,amsthm}
\theoremstyle{plain}
\newtheorem{theorem}{Theorem}
\usepackage{booktabs}
\usepackage{bm}

\usepackage{multirow}
\usepackage{geometry}
\usepackage{tabularx}
\usepackage{makecell}

\usepackage{fancyvrb}
\usepackage{graphicx}
\usepackage{url}
\usepackage{hyperref}
\hypersetup{
    colorlinks,
    linkcolor={red!50!black},
    citecolor={blue!50!black},
    urlcolor={blue!80!black}
}

\usepackage{listings}
\usepackage{longtable}
\usepackage{pdflscape}
\usepackage{color}
\usepackage{colortbl}

\usepackage{placeins}

\usepackage{abstract}

\definecolor{gray}{rgb}{0.4,0.4,0.4}
\definecolor{darkblue}{rgb}{0.0,0.0,0.6}
\definecolor{cyan}{rgb}{0.0,0.6,0.6}

\definecolor{lighttuerkis}{HTML}{F7FCFC}
\definecolor{tuerkis}{HTML}{D9EEEE}
\definecolor{RoyalBlue}{RGB}{22,111,191} 
\definecolor{RoyalBluedark}{RGB}{10,80,191}
\definecolor{lightgray}{RGB}{240,240,240}
\definecolor{darkgray}{RGB}{225,225,225}

\definecolor{lightantituerkis}{HTML}{FFFCFA}
\definecolor{antituerkis}{HTML}{FFE5CF}
\definecolor{darkantituerkis}{HTML}{FF9A0C} 

\usepackage{tikz}
\usepackage{subcaption}
\usepackage{pgfplots}
\pgfplotsset{compat=newest} 
\usetikzlibrary{
    automata, 
    arrows.meta, 
    calc, 
    shapes.geometric, 
    graphs, 
    positioning, 
    calligraphy, 
    decorations.pathreplacing, 
    angles, 
    external,
    quotes
}

\usepackage[
  table-parse-only,
  detect-all,
  list-separator = \text{;\ },
  list-units = single,
  range-units = single
]{siunitx}

\lstdefinelanguage{XML} 
{
  numbers=none,
  captionpos=b, 
  basicstyle=\footnotesize\singlespacing,
  showstringspaces=false,
  morestring=[b]",
  morestring=[s]{>}{<},
  morecomment=[s]{<?}{?>},
  stringstyle=\color{black},
  identifierstyle=\color{darkblue},
  keywordstyle=\color{cyan},
  morekeywords={timestep, index, group, part, byte_order, file ,version, type}
}

\lstdefinelanguage{bash} 
{
  numbers=none,
  captionpos=b, 
  basicstyle=\small\ttfamily,
  showstringspaces=false,
  keywordstyle=\color{RoyalBlue}\bfseries, 
  morekeywords={convert,ffmpeg}
}

\lstdefinelanguage{myc++} 
{
  language=C++,
  captionpos=b, 
  keywordstyle=\color{RoyalBlue}\bfseries\em, 
  keywordstyle=[2]\color{RoyalBlue}\bfseries\em,
  basicstyle=\footnotesize\ttfamily,
  commentstyle=\color{gray}\ttfamily, 
  stringstyle=\rmfamily, 
  numbers=left, 
  numberstyle=\scriptsize, 
  stepnumber=1, 
  numbersep=8pt, 
  showstringspaces=false, 
  breaklines=true, 
  lineskip=1pt, 
  frame=tB, 
  belowcaptionskip=.75\baselineskip, 
}

\lstdefinestyle{floater}{
    float=tb, 
    frame=tB, 
    language=myc++
}

\lstdefinestyle{intext}{
    frame=leftline, 
    numbers=none, 
    xleftmargin=10pt, 
    language=myc++
}

\usepackage[T1]{fontenc}

\usepackage[
  backend       = biber,
  citetracker   = true,
  defernumbers  = true, 
  sorting      = nyt,
  natbib        = true,
  style         = numeric-comp,    
  giveninits    = true,
  maxnames      = 2,
  minnames      = 1,
  maxbibnames   = 10,
  url           = false,
  isbn          = false,
  doi           = true,
  eprint        = false
  ]
  {biblatex}

\makeatletter 
\let\blx@rerun@biber\relax
\makeatother

\apptocmd{\UrlBreaks}{\do\f\do\m}{}{}
\newcommand{\bigO}[1]{\mathbb{O}(#1)}
\newcommand{\field}[1]{\mathbb{#1}}

\newcommand{\RR}{\field{R}}
\newcommand{\Rplus}{\RR^+}
\newcommand{\EOC}{\text{EOC}}
\newcommand{\Err}{\text{Err}}
\newcommand{\class}[1]{\sloppy\lstinline[language=myc++]{#1}}
\renewcommand{\boldsymbol}{\bm}

\let\OldTexttt\texttt 
\renewcommand{\texttt}[1]{\sloppy{\footnotesize\OldTexttt{#1}}}

\newcommand{\eg}{e.g.\ }
\newcommand{\ie}{i.e.\ }

\usepackage{tocloft} 

\begin{document}

\setlength\cftchapnumwidth{2em}
\setlength\cftsecnumwidth{3em}
\setlength\cftsubsecnumwidth{4em}

\begin{titlepage}
	\begin{tikzpicture}[remember picture,overlay]
	\node[text width=1.0\textwidth] at ($ (current page.north) + (0mm,-40mm) $){
		\begin{center}
			{\Huge\sffamily\textbf{OpenLB User Guide}}\\[0.5cm]
			{\Large\sffamily Associated with Release 1.6 of the Code}
		\end{center}
	};
		\node[text width=1.0\textwidth] at ($ (current page.center) + (0mm,15mm) $){
		\begin{center}
		\includegraphics[width=\textwidth,angle=-90]{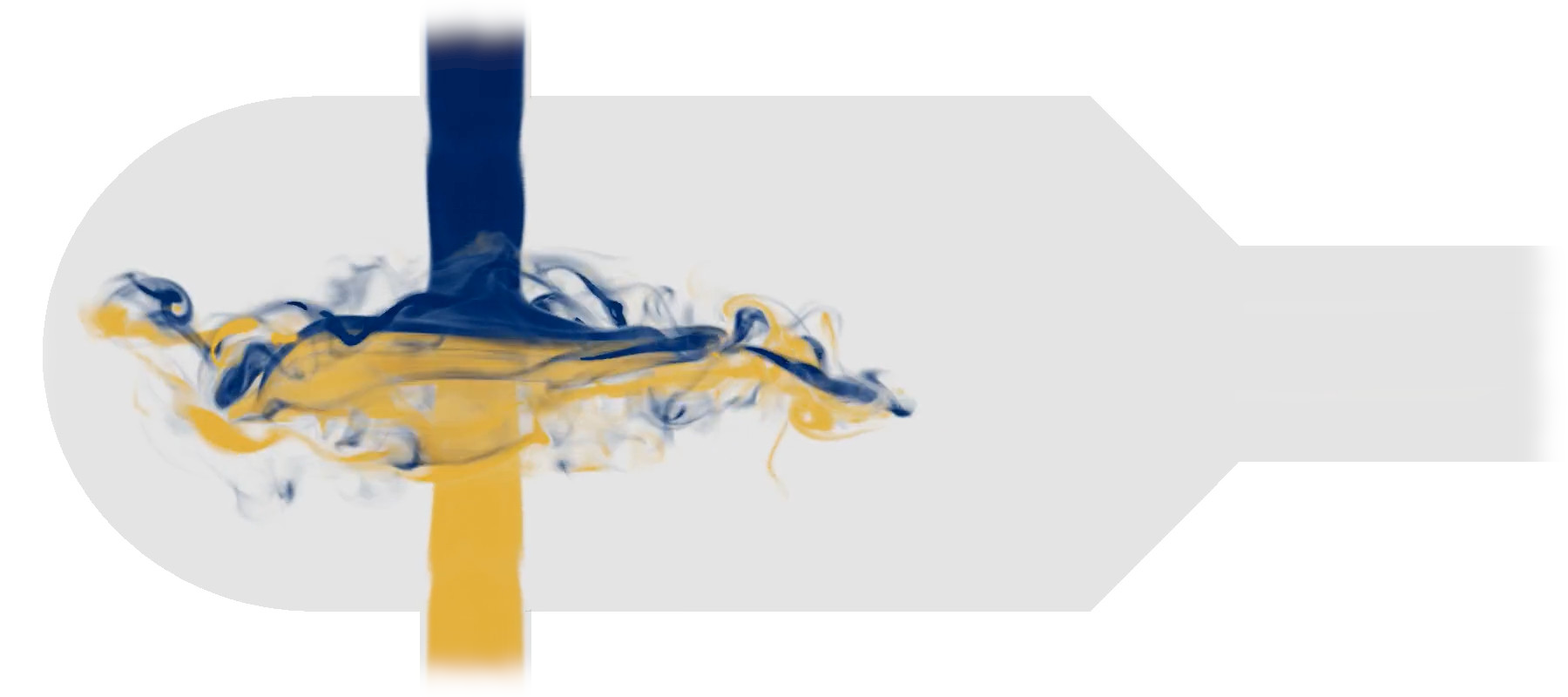}
		\end{center}
	};
	\node[text width=1.0\textwidth] at ($ (current page.south) + (0mm,+70mm) $){
		\begin{center}
		{\Large \textbf{\sffamily{Authors:}}
                Adrian Kummerl\"ander,
                Samuel J. Avis,
                Halim Kusumaatmaja,
                Fedor Bukreev,
                Michael Crocoll,
                Davide Dapelo,
                Simon Großmann,
                Nicolas Hafen,
                Shota Ito,
                Julius Jeßberger,
                Eliane Kummer,
                Jan E. Marquardt,
                Johanna Mödl,
                Tim Pertzel,
                František Prinz,
                Florian Raichle,
                Martin Sadric,
                Maximilian Schecher,
                Dennis Teutscher,
                Stephan Simonis,
                Mathias J. Krause
			}
		\end{center}
	};
	\node[text width=1.0\textwidth] at ($ (current page.south) + (0mm,+35mm) $){
    \begin{flushright}
		\includegraphics[width=5cm]{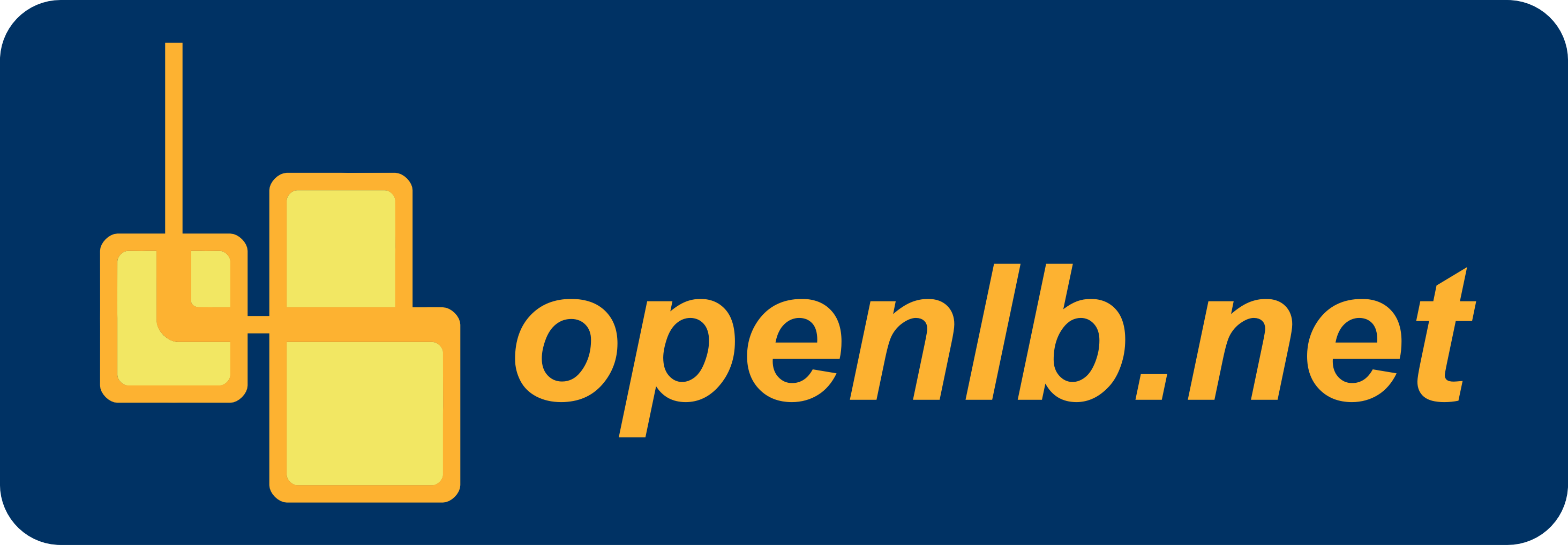}
	\end{flushright}
	};
	\end{tikzpicture}
\end{titlepage}

\centerline{Copyright \copyright{} 2006-2008 Jonas Latt}
\centerline{Copyright \copyright{} 2008-2023 Mathias J.\ Krause}
\centerline{\href{mailto:info@openlb.net}{\nolinkurl{info@openlb.net}}}\vspace{1cm}

Permission is granted to copy, distribute and/or modify this document
under the terms of the GNU Free Documentation License, Version 1.2
or any later version published by the Free Software Foundation;
with no Invariant Sections, no Front-Cover Texts, and no Back-Cover
Texts.  A copy of the license is included in the Section entitled ``GNU
Free Documentation License'' (Section~\ref{sec:license}).

\begin{abstract}

OpenLB is an object-oriented implementation of LBM. 
It is the first implementation of a generic platform for LBM programming, which is shared with the open source community (GPLv2). 
Since the first release in 2007 \cite{olb03}, the code has been continuously improved and extended which is documented by thirteen releases~\cite{olb16,olb15,olb14,olb13,olb12,olb11,olb10,olb09,olb08,olb07,olb06,olb05,olb04} as well as the corresponding release notes which are available on the OpenLB website \url{https://www.openlb.net}. 
The OpenLB code is written in C++ and is used by application programmers as well as developers, with the ability to implement custom models OpenLB supports complex data structures that allow simulations in complex geometries and parallel execution using MPI, OpenMP and CUDA on high-performance computers. 
The source code uses the concepts of interfaces and templates, so that efficient, direct and intuitive implementations of the LBM become possible. 
The efficiency and scalability has been checked and proved by code reviews. 
This user manual and a source code documentation by DoxyGen are available on the OpenLB project website.

\end{abstract}

\tableofcontents

\chapter{Introduction}

\section{Lattice Boltzmann Methods}

In general, the lattice Boltzmann method (LBM) can be interpreted as a bottom-up procedure to numerically approximate a given partial differential equation (PDE).
In contrast to the conventional top-down discretization of the target equation (TEQ), the LBM emerges from the space and time discretization of a discrete velocity Boltzmann-type equation.
Subsequently, the solution to the TEQ is recovered in a limiting process.
Different PDEs can be approximated by altering certain features of the method.
Exemplary LBM building blocks which induce the recovery of the aimed at transport equation are discussed in the following sections.
Thorough derivations of specific LBM approaches are summarized more comprehensively for example in~\cite{krause:21,kruger2017lattice}.

The LB algorithm is typically divided into two steps: collision and streaming.
During the collision step every distribution function $f_i$ at a grid node $\bm{x} \in \Omega_{\triangle x} \subseteq \mathbb{R}^{d}$ receives a collision term $J_{i} $, \ie 
\begin{equation}
f_i^*(t,\bm{x}) = f_i(t,\bm{x}) + \triangle t J_i  (t, \bm{x}) ~,
\end{equation}
where \(\triangle t\) is the time step size and $t \in I_{\triangle t} \subseteq \mathbb{R}_{\geq 0}$.
The most prominent collision operator introduced by Bhatnager, Gross and Krook (BGK)~\cite{bhatnagar1954model} reads
\begin{equation}
J_i (t, \bm{x})  = - \frac{1}{\tau}\left[ f_i (t, \bm{x})-f_i^{eq} (t, \bm{x}) \right] ~.
\end{equation}
The relaxation time $\tau > 0 $ determines the speed at which the populations approach equilibrium and can for example be related to the viscosity in the hydrodynamic limit.
The local equilibrium function $f_i^{eq}$ approximates the Maxwell\textendash Boltzmann distribution.
At the streaming step all populations are shifted to the next grid point
\begin{equation}
f_i \left( t + \triangle t, \bm{x} + \triangle t \bm{c}_i \right) = f_i^*(t, \bm{x}) ~,
\end{equation}
where \(i=0,1, \ldots, q-1\) enumerates the discrete velocities \(\bm{c}_i\) of dimension \(d\) contained in the underlying set \(DdQq\). 
For the purpose of illustration, the \(D2Q9\) and the \(D3Q19\) discrete velocity sets are shown in Figures~\ref{fig:D2Q9} and~\ref{fig:D3Q19}, respectively. 
\begin{figure} 
  \centering
  \subfloat[\(D2Q9\)]{
      \includegraphics[scale=1]{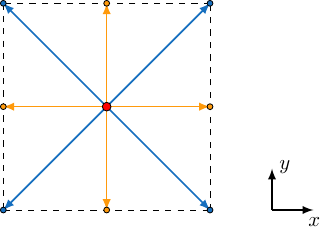}
      \label{fig:D2Q9}
  } \hspace{1em}
  \subfloat[\(D3Q19\)]{ 
      \includegraphics[scale=1]{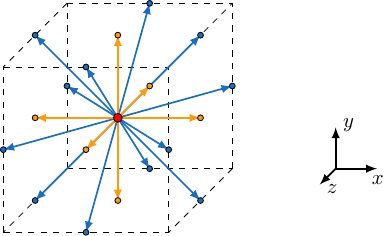}
      \label{fig:D3Q19}
  }
  \caption{Exemplary discrete velocity sets used for the descriptors implemented in OpenLB.}
  \label{fig:discreteVelocitySets}  
\end{figure}
As a linkage from the LBM to the TEQ, the macroscopic conservative variables are obtained by respective moment summation over the populations and typically yield an approximation up to second order in space and time. 

The following references are suggested for further insight into LBM.
\begin{itemize}
\item The book \textit{The Lattice Boltzmann Method: Principles and Practice} [2017] by Kr\"uger \textit{et al.}~\cite{kruger2017lattice} delivers a clear and complete introduction for beginners.
\item A concise introduction is given by Mohamad \cite{mohamad:11}.
The book \textit{Lattice Boltzmann Method} [2011] documents for example the formal derivation of macroscopic limit equations using Chapman--Enskog expansion.
\item An LBM predecessor\textemdash the Lattice-Gas Cellular Automata (LGCA)\textemdash is extensively described in Wolf-Gladrow's book \textit{Lattice-Gas Cellular Automata and Lattice Boltzmann Models} [2000].
Starting with the derivation of LGCA, the author derives the LBM step by step.
Furthermore, a helpful interpretation of LBM is given in the beginning of the book.
\item A quick overview of LBM, can be obtained from the paper \textit{Lattice Boltzmann Method for Fluid Flows} [1998] by Chen and Doolen \cite{chen-doolen:98}.
\end{itemize}

\section{OpenLB Project}

\subsection{Scope and Overview}
\label{sec: openlb overview}

OpenLB is a generic implementation of LBM written using the C++ programming language and is used by application programmers as well as method developers. It is the first implementation of a comprehensive platform for LBM, that is freely shared with the open source community. Since its conception, OpenLB is developed as open source under the GPL2 license. Therefore, everyone with access to the source code has the right to use, adapt and publish the software under the same license. Since the first release in 2007 ~\cite{olb03}, the code has been continuously improved and extended across thirteen major releases. User guides and source code documentation for developers are available on the project website. In summary, OpenLB is a technically complex and feature-rich CFD software library with an extensive set of example cases also written in modern C++. 

\begin{figure}[ht]
    \centering
    \includegraphics[width=\textwidth]{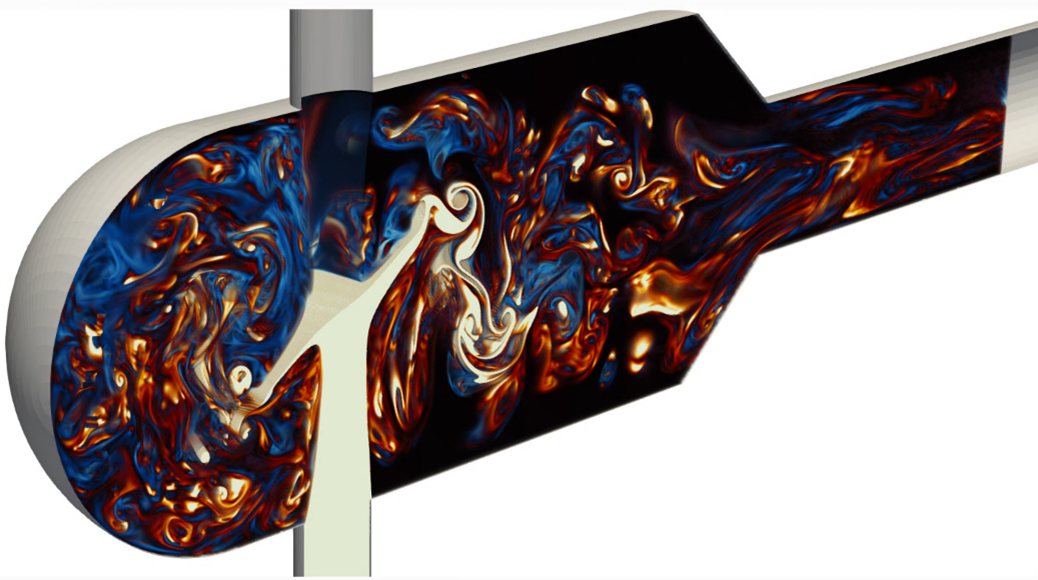}
    \caption{Resolved simulation with 250 million grid cells of turbulent mixing and chemical reactions done by OpenLB on 24 GPUs on the HoreKa supercomputer funded by the Ministry of Science, Research and the Arts Baden-Württemberg and by
the Federal Ministry of Education and Research.}
    \label{fig:turbulent_mixer}
\end{figure}

OpenLB is actively used in both developing novel methodologies and their applications. Developers of OpenLB, especially researchers from mathematics and computer science, focus on the methodology-oriented approach. Here a selection of recent (after 2016) publications is given:

\begin{itemize}
    \item General LBM:~\cite{mink:16, hoecker:18, haussmann:20b, simonis:20, simonis:21, simonis:23}
    \item HPC LBM:~\cite{mohrhard:19,kummerlaender:23,kummerlander2022advances} 
    \item Particle LB -- Methods:~\cite{trunk:16, maier:17, krause:17, trunk:18, dapelo:21, maier:21, trunk:21}
    \item Turbulence LB -- Methods:~\cite{nathen:17, haussmann:19a, haussmann:20, simonis:21, simonis:22}
    \item Optimal Control LB -- Methods:~\cite{klemens:18, klemens:18a, zarth:21, reinke:22}
\end{itemize}

The application-oriented approach is chosen by esp. engineering researchers. At the top level, each simulation in OpenLB currently consists of a fully parallel mesh construction (without any need for external tools ), physical parametrization, assignment of dynamics encoding the cell-specific models , configuration of operators for multi-physics coupling, non-local boundary treatment and a selection of post-processing functions for extracting the relevant flow features for further evaluation~\cite{krause:21}. OpenLB enables application in a broad set of areas (physical characterizations) such as incompressible Newtonian and non-Newtonian fluid flows, multi–phase and multi–component flows (cf.\ Figure~\ref{fig:turbulent_mixer}), light and thermal radiation as well as thermal flows, particulate flows using both Euler–Euler and Euler–Lagrange with resolved and sub-grid scale models, turbulent flow models (LES models and wall function approaches) and porous media models~\cite{krause:21}. A selection of recent (after 2016) application-oriented publications is listed here:
\begin{itemize}
    \item Particle LB -- Applications:~\cite{henn:16, augusto:17, maier:18, bretl:21, marquardt:21, trunk:21a, , simonis:22a, ditscherlein:22, hafen:22, mink:22, thieringer:22, bukreev:23, hafen:23a}
    \item Turbulence LB -- Applications:~\cite{mirzaee:16, nadim:16, gaedtke:18, gaedtke:19, haussmann:19, gaedtke:19a, haussmann:20, haussmann:21, siodlaczek:21}
    \item Optimal Control LB -- Application:~\cite{klemens:18a, rossjones:19, klemens:20, jessberger:22}
\end{itemize}

OpenLB is not only an open source code, but also a community that aims to share resources for software development, produce reproducible scientific results, and provide a good start for students in LBM. Through the years, there have been about over 150 scientific articles published, over 90 with and over 60 without the OpenLB developers’ participation (2022). Until January 2023, more than 50 individuals from Germany, Switzerland, France, United Kingdom, Brazil, Czech Republic and USA have contributed to OpenLB.

\subsection{What Makes OpenLB Special?}
OpenLB is a numerical framework for lattice Boltzmann simulations, created by students and researchers with different backgrounds in computational fluid dynamics. The code can be used by application programmers to implement specific flow geometries in a straightforward way, and by developers to formulate new models. For this first group of users, OpenLB offers a neat interface through which it is possible to set up a simulation with little effort. For the second group, the structure of the code is kept conceptually simple, implementing basic concepts of the lattice Boltzmann theory step-by-step. Thanks to this, the code is an excellent framework for programmers to develop pieces of reusable code that can be exchanged in the community.

One key aspect of the OpenLB code is genericity in its many facets. The core concept of generic programming is to offer a single code that can serve many purposes. On one hand, the code implements dynamic genericity through the use of object-oriented interfaces. One use of this is that the behavior of lattice sites can be modified during program execution, to distinguish for example between bulk and boundary cells, or to modify the fluid viscosity or the value of a body force dynamically. Furthermore, C++ templates are used to achieve static genericity. As a result, it is sufficient to write a single generic code for various 3D lattice structures, such as D3Q15, D3Q19, and D3Q27 (for more information on lattice structures, see Section~\ref{ssec:Descriptor}).

\subsection{Which Features are Currently Implemented?}
An excerpt of the the current features is given below.
Note that an extended list is provided in~\cite{krause:21}. 
For detailed description of individual features, see the respective sections below. 

\subsubsection{Pre-processing}
\begin{table}[ht!]
\begin{tabular}{|l|l|l|}
\hline
Simulation domain construction with volume meshes (VTI) & \\
Simulation domain construction with surface meshes (STL) & Section~\ref{sec:geometry_creation}\\
Simulation domain construction with geometry primitives  & Section~\ref{sec:geometry_creation}\\
Automatic parallel meshing  & Section~\ref{sec:geometry_creation}\\
Automatic load balancing  & \\
Segmentation for boundaries setting & Section~\ref{sec:materialnumbers_setting}\\
Assigning boundary and initial values & Section~\ref{sec:functors},~\ref{sec:functor_aplication} \\
\hline
\end{tabular}
\end{table}
\FloatBarrier

\subsubsection{Post-processing}
\begin{table}[ht!]
\begin{tabular}{|l|l|l|}
\hline
Interfaces to parallel  VTK formats  & Section~\ref{sec:outPut-vtk}   \\
Interface CSV & Section~\ref{sec:csvWriter}  \\
Writing images build-in& Section~\ref{sec:writing_image}  \\
Evaluating with Gnuplot & Section~\ref{sec:gnuplot} \\
Analytical functors for \eg error norms, integrals, ...& Section~\ref{sec:functors},~\ref{sec:functor_aplication} \\
\hline
\end{tabular}
\end{table}
\FloatBarrier

\subsubsection{Lattice Boltzmann Models}
\begin{table}[ht!]
\begin{tabular}{|l|l|l|}
\hline
BGK model for fluids               & Section~\ref{sec:dynamics} & Reference~\cite{chen-doolen:98} \\
Regularized model for fluids       & Section~\ref{sec:dynamics} & Reference~\cite{latt-regularized:06} \\
Multiple Relaxation Times (MRT)    & Section~\ref{sec:dynamics} & References~\cite{mrt-3d:02,yu:03}\\
Entropic Lattice Boltzmann         & Section~\ref{sec:dynamics} & Reference~\cite{ansumali-thesis}\\
BGK with adjustable speed of sound & Section~\ref{sec:dynamics} & References~\cite{tr2,chopard:02}\\
BGK and MRT with Smagorinsky model & Section~\ref{sec:dynamics} & References~\cite{nathen:13}\\
Porous media model                 & Section~\ref{sec:dynamics} &\\
Power law model                    & Section~\ref{sec:dynamics} &\\
\hline
\end{tabular}
\end{table}
\FloatBarrier

\subsubsection{Multiphysics Coupling}
See Section~\ref{sec:couplings}. 
\begin{table}[ht!]
\begin{tabular}{|l|l|l|}
\hline
Shan-Chen two-component fluid & Section~\ref{sec:couplings} & Reference~\cite{shan_chen:93} \\
Free energy model for multicomponent fluids & Section~\ref{sec:couplings} & Reference~\cite{semprebon:16} \\
Thermal fluid with Boussinesq approximation & Section~\ref{sec:couplings} & Reference~\cite{guo:02c} \\
\hline
\end{tabular}
\end{table}
\FloatBarrier

\subsubsection{Lattice Structures}
See Section~\ref{ssec:Descriptor}. 
Exemplary discrete velocity sets (lattices) available in OpenLB are \(D2Q5\), \(D2Q9\), \(D3Q7\), \(D3Q13\), \(D3Q15\), \(D3Q19\), \(D3Q27\). 

\subsubsection{Boundary Conditions for Straight Boundaries (Including Corners)}
See Section~\ref{sec:defineBoundaryMethod}.
\begin{table}[ht!]
\begin{tabular}{|l|l|l|}
\hline
Regularized & local & Default choice in examples\\
Finite difference (FD) velocity gradients & non-local & Default choice in examples \\
Inamuro & local & \\
Zou/He & local & \\
Non-linear FD velocity gradients & non-local & \\
\hline
\end{tabular}
\end{table}
\FloatBarrier

\subsubsection{Boundary Conditions for Curved Boundaries}
See Section~\ref{sec:defineBoundaryMethod}.
\begin{table}[ht!]
\begin{tabular}{|l|l|l|l|}
\hline
Bouzidi & non-local & References~\cite{bouzidi:01}\\
\hline
\end{tabular}
\end{table}
\FloatBarrier

\subsubsection{Input / Output}
The basic mechanism behind I/O operations in OpenLB is the serialization and unserialization of a \class{BlockLatticeXD}. 
This mechanism is used to save the state of a simulation and to produce VTK output for data post-processing with external tools. 
In both cases, the data is saved in the binary Base64 format, which ensures compact and (relatively) platform-independent data storage. 
For operational details, see for example Section~\ref{sec:vti_pvd-format}.

\subsection{Project Participants}

The OpenLB project was initiated in 2006. 
Between 2006 and 2008 Jonas Latt was the project coordinator. 
As of 2009, Mathias J.\ Krause is coordinating the project. 
Since 2006 more than 45 persons contributed to OpenLB. 
A list is provided in Section~\ref{sec: List of Project Participants} of the Appendix.

\subsection{Getting Help with OpenLB}\label{sec:getting-help}
The following resources are available for OpenLB users:
\begin{description}
\item[Web site:] Most recent releases of the code and documentation, including this user guide, can be found on the website \url{https://www.openlb.net/} .
\item[Forum:] If you experience troubles with OpenLB, you may wish to post your concerns to the Lattice Boltzmann community in the forum on the OpenLB homepage.
\item[Bug reports:] If you think you found a bug in OpenLB, we encourage you to submit a report to \href{mailto:bug@openlb.net}{\nolinkurl{bug@openlb.net}}. Useful bug reports include the full source code of the program in question, a description of the problem, an explanation of the circumstances under which the problem occurred, and a short description of the hardware and the compiler used.
Moreover, other makefile switches, such as buildtype and mode of parallelization found in \path{config.mk} can provide useful information too.
\item[Spring School:] To lower the hurdle and the assistance of people new to this field is the motivation in organizing yearly spring schools, namely \textbf{Lattice Boltzmann Methods with OpenLB Software Lab}. The first spring school was organized in \href{https://www.openlb.net/spring-school-2017/}{2017 in Hammamet (Tunisia)}, followed by \href{https://www.openlb.net/spring-school-2018/}{Karlsruhe in 2018}, \href{https://www.openlb.net/spring-school-2019/}{Mannheim in 2019}, \href{https://www.openlb.net/spring-school-2020/}{Berlin in 2020}, \href{https://www.openlb.net/spring-school-2022/}{Krakow (Poland) in 2022} and \href{https://www.openlb.net/spring-school-2023/}{London (UK) in 2023}. All spring schools are organized as open international workshops. The intention of creating an international platform with courses for beginners in LBM, which are researchers from academia and industry, was fulfilled with a response of 40 participants from 15 countries, averaged over the years. The event is based on the interlaced educational concept of comprehensive and methodology-oriented courses offered by the core developer team of OpenLB together with the local partner group as well as professional guest lecturers within the field of LBM. The first half of the week is dedicated to the theoretical fundamentals of LBM up to ongoing research on selected topics. Followed by mentored training on case studies using OpenLB in the second half, where the participants gain deep insights into LBM and its various applications in different disciplines. More information about the next spring school can be found on the OpenLB website. 
\item[Consortium:] If you want more or deeper support or get involved in further development of OpenLB, become a member of the OpenLB consortium. If you are interested, please send an email to \href{mailto:consortium@openlb.net}{\nolinkurl{consortium@openlb.net}}. The consortium promotes and aims at continuous development and maintenance of OpenLB, especially for (1) improvement of the general usage according to GNU GPL2 License and regular publication, (2) maintenance of the executability on current standard HPC-Hardware and (3) conservation of the current status of the LBM research in OpenLB.

\end{description}

\subsection{How to Cite OpenLB}

For references to OpenLB in general, we recommend to cite the most recent overview paper~\cite{krause:21}.
\begin{Verbatim}[fontsize=\small]
@article{openlb-2020,
    author    = "Krause, M. J. and Kummerl{\"a}nder, A. and Avis, S. J.
                 and Kusumaatmaja, H. and Dapelo, D. and Klemens, F.
                 and Gaedtke, M. and Hafen, N. and Mink, A.
                 and Trunk, R. and Marquardt, J. E. and Maier, M.-L.
                 and Haussmann, M. and Simonis, S.",
    title     = "OpenLB—Open source lattice Boltzmann code",
    year      = "2020",
    journal   = "Computers \& Mathematics with Applications",
    doi       = "10.1016/j.camwa.2020.04.033"
}
\end{Verbatim}
For references to a specific release of the code, each version is associated with a citable Zenodo publication. 
The latest release 1.6~\cite{olb16} is covered by the present user guide and citable via:
\begin{Verbatim}[fontsize=\small]
@software{olbRelease16,
  title        = {{OpenLB Release 1.6: Open Source Lattice Boltzmann Code}},
  author       = {Kummerl\"{a}nder, Adrian
                  and Avis, Sam
                  and Kusumaatmaja, Halim
                  and Bukreev, Fedor
                  and Crocoll, Michael
                  and Dapelo, Davide
                  and Hafen, Nicolas
                  and Ito, Shota
                  and Je\ss{}berger, Julius
                  and Marquardt, Jan E.
                  and M\"{o}dl, Johanna
                  and Pertzel, Tim
                  and Prinz, Franti\v{s}ek
                  and Raichle, Florian
                  and Schecher, Maximilian
                  and Simonis, Stephan
                  and Teutscher, Dennis
                  and Krause, Mathias J.},
  year         = 2023,
  month        = apr,
  publisher    = {Zenodo},
  doi          = {10.5281/zenodo.7773497},
  url          = {https://doi.org/10.5281/zenodo.7773497},
  version      = {1.6}
}
\end{Verbatim}
Starting with release 1.6, metadata is also defined in the standardized \emph{Citation File Format} and included as a \path{CITATION.cff} in the release tarball.

\chapter{Core Concepts} 

\label{sec:modelsAndCoreDataStructure}

The basic data structure of any \emph{Lattice Boltzmann} code is the eponymous \emph{Lattice} or grid.
For the most widely accepted formulations, such a grid is a regular, homogeneous lattice $\Omega_h$ with equal spacing $h\in\mathbb{R}_{>0}$ in all directions.
Each spatial location on this lattice can be referred to as a \emph{Cell} which is the core unit of the LB algorithm.

One of the main aspects motivating the use of LBM is its unique suitability for massive parallel processing.
This parallelization aspect naturally leads to the need for further spatial decomposition of the lattice data structure in order to be able to distribute the processing to multiple independent processing units.
In OpenLB, this spatial decomposition into so called \emph{Blocks} motivates the core naming convention of distinguishing between \emph{Super}- and \emph{Block}-level structures.

This approach respects the spirit of LBM well and leads to elegant and efficient implementations.
For complex geometries, a multi-block approach provides another advantage:
A given domain can be represented by a certain number of regular blocks, which delivers a good compromise between highly efficient block-local processing and sparse memory consumption.

For most practical applications, the lattice data structure not only manages the population values assigned to the individual cells and declared by the descriptor (cf.\ Section~\ref{ssec:Descriptor}), but also associated per-cell data such as force fields or reaction coefficients.

\section{Descriptor}\label{ssec:Descriptor}

Descriptors are the structure that is used to define and access LB model specific information such as the number of dimensions and discrete velocities as well as weights and declarations of additional fields.
As such, a descriptor is a central choice in any OpenLB application and passed as a template argument throughout the code base.
\begin{lstlisting}[language=myc++]
using T = FLOATING_POINT_TYPE;
using DESCRIPTOR = descriptors::D2Q9<>;

// number of spatial dimensions
const int d = descriptors::d<DESCRIPTOR>(); // == 2
// number of discrete velocities
const int q = descriptors::q<DESCRIPTOR>(); // == 9

// second discrete velocity vector
const Vector<int,2> c1 = descriptors::c<DESCRIPTOR>(1); // == {-1,1}

// weight of the first discrete velocity
const T w = descriptors::t<T,DESCRIPTOR>(0); // == 4./9.
\end{lstlisting}
OpenLB provides a rich set of such descriptors, most of which are defined in \path{src/dynamics/latticeDescriptors.h}. 
Despite the central role of this concept, the concrete definitions of for example the \class{D2Q9} descriptor, are quite compact. 
To illustrate this point, Listing~\ref{lst:d2q9} provides the full definition of this descriptor including all of its data.
\begin{lstlisting}[language=myc++,caption=Definition of the \class{D2Q9} descriptor,label=lst:d2q9]
template <typename... FIELDS>
struct D2Q9 : public LATTICE_DESCRIPTOR<2,9,POPULATION,FIELDS...> {
  D2Q9() = delete;
};

namespace data {

template <>
constexpr int vicinity<2,9> = 1;

template <>
constexpr int c<2,9>[9][2] = {
  { 0, 0},
  {-1, 1}, {-1, 0}, {-1,-1}, { 0,-1},
  { 1,-1}, { 1, 0}, { 1, 1}, { 0, 1}
};

template <>
constexpr int opposite<2,9>[9] = {
  0, 5, 6, 7, 8, 1, 2, 3, 4
};

template <>
constexpr Fraction t<2,9>[9] = {
  {4, 9}, {1, 36}, {1, 9}, {1, 36}, {1, 9},
  {1, 36}, {1, 9}, {1, 36}, {1, 9}
};

template <>
constexpr Fraction cs2<2,9> = {1, 3};

}
\end{lstlisting}
Many LBM based solutions to practical problems require the ability to store not just the populations of each cell but also additional data such as an external force (see \eg Section~\ref{sec:externalForces}). For this reason every descriptor may explicitly declare such fields.
\begin{lstlisting}[language=myc++]
using DESCRIPTOR = descriptors::D2Q9<descriptors::FORCE>;

// Check whether DESCRIPTOR contains the field FORCE
DESCRIPTOR::provides<descriptors::FORCE>(); // == true
// Get cell-local memory location of the FORCE field
const int offset = DESCRIPTOR::index<descriptors::FORCE>(); // == 9
// Get size of the descriptor's FORCE field
const int size = DESCRIPTOR::size<descriptors::FORCE>(); // == 2
\end{lstlisting}
A list of various predefined fields can be found in the \path{src/dynamics/descriptorField.h} header.
Note that one may add an arbitrary list of such fields to a given descriptor.
Even more so it is fully supported to add fields on a per-app basis.
One only needs to make sure that the type is defined prior to any custom descriptor that depends on it, \ie a user could write inside of a case.
\begin{lstlisting}[language=myc++]
struct MY_CUSTOM_FIELD: public FIELD_BASE<42,0,0> { };
using DESCRIPTOR = D2Q9<FORCE,MY_CUSTOM_FIELD>;
\end{lstlisting}
This custom descriptor can then be used in the same way as any other descriptor.
This might even be preferable for very specific fields that are only used by a small number of apps or specific user-provided features.

\section{(Super,Block)Lattice}\label{sec:superblocklattice}

The \class{SuperLattice} is OpenLB's main class for representing and maintaining a concrete lattice with all associated data, cell local models and non-local operators.
It is constructed from a load balanced geometry, \ie given a \class{SuperGeometry} instance which is in turn obtained from a \class{CuboidGeometry} and its associated \class{LoadBalancer}.

\subsection{Context}

The \class{CuboidGeometry} decomposes a given indicator geometry into individual cuboids.
By design, there is a bijection between the set of cuboids and the set of block lattices resp.\ block geometries.
Given a cuboid geometry, the \class{LoadBalancer} assigns the contained cuboids to the available number of MPI processes.
Based on this structure, each spatial cell location of each lattice location encompassed by the cuboid geometry is assigned a so called \emph{material number} managed by the \class{SuperGeometry} (cf.\ Chapter~\ref{sec:geometry} on meshing).
Finally, a concrete \class{BlockLattice} is instantiated by the \class{SuperLattice} instance of the responsible \class{LoadBalancer}-assigned MPI process.
This class maintains the actual data associated with all its lattice cells.
As such, each process only has access to the data of \emph{its own} blocks extended by a small overlap for inter-block communication.
This communication (\ie synchronization and overlap data) is performed by \class{SuperCommunicator} instances maintained in \emph{stages} of the \class{SuperLattice}.

\subsection{BlockLattice}

The \class{BlockLattice} class is the high-level interface abstracting the platform-specific implementation details of the actual core structures managed by \class{ConcreteBlockLattice} class template specializations.

Specifically, there are specializations for each supported platform: \texttt{CPU\_SISD}, \texttt{CPU\_SIMD}, and \texttt{GPU\_CUDA}. At its core, adding support for a new platform consists of specializing all platform-templated classes.

\subsubsection{Core Data}\label{sec:coredata}

The very core data structure of OpenLB is the per-field (cf.~Section~\ref{ssec:Descriptor}) column or \emph{Structure-of-Arrays} (SoA) storage \class{FieldArrayD}.
This \class{FieldArrayD} maintains \class{(Cyclic)Column} instances for each component of its \texttt{FIELD} template argument in a fashion specific to the concrete platform.
In other words, the same \texttt{FIELD} is stored differently for each block / platform, while sharing the same interface of column-wise access.
This fundamental structure allows for exchanging the entire core data structure implementation without touching the higher level interfaces against which for example cell specific models are implemented (cf.~Section~\ref{sec:dynamics}).

On the next level, all \class{FieldArrayD} instances of a block are maintained within the \class{Data} class.
This class provides an on-demand allocating \texttt{getFieldArray} method for access to the field array specific to a given \texttt{FIELD}.

The details of rendering this most performance-critical structure efficient across all platforms exceeds the scope of this guide, but you can be assured that the developer team invests lots of effort in preserving and increasing OpenLB's performance.

\subsection{Communication}

While the cuboid decomposition maintained by the \class{CuboidGeometry} is disjoint, block geometries and lattices overlap by a configurable overlap layer.
Synchronizing this overlap layer, that is communicating the data owned by a given block to its neighbor's overlap layer, is how any inter-block communication is handled in OpenLB.

The most frequent type of overlap communication is the synchronization of a single layer of populations during the \emph{PostCollide} stage of the core \emph{collideAndStream} iteration.

Communication of given stages can be manually triggered using the following.
\begin{lstlisting}[language=myc++]
sLattice.getCommunicator(PostCollide{}).execute();
\end{lstlisting}
Similarly, communication stages may be configured w.r.t.\ to their extent and included fields.
For example, we can consider the setup of the default \emph{PostCollide} stage.
\begin{lstlisting}[language=myc++]
auto& communicator = sLattice.getCommunicator(PostCollide{});
communicator.template requestField<descriptors::POPULATION>();
communicator.requestOverlap(1);
communicator.exchangeRequests();
\end{lstlisting}

\subsection{Cell}

Users familiar with previous versions of OpenLB may remember the original hierarchy of the \class{SuperLattice} containing \class{BlockLattice} instances containing a D-dimensional array of \class{Cell} instances containing a Q-dimensional array of population values.

This data-holding \class{Cell} was refactored by OpenLB 1.4~\cite{olb14} to only provide a view of the data maintained in field arrays (cf.\ Section~\ref{sec:coredata}).
This provided an essential ingredient for the implementation of both SIMD and GPU support in OpenLB 1.5~\cite{olb15} (cf.\ Talk~\cite{kummerlaender:deRSE:23} for a summary of the refactoring journey).

Today, there are multiple per-platform and use case versions of the cell interface.
Core structures such as Dynamics and non-local operators are templated against the \emph{concept of a cell} instead of a specific implementation thereof. This is essential for being able to instantiate the same abstract model implementations on diverse hardware, \eg both in vectorized CPU collision loops and CUDA GPU kernel functions.

\section{Dynamics}\label{sec:dynamics}

In the context of OpenLB, the term \emph{Dynamics} refers to the implementation of cell-local LB models. \ie equilibrium distribution, collision step and momenta computation.
Each \emph{cell} of a lattice is assigned a \emph{dynamics} describing its local behavior.
This way, it is easy to implement inhomogeneous fluids which use a different type of physics from one cell to another.

The non-local streaming step between cells is modeled independent of the choice of cell-local models and remains invariant.
OpenLB's block-local propagation is implemented using the \emph{Periodic Shift} pattern~\cite{kummerlaender:23} on all platforms.

This concept of per-cell dynamics is tied together in the \class{Dynamics} interface.
Every dynamics implementation describes the various components required to model the local behavior of the assigned cell locations.
Specifically, this includes the set of momenta, the equilibrium distribution as well as the actual collision operator.
Fittingly, most dynamics are declared as a tuple of those three components.
\begin{lstlisting}[language=myc++]
template <typename T, typename DESCRIPTOR>
using TRTdynamics = dynamics::Tuple<
  T, DESCRIPTOR,
  momenta::BulkTuple,
  equilibria::SecondOrder,
  collision::TRT
>;
\end{lstlisting}
In this example declaration of the \class{TRTdynamics} we can tell at a glance that its assigned cells will expose bulk momenta reconstructed directly from the population values and relax towards the second order equilibrium using the TRT collision operator.
The dynamics tuple concept is quite powerful, \eg we may apply a forcing scheme by adding a single line.
\begin{lstlisting}[language=myc++]
template <typename T, typename DESCRIPTOR>
using ForcedTRTdynamics = dynamics::Tuple<
  T, DESCRIPTOR,
  momenta::BulkTuple,
  equilibria::SecondOrder,
  collision::TRT,
  forcing::Guo
>;
\end{lstlisting}
This \class{forcing::Guo} \emph{combination rule} is the fourth and optional component of the dynamics tuple.
Combination rules may arbitrarily manipulate the previous momenta, equilibrium and collision components.
In the case of Guo, forcing this means that the momenta argument is shifted via \class{momenta::Forced} and the TRT collision operator is wrapped to force the post collision populations according to the Guo scheme.

OpenLB offers a comprehensive library of momenta, equilibria, collision operators and combination rules that can be easily combined into many different dynamics tuples. See 
\begin{itemize}
    \item \path{src/dynamics/momenta/aliases.h}
    \item \path{src/dynamics/equilibrium.h}
    \item \path{src/dynamics/collision.h}
    \item \path{src/dynamics/forcing.h}
\end{itemize}
for some examples.

While this framework covers most collision steps, a fallback option is provided via \class{dynamics::CustomCollision}.
This class enables the combination of momenta with arbitrary collision and equilibrium implementations.
One example for this are the \class{CombinedRLBdynamics} that are used as the foundation for various boundary conditions.

The full definition of the interface is available in \path{src/dynamics/interface.h}.
Note that due to recent large refactoring to support execution on GPUs and SIMD CPUs, \class{Dynamics} currently contains
various legacy methods that will be deprecated in future releases. New dynamics implementations should be formulated either as a \class{dynamics::Tuple} or \class{dynamics::CustomCollision} template.

\subsection{Collision Operators}

A partial summary of the currently supported collision operators can be found in the recent publication on OpenLB~\cite{krause:21}.
Furthermore, five commonly used collision schemes are compared in \cite{haussmann:19a} for a typical 3D benchmark case of homogeneous isotropic turbulence.
For derivations, analysis and theoretical comparisons the interested reader is referred to \cite{kruger2017lattice}.

\subsubsection{Implementation in Dynamics}

As was touched upon in Section~\ref{sec:dynamics}, local collision operators are expressed as \class{Dynamics} in the context of OpenLB. Specifically, the common dynamics tuple concept expresses collision operators as reusable elements alongside equilibria and momenta.

\begin{lstlisting}[language=myc++]
struct BGK {
  using parameters = typename meta::list<descriptors::OMEGA>;

  static std::string getName() {
    return "BGK";
  }

  template <typename DESCRIPTOR, typename MOMENTA, typename EQUILIBRIUM>
  struct type {
    using EquilibriumF = typename EQUILIBRIUM::template type<DESCRIPTOR,MOMENTA>;

    template <typename CELL, typename PARAMETERS, typename V=typename CELL::value_t>
    CellStatistic<V> apply(CELL& cell, PARAMETERS& parameters) any_platform {
      V fEq[DESCRIPTOR::q] { };
      const auto statistic = EquilibriumF().compute(cell, parameters, fEq);
      const V omega = parameters.template get<descriptors::OMEGA>();
      for (int iPop=0; iPop < DESCRIPTOR::q; ++iPop) {
        cell[iPop] *= V{1} - omega;
        cell[iPop] += omega * fEq[iPop];
      }
      return statistic;
    };
  };
};
\end{lstlisting}

This is the complete listing of the well known BGK collision operator that is used by many different dynamics. Each collision operator consists of three elements: A \texttt{parameters} type list of fields that is used to parametrize the collision, a \texttt{getName} method that is used to generate human readable names for dynamics tuples and a nested \texttt{type} template that contains the actual \texttt{apply} method specific to each operator.

The nested \texttt{type} template is used to enable composition into dynamics tuples and will be automatically instantiated for the required descriptor, momenta and equilibrium types. Additionally, this is used as a place for injecting partial specialization which enable usage of autogenerated CSE-optimized kernels.

As is the case for all other elements, the \texttt{apply} template method follows a fixed signature for all collision operators. Each call is provided an instance of some platform-specific implementation of the cell concept alongside a parameters structure containing all requested values. Using these two inputs the method can perform the local collision and return the computed density and velocity magnitudes for usage in lattice statistics.
This pattern is repeated at various places of the library. Examples for other instances are equilibria and momenta elements as well as post processors. Any implementations of this style are usable on any of OpenLB's target platforms (currently this means the scalar and vectorized CPU code as well as GPU support). They are also amenable to automatic code generation.

For an introduction on how to write your own dynamics, see the FAQ in Section~\ref{faq:dynamics}.

\section{Post Processors}\label{sec:postprocessor}

While the basic concept of dynamics assigned to cells in a block lattice is conceptually close to the theory of LBM, it is not sufficiently general to address all possible requirements arising in complex applications.
As a case in point, some boundary conditions are non-local and need to access neighboring nodes. 
Their execution is taken care of by a \emph{post processing} step, which instead of traversing the entire lattice a second time, applies to selected cell locations only.

While collision steps are easily parallelized due to their local nature, this doesn't hold once neighborhood access is required.
For this reason, two adjacent concepts are used to group post processors. Each post processor is assigned to a \emph{stage} and a \emph{priority} within this stage.
For example both \class{OuterVelocityCornerProcessor3D} and \class{OuterVelocityEdgeProcessor3D} are processed in the \texttt{PostStream} stage, but the former is executed after the latter to avoid access conflicts.
In turn, post processors within the same stage and priority may be executed in parallel depending on the execution platform.
Note that both the number of priorities and stages may be freely customized -- \eg the free surface code introduces a number of additional stages to interleave post processing and custom communications steps.

Each post processor consists of scope and priority declarations in addition to an \texttt{apply} template method.
For the aforementioned \class{OuterVelocityCornerProcessor3D} this is declared as follows (see \path{src/boundary/boundaryPostProcessors3D.hh} for the full implementation).
\begin{lstlisting}[language=myc++]
template<typename T, typename DESCRIPTOR,
         int xNormal, int yNormal, int zNormal>
struct OuterVelocityCornerProcessor3D {
  static constexpr OperatorScope scope = OperatorScope::PerCell;

  int getPriority() const {
    return 1;
  }

  template <typename CELL>
  void apply(CELL& cell) any_platform;
};
\end{lstlisting}
Here, the \class{OperatorScope::PerCell} declares that the apply function will be provided a cell implementation with neighborhood access for each assigned location. Other scopes such as \class{OperatorScope::PerBlock} (used for example for statistics computation) enable different access patterns.
In any case, the assigned cell locations are maintained by OpenLB's post processor framework.
For example
\begin{lstlisting}[language=myc++]
sLattice.addPostProcessor<PostStream>(indicator, meta::id<SomePerCellPostProcessor>{});
\end{lstlisting}
schedules a post processor for application to all indicated cells during the \texttt{PostStream} stage.

Note that this describes the new post processor concept adopted by OpenLB 1.5.
Many existing \emph{legacy} post processors use a different paradigm.
Specifically, they are derived from \class{PostProcessor(2,3)D} and override virtual methods such as 

\class{void PostProcessor(2,3)D::process(BlockLattice<T,DESCRIPTOR>\&)}.

All these post processors will be ported to the new approach in time.
For CPU targets, legacy post processors can be used without restrictions but they are not supported on the GPU platform.
For an introduction on how to write your own post processors, see the FAQ in Section~\ref{faq:postprocessor}.

\section{Functors}\label{sec:functors}
In general, a \emph{functor} is a class that behaves like a function.
Objects of a functor class perform computations by overloading the \class{operator()} in a standardized fashion, allowing for arithmetic combination.
One big advantage of functors over functions is that they allow the creation of a hierarchy and bundle ''classes of functions''.

\subsection{Basic Functor Types}
The functor concept is a user friendly and efficient technique to process lattice data and extract relevant data for post-processing.
In the meantime, OpenLB deploys the functors also for the geometry, which is a very intuitive and powerful choice.

Basically, functors are applications that operate either on the lattice $\mathbb{N}^3$ or more generally on $\mathbb{R}^3$.
The values of such a functor may be three dimensional, as for example for the velocity, where the mapping is defined as
\begin{equation}
\text{\class{Functor}}: \Omega \to \mathbb{R}^d, \quad \text{for } d \in\mathbb{N} ~.
\end{equation}
The nomenclature is based on the dimension of the domain.
Let's say the functor acts on a $3d$ (super) lattice, the the functor is called \class{SuperLatticeF3D}.
If the functor value is density, then this functor is called \class{SuperLatticeDensity3D}.

\subsubsection{GenericF}
The GenericF functor stands at the top of the hierarchy and is a virtual base class that provides interfaces.
Template parameter \texttt{S} defines the input data type and template parameter \texttt{T}, the output.
The essential interface is the unwritten (pure virtual function) \texttt{operator()}.
Commonly, this $()-$operator is used as an evaluation of a certain functor, \eg pressure at position $\bm{x}$.

\subsubsection{AnalyticalF}
This a subclass of \class{GenericF}, for functions that live in SI-units, \eg for setting velocities in $m/s$.
Parts of this class are, for example, constant, linear, interpolation and random functors, which can be evaluated by the $()$-operator.
There is an \class{AnalyticalCalc} class, which inherits from \class{AnalyticalF} and establishes arithmetic operations $(+,-,\cdot,/)$ between every type of 
\begin{equation}
\text{\class{AnalyticalF3D}}: \mathbb{R}^3 \to \mathbb{R}^d, \quad d\in\mathbb{N} ~.
\end{equation}

\subsubsection{IndicatorF}
This is another subclass of \class{GenericF} that returns a vector with elements \(0\) or \(1\), \ie
\begin{equation}
\text{\class{IndicatorF3D}}: \mathbb{R}^3 \to \{0,1\} ~.
\end{equation}
These are used to construct geometries, \eg \class{IndicatorSphere3D} creates a sphere using an origin and radius.
Evaluation returns \(1\), if the vector is inside the sphere and \(0\) elsewise.
In analogy to the \class{AnalyticalF}, there are arithmetic operations as well, but with a slightly different definition.
The returned object of an addition is the union, multiplication returns the intersection, and subtraction represents the relative complement.

\subsubsection{SmoothIndicatorF}
Smooth indicators define a \emph{small} epsilon region around the object such that it has a smooth transition form $0$ to $1$. 
In general, the mapping is defined as 
\begin{equation}
\text{\class{SmoothIndicatorF3D}}: \mathbb{R}^3 \to [0,1] ~.
\end{equation}

\subsubsection{SuperLatticeF}
These functors are defined on the lattice via 
\begin{equation}
\text{\class{SuperLatticeF3D}}: \mathbb{N}^3 \to \mathbb{R}^d, \quad d\in\mathbb{N} ~,
\end{equation}
and commonly represent the raw simulation data, \eg macroscopic moments such as pressure and velocity.
\class{SuperLattice} functors are part of the parallelism layer and they delegate the calculations to the corresponding BlockLattice functors. 

\subsubsection{InterpolationF} Interpolation functors establish conversion between the analytical and lattice functors.
They are very important in setting analytical boundary conditions, by evaluating the given analytical function on the lattice points.
The reverse direction -- from lattice to analytical functors -- is where this functor receives its name, as the conversion is achieved by interpolation between the lattice points.

\subsection{Functor Arithmetic}

Functor arithmetic expressions are based on functor instances wrapped in \texttt{std::shared\_ptr<F>} smart pointers.
This layer is necessary to enable automatic memory management for trees of interdependent functors.

\begin{lstlisting}[language=myc++,caption={Basic showcase for \texttt{std::shared\_ptr} based functor arithmetic},label=lst:sharedPtrFunctorArithmetic]
std::shared_ptr<SuperF3D<T>> aF(
  new SuperConst3D<T>(superStructure, {1.0, 2.0}));
std::shared_ptr<SuperF3D<T>> bF(
  new SuperConst3D<T>(superStructure, {2.0, 1.0}));
std::shared_ptr<SuperF3D<T>> cF = aF + bF;
// cF->operator() returns {3.0, 3.0}
\end{lstlisting}

Note that \texttt{cF} can be passed out of scope without any regard for \texttt{aF} and \texttt{bF}, as managed pointers are stored internally.
At first glance this new functor arithmetic may seem unnecessarily verbose for basic usage such as simply adding two functors and directly using the result.
As such legacy functor arithmetic is still available for basic use cases.
Usage of \texttt{std::shared\_ptr} functor arithmetic is supported by both \class{FunctorPtr} and a DSL to ease development of more complex functor compositions.

The \class{FunctorPtr} helper template is used throughout the functor codebase to transparently accept functors independently of how their memory is managed. This means that functors managed by \texttt{std::shared\_ptr} are accepted as arguments in any place where raw functor references were used previously. As a nice benefit \texttt{FunctorPtr} transparently forwards any calls of its own operator function to the operator of the underlying functor.

\begin{lstlisting}[language=myc++,caption={\texttt{FunctorPtr} and \texttt{std::shared\_ptr} based functor arithmetic},label=lst:sharedPtrFunctorArithmeticFunctorPtr]
T error(FunctorPtr<SuperF3D<T>>&& f, T reference) {
  T         output[1] = { };
  const int origin[4] = {0,0,0,0};
  f(output, origin);
  return fabs(output[0] - reference);
}

std::shared_ptr<SuperF3D<T>> managedF(
  new SuperConst3D<T>(superStructure, 1.0));
SuperConst3D<T> rawF(superStructure, 1.0);

// error(managedF, 1.1) == error(rawF, 1.1) == 0.1
\end{lstlisting}

Functor arithmetic expressions may also contain constants in addition to functors.
Any scalar constant used in the context of managed functor arithmetic is implicitly wrapped into a \class{SuperConst(2,3)D} instance.

\begin{lstlisting}[language=myc++,caption={Constant scalars in managed functor arithmetic},label=lst:sharedPtrFunctorArithmeticConstants]
std::shared_ptr<SuperF3D<T>> aF(/*...*/);
auto bF = 0.5 * aF + 2.0; // scalar multiplication and addition
\end{lstlisting}

\begin{sloppypar}
Constant vectors are also supported if they are explicitly passed to the \class{SuperConst(2,3)D} constructor.
Note that arithmetic operations of equidimensional functors are performed componentwise (\ie \texttt{aF * aF} is not the scalar product).
\end{sloppypar}

\begin{lstlisting}[language=myc++,caption={Constant vectors in managed functor arithmetic},label=lst:sharedPtrFunctorArithmeticConstantVector]
std::shared_ptr<SuperF3D<T>> vectorF(
  new SuperConst3D<T>(superStructure, {1.0, 2.0, 3.0}));
auto cF = aF / vectorF; // componentwise division
\end{lstlisting}

\subsubsection{Functor Composition}
\label{sec:functorComposition}

Composing multiple managed functors which in turn need multiple arguments by themselves such as when calculating error norms in a reusable fashion can quickly lead to expressions that are fairly hard to read.
For this reason, the \texttt{functor\_dsl} namespace offers a set of conveniently named helper functions in order to deobfuscate such functor compositions.

Consider for example the following snippet which constructs and evaluates a relative error functor based on the L2 norm.
\begin{lstlisting}[language=myc++,caption={\texttt{functor\_dsl} supported functor composition},label=lst:functorCompositionDSL]
using namespace functor_dsl;
// decltype(wantedF) == std::shared_ptr<AnalyticalF3D<double,double>>
// decltype(f)       == std::shared_ptr<SuperF3D<double>>
// decltype(indicatorF) == std::shared_ptr<SuperIndicatorF3D<double>>
auto wantedLatticeF = restrictF(wantedF, sLattice);
auto relErrorNormF  = norm<2>(wantedLatticeF - f, indicatorF))
                    / norm<2>(wantedLatticeF, indicatorF);
const int input[4];
double result[1];
relErrorNormF->operator()(result, input);
std::cout << "Relative error: " << result[0] << std::endl;
\end{lstlisting}
Note that lines~5--7 contain the full implementation of the expression 
\begin{align}
\frac{\|\texttt{wantedF} - \texttt{f}\|_2}{\|\texttt{wantedF}\|_2} ~,
\end{align}
\ie the L2-normed relative error of an arbitrary functor \texttt{f} as compared to the analytical solution \texttt{wantedF}.
This simply allows for moving even basically one-off functor compositions into reusable and easily verifiable functors whose implementation is as close to the actual mathematical definition as is reasonably possible. Correspondingly a more developed version of Listing~\ref{lst:functorCompositionDSL} can be found in \class{SuperRelativeErrorLpNorm3D} which is used extensively by the \path{poiseuille3d} example to compare simulated and analytical solutions.

\begin{lstlisting}[language=myc++,caption={Functor composition in \texttt{SuperRelativeErrorLpNorm3D}'s constructor},label=lst:relErrorLpNormConstructor]
template <typename T, typename W, int P>
template <template <typename U> class DESCRIPTOR>
SuperRelativeErrorLpNorm3D<T,W,P>::SuperRelativeErrorLpNorm3D(
  SuperLattice<T,DESCRIPTOR>&      sLattice,
  FunctorPtr<SuperF3D<T,W>>&&        f,
  FunctorPtr<AnalyticalF3D<T,W>>&&   wantedF,
  FunctorPtr<SuperIndicatorF3D<T>>&& indicatorF)
  : SuperIdentity3D<T,W>([&]()
{
  using namespace functor_dsl;

  auto wantedLatticeF = restrictF(wantedF.toShared(), sLattice);

  return norm<P>(wantedLatticeF-f.toShared(), indicatorF.toShared())
         / norm<P>(wantedLatticeF, indicatorF.toShared());
}())
{
  this->getName() = "relErrorNormL" + std::to_string(P);
}
\end{lstlisting}

Disregarding the addition of \texttt{FunctorPtr} as well as further templatization, lines 12--15 are equivalent to the ad-hoc error norm in Listing~\ref{lst:functorCompositionDSL}. Also note how the actual composition happens inside of a lambda expression and is then returned to be stored by \class{SuperIdentity3D}. This allows for assigning composed functors their own name and renders them indistinguishable from \emph{primitive} functors.

\section{Parallelization}\label{sec:mpi-intro}

As applications in computational fluid dynamics require a large amount of resources, it is essential to have the flexibility to switch to a parallel platform easily.
This section concentrates on parallelism on distributed memory machines using MPI, as distributed memory is the most common model on large-scale, parallel machines.
All example cases in the OpenLB distribution can be compiled with MPI and executed in parallel.
Data which is spatially distributed, such as lattice fields, is handled through a data-parallel paradigm.
The data space is partitioned into smaller regions that are distributed over the nodes of a parallel machine.
In the following, these types of structures are referred to as data-parallel structures.
Other data types that require a small amount of storage space are duplicated on every node. These are referred to as duplicated data.
All native C++ data types are automatically duplicated by virtue of the Single-Program-Multiple-Data model of MPI.
These types should be used to handle scalar values, such as the parameters of the simulation, or integral values over the solution (\eg the average energy).

For output on the console it is recommended to use \class{OstreamManager} since it transparently manages output in case of parallel execution (cf Chapter~\ref{sec:consoleOutput}).

\subsection{Supported Platforms}

OpenLB models heterogeneous parallelization using its predominant block decomposition architecture. Each individual block of a super lattice is assigned one of currently three possible target plaforms: \texttt{Platform::CPU\_SISD}, \texttt{Platform::CPU\_SIMD} or \texttt{Platform::GPU\_CUDA}.
The availability of each platform is controlled by adding its name to the \texttt{PLATFORMS} variable in the \path{config.mk} build configuration.
Note that further system specific changes to the compiler settings are almost certainly required for \texttt{CPU\_SIMD} and \texttt{GPU\_CUDA}.
See the \path{config/} folder for some examples.

Note that \texttt{CPU\_SISD} must always be enabled as some host-side data structures rely on this platform.
In the absence of heterogeneous load balancing, the \emph{most efficient} enabled platform is used by default for all blocks. Specifically, \texttt{CPU\_SIMD} if only the two CPU platforms are enabled or \texttt{GPU\_CUDA} (even if \texttt{CPU\_SIMD} is also enabled at the same time).

\subsubsection{SIMD}

Modern CPUs offer vector instructions with a width of up to 512 bytes in the case of AVX-512. This means that 8 respectively 16 individual scalar values can be processed in a single instruction. In some situations this can significantly speed up the bulk collision step, which is why OpenLB supports this option for processing its \class{Dynamics}.
This option is at its most powerful when combined with the \texttt{HYBRID} parallelization mode s.t. OpenMP is used to further parallelize the vectorized collision on each shared memory node. However, setting up the hybrid mode correctly on a HPC system is non trivial, which is why we suggest to stick to the MPI-only mode for users unexperienced in working with HPC systems.

Once enabled in the build configuration, vectorization is applied to the dominant collision of each block lattice transparently without requiring any additional code changes.

\subsubsection{GPU}

General purpose graphics processing units (GPGPUs) are an ideal platform for many LBM-based simulations due to their high memory bandwidth and high degree of parallelization. OpenLB currently supports transparent usage of Nvidia GPUs via CUDA for almost all dynamics and a core set of boundary post processors and coupling operators.
Similarly to both other available platforms, enabling GPU support requires only adding \texttt{GPU\_CUDA} to the \texttt{PLATFORMS} variable in \path{config.mk} and some system-specific updates to the compiler settings. MPI parallelization is fully supported for GPU blocks, enabling simulations on multi-GPU clusters.

Note that the resolution commonly needs to be increased significantly compared to the CPU-accommodating default values in order to take full advantage of GPU acceleration.
Especially on desktop-grade GPUs changing the fundamental floating point type \class{T} to \class{float} is advisable.
Also see Section~\ref{sec:nvidiaGpuOlb} for guidance on how to install, compile and run GPU-aware OpenLB cases.


\chapter{Geometry Creation and Meshing} 

\label{sec:geometry}
This chapter presents how geometry data can be loaded in or created by OpenLB.
Furthermore, the concept of material numbers is shown.

\section{Creating a Geometry}\label{sec:geometry_creation}
OpenLB provides an interface for STL based geometry data and generates fully automated a structured regular mesh.
On the other hand, geometries can be build out of primitive shapes such as cuboids, spheres and cylinders.
By the implemented arithmetic that includes intersection, union and complement, those primitives can be assembled very generally.
A computational domain such as the \class{SuperLattice} is created in $6$ simple steps (see also Fig~\ref{fig:geometryCreation}):
\begin{description}
\item[Step 1:] Create an \class{Indicator3D} instance by
  \begin{enumerate}
    \item Reading an STL file, see example aorta3d~\ref{sec:aorta3d} and section~\ref{sec:readstl}.
    \item Pre-defined primitive shapes (cuboid, circle, cylinder, sphere) and their combinations $(+,-,\cdot,/)$, as described in the example \path{venturi3d}, see Section~\ref{sec:venturi3d}.
  \end{enumerate}
\item[Step 2:] Construct \class{CuboidGeometry3D}. During construction, the geometry from step 1 is divided into the predefined number of cuboids that are thereafter automatically removed, shrunk and weighted for a good load balance. By weighting user can choose between \textit{weight} and \textit{volume} strategies. For complex shapes the last option is preferable. A larger cuboids number removes more unnecessary nodes, but implies higher communication costs. 
\item[Step 3:] Construct \class{LoadBalancer} that assigns cuboids to threads. The standard option is the HeuristicLoadBalancer, whereby there are also other variants. 
\item[Step 4:] Construct \class{SuperGeometry3D} that links material numbers to voxels (see Section~\ref{sec:materialnumbers_setting}).
\item[Step 5:] Set \emph{material numbers} to different simulation space regions, where afterwards dynamics and boundaries are defined.
\item[Step 6:] Construct \class{SuperLattice} to perform stream and collide algorithm.
\end{description}

\begin{lstlisting}[language=myc++, caption={Create geometry based on STL or primitive shapes. All six steps are presented briefly as source code.}]
// Step 1: Create Indicator
STLreader<T> stlreader("filename.stl", voxelSize, stlSize, method);
IndicatorCuboid3D<T> indicator( extend, origin );
// Step 2: Construct cuboidGeometry.
CuboidGeometry3D<T> cuboidGeometry(stlReader / indicator, voxelSize, noOfCuboids, weightingStrategy);
// Step 3: Construct LoadBalancer.
HeuristicLoadBalancer<T> loadBalancer(cuboidGeometry);
// Step 4: Construct SuperGeometry.
SuperGeometry<T,3> superGeometry(cuboidGeometry, loadBalancer);
// Step 5: Set material numbers.
// set material number 2 for whole geometry
superGeometry.rename(0,2);
// change material number from 2 to 1 for inner (fluid) cells, so that only boundary cells have material nunmer 2
superGeometry.rename(2,1,{1,1,1});
// or simply use an indicator that changes its lattices to one
superGeometry.rename(2,1,fluidIndicator);
// additional material numbers for other boundary conditions, the 3rd argument in the brackets is the material number which the boundary cells should face
superGeometry.rename(2,3,1,cylinderInFLow);
superGeometry.rename(2,4,1,outflowIndicator0);
superGeometry.rename(2,5,1,outflowIndicator1);
// Step 6: Construct SuperLattice.
SuperLattice<T,DESCRIPTOR> sLattice(superGeometry);
\end{lstlisting}

The powerful application of the geometry generation of OpenLB can be demonstrated on the example \path{aorta3d}.
This example is based on a very complex geometry and illustrates the highly user friendly and automated process from STL to computation grid the \class{SuperLattice}, see Figure~\ref{fig:geometryCreation}.
\begin{figure}[ht]
  \centering
  \includegraphics[width=.8\textwidth]{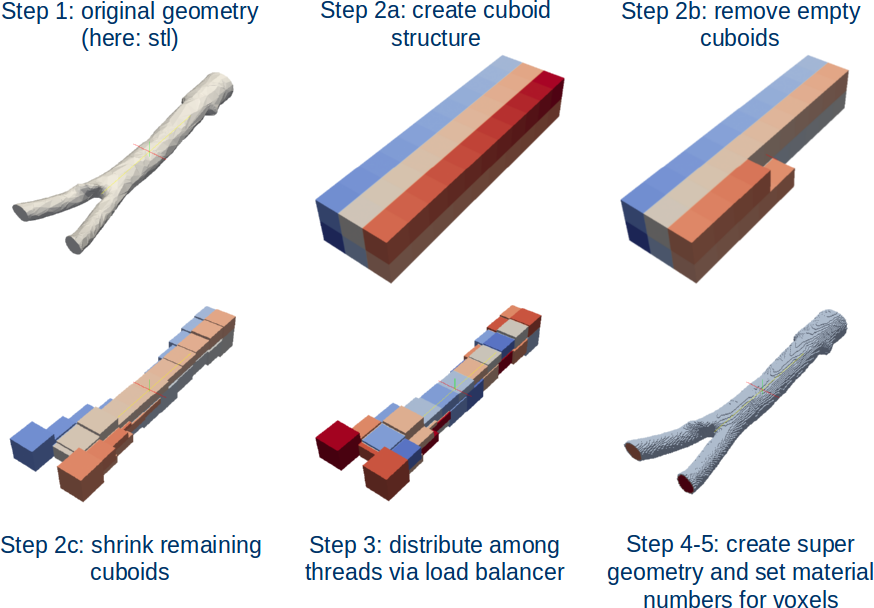}
  \caption{
  Six steps to create a Geometry.
  It starts by reading an STL file with the help of an \class{STLreader} and ends with the creation of a \class{SuperLattice}.}
  \label{fig:geometryCreation}
\end{figure}

\section{Setting Material Numbers} \label{sec:materialnumbers_setting}
OpenLB has a general concept for representation of a geometry.
A specific number called the \emph{material number} is assigned to each cell, defining whether that cell belongs to the boundary or to the fluid domain or whether it is superfluous in the computations.
Figure~\ref{fig:materialNumbers} illustrates this using the example of a channel flow with an obstacle.
The different collision and streaming steps on the boundary and the fluid are defined with respect to the material number.
\begin{figure}[ht]
  \center
  \includegraphics[width=0.8\textwidth]{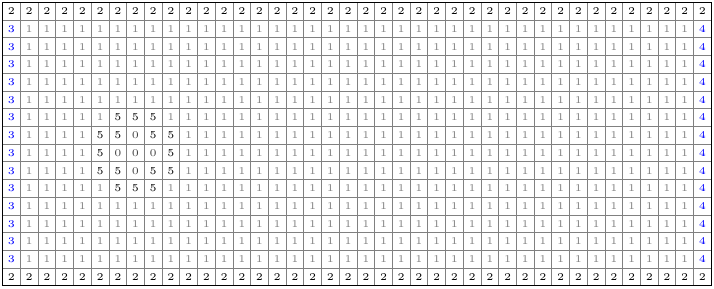}
  \caption{
  Lattice nodes of the geometry are associated with material numbers.
  Material number zero, one, two, three, four and five correspond to external region, fluid, bounce back boundary, inflow, outflow and obstacle cells, (1=fluid, 2=no-slip boundary, 3=velocity boundary, 4=constant pressure boundary, 5=curved boundary, 0=do nothing).}
  \label{fig:materialNumbers}
\end{figure}
The benefit of using material numbers in CFD simulations is the automatic determination of streaming directions on boundary nodes, as this is not always practical by hand e.\,g.\ if material numbers of a complex geometry are obtained from a STL file.

Besides creating the domain, \class{IndicatorFXD} functions can be used to set \emph{material numbers} with the help of one of the \texttt{rename} functions in \class{SuperGeometryXD}.
\begin{lstlisting}[language=myc++, caption=Different rename functions to set material numbers.]
/// replace one material with another
void rename(int fromM, int toM);
/// replace one material that fulfills an indicator functor condition with another
void rename(int fromM, int toM, IndicatorF3D<bool,T>& condition);
/// replace one material with another respecting an offset (overlap)
void rename(int fromM, int toM, { unsigned offsetX, unsigned offsetY, unsigned offsetZ });
/// renames all voxels of material fromM to toM if the number of voxels given by testDirection is of material testM
void rename(int fromM, int toM, int testM, std::vector<int> testDirection);
/// renames all voxels of material fromM to toM if two neighbour voxels in the direction of the discrete normal are voxels with material testM in the region where the indicator function is fulfilled
void rename(int fromM, int toM, int testM, IndicatorF3D<bool,T>& condition);
\end{lstlisting}

\section{Building Simulation Domains with Geometry Primitive Functors}
\label{subsec:indic}
For the purpose of setting up simulation domains (here called geometries), OpenLB provides several functors (see Section~\ref{sec:functors}) for the creation of basic geometric primitives such as cuboids, circles, spheres, cones etc. 
These can be combined using the mathematical operators ($+$ union, $-$ set difference, $\cdot$ intersection) to create more complex domains. 
It can be done in the application setup (see below) or in the XML interface (see Section~\ref{sec:venturi3d}).
\begin{lstlisting}[language=myc++, caption=Geometry operations.]
Vector<T,3> C0(0,50,50);
Vector<T,3> C1(5,50,50);
Vector<T,3> C2(40,50,50);
///...
Vector<T,3> C13(115,7,50);

T radius1 = 10 ;
T radius2 = 20 ;
T radius3 = 4 ;

IndicatorCylinder3D<T> inflow(C0, C1, radius2);
std::shared_ptr<IndicatorCylinder3D<T>> cyl1 (
 new IndicatorCylinder3D<T> ( (C1, C2, radius2));
std::shared_ptr<IndicatorCone3D<T>> co1 (
 new IndicatorCone3D<T> ( (C2, C3, radius2, radius1));
///...
IndicatorCylinder3D<T> outflow1(C12, C13, radius1);
/// IndicatorIdentity3D collects indicator functors in one object
IndicatorIdentity3D<T> venturi(cyl1 + co1 + others );
\end{lstlisting}

\section{Reading STL-files}
\label{sec:readstl}
For a correct STL representation, the STL reader should be properly set up.
\begin{lstlisting}[language=myc++, caption={STL reader}]
STLreader<T> stlreader("filename.stl", voxelSize, stlSize, method, verbose);
\end{lstlisting}
The STL file can be stored not only in the current application folder, but also somewhere else. The path there can be written in the first argument of the reader. The scaling factor \textit{stlSize} should be set to the units of the STL part. If it is exported in meters, the scaling factor is $1$, if in millimeters then it is $0.001$. The reading methods can be chosen depending on the geometry complexity. For easy geometries the option $0$ can be chosen, for the complex and possibly untight shapes the option $1$ is to be set, whereby it is slower. The \textit{verbose} argument can be true or false. It prints the information about the STL file in the terminal.

\section{Excursus: Creating STL-files}
The general process chain assumes that the geometry is already given in form of an STL file, if not created by the \class{IndicatorFXD}-functions.
Simple geometries can be created using a CAD tool like FreeCAD (\url{www.freecad.org}).
An introduction to modeling with FreeCAD can be found for example in
\url{http://www.youtube.com/watch?v=geIrH1cOCzc}.
The general procedure is mostly similar to the following description.

Firstly, a 2D sketch is created on a selected plane (\eg the \(xy\)-plane) using different bi-dimensional shapes.
In the next step the sketch is extruded in the third dimension. Several such 3d objects can be combined using operations like union, cut, intersection, rotation, trace, etc.\ to obtain the target geometry.
Creating a square and a circle for the example \path{cylinder3d} in Figure~\ref{fig:freecad-cylinder3d} is an easy task. Also, complex geometries as that of a filter or a porous media can be set up easily with OpenLB's indicator approach.
\begin{figure}[ht]
  \centering
  \includegraphics[width=10cm]{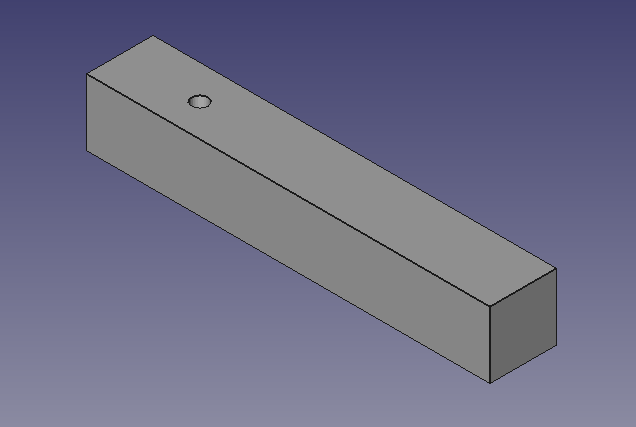}
  \caption{
    The geometry of the example \protect\path{cylinder3d} from Section~\ref{sec:cylinder2dAndCylinder3d} opened in FreeCAD.
    }
  \label{fig:freecad-cylinder3d}
\end{figure}


\chapter{Simulation Models}

\section{Non-dimensionalization and Choice of Simulation Parameters}

Basically, to describe a physical quantity you need a reference scale. By dividing a quantity by a variable of the same dimension, a dimensionless quantity is derived. The result is a number, which is called "the lattice value of the quantity" or the value of the amount "in lattice units" (denoted with \(\cdot_{\mathrm{LB}}\)), if you consider the Lattice Boltzmann Method. The reference scale is named conversion factor between two reference systems. Furthermore, the conversion factor for the quantity $\phi$ is called $C_\phi$ and the non-dimensional quantity gets the symbol $\phi^*$.

To get into this topic, the book \textit{The Lattice Boltzmann Method}~\cite{kruger2017lattice}, especially Chapter 7.1 and 7.2 therein, and \url{https://en.wikipedia.org/wiki/Dimensional_analysis} are strongly recommended.

\section{Porous Media Model}

The permeability parameter $K$ is a physical parameter
that describes the macroscopic drag in a porous media model.
For laminar flows it is defined by Darcy's law
\begin{equation}
  K = - \frac{Q \mu L}{\triangle P}~,
\end{equation}
where $Q = UA$ is the flow rate, $U$ is a characteristic velocity, $A$ is a cross-sectional area, $\mu$ denotes the dynamic viscosity, $L$ is a characteristic length, and $\triangle P$ specifies the pressure difference in between starting point and endpoint of the volume.

The porosity-value $d \in [0,1]$ is a lattice-dependent value, where $d = 0$ means the medium is solid and $d = 1$ denotes a liquid. 
According to Brinkman~\citet{brinkman-permeability:49, brinkman-viscousforce:1949},
\citet{borrvall:2003} and~\citet{pingen:2007}, the Navier-Stokes equation is transformed (see~\citet{dornieden:2013} and~\citet{stasius:2014}).
The discrete formulation of $d$ describes a flow region by its permeability
\begin{equation}
  d = 1 - h^{dim-1} \frac{\nu_{\mathrm{LB}} \tau_{\mathrm{LB}}}{K}~, 
\end{equation}
where $\tau_{\mathrm{LB}}$ is the relaxation time, $\nu_{\mathrm{LB}}$ is the discrete kinematic viscosity and $h$ is the length.
Therefore $K \in [ \nu_{\mathrm{LB}} \tau_{\mathrm{LB}} h^{dim-1}, \infty]$.
To describe the porosity or permeability of a medium, a descriptor for porosity must be used, such as the one below.
\begin{lstlisting}[language=myc++]
using DESCRIPTOR = descriptors::PorousD3Q19Descriptor;
\end{lstlisting}
Be aware that the porous media model only works in the generic compilation mode.
In the function \class{prepareLattice}, dynamics for the corresponding number of the porous material are defined
for example as follows.
\begin{lstlisting}[language=myc++]
void prepareLattice(..., Dynamics<T, DESCRIPTOR>& porousDynamics, ...){
  /// Material=3 --> porous material
  sLattice.defineDynamics(superGeometry, 3, &porousDynamics);
  ...
}
\end{lstlisting}
In the function \class{setBoundaryValues}, the initial porosity value and external field are defined.
\begin{lstlisting}[language=myc++]
void setBoundaryValues(..., T physPermeability, int dim, ...){
  // d in [0,1] is a lattice-dependent porosity-value
  // depending on physical permeability K = physPermeability
  T d = converter->latticePorosity(physPermeability);
  AnalyticalConst3D<T,T> porosity(d);
  sLattice.defineField<descriptors::POROSITY>(superGeometry, 3, porosity);
  ...
}
\end{lstlisting}
In the \class{main} function, the required parameters as well as the porous media dynamics are defined.
\begin{lstlisting}[language=myc++]
int main(int argc, char* argv[]) {
  ...
  T physPermeability = 0.0003;
  ...
  PorousBGKdynamics<T, DESCRIPTOR> porousDynamics(converter->getOmega(),
    instances::getBulkMomenta<T, DESCRIPTOR>());
  ...
}
\end{lstlisting}
Additionally, the class \class{UnitConverter} in \path{src/core/units.h} provides useful functions for conversion between physical and lattice values:
\begin{lstlisting}[language=myc++]
/// converts a physical permeability K to a lattice-dependent porosity d
/// (a velocity scaling factor depending on Maxwellian distribution
/// function), needs PorousBGKdynamics
T latticePorosity(T K) const
{ return 1 - pow(physLength(),getDim()-1)*getLatticeNu()*getTau()/K; }

/// converts a lattice-dependent porosity d (a velocity scaling factor
/// depending on Maxwellian distribution function) to a physical
/// permeability K, needs PorousBGKdynamics
T physPermeability(T d) const
{ return pow(physLength(),getDim()-1)*getLatticeNu()*getTau()/(1-d); }
\end{lstlisting}

\section{Power Law Model}

The two most common deviations from Newton's Law observed in real systems are pseudo-plastic fluids and dilatant fluids. 
For pseudo-plastic fluids the viscosity of the system decreases as the shear rate is increased. 
On the other hand, as the shear rate by dilatant fluids is increased, the viscosity of the system also increases. 
The simplest model, that describes these two types of deviations, is called the (Ostwald--de Waele) Power Law model and defines the viscosity as
\begin{equation}
\label{eq:powerLaw}
\mu=m\dot{\gamma}^{n-1}~,
\end{equation}
where $m$ is the flow consistency index, $\dot{\gamma}$ is the shear rate, and $n$ denotes the flow behavior index. 
Then
\begin{itemize}
\item $n < 1$: pseudoplastic fluids,
\item $n = 1$: Newtonian fluids,
\item $n > 1$: dilatant fluids.
\end{itemize}

To simulate a power law fluid, a descriptor for dynamic \class{omega} must be used, such as:
\begin{lstlisting}[language=myc++]
using DESCRIPTOR = descriptors::DynOmegaD2Q9Descriptor;
\end{lstlisting}
In the function \class{setBoundaryValues}, the initial same \class{omega}-argument is defined.
\begin{lstlisting}[language=myc++]
AnalyticalConst2D<T,T> omega0(converter.getOmega());
sLattice.defineField<descriptors::OMEGA>(
  superGeometry.getMaterialIndicator({1,2,3,4}), omega0);
\end{lstlisting}
In the \class{main} function, the power law dynamics are defined.
\begin{lstlisting}[language=myc++]
int main(int argc, char* argv[]) {
  ...
  PowerLawBGKdynamics<T, DESCRIPTOR> bulkDynamics(converter.getOmega(), instances::getBulkMomenta<T, DESCRIPTOR>(), m, n, converter.physTime());
}
\end{lstlisting}
To recover a nonconstant kinematic viscosity in the targeted PDE, also the \class{omega}-argument has to be no longer a constant. 
With using the power law model~\eqref{eq:powerLaw}, the kinematic viscosity is computed in each step as
\begin{equation}
\label{eq:powerLawKinematicViscosity}
\nu=\dfrac{1}{\rho}m\dot{\gamma}^{n-1}~.
\end{equation}
The shear rate $\dot{\gamma}$ is computable via using the second invariant of the strain rate tensor $D_{II}$, \ie
\begin{align}
\label{yan}
\dot{\gamma}=\sqrt{2D_{II}}~,
\end{align}
where
\begin{align}
\label{D}
D_{II}=\sum_{\alpha,\beta=1}^{d}E_{\alpha \beta}E_{\alpha \beta}~,
\end{align}
and 
\begin{align}
\label{strainRateTensor1}
E_{\alpha \beta} = - \Bigg(1-\dfrac{1}{\tau}\Bigg)\dfrac{1}{2\varrho \nu}\sum_{i=0}^{q-1}f_{i}^{h}\boldsymbol{\xi}_{i\alpha}\boldsymbol{\xi}_{i\beta}~.
\end{align}
This concept is very significant, since $f_{i}^{h}\boldsymbol{\xi}_{i\alpha}\boldsymbol{\xi}_{i\beta}$ is usually computed during the collision process and therefore, in comparison to other CFD methods, adds almost no additional computational cost. 
The computation of a new \class{omega}-argument is done in \path{src/dynamics/powerLawBGKdynamics.h} via: 
\begin{lstlisting}[language=myc++]
T computeOmega(T omega0_, T preFactor_, T rho_, T pi_[util::TensorVal<DESCRIPTOR >::n] );.
\end{lstlisting}

\section{Multiphysics Couplings}\label{sec:couplings}

\subsection{Shan--Chen Model}
For the simulation of both multiphase and multicomponent flow the Shan--Chen model is implemented in OpenLB. Since its first introduction~\cite{shan_chen:93}, many variants of the model have been developed. In this implementation, there are several forcing schemes~\cite{shan_doolen:95,guo-forcing:02} and interaction potentials to choose from.

\subsubsection{Implementation of Shan--Chen Two-phase Fluid}
The two phases can be simulated on the same lattice instance.
\begin{lstlisting}[language=myc++]
SuperLattice<T, DESCRIPTOR> sLattice(superGeometry);
\end{lstlisting}
Then the dynamics are chosen, which have to support external forces.
\begin{lstlisting}[language=myc++]
ForcedShanChenBGKdynamics<T, DESCRIPTOR> bulkDynamics1 (
omega1, instances::getExternalVelocityMomenta<T,DESCRIPTOR>() );
\end{lstlisting}
Possible choices for the dynamics are \class{Forced\-BGK\-dynamics} and \class{ForcedShanChen\-BGK\-dynamics}.

Then the interaction potential is chosen.
\begin{lstlisting}[language=myc++]
ShanChen93<T,T> interactionPotential;
\end{lstlisting}
Viable interaction potentials for one component multiphase flow are \class{ShanChen93},  \class{ShanChen94},  \class{CarnahanStarling} and  \class{PengRobinson}. In this model \class{PsiEqualsRho} should not be used, because this would make all the mass gather in the same place.

To enable interaction between the fluid, they have to be coupled, so the kind of coupling has to be chosen (here: \class{ShanChenForcedSingleComponentGenerator3D}) and the material numbers to which it applies. Since in the case of single component flow there is only one lattice, it is coupled with itself.
\begin{lstlisting}[language=myc++]
const T G      = -120.;
ShanChenForcedSingleComponentGenerator3D<T,DESCRIPTOR> coupling(
 G,rho0,interactionPotential);
sLattice.addLatticeCoupling(superGeometry, 1, coupling, sLattice);
\end{lstlisting}
The interaction strength \(G\) has to be negative and the correct choice depends on the chosen interaction potential. When using \texttt{PengRobinson} or \texttt{CarnahanStarling} interaction potential, \(G\) is canceled out during computation, so the result is not affected by it (though it still has to be negative).

Finally, during the main loop the lattices have to interact with each other (or in the case of only one fluid component the lattice with itself).
\begin{lstlisting}[language=myc++]
sLattice.communicate();
sLattice.executeCoupling();
\end{lstlisting}
These steps are placed immediately after the \class{collideAndStream} command.

Examples for the implementation of a LB simulation using the Shan--Chen model for two-phase flow are \path{examples/phaseSeparation2d} and \path{examples/phaseSeparation3d}.

\subsubsection{Implementation of Shan--Chen Two-component Fluid}
Two lattice instances are needed, one for each component (though operating on one geometry).
\begin{lstlisting}[language=myc++]
SuperLattice<T, DESCRIPTOR> sLatticeOne(superGeometry);
SuperLattice<T, DESCRIPTOR> sLatticeTwo(superGeometry);
\end{lstlisting}
Then the dynamics are chosen, which have to support external forces.
\begin{lstlisting}[language=myc++]
ForcedShanChenBGKdynamics<T, DESCRIPTOR> bulkDynamics1 (
omega1, instances::getExternalVelocityMomenta<T,DESCRIPTOR>() );
ForcedShanChenBGKdynamics<T, DESCRIPTOR> bulkDynamics2 (
omega2, instances::getExternalVelocityMomenta<T,DESCRIPTOR>() );
\end{lstlisting}
Possible choices for the dynamics are \class{ForcedBGKdynamics} and \class{ForcedShanChenBGKdynamics}. 
One should keep in mind that tasks like definition of dynamics, external fields and initial values and the collide and stream execution have to be carried out for each lattice instance separately. 
The same is true for data output.
Then the interaction potential is chosen.
\begin{lstlisting}[language=myc++]
PsiEqualsRho<T,T> interactionPotential;
\end{lstlisting}
In the multicomponent case, the most frequently used interaction potential is \class{PsiEqualsRho}, but \class{ShanChen93} for example, would also be a viable choice.

To enable interaction between the fluid, they have to be coupled, so the kind of coupling has to be chosen (here: \class{ShanChenForcedGenerator3D})and the material numbers to which it applies.
\begin{lstlisting}[language=myc++]
const T G      = 3.;
ShanChenForcedGenerator3D<T,DESCRIPTOR> coupling(
 G,rho0,interactionPotential);
sLatticeOne.addLatticeCoupling(superGeometry, 1, coupling, sLatticeTwo);
sLatticeOne.addLatticeCoupling(superGeometry, 2, coupling, sLatticeTwo);
\end{lstlisting}
The interaction strength G has to be positive. 
If the chosen interaction potential is \class{PsiEqualsRho}, $G>1$ is needed for separation of the fluids, but it should not be much higher than $3$ for stability reasons.

Finally, during the main loop the lattices have to interact with each other.
\begin{lstlisting}[language=myc++]
sLatticeOne.communicate();
sLatticeTwo.communicate();
sLatticeOne.executeCoupling();
\end{lstlisting}
These steps are placed immediately after the \class{collideAndStream} command.

Examples for the implementation of a LB simulation using the Shan--Chen model for two-component flow are \path{examples/multiComponent2d} and \path{examples/multiComponent3d}.

\subsection{Free Energy Model}\label{sec:freeEnergy}
As an alternative option for simulating multi-component flow, the free energy model has been implemented into OpenLB, and can be used for either two or three fluid components. 
Examples for the binary case are given in \path{youngLaplaceXd} and \path{contactAngleXd}, while an example of the ternary case with boundaries is provided in \path{microFluidics2d}. These are all contained within the \path{examples/multiComponent/} folder.

The approach taken in OpenLB is similar to that given in~\cite{semprebon:16} and assumes equal densities and viscosities for each of the fluids. 
In the next sections the method will be outlined briefly for three components. 
The two component case is identical to taking the third fluid component to be zero and instead only uses two lattices.

\subsubsection{Bulk Free Energy Model}
Three lattices are required to track the density $\rho$, and order parameters $\phi$ and $\psi$. 
These are related to the individual component densities $C_i$ by
\begin{equation}
  \rho = C_1 + C_2 + C_3 ~,  \qquad  \phi = C_1 - C_2 ~,  \qquad  \psi = C_3 ~, 
\end{equation}
respectively. 
By considering the free energy, a force is derived to drive the fluid towards the thermodynamic equilibrium. 
The density therefore obeys the Navier--Stokes equation with this added force. The equation of motion for the order parameters is the Cahn--Hilliard equation. 
The dynamics chosen for the first lattice must therefore include an external force, such as \class{ForcedBGKdynamics}, while for the second and third lattices \class{FreeEnergyBGKdynamics} is required.
\begin{lstlisting}[language=myc++]
ForcedBGKdynamics<T, DESCRIPTOR> bulkDynamics1 (
    omega, instances::getBulkMomenta<T,DESCRIPTOR>() );
FreeEnergyBGKdynamics<T, DESCRIPTOR> bulkDynamics23 (
    omega, gamma, instances::getBulkMomenta<T,DESCRIPTOR>() );
\end{lstlisting}

To compute the force, two lattice couplings are required. The first computes the chemical potentials for each lattice using the equations
\begin{align}
  \mu_{\rho}  = &
      A_1 + A_2 + \frac{\alpha^2}{4} \left[ (\kappa_1 + \kappa_2)
      \left( \bm{\nabla}^2\psi - \bm{\nabla}^2\rho \right) +
      (\kappa_2 - \kappa_1) \bm{\nabla}^2\phi \right] ~,  \label{eq:muRho} \\
  \mu_{\phi}  = &
      A_1 - A_2 + \frac{\alpha^2}{4} \left[ (\kappa_2 - \kappa_1)
      \left( \bm{\nabla}^2\rho - \bm{\nabla}^2\psi \right) -
      (\kappa_1 + \kappa_2) \bm{\nabla}^2\phi \right] ~,  \label{eq:muPhi} \\
  \mu_{\psi} = & - A_1 - A_2 + \kappa_3 \psi (\psi - 1) (2\psi - 1)
        + \frac{\alpha^2}{4} \left[(\kappa_1 + \kappa_2)
        \bm{\nabla}^2\rho \right. \nonumber \\
        & - \left. (\kappa_2 - \kappa_1) \bm{\nabla}^2\phi -
        (\kappa_1 + \kappa_2 + 4\kappa_3) \bm{\nabla}^2 \psi \right] ~,  \label{eq:muPsi}
\end{align}
where $A_1$ and $A_2$ are defined as
\begin{align}
  A_1 = \frac{\kappa_1}{8} (\rho + \phi - \psi)
            (\rho + \phi - \psi - 2) (\rho + \phi - \psi - 1) ~, \\
  A_2 = \frac{\kappa_2}{8} (\rho - \phi - \psi)
            (\rho - \phi - \psi - 2) (\rho - \phi - \psi - 1) ~, 
\end{align}
respectively. 
The $\alpha$ the $\kappa$ parameters are input parameters for the lattice coupling and can be used to tune the interfacial width and surface tensions. 
The interfacial width is given by $\alpha$ and the surface tensions are $\gamma_{mn} = \alpha(\kappa_m + \kappa_n)/6$.

The second lattice coupling then computes the force using
\begin{equation}
  \boldsymbol{F} = - \rho \bm{\nabla} \mu_\rho
                   - \phi \bm{\nabla} \mu_\phi
                   - \psi \bm{\nabla} \mu_\psi ~.
\end{equation}
This depends non-locally upon the result of the first lattice coupling, and so the chemical potential must be communicated in between. To accommodate this, the first coupling is assigned to the $\rho$ lattice and the force coupling is assigned to the $\phi$ lattice.
\begin{lstlisting}[language=myc++]
FreeEnergyChemicalPotentialGenerator3D<T, DESCRIPTOR> coupling1(
    alpha, kappa1, kappa2, kappa3 );
FreeEnergyForceGenerator3D<T, DESCRIPTOR> coupling2;

sLattice1.addLatticeCoupling<DESCRIPTOR>(
    superGeometry, 1, coupling2, {&sLattice2, &sLattice3} );
sLattice2.addLatticeCoupling<DESCRIPTOR>(
    superGeometry, 1, coupling3, {&sLattice1, &sLattice3} );
\end{lstlisting}
The following is then used in the main loop to calculate the force at each timestep.
\begin{lstlisting}[language=myc++]
sLattice1.executeCoupling();
sExternal1.communicate();
sExternal2.communicate();
sExternal3.communicate();
sLattice2.executeCoupling();
\end{lstlisting}

\subsubsection{Boundaries in Free Energy Models}
Using the \class{setFreeEnergyWallBoundary} function, bounce-back wall boundaries with controllable contact angles can be added via the following commands: 
\begin{lstlisting}[language=myc++]
setFreeEnergyWallBoundary<T,DESCRIPTOR>(sLattice1, superGeometry, 2,
    alpha, kappa1, kappa2, kappa3, h1, h2, h3, 1);
setFreeEnergyWallBoundary<T,DESCRIPTOR>(sLattice2, superGeometry, 2,
    alpha, kappa1, kappa2, kappa3, h1, h2, h3, 2);
setFreeEnergyWallBoundary<T,DESCRIPTOR>(sLattice3, superGeometry, 2,
    alpha, kappa1, kappa2, kappa3, h1, h2, h3, 3);
\end{lstlisting}
The final parameter, \class{latticeNumber}, is necessary to change each lattice differently. While the $h_i$ parameters are related to the contact angles at the boundary. The contact angles $\theta_{mn}$, are given by the following relation
\begin{equation}
    \cos\theta_{mn} =
        \frac{(\alpha\kappa_n \! + \! 4h_n)^{3/2} -
              (\alpha\kappa_n \! - \! 4h_n)^{3/2}}
             {2 (\kappa_m + \kappa_n) \sqrt{\alpha\kappa_n}} -
        \frac{(\alpha\kappa_m \! + \! 4h_m)^{3/2} -
              (\alpha\kappa_m \! - \! 4h_m)^{3/2}}
             {2 (\kappa_m + \kappa_n) \sqrt{\alpha\kappa_m}} ~.
\label{eq:contactAngles}
\end{equation}
Notably, to set neutral wetting ($90^\circ$ angles), the values can be set to $h_i=0$.

A demonstration of using these solid boundaries for a binary fluid case is provided in the \path{contactAngle(2,3)d} examples. 
The examples compare the simulated angles to those given by equation~\eqref{eq:contactAngles}, respectively for dimensions \(d=2,3\).

Open boundary conditions can also be implemented using the \class{setFreeEnergyInletBoundary} and \class{setFreeEnergyOutletBoundary} functions. 
These can be used to specify constant density or velocity boundaries. 
The first lattice is used to define the density or velocity boundary condition, while on the second and third lattices $\phi$ and $\psi$ must instead be defined. 
For example, to set a constant velocity inlet, see the code snippet below.
\begin{lstlisting}[language=myc++]
setFreeEnergyInletBoundary(
    sLattice1, omega, inletIndicator, "velocity", 1 );
setFreeEnergyInletBoundary(
    sLattice2, omega, inletIndicator, "velocity", 2 );
setFreeEnergyInletBoundary(
    sLattice3, omega, inletIndicator, "velocity", 3 );

sLattice1.defineU( inletIndicator, 0.002 );
sLattice2.defineRho( inletIndicator, 1. );
sLattice3.defineRho( inletIndicator, 0. );
\end{lstlisting}
However, this alone is insufficient to set a constant density outlet because $\rho$, $\phi$, and $\psi$ are redefined by a convective boundary condition on each time step. 
In this case an additional lattice coupling is required, using \class{FreeEnergyDensityOutletGeneratorXD}.

There are two additional requirements for open boundaries. 
The first is that the velocity must be coupled between the lattices using \class{FreeEnergyInletOutletGeneratorXD}, because this is required for the collision step. 
The second is that the communication of the external field must now include two values. 
This ensures that $\rho$, $\phi$, and $\psi$ are properly set on block edges at the outlet. 
To see a full example of applying these boundary conditions, see the \path{microFluidics2d} example.

\subsection{Coupling Between Momentum and Energy Equations}
As explained in reference \cite{mohamad2010critical}, there are different schemes to couple the momentum and energy equations by means of a buoyancy force (also called Boussinesq approximation). 
Some schemes add an extra force term to the collision term, other methods shift the velocity field according to Newton's second law, and others combine an extra force term and a velocity shift. 
The implementation applied in OpenLB belongs to this last group of schemes.

Once the boundary values for the velocity and temperature fields are set, collision and streaming functions are called. 
The dynamics with an external force \(\bm{F}\) used for the velocity calculation (\eg \class{ForcedBGKdynamics}) shifts the velocity \(\bm{v}\) before executing the collision step. 
The shift follows equation~\eqref{eq:FBD}, 
\begin{equation}
\bm{u}_{\mathrm{shift}}= \bm{u} + \frac{\bm{F}}{2} ~.
\label{eq:FBD}
\end{equation}
The code snippet responsible for this shift is defined in the collision function of the file \path{/dynamics/dynamics.hh} for the class \class{ForcedBGKdynamics}.
\begin{lstlisting}[language=myc++,caption={Velocity shift}]
this->momenta . computeRhoU( cell , rho , u) ;
FieldPtr<T,DESCRIPTOR,FORCE> force = cell.template getFieldPointer<descriptors::FORCE>();
for ( int iVe l=0; iVel<DESCRIPTOR >::d; ++iVel ) {
  u[iVel] += force[iVel] / T{2};
}
\end{lstlisting}
After the corresponding collision step using the shifted velocity, the value of the density distribution functions $f_i$ is modified by the external force with the call to the function.
\begin{lstlisting}[language=myc++]
lbm< Lattice >::addExternalForce( cell , u , omega )
\end{lstlisting}
This function follows
\begin{equation}
\bar{f_i} 
= 
f_i 
+ \left( 1 - \frac{\omega}{2} \right)
  w_i 
  \left( \frac{\bm{c}_i - \bm{u}}{c_s^2} + \frac{\bm{c}_i \cdot \bm{u}}{c_s^4}\bm{c}_i \right) 
\cdot \bm{F} ~, 
\label{eq:dF}
\end{equation}
where $\tilde{f_i}$ represents the new distribution function (see reference~\cite{guo-forcing:02} for the BGK model and \cite{fakhari2010phase} for MRT models), \(w_{i}\) are the weights of the discrete velocities \(\bm{c}_{i}\), and \(\omega = \tau^{-1}\) denotes the relaxation frequency. 
The coupling in the collision step for the temperature field is given by the use of the velocity from the isothermal field.
\begin{lstlisting}[language=myc++]
auto u = cell.template getFieldPointer<descriptors::VELOCITY>();
\end{lstlisting}

The equilibrium density distribution function for the temperature only has terms of first order (see \eg \cite{siodlaczek:21}). 
After the collision step, the coupling function is called \class{NSlattice.executeCoupling()}, where the values of the external force in the \class{NSlattice} and of the advected velocity in the \class{ADlattice} are updated.
\begin{lstlisting}[language=myc++,caption={Velocity coupling}]

auto u = tPartner->get(iX, iY).template getFieldPointer<descriptors::VELOCITY>();
blockLattice.get(iX, iY).computeU(u);
\end{lstlisting}

The new force is computed via the Boussinesq approximation
\begin{equation}
\bm{F} = \rho \frac{T-T_0}{\triangle T} \bm{g}~.
\label{eq:bA}
\end{equation}
The temperature \(T\) is obtained from the \class{ADlattice}, $T_0$ is the average temperature between the defined cold and hot temperatures, whereas $\triangle T$ is the difference between the hot and cold temperatures.
\begin{lstlisting}[language=myc++,caption={Computation of the Boussinesq force}]

auto force = blockLattice.get(iX, iY).template getFieldPointer<descriptors::FORCE>();
T temperature = tPartner->get(iX, iY).computeRho();
T rho = blockLattice.get(iX, iY).computeRho();
for (unsigned iD = 0 ; iD < L :: d ; ++iD) {
  force[iD] = gravity * rho * (temperature - T0) / deltaTemp * dir[iD];
}
\end{lstlisting}

\section{Advection--Diffusion Equation}\label{sec:ade}

The advective and diffusive transport of a macroscopic density, energy or temperature is governed by the advection--diffusion equation
\begin{equation}
  \frac{\partial c}{\partial t} = \bm{\nabla} \cdot (D\bm{\nabla} c) - \bm{\nabla} \cdot (\bm{v}c) \quad \text{in } \Omega \times I ~,
\end{equation}
where $c\colon \Omega \times I \to \mathbb{R}, (\bm{x}, t)\mapsto c(\bm{x}, t)$ is the considered physical quantity (temperature, particle density), $D>0$ is the diffusion coefficient and $\bm{v}$ is a velocity field affecting $c$. 
It is possible to approximate this equation with LBM by using an equilibrium distribution function different from the one for the Navier--Stokes equations~\cite{chopard:02,simonis:20,simonis:23}
\begin{equation}\label{eq:equilibriumfirstOrder}
g_i^{eq} = w_i \rho \bigg(1+\frac{\bm{c}_i \cdot \bm{v}}{c_s^2}\bigg) ~,
\end{equation}
that takes the advective transport into account. 
In equation~\eqref{eq:equilibriumfirstOrder}, $w_i$ is a weighting factor, $\bm{c}_i$ a unit vector along the lattice directions, $c_s$ the speed of sound, and \(i\) denotes the discrete velocity counter. 
To use this implementation the dynamics object has to be replaced by special advection--diffusion dynamics:
\begin{lstlisting}[language=myc++,caption=Advection diffusion dynamics object]
AdvectionDiffusionBGKdynamics<T, DESCRIPTOR> bulkDynamics(
  converter.getOmega(),
  instances::getBulkMomenta<T,DESCRIPTOR>());
\end{lstlisting}
Additionally, a different descriptor with fewer lattice velocities is used~\cite{huang:11}:
\begin{lstlisting}[language=myc++,caption=Advection diffusion descriptor]
using DESCRIPTOR = descriptors::AdvectionDiffusionD3Q7Descriptor;
\end{lstlisting}
In OpenLB, the descriptors \class{D2Q5} and \class{D3Q7} are implemented for the advection--diffusion equation.
Since the advection--diffusion equation simulates different physical conditions than the Navier--Stokes equations, another set of boundary conditions is needed.
A Dirichlet condition for the density is already implemented, for example to simulate a boundary with a constant temperature.
\begin{lstlisting}[language=myc++,caption=Advection diffusion dynamics object]
ParticleAdvectionDiffusionBGKdynamics<T, ADDESCRIPTOR> bulkDynamicsAD ( omegaAD,
      instances::getBulkMomenta<T,ADDESCRIPTOR>() );
\end{lstlisting}
Finally the boundary condition is set to the desired material number.
\begin{lstlisting}[language=myc++,caption=Advection diffusion descriptor]
void prepareLattice(...) {
...
/// Material=3 -> boundary with constant temperature
setAdvectionDiffusionTemperatureBoundary<T,ADDESCRIPTOR>(
      sLatticeAD, omegaAD, superGeometry, 3);
...
}
\end{lstlisting}
To apply convective transport, a velocity vector has to be passed. 
This can either be done individually on each cell by using the following.
\begin{lstlisting}[language=myc++,caption=Add advective velocity on a cell]
T velocity[3] = {vx,vy,vz};
...
cell.defineField<descriptors::VELOCITY>(velocity);
\end{lstlisting}
Alternatively, it can be passed to the whole \class{SuperLattice} using:
\begin{lstlisting}[language=myc++,caption=Add advective velocity on a superlattice]
AnalyticalConst3D<T,T> velocity(vel);
...
/// sets advective velocity for material 1
superLattice.defineField<descriptors::Velocity>(superGeometry, 1, velocity);
\end{lstlisting}
Here, \class{vel} is a \class{std::vector<T>}.

\subsection{AdvectionDiffusion Boundary Conditions}

\subsubsection{Dirichlet Boundary Condition}
At the boundaries of a lattice, only the outgoing directions of the distribution functions are known, while those towards the domain need to be computed. 
Several types of implementations for Dirichlet boundary conditions are summarized in Section~\ref{sec:defineBoundaryMethod}. 
At a Dirichlet boundary for the advection--diffusion equation, the observable, \eg temperature, is set to a constant value. 
This boundary condition can be applied to flat walls, corners and edges (for three-dimensional
domains). 
The algorithm to set a certain temperature on a wall is defined in the dynamics class \class{AdvectionDiffusionBoundariesDynamics} in the file \path{boundary/advectionDiffusionBoundaries.hh} and works as described below.

First, the index \(i\) of the unknown distribution function $g_i$ incoming to the fluid domain is determined.
\begin{lstlisting}[language=myc++,caption={Collision step for a temperature boundary}]
int missingNormal = 0 ;
constexpr auto missingDiagonal = util::subIndexOutgoing<L,direction,orientation>();
std::vector<int> knownIndexes = util::remainingIndexes<L>(missingDiagonal);
for (unsigned iPop = 0; iPop < missingDiagonal.size(); ++iPop)
{
  int numOfNonNullComp = 0;
  for (int iDim = 0 ; iDim < L::d; ++iDim)
    numOfNonNullComp += abs (L::c[missingDiagonal[iPop]][iDim]);

  if (numOfNonNullComp == 1)
  {
    missingNormal = missingDiagonal[iPop];
    missingDiagonal.erase(missingDiagonal.begin()+iPop);
    break;
  }
}
\end{lstlisting}
Then, the sum of the rest of the populations is computed.
\begin{lstlisting}[language=myc++]
T sum = T( );
for (unsigned iPop = 0 ; iPop < knownIndexes.size(); ++iPop)
{
  sum += cell[knownIndexes[iPop]];
}
\end{lstlisting}

The difference between the desired temperature value (given when setting the boundary condition) and this sum is the value assigned to the unknown distribution.
\begin{lstlisting}[language=myc++]
T temperature = this->momenta.computeRho( cell );
cell[missingNormal] = temperature - sum -(T)1;
\end{lstlisting}

After that, all distribution functions are determined and a regular collision step is performed.
\begin{lstlisting}[language=myc++]
boundaryDynamics.collide(cell, statistics);
\end{lstlisting}

As an example, take the case of a left wall in 2D. 
After the streaming step, all populations are known except for $g_3$. 
Once the desired temperature $T_{wall}$ at the wall is known, the value of the unknown distribution is computed via  
\begin{align}
T_{wall} & = \sum \limits_{i=0}^{4} g_i  \\
\Leftrightarrow 
g_3 & = T_{wall} - (g_0+g_1+g_2+g_4) ~.
\end{align}

\subsubsection{Neumann Boundary Condition}
For flat walls there is also an opportunity to prescribe the flux through the wall as a boundary condition which is 
\begin{align}
\bm{\nabla} c \cdot \bm{n} = \text{flux} ~.
\end{align} 
Here $\bm{n}$ stands for the outer normal vector of the boundary. 
The flux multiplied with the size of the length discretization has to be written in the field \class{BOUNDARY}. 
\begin{lstlisting}[language=myc++]
AnalyticalConst2D<T,T>flux_ (converter.getLatticeDensity(flux)*converter.getPhysDeltaX());
sLatticeAD.defineField<descriptors::BOUNDARY>(boundaryMaterialIndicator,flux_); \end{lstlisting}
The setter \class{setAdvectionDiffusionNeumannBoundary} for \class{AdvectionDiffusionNeumannBoundary} adds a preprocessing step which calculates according to the direction of the normal the Dirichlet value which should be at the wall in order to satisfy the flux condition. 
Therefore a simple backward or forward difference quotient is used as suggested in \cite{kruger2017lattice}. 
As an example, we take again a left wall in 2D and have that
\begin{align}
c_\text{wall} = \triangle x \, \text{flux} + c \left( x_\text{wall}+\triangle x \right) ~.
\end{align}

\begin{lstlisting}[language=myc++]
T analyticalBoundaryValue = cell. template getField<descriptors::BOUNDARY>();

            if constexpr(normal1 !=0 && normal2 == 0 && normal3 == 0){
              if constexpr(normal1 >0){
             // right boundary, difference quotient with x-1
             neumannBoundary = analyticalBoundaryValue+cell.neighbor({-1,0,0}).computeRho();
             }
              else if constexpr(normal1<0){
              //left boundary, difference quotient with x+1
              neumannBoundary = cell.neighbor({1,0,0}).computeRho() - analyticalBoundaryValue;
            }
          }
\end{lstlisting}
Finally we have to set the calculated value for the cell with the \class{defineRho} command. 
After this precalculation step the Dirichlet Boundary Condition is applied.

\subsubsection{Adiabatic Boundary Condition}
Additionally, thanks to the simplified lattice velocity sets used in the thermal descriptors (\class{D2Q5} and \class{D3Q7}, cf.\ the discrete velocities in Figure~\ref{fig:discreteVelocitySets} colored in red and orange), it is possible to implement an adiabatic boundary using bounce-back dynamics \cite{mezrhab2010double}. 

An adiabatic boundary condition requires no heat conduction in the normal direction of the boundary. 
In a general situation the adiabatic boundary is set on a solid wall, meaning that the normal velocity to the wall is zero. 
To implement an adiabatic wall, take a 2D south wall as example. 
The distribution function $g_4$ and the temperature at the wall are undetermined. 
The population $g_4$ can be computed from the distribution function in the opposite direction, in order to ensure that at the macroscopic level there is no heat conduction, \ie
\begin{equation}
g_4=g_2 ~.
\end{equation}
This procedure corresponds to the the bounce-back scheme. 
With all the distribution functions known, the temperature at the wall can be determined from its definition
\begin{equation}
T_{wall}= \sum \limits_{i=0}^{4} g_i ~.
\end{equation}

\subsection{Convergence Criterion}
\label{subsec:convergenceCrit}
For thermal applications, the following convergence criterion can be applied to one of the computed fields or to both of them. 
Generally, a value tracer on the average energy is used, which is also available for any other quantity and in turn applicable for any TEQ approximated with LBM in OpenLB (see Section~\ref{sec:lessonConvergenceCheck}). 
Here, the average energy is defined proportional to the velocity squared, which makes it independent to use the \texttt{NSlattice} or the \texttt{ADlattice}, since both share the same macroscopic velocity field \(\bm{u}\).

The parameters to initialize the tracer object are the characteristic velocity of the system \class{converter.getU( )}, the characteristic length of the system \class{converter.getNy( )}, and the desired precision of the convergence \class{eps}. 
The listing~\ref{lst:C C} shows how the object is defined in the main function, and how its value is updated and checked at each time step.
\begin{lstlisting}[language=myc++,caption={Convergence check},label={lst:C C}]
util::ValueTracer<T> converge( converter.getU( ), converter.getNy( ), eps );
...
for ( iT=0; iT<maxIter ; ++iT) {
  converge.takeValue( ADlattice.getStatistics( ).getAverageEnergy( ), true );
  ...
  if ( converge.hasConverged() ) {
    break;
  }
}
\end{lstlisting}

\subsection{Creating an Application with AdvectionDiffusionDynamics}
When creating a program \eg for a thermal applications, the following points should be regarded.

\subsubsection{Lattice Descriptors}
Lattice descriptors used for the \class{AdvectionDiffusionDynamics} are \class{D2Q5} and \class{D3Q7}, which have less degrees of freedom for the velocity space then the classical discrete velocity sets for the Navier--Stokes equations (see Figure~\ref{fig:discreteVelocitySets}).
With the Chapman--Enskog expansion it can be shown that approximating the advection--diffusion equation as a target does not require fourth order isotropic lattice tensors (see for example \cite{zheng2006lattice}), therefore descriptors with less discrete velocities can be used without loss of accuracy.

To approximate the Navier--Stokes velocity field \(\bm{u}\), a first descriptor should be defined that can include an external force, \eg \class{ForcedD3Q19Descriptor}. 
Another descriptor is necessary to approximate the temperature \(T\) governed by an advection--diffusion equation, \eg \class{AdvectionDiffusionD3Q7Descriptor}.

\subsubsection{Preparing the Geometry}
This step is similar to the isothermal procedure (without the second lattice for \(T\)). 
With help of indicators and STL files, the desired geometry can be created. 
Several material numbers are assigned to the cells that in turn can be implemented to specify several bulk, initial and boundary collision schemes and/or dynamics.  

\paragraph{Reading STL files}
To conveniently read STL files, the OpenLB class \class{stlReader} is provided. 
There are differences, however, when compared to the isothermal case, due to the differing converter objects. 
An example of its use could be:
\begin{lstlisting}[language=myc++,caption={Initialization of a STLreader object}]
STLreader<T> nameIndicator( " fileName.stl ", converter.getDeltaX ( ), conversionFactor );
\end{lstlisting}
The offsets between the STL file and the global geometry are much easier handled, if they are directly defined when creating the STL file, rather than trying to do it in the application code afterwards.
Note that, if the geometry is adjusted for a grid resolution of \(N=100\) and a conversion factor of \(1\), and the resolution is increased to \(N=200\), the corresponding conversion factor also has to be increased by a factor \(2\), in order to keep the correct proportions of the model.

\subsubsection{Preparing the Lattices}
Recall that in a typical thermal application there are two independent lattices: one for the isothermal flow (usually referred to as \class{NSlattice}), and one for the thermal variables (\eg the temperature, usually referred to as \class{ADlattice}). 
For each material number the desired dynamics behavior has to be defined. 
Commonly used possibilities are \class{instances::getNoDynamics} (do nothing), \class{instances::getBounceBack} (no slip), or \class{bulkDynamics} (previously set collision for the bulk).
For a thermal lattice the we can define a boundary with a given temperature.
\begin{lstlisting}[language=myc++,caption={Definition of a temperature boundary}]
setAdvectionDiffusionTemperatureBoundary<T,TDESCRIPTOR>(
     ADlattice, Tomega, superGeometry, 2);
\end{lstlisting}
The chosen dynamics for a material number may differ between the isothermal and the thermal lattices, \eg an obstacle with a given temperature inside a flow channel would have a no-slip
behavior for the fluid part, but be part of the bulk and have a given temperature in the thermal lattice.

\subsubsection{Initialization of the Lattices (iT=0)}

\paragraph{NSlattice} 
For all material numbers defined as \class{bulkDynamics}, an initial velocity and density has to be set (usually fluid flow at rest). 
Additionally, since the velocity field and the temperature are related by a force term, an external field has to be defined. 
The easiest way to do this is by the material number.
\begin{lstlisting}[language=myc++,caption={Initialization of an external force field}]
NSlattice.defineField<descriptors::FORCE>( superGeometry, 1,  force );
\end{lstlisting}
Here \class{force} is an element of type \class{AnalyticalF}, which can initially be set to zero. 

\paragraph{ADlattice} 
For the advection--diffusion lattice, an initial temperature is set (similar to the density variable on the Navier--Stokes lattice), as well as the distribution functions corresponding
to this temperature value:
\begin{lstlisting}[language=myc++,caption={Initialization of the temperature field}]
T Texample = 0.5;
T zerovel[descriptors::d<T,DESCRIPTORS>()] = {0., 0.};
ConstAnalyticalF2D<T,T> Example( Texample );
std::vector<T> tEqExample(descriptors::q<T,DESCRIPTORS>() );
for ( int iPop = 0; iPop < descriptors::q<T,DESCRIPTORS>(); ++iPop )
  {tEqExample [ iPop ] = advectionDiffusionLbHelpers<T,TDESCRIPTOR>::
    quilibrium( iPop, Texample, zerovel ); }
ConstAnalyticalF2D<T,T> EqExample( tEqExample );
ADlattice.defineRho( superGeometry, 1 ,Example );
ADlattice.definePopulations( superGeometry, 1, EqExample );
\end{lstlisting}
To apply convective transport, a velocity vector has to be passed, which can be also done by material number.
\begin{lstlisting}[language=myc++]
std::vector<T> zero ( 2, T( ) );
ConstAnalyticalF2D<T,T> velocity ( zero );
ADlattice.defineField<descriptors::VELOCIY>(superGeoemtry, 1,velocity);
\end{lstlisting}
The last step is to make the lattice ready for the simulation:
\begin{lstlisting}[language=myc++,caption={Initialization of the lattices}]
NSlattice.initialize( );
ADlattice.initialize( );
\end{lstlisting}

\subsubsection{Setting the Boundary Conditions}
If the value of a boundary condition has to be updated during the simulation, \eg via increasing the velocity at the inflow or changing the temperature of a boundary, this can be achieved following the same procedure as for the initial conditions (see Section~\ref{sec:defineInitialConditions} further below).

\subsubsection{Getting the Results}
The desired data is saved using the \class{VTKwriter} objects, which can write the value of functors in VTI files (VTK format used \eg by ParaView~\cite{paraview-web}). 
The functors which are usually saved are the velocity field from the \class{NSlattice}, and the
temperature field (referred to as density) of the \class{ADlattice}. 
Thermal and isothermal information must be saved in two different objects, since they have two different descriptors.
\begin{lstlisting}[language=myc++,caption={Saving results in VTK files}]
SuperLatticeVelocity2D<T,NSDESCRIPTOR> velocity( NSlattice );
SuperLatticeDensity2D<T,TDESCRIPTOR> density( ADlattice );
vtkWriterNS.addFunctor( velocity );
vtkWriterAD.addFunctor( density );
vtkWriterNS.write(iT);
vtkWriterAD.write(iT);
\end{lstlisting}
It is important to emphasize that the data saved is in \textit{lattice units}. 
Some conversions can be made in order to obtain physical magnitudes, which are described later in Section~\ref{sec:oRITS}.

\subsubsection{Structure of the Program}
It is advisable to structure the main loop of the \path{.cpp}-file along the following steps:  
\begin{description}
  \item[1. Initialization] The converter between dimensionless and lattice units is set via \eg \(N\), \(\triangle t\) and the parameters for the simulation \(Ra\), \(Pr\), \(T_{\mathrm{cold}}\), \(T_{\mathrm{hot}}\), \(L_{x,y,z}\).
  \item[2. Prepare geometry] The mesh is created and voxels are classified with different material numbers according to their behavior (inflow, outflow, etc.).
  \item[3. Prepare lattice] The lattice dynamics are set according to the material numbers assigned before. The boundary conditions are initialized. Since there are two different lattices, the definition of the dynamics and the kind of boundary conditions (though not the actual values yet) have to be made for each of them separately. At this point the coupling generator is initialized (usually on the \class{NSlattice}) and then it is indicated which material numbers are to be coupled with the \class{ADlattice}.
  \item[4. Main loop with timer] The functions \class{setBoundaryValues}, \class{collideAndStream}, and \class{getResults} are called repeatedly until the maximum of iterations is reached or the simulation has converged (if a convergence criteria is set).
  \item[5. Definition of initial and boundary conditions] The values for the boundary functions are set. In some applications the values are to be refreshed at each time step. Thermal and isothermal lattices are treated separately. As indicated before, velocity and density (\class{NSlattice}), as well as temperature (\class{ADlattice}) have to be defined. Additionally, the couplings, and external forces and velocities should be initialized and reused as required.
  \item[6. Collide and stream execution] The collision and the streaming steps are performed. This function is called for each of the lattices separately. After the streaming step, the coupling between the lattices (here, based on the Boussinesq approximation) is executed.
  \item[7. Computation and output of the results] Console and data outputs of the results at certain time steps are created.
\end{description}

\subsection{Obtaining Results in Thermal Simulations \label{sec:oRITS}}
Here, the Rayleigh and the Prandtl numbers are the dimensionless numbers which control the physics of a convection problem. 
The Rayleigh number for a fluid is associated with buoyancy driven flow. 
When the Rayleigh number is below the critical value for that fluid, heat transfer is primarily in the form of conduction. 
When it exceeds the critical value, heat transfer is primarily in the form of convection. 
For natural convection, it is defined as
\begin{equation}
Ra=\frac{g \beta}{\nu \alpha} \triangle T L^3 ~, 
\end{equation}
where $g$ is the acceleration magnitude due to gravity, $\beta$ is the thermal expansion coefficient, $\nu$ is the kinematic viscosity, $\alpha$ is the thermal diffusivity, $\triangle T$ is the temperature difference, and $L$ denotes the characteristic length.
The Prandtl number is defined as the ratio of momentum diffusivity $\nu$ to thermal diffusivity $\alpha$ 
\begin{equation}
Pr=\frac{\nu}{\alpha}~.
\end{equation}
To handle differences between the converter objects for isothermal and thermal simulations some of the isothermal functions have been re-implemented for thermal simulations in a way that they only depend on lattice parameters and the Rayleigh and Prandtl numbers by modifying existing functors in the following files:
\begin{itemize}
    \item \path{/functors/lattice/blockLatticeIntegralF3D.(h,hh)}
    \item \path{/functors/lattice/superLatticeIntegralF3D.(h,hh)}
\end{itemize}

\subsubsection{Velocity}
The resulting velocity magnitude \(v_{res}\), independent of the lattice velocity \texttt{latticeU} selected, can be computed by
\begin{equation}
v_{res}
=
\frac{v_{\mathrm{LB}}}{\texttt{latticeU}}\sqrt{Ra \, Pr}
=
\frac{v_{\mathrm{LB}}}{N \triangle t}\sqrt{Ra \, Pr} ~.  
\end{equation}
The lattice velocity \class{latticeU} is obtained from the function \class{converter.getCharLatticeVelocity()} in the thermal converter object.
\subsubsection{Pressure}
The pressure in physical units is derived from the lattice pressure by using its definition from the isothermal converter object
\begin{align}
p_{\mathrm{phys}} & = p_{\mathrm{LB}}\frac{\texttt{physcForce}}{\texttt{physL}^{d-1}} \\
& =p_{\mathrm{LB}}\frac{\texttt{physL}^{d+1}}{\texttt{physT}^2}\frac{1}{\texttt{physL}^{d-1}} \\
& =p_{\mathrm{LB}}\left( \frac{\texttt{physL}}{\texttt{physT}}\right)^2 \\
& =p_{\mathrm{LB}}\left( \frac{\texttt{charU }}{\texttt{latticeU}} \right)^2 \\
& =p_{\mathrm{LB}}\frac{Ra \, Pr}{\texttt{latticeU}^2} ~.
\end{align} 
The lattice pressure can easily be computed from the lattice density using
\begin{equation}
p_{\mathrm{LB}}=\frac{\rho -1}{3} ~.
\end{equation}
The physical force can also be obtained from the computed lattice force
\begin{align}
F_{\mathrm{phys}} 
& = 
F_{\mathrm{LB}}\frac{\texttt{physL}^{d+1}}{\texttt{physT}^2} \\
& = 
F_{\mathrm{LB}} \frac{ \left( \texttt{charL} \frac{\texttt{latticeL}}{\texttt{charL}} \right)^{d+1}}{\left( \frac{\texttt{charL}}{\texttt{charU}}\frac{\texttt{latticeU} \texttt{latticeL}}{\texttt{charL}} \right)^2}\\
& = 
F_{\mathrm{LB}} \, \texttt{latticeL}^{d-1}\left(\frac{\texttt{charU}}\,{\texttt{latticeU}} \right)^2 \\
& = 
F_{\mathrm{LB}}\frac{\texttt{latticeL}^{d-1}}{\texttt{latticeU}^2} Ra \,  Pr ~,  
\end{align} 
where $d$ is the number of dimensions in the problem.

For most applications the value of the force coefficients in the different coordinate directions can be of interest, which can be computed with
\begin{equation}
C_{F_i}= F_{i,\mathrm{phys}}\frac{1}{\frac{1}{2} \texttt{charU}^2 \cdot \texttt{count}_i \cdot \texttt{latticeL}^{d-1}}=F_{i,\mathrm{LB}} \frac{1}{\frac{1}{2} \texttt{latticeU}^2 \cdot \texttt{count}_i} ~,
\end{equation}
where $\texttt{count}_i$ is the number of cells in the surface perpendicular to the direction $i$ of the force coefficient computed.

\subsection{Conduction Problems}
For heat conduction problems there is no velocity field that advects the temperature (see \cite{mishra2007solving} and \cite{kaluza2012numerical} ). 
In absence of convection, radiation and heat generation, the energy equation for a
homogeneous medium is given by
\begin{equation}
\frac{\partial T}{\partial t}=\alpha \bm{\Delta} T ~.
\end{equation}
A conduction simulation can be executed using an independent advection--diffusion lattice, without any velocity field coupled. 
In the same way as in the convection--diffusion heat transfer, the temperature is obtained after summing the distribution functions over all directions. 
The equilibrium distribution function in this case with the BGK approximation is given by
\begin{equation}
g_i^{eq}=\omega_i\rho=\omega_i T ~,
\end{equation}
which is equivalent to the one used in advection--diffusion simulations with the flow velocity set to zero. 
It means that conduction problems could be computed based on the available OpenLB code by only using a lattice with advection--diffusion dynamics and by setting the external velocity field to zero at any time.

\subsubsection{Multiple-Relaxation-Time (MRT)}
The implementation of the thermal lattice Boltzmann equation using the multiple-relaxation-time collision model is done similarly to the procedure used with the BGK collision model. 
A double MRT-LB is used, which consists of two sets of distribution functions: an isothermal MRT model for the mass-momentum equations, and a thermal MRT model for the temperature equation.
Both sets are coupled by a force term according to the Boussinesq approximation. 
The macroscopic governing equation for the temperature is
\begin{equation}
\frac{\partial T}{\partial t} +v\bm{\nabla} T = \alpha \bm{\Delta} T~,
\end{equation}
where $\alpha$ is the thermal diffusivity coefficient.

The isothermal MRT model with an external force is already implemented in OpenLB for the \class{D2Q9} and \class{D3Q19} lattice models (dynamics class \class{ForcedMRTdynamics}). 
This means that only the thermal MRT counterparts for 2D and 3D have to be developed.

The computation of the force term in the MRT model in the \class{ForcedMRTdynamics} class uses the body force as described in \cite{ladd2001lattice}. 
It does not include however a velocity shift like the BGK model, due to negligible differences in benchmark tests. 

\paragraph{D2Q5 thermal model} 
The formulation for the \class{D2Q5} thermal MRT model is based on~\cite{liu2015double}. 
The temperature field distribution functions $g_i$ are governed by the following equation
\begin{equation}
g_{i} (\bm{x} + \bm{c}_{i} \triangle t, t+ \triangle t ) - g_{i} (\bm{x} , t ) 
= 
- \bm{N}^{-1}_{i} \theta_{i} [\bm{n} (\bm{x},t) - \bm{n}^{eq} (\bm{x},t)] ~,
\end{equation}
where \(\bm{g}\) and \(\bm{n}\) are column vectors with entries \(g_{i}\) and \(n_{i}\) for \(i=0, 1, \ldots, q-1\), and denote the distribution functions and the moments, respectively. 
The vectors \(\bm{N}_{i}\) are the rows of the orthogonal transformation matrix \(\mathbf{N}\) and $\theta_{i}$ are the entries of a non-negative, diagonal relaxation matrix.
The macroscopic temperature T can be calculated by
\begin{equation}
T =\sum \limits_{i=0}^{4} g_i ~.
\end{equation}
The weight coefficients for each lattice direction are given in equation~\eqref{eq:wC} with
\begin{equation}\label{eq:wC}
\omega_i = \begin{cases}
\frac{3}{5}, & \quad \text{if } i = 0 ~,\\
\frac{1}{10}, & \quad \text{if } i = 1,2,3,4 ~.
\end{cases}
\end{equation}
The transformation matrix N maps the distribution functions for the temperature $g_i$ to the
corresponding moments $n_i$, \ie
\begin{equation}
\bm{n}= \mathbf{N} \bm{g} .
\end{equation}
The transformation matrix \(\mathbf{N}\) and its inverse matrix $\mathbf{N}^{-1}$ are shown in equations~\eqref{ma:ma1} and \eqref{ma:ma2}. 
There are some differences in the order of the columns with respect to what is specified in the reference \cite{liu2015double}. 
This is due to the different sequence used in numbering the velocity directions.
\begin{align}
    \mathbf{N} 
    & =
    \begin{pmatrix} 
       \langle 1| \\
       \langle e_x| \\
       \langle e_y| \\
       \langle 5e^2 - 4| \\
       \langle e_x^2-e_y^2|
    \end{pmatrix}
    =
    \begin{bmatrix} 
     1 & 1 & 1 & 1 & 1 \\
     0  &-1 & 0 & 1 & 0 \\
     0 & 0 & -1 & 0 & 1 \\
     -4 & 1 & 1 & 1 & 1 \\
     0 & 1 & -1 & 1 & -1
    \end{bmatrix} \label{ma:ma1} ~,\\
    \mathbf{N}^{-1} 
    & =
    \begin{pmatrix}
     \frac{1}{5} & 0 & 0 & -\frac{1}{5}  & 0\\
     \frac{1}{5}   &-\frac{1}{2} & 0 & \frac{1}{20} & \frac{1}{4} \\
     \frac{1}{5}& 0 & -\frac{1}{2} & \frac{1}{20}  & -\frac{1}{4} \\
     \frac{1}{5} & \frac{1}{2}& 0 & \frac{1}{20}  & -\frac{1}{4} \\
     \frac{1}{5}  & 0 & \frac{1}{2}  & \frac{1}{20}  & -\frac{1}{4}
    \end{pmatrix} \label{ma:ma2} ~.
\end{align}
The equilibrium moments $\bm{n}^{eq}$ are defined as
\begin{equation}
    \bm{n}^{eq} = \left( T, u_{x} T, u_{y} T, \varpi T , 0 \right)^T,
\end{equation}
where $\varpi$ is a constant of the \(D2Q5\) model, which we set to \(-2\). 
The diagonal relaxation matrix $\mathbf{\theta}$ is defined by
\begin{equation}
\mathbf{\theta} 
= 
\mathrm{diag} \left( 0, \zeta_a, \zeta_a, \zeta_e, \zeta_{\nu} \right). 
\end{equation}
The first relaxation rate, corresponding to the temperature, is set to zero for simplicity, since the first moment is conserved. 
The relaxation rates $\zeta_e$ and $\zeta_{\nu}$ are set to \(1.5\), whereas the relaxation rates $\zeta_a$ are functions of the thermal diffusivity $\alpha$ in \eqref{eq:tD}. 
The speed of sound of the \(D2Q5\) model is $c_s^2 = 0.2$.
\begin{align}
\alpha & = c_s^2 \left( \tau_a -\frac{1}{2} \right) = \frac{1}{5} \left( \tau_a - \frac{1}{2} \right)  \\
\Rightarrow 
\zeta_a & = \frac{1}{\tau_a} = \frac{1}{5\alpha + \frac{1}{2}} ~.
\label{eq:tD}
\end{align}

\subsubsection{Particle Flows as advection--diffusion Problem}\label{sec:particleADE}

The quantity $c$ in the advection--diffusion equation can be considered as a particle density, thereby giving a continuous ansatz for simulating particle flows. 
To solve for the particle distribution, an additional lattice is required with an appropriate descriptor and dynamics, which are only implemented for the 3D case.
\begin{lstlisting}[language=myc++,caption=Advection diffusion descriptor for particle flows, label=lst:particleADEdescriptor]
using ADDESCRIPTOR = descriptors::ParticleAdvectionDiffusionD3Q7Descriptor;
\end{lstlisting}
The descriptor in Listing~\ref{lst:particleADEdescriptor} allocates additional memory since for the computation of the particle velocity the velocity of the last time step has to be stored as well. These calculations are also non-local, therefore the communication of the additional data has to be ensured by an additional object, which is constructed according to line~1 of Listing~\ref{lst:ADEsExternal} and communicates the data by a function as shown in line~2, which has to be called in the time loop. 
\begin{lstlisting}[language=myc++,caption=SuperExternal3D object for the communication of additional data, label=lst:ADEsExternal]
SuperExternal3D<T, ADDESCRIPTOR, descriptors::VELOCITY> sExternal( superGeometry, sLatticeADE, sLatticeAD.getOverlap() );
sExternal.communicate();
\end{lstlisting}
Although the same unit converter can be used for the advection--diffusion lattice, another relaxation parameter has to be handed to the dynamics, as shown in Listing~\ref{lst:particleADEdynamics}. 
Additionally, some of the boundary conditions have to take the diffusion coefficient into account. 
Therefore a new $\omega_{ADE}$ is computed by
\begin{equation}
\omega_{ADE} = \left(4D \frac{U_{\mathrm{LB}}}{L_{\mathrm{LB}} U_C}+0.5\right)^{-1}. ~
\end{equation}
with characteristic lattice velocity $U_{\mathrm{LB}}$, characteristic velocity $U_C$, lattice length $L_{\mathrm{LB}}$, as well as the desired diffusion coefficient $D$. 
\begin{lstlisting}[language=myc++,caption=Dynamics for the simulation of particle flows via advection--diffusion equations, label=lst:particleADEdynamics]
ParticleAdvectionDiffusionBGKdynamics<T, ADDESCRIPTOR> bulkDynamicsAD ( omegaAD, instances::getBulkMomenta<T, ADDESCRIPTOR>() );
\end{lstlisting}
Applying the advection--diffusion equation to particle flow problems requires a new dynamics due to the handling of the particle velocity by the coupling processor of the two lattices, which differs for reasons of efficiency. When constructing the coupling post-processor as shown in Listing~\ref{lst:particleADEcoupling}, forces acting on the particle can be added like the Stokes drag force as shown in line~ 2 and 3 of Listing~\ref{lst:particleADEcoupling}. The implementation of new forces is straight forward, since only a new class which provides a function \texttt{applyForce(...)}, computing the force in a cell, needs to be written analogously to the existing \class{advDiffDragForce3D}. Finally the lattices are linked by line~4 of Listing~\ref{lst:particleADEcoupling}, which needs to be applied to the Navier--Stokes lattice for reasons of accessibility.
\begin{lstlisting}[language=myc++,caption=Coupling of an advection--diffusion and a Navier--Stokes lattice for particle flow simulations, label=lst:particleADEcoupling]
AdvectionDiffusionParticleCouplingGenerator3D<T,NSDESCRIPTOR> coupling( ADDESCRIPTOR::index<descriptors::VELOCITY>());
advDiffDragForce3D<T, NSDESCRIPTOR> dragForce( converter,radius,partRho );
coupling.addForce( dragForce );
sLatticeNS.addLatticeCoupling( superGeometry, 1, coupling, sLatticeAD );
\end{lstlisting}
For the boundary conditions the same basic objects as for the advection--diffusion equation can be used, however there is an additional boundary condition shown on Listing~\ref{lst:particleADEboundary} which has to be applied at all boundaries to ensure correctness of the finite differences scheme used to compute the particle velocity.
Further information as well as results can be found in Trunk \textit{et al.}~\cite{trunk:16} as well as in the examples section.
\begin{lstlisting}[language=myc++,caption=Example of a boundary condition for the particle velocity for particle flow simulations, label=lst:particleADEboundary]
setExtFieldBoundary<T,ADDESCRIPTOR,descriptors::VELOCITY,descriptors::VELOCITY2>(
      sLatticeAD, superGeometry.getMaterialIndicator({2, 3, 4, 5, 6}));
\end{lstlisting}


\section{Particles}\label{sec:dpm}

The present section summarizes OpenLB’s functionality regarding the consideration of discrete particles in a Lagrangian framework.
This includes both \textit{sub-grid} particles assuming spherical shapes and \textit{surface resolved} particles with arbitrary shapes, which can be handled by a common particle framework.
As the framework follows advances in the data concept of the lattice (cf.\ Section~\ref{sec:modelsAndCoreDataStructure}), it provides a dimension agnostic, flexible and easily extendable implementation.
While abstract template meta functionality characterizes the data handling level, accessible high-level user-functions are provided for \eg particle creation or coupling handling.
In order to guarantee support for previously developed applications, the 3D-only sub-grid particle framework from the previous releases is included as \textit{sub-grid (legacy) framework} (cf.\ Section~\ref{sec:subGridLegacy}) as well.

To get a good overview of the particle framework, the code of examples \texttt{settlingCube3d} (Section~\ref{sec:settlingCube3d}) and \texttt{bifurcation3d} (Section~\ref{sec:bifurcation3d}) is reviewed, focusing on the simulation of particles. 
The example \texttt{settlingCube3d} examines the settling of a cubical silica particle under the influence of gravity in surrounding water.
It starts with integrating some libraries and namespaces, followed by the definition of different types (Listing~\ref{lst:SetCube3D_1}), \eg the \textit{descriptor} and the \textit{particle type}. 
Afterwards, some variables are set to a concrete value, used in the fluid and particle calculation (Listing~\ref{lst:SetCube3d_2}). 
Particle settings include all the data to solve the equations of motion, such as the particle’s starting position and density.

\begin{lstlisting}[language=myc++,mathescape=true, caption=Following the new particle system with following the example settlingCube3d, label=lst:SetCube3D_1]
#include "olb3D.h"
#include "olb3D.hh"

using namespace olb;
using namespace olb::descriptors;
using namespace olb::graphics;
using namespace olb::util;
using namespace olb::particles;
using namespace olb::particles::dynamics;

using T = FLOATING_POINT_TYPE;
typedef D3Q19<POROSITY,VELOCITY_NUMERATOR,VELOCITY_DENOMINATOR> DESCRIPTOR;

//Define lattice type
typedef PorousParticleD3Q19Descriptor DESCRIPTOR;

//Define particleType
typedef ResolvedParticle3D PARTICLETYPE;
\end{lstlisting}

\begin{lstlisting}[language=myc++,mathescape=true, caption=Particle settings in example settlingCube3d, label=lst:SetCube3d_2]
//Particle Settings
T centerX = lengthX*.5;
T centerY = lengthY*.5;
T centerZ = lengthZ*.9;
T const cubeDensity = 2500;
T const cubeEdgeLength = 0.0025;
Vector<T,3> cubeCenter = {centerX,centerY,centerZ};
Vector<T,3> cubeOrientation = {0.,15.,0.};
Vector<T,3> cubeVelocity = {0.,0.,0.};
Vector<T,3> externalAcceleration = {.0, .0, -T(9.81) * (T(1) - physDensity / cubeDensity)};


// Characteristic Quantities
T const charPhysLength = lengthX;
T const charPhysVelocity = 0.15;    // Assumed maximal velocity
\end{lstlisting}

Like other simulations, particle flow simulations need basic, non particle-specific functions like \texttt{prepareGeometry} or \texttt{prepareLattice}.
After initializing those functions, the main function starts. 
The main section begins with initialization of physical units in the unit converter, which is explained in the Q\&A in Section~\ref{sec:QandA}.
The unit converter is followed by the preparation of the geometry using the \texttt{prepareGeometry}-function and afterwards the \texttt{prepareLattice}-function.
After those general simulation functions, the particle simulation starts.
First, the \class{ParticleSystem} (Listing~\ref{lst:SetCube3d_3}, explained in Section~\ref{subsec:particleSystem}) is called followed by the calculation of the particle quantities like a smoothing factor and the extent of the particles.
After those calculations, the particles are created. 
In the following lines, dynamics are assigned to the particles.

\begin{lstlisting}[language=myc++,mathescape=true, caption=Creation of particles and assigning dynamics, label=lst:SetCube3d_3]
  // Create ParticleSystem
  ParticleSystem<T,PARTICLETYPE> particleSystem;

  //Create particle manager handling coupling, gravity and particle dynamics
  ParticleManager<T,DESCRIPTOR,PARTICLETYPE> particleManager(
    particleSystem, superGeometry, sLattice, converter, externalAcceleration);

  // Create and assign resolved particle dynamics
  particleSystem.defineDynamics<
    VerletParticleDynamics<T,PARTICLETYPE>>();

  // Calculate particle quantities
  T epsilon = 0.5*converter.getConversionFactorLength();
  Vector<T,3> cubeExtend( cubeEdgeLength );

  // Create Particle 1
  creators::addResolvedCuboid3D( particleSystem, cubeCenter,
                                 cubeExtend, epsilon, cubeDensity, cubeOrientation );

  // Create Particle 2
  cubeCenter = {centerX,lengthY*T(0.51),lengthZ*T(.7)};
  cubeOrientation = {0.,0.,15.};
  creators::addResolvedCuboid3D( particleSystem, cubeCenter,
                                 cubeExtend, epsilon, cubeDensity, cubeOrientation );

  // Check ParticleSystem
  particleSystem.checkForErrors();

\end{lstlisting}
Before the main loop starts, Listing~\ref{lst:SetCube3d_4}, we create a timer, Listing~\ref{lst:SetCube3d_4} and set initial values to the distribution functions by calling \class{setBoundaryValues}.
After this, the following is processed at every time step.
The fluid's influence on the particles is calculated by evaluating hydrodynamic forces acting on the particle surface.
Afterwards, an external acceleration, \eg gravity, is applied onto the particles (Listing~\ref{lst:SetCube3d_4}) and the equations of motion are solved for each one. The back coupling from the particles to the fluid follows afterwards.
Finally, the main loop ends with the \class{getResults}-function, which prints the results to the console and writes VTK data for post-processing with ParaView (Section~\ref{sec:paraview}) at previously defined time intervals.

\begin{lstlisting}[language=myc++,mathescape=true, caption=Main Loop with Timer, label=lst:SetCube3d_4]
 /// === 4th Step: Main Loop with Timer ===
  Timer<T> timer(converter.getLatticeTime(maxPhysT), superGeometry.getStatistics().getNvoxel());
  timer.start();


  /// === 5th Step: Definition of Initial and Boundary Conditions ===
  setBoundaryValues(sLattice, converter, 0, superGeometry);

  clout << "MaxIT: " << converter.getLatticeTime(maxPhysT) << std::endl;

  for (std::size_t iT = 0; iT < converter.getLatticeTime(maxPhysT)+10; ++iT) {

    // Execute particle manager
    particleManager.execute<
      couple_lattice_to_particles<T,DESCRIPTOR,PARTICLETYPE>,
      apply_gravity<T,PARTICLETYPE>,
      process_dynamics<T,PARTICLETYPE>,
      couple_particles_to_lattice<T,DESCRIPTOR,PARTICLETYPE>
    >();

    // Get Results
    getResults(sLattice, converter, iT, superGeometry, timer, particleSystem );

    // Collide and stream
    sLattice.collideAndStream();
  }

  timer.stop();
  timer.printSummary();
\end{lstlisting}

As we followed the example for particle simulation \path{settlingCube3d}, some functions necessary for the simulations were introduced.
Therefore, in the next chapters, the individual parts of the framework are examined.


\subsection{Class ParticleSystem}
\label{subsec:particleSystem}

The \class{ParticleSystem} stores all data concerning the particles in containers.
Therefore, the class is used multiple times in a particle simulation.
First, the \class{ParticleSystem} is created according to the desired \class{PARTICLETYPE} (Listing~\ref{lst:SetCube3d_3}).
However, the container of particles is empty.
Therefore, we add two particles to it using creator functions and add dynamics via the \class{ParticleSystem}.
Additionally, it is utilized in the \class{ParticleManager} (cf.\ Section~\ref{subsec:particleManager}) to access the particles and perform predefined operations on them.

One focus of the new particle system is the separation of data and operations according to the lattice framework (cf.\ Section~\ref{sec:modelsAndCoreDataStructure}).
Therefore, only the data is stored in the \class{ParticleSystem}. For the operations, it is non-relevant and only used to store data of the calculations.

\subsection{Class SuperParticleSystem}
The example \path{bifurcation3d} (Section~\ref{sec:bifurcation3d}) makes use of OpenLB's domain decomposition approach. 
In order to use this, a \class{SuperParticleSystem} has to be created by passing the \class{SuperGeometry}, which holds all information regarding the lattice decomposition:

\begin{lstlisting}[language=myc++,mathescape=true, caption=Initialization of a SuperParticleSystem, label=lst:createSuperParticleSystem] 
  SuperParticleSystem<T,PARTICLETYPE> superParticleSystem(superGeometry);

\end{lstlisting}

\subsection{Class ParticleManager}
\label{subsec:particleManager}
The \class{ParticleManager} can be used to encapsulate relevant reoccurring particle tasks as \eg the particle-lattice-coupling.
After its initial instantiation by providing the access to relevant particle, lattice and set-up specific data, its \class{execute()} method can be called with the respective tasks specified as template arguments in the desired order.
The individual tasks (included in \path{particleTasks.h}) provide an \class{execute()} method as well and a parameter set specifying the coupling type and the potential embedding into a loop over all available particles.
The \class{ParticleManager} also takes care of combining respective tasks into a single particle loop.

When using a domain decomposition, the particle core distribution has to be updated in every time step. 
The use of the \class{ParticleManager} in the \path{bifurcation3d} example (Section~\ref{sec:bifurcation3d}) respectively looks as follows:

\begin{lstlisting}[language=myc++,mathescape=true, caption=Initialization of a ParticleManager, label=lst:createParticleManager] 
ParticleManager<T,DESCRIPTOR,PARTICLETYPE> particleManager(
    superParticleSystem, superGeometry, superLattice, converter);
\end{lstlisting}

\begin{lstlisting}[language=myc++,mathescape=true, caption=Execution of the ParticleManager, label=lst:executeParticleManager] 
particleManager.execute<
      couple_lattice_to_particles<T,DESCRIPTOR,PARTICLETYPE>,
      process_dynamics<T,PARTICLETYPE>,
      update_particle_core_distribution<T,PARTICLETYPE>
    >();
\end{lstlisting}

\subsection{Resolved Lattice Interaction}
In the directory \path{resolved}, all surface resolved specific functionality is bundled. The \path{blockLatticeInteraction.h} (only header file) and \path{blockLatticeInteraction.hh} files consist of five functions.
All of those functions are needed to calculate and check the position of the particles inside the geometry. For example, the \class{checkSmoothIndicatorOutOfGeometry}-function checks if every part of the particle is inside the cell barrier. If the particle is partly outside of the geometry, the position needs to be changed. Another important functions is the \class{setBlockParticleField}, where a all cells, which are inside the particle are set as a particle field.
Similar to the \path{blockLatticeInteraction.hh} also the \path{superLatticeInteraction.hh} exist. In this file the \class{setBlockParticleField} gets converted to the lattice structure with the function \class{setSuperParticleField}.
The file \path{momentumExchangeForce.h} provides functions to calculate hydrodynamic forces on the particle's surface via an adapted momentum exchange algorithm.
The file (\path{smoothIndicatorInteraction.h}) is needed for the simulation of the area directly at the surface of the particles.

\subsection{Particle Descriptors}
\label{descriptors}
The first file is the \path{particleDescriptorAlias.h} file, in which the alias' are given to different types of \class{PARTICLETYPE}s. In the example \path{settlingCube3d}, right in the beginning, the \class{PARTICLETYPE} \class{ResolvedParticle3D} is chosen.
After the choice of this alias, the dynamics and other main properties are set. Other possible particle types that can be chosen are \class{ResolvedParticle2D} or \class{ResolvedSphere3D}.

For the \path{bifurcation3d} example (Section~\ref{sec:bifurcation3d}), a respective descriptor is chosen:
\begin{lstlisting}[language=myc++,mathescape=true, caption=Particle type and descriptor] 
typedef D3Q19<> DESCRIPTOR;
typedef SubgridParticle3D PARTICLETYPE;
\end{lstlisting}

\subsection{Particle Dynamics}
\label{dynamics}
Another important part of the particle system are the dynamics. The files are used to define those properties for the chosen particle type. For example, in the \path{particleDynamics.h} and \path{particleDynamics.hh} all dynamic functions for the particle type are implemented.
Therefore, all information for calculation of dynamic values can be found here \eg acceleration or angular acceleration.
Those functions get called in the main part of simulations, \eg in the example \path{settlingCube3d} (\class{VerletParticleDynamics}), in Listing~\ref{lst:SetCube3d_3}, as the dynamics are assigned to the particle type.

In the example \path{bifurcation3d} (Section~\ref{sec:bifurcation3d}), different capture methods can be used by choosing the respective setting beforehand:
\begin{lstlisting}[language=myc++,mathescape=true, caption=Enum ParticleDynamicsSetup and example initialization , label=lst:SetCube3d_4b] 
//Set capture method:
// materialCapture: based on material number
// wallCapture:     based on more accurate stl description
typedef enum {materialCapture, wallCapture} ParticleDynamicsSetup;
const ParticleDynamicsSetup particleDynamicsSetup = wallCapture;
\end{lstlisting}

If the \texttt{wallCapture} is chosen, a \class{SolidBoundary} object has to be created and passed to the respective dynamics:

\begin{lstlisting}[language=myc++,mathescape=true, caption=Solid wall creation, label=lst:SetCube3d_4c] 
STLreader<T> stlReader( "../bifurcation3d.stl", converter.getConversionFactorLength() );
  IndicatorLayer3D<T> extendedDomain( stlReader,        converter.getConversionFactorLength() );
  // Create solid wall
  const unsigned latticeMaterial = 2; //Material number of wall
  const unsigned contactMaterial = 0; //Material identifier (only relevant for contact model)
  SolidBoundary<T,3> wall( std::make_unique<IndicInverse<T,3>>(stlReader),
                           latticeMaterial, contactMaterial );
\end{lstlisting}

When using the \texttt{materialCapture} instead, a \class{MaterialIndicator} is necessary to identify material numbers that initialize the capture treatment.

\begin{lstlisting}[language=myc++,mathescape=true, caption=Initialization of a SuperIndicatorMaterial, label=lst:SetCube3d_4d] 
std::vector<int> materials {2,4,5};
  SuperIndicatorMaterial<T,3> materialIndicator (superGeometry, materials);
\end{lstlisting}

Both \class{SolidBoundary} and \class{MaterialIndicator} are then used in the function \class{prepareParticles} to define the chosen dynamics:

\begin{lstlisting}[language=myc++,mathescape=true, caption=Creation of verlet dynamics, label=lst:SetCube3d_4e] 
if (particleDynamicsSetup==wallCapture){
    //Create verlet dynamics with material aware wall capture
    superParticleSystem.defineDynamics<
      VerletParticleDynamicsMaterialAwareWallCapture<T,PARTICLETYPE>>(
        wall, materialIndicator);
  } else {
    //Create verlet dynamics with material capture
    superParticleSystem.defineDynamics<
      VerletParticleDynamicsMaterialCapture<T,PARTICLETYPE>>(materialIndicator);
  }
  \end{lstlisting}

\subsection{Particle Functions}
In the functions-directory, additional free functions are defined. These functions are callable anywhere in the code.
The first set of files including \path{particleCreatorFunctions.h}, \path{particleCreatorFunctions2D.h}, \path{particleCreatorFunctions3D.h} and \path{particleCreatorHelperFunctions.h} concentrate on functions concerning the creation of particles with different types of surface structures. These functions are therefore called first to create particles in the desired shape. In the example \path{settlingCube3d}, the function \class{addResolvedCuboid} is called and creates a particle in the shape of a cuboid. Also other geometries, like circles in 2D or cylinders in 3D, can be created. All of those functions are implemented in these files.

The file (\path{particleMotionfunctions.h}) concentrates on the main algorithms for solving the equations of motion.
Two functions exist, using different integration-types, velocity Verlet algorithm (\class{velocityVerletIntegration}) or Euler-Integration (\class{eulerIntegration}).
The former function is used in the \class{VerletParticleDynamics} class (Section~\ref{dynamics}) and are therefore called in the main part of the example.
Often two functions for the same calculation exist as they need to match the dimension of a problem. Those are differentiated by partial template specialization.

The \path{particleDynamicsFunctions.h} also contains other important functions to simulate particle flows. Tasks are included in the \path{particleTasks.h} like \eg the \class{couple\_lattice\_to\_particles} or \class{couple\_resolved\_particles\_to\_lattice}.
Both functions are used in the main loop of the example, to realize a two-way-coupling.

To sum up, many of the most important functions for the simulations of the particle flow are implemented in the \path{ParticleDynamicsFunctions.h}.
Other functions, \eg concerning the calculation of rotation of the particle body, are implemented in the \path{bodyMotionfunctions.h}.
The \path{particleIoFunctions.h} contains functions to get the output of the calculation to the console.
It consist of two important functions (\class{printResolvedParticleInfo} an \class{printResovedParticleInfoSimple}), which are used in the \class{getResults}-function.
The \class{getResults}-function is called at the end of the main part of every simulation.

\subsection{Discrete Contact Model for Surface Resolved Particles}
\label{sec:particle:discrete_contact}

In order to simulate particulate flows accurately, it is often necessary to incorporate a contact model. The here used discrete contact model~\cite{marquardt:23} allows for the treatment of particle-particle and particle-wall interactions, enabling the calculation of contact forces and their application to the particles. The discrete contact model consists of several steps that are integrated into the general algorithm. Let's discuss each step in detail:

\begin{itemize}
\item \textit{Rough contact detection during coupling:}
During the coupling stage, where particle information is transferred to the fluid lattice, a rough contact detection mechanism is employed. This step identifies potential contact regions between particles and the fluid. It determines potential contacts by identifying particles that couple to the same cell.
\item \textit{Communication of found contacts:}
Once the potential contact regions are identified, the information regarding the found contact is communicated across all processes.
\item \textit{Correction of contact bounding box:}
To improve the accuracy of the contact treatment, the contact bounding boxes are refined based on the information obtained during the communication step. This correction step helps in precisely defining the contact regions, ensuring that the subsequent calculations consider the actual contacts.
\item \textit{Determination of contact properties:}
With the refined contact bounding boxes, the discrete contact model determines various contact properties. These properties include the contact volume, contact point, contact normal and other relevant parameters.
\item \textit{Calculation of contact force and application to particles:}
Using the contact properties, the contact force is calculated from the parameters determined before and applied to the particles so that it's available when solving the equations of motion.
\item \textit{Removal of empty contact objects (optional):}
After the contact forces have been determined and applied, empty contact objects, which no longer represent an existing contact, may be removed. This step helps in optimizing the computational efficiency by eliminating unnecessary iterations.
\end{itemize}
\par
The usage within OpenLB is exemplified by \textit{dkt2d}. First, we set types for the particle-particle and particle-wall interactions, which define how to the contact is treated. This is represented in Listing~\ref{lst:contact_types}).

\begin{lstlisting}[language=myc++,mathescape=true, caption=Contact types, label=lst:contact_types]
typedef ParticleContactArbitraryFromOverlapVolume<T, DESCRIPTOR::d, true> PARTICLECONTACTTYPE;
typedef WallContactArbitraryFromOverlapVolume<T, DESCRIPTOR::d, true> WALLCONTACTTYPE;
\end{lstlisting}

Additionally, we define the walls for the contact treatment, as shown in Listing~\ref{lst:solid_boundary}). Here, we create a \class{SolidBoundary} from an indicator and specify the minimal and maximal coordinates as well as the material number that represents the wall on he lattice and the material identifier, which defines the walls mechanical properties.

\begin{lstlisting}[language=myc++,mathescape=true, caption=Solid boundaries, label=lst:solid_boundary]
std::vector<SolidBoundary<T, DESCRIPTOR::d>> solidBoundaries;
solidBoundaries.push_back(  SolidBoundary<T, DESCRIPTOR::d>(
    std::make_unique<IndicInverse<T, DESCRIPTOR::d>>(
      cuboid, 
      cuboid.getMin() - 5 * converter.getPhysDeltaX(),
      cuboid.getMax() + 5 * converter.getPhysDeltaX()),
      wallLatticeMaterialNumber, 
      wallContactMaterial));
\end{lstlisting}

Similarly, we set a number that relates the particles to mechanical properties, see Listing~\ref{lst:solid_boundaryb}.
\begin{lstlisting}[language=myc++,mathescape=true, caption=Particle material number, label=lst:solid_boundaryb]
for (std::size_t iP = 0; iP < particleSystem.size(); ++iP) {
    auto particle = particleSystem.get(iP);
    setContactMaterial(particle, particleContactMaterial);
}
\end{lstlisting}

To store contact objects later on, we create an empty \class{ContactContainer} as shown in Listing~\ref{lst:contact_container}).

\begin{lstlisting}[language=myc++,mathescape=true, caption=Contact container, label=lst:contact_container]
ContactContainer<T, PARTICLECONTACTTYPE, WALLCONTACTTYPE> contactContainer;
\end{lstlisting}

By creating a lookup table \class{ContactProperties} that contains constant parameters which solely depend on the material combination in Listing~\ref{lst:contact_properties}), we save computational effort. Setting properties must be done for each material combination separately.

\begin{lstlisting}[language=myc++,mathescape=true, caption=Contact properties, label=lst:contact_properties]
ContactProperties<T, 1> contactProperties;
contactProperties.set(particleContactMaterial,
                    wallContactMaterial,
                        evalEffectiveYoungModulus(
                            youngsModulusParticle, 
                            youngsModulusWall,
                            poissonRatioParticle, 
                            poissonRatioWall),
                    coefficientOfRestitution,
                    coefficientKineticFriction,
                    coefficientStaticFriction);
\end{lstlisting}

Finally, we process the contacts as shown in Listing~\ref{lst:process_contacts}).

\begin{lstlisting}[language=myc++,mathescape=true, caption=Contact processing, label=lst:process_contacts]
processContacts<T, PARTICLETYPE, PARTICLECONTACTTYPE, WALLCONTACTTYPE, ContactProperties<T, 1>>(
    particleSystem, solidBoundaries, 
    contactContainer, contactProperties,
    superGeometry, contactBoxResolutionPerDirection);
\end{lstlisting}

\subsection{Sub-grid Legacy Framework}
\label{sec:subGridLegacy}
In this Section the use of Lagrangian particles with the legacy framework is shown.
Due to similar naming of classes and functions in the new common framework, it is worth noting that all terms are primarily referring to the naming convention used in the legacy framework itself and should not be mixed up with those of the new one.

Similar to the \class{BlockLattice} and \class{SuperLattice} structure, a \class{ParticleSystem3D} and \class{SuperParticleSystem3D} structure exists. In line~2 of Listing~\ref{lst:sPartSystem} the \class{SuperParticleSystem3D} is instantiated, taking a \class{SuperGeometry} as a parameter. In line~4 the \class{SuperParticleSysVtuWriter} is instantiated. It takes the \class{SuperParticleSystem3D}, a filename as \class{string}, and the wanted particle properties as arguments. Calling the function \class{SuperParticleSysVtuWriter.write(int timestep)} creates \path{.vtu} files of the particles positions for the given timestep. These files can be visualized with ParaView.

Line~10 of the listing instantiates an interpolation functor for the fluids velocity, which is used in line~13 during the instantiation of \class{StokesDragForce3D}. Particles need boundary conditions as well. In the listing, the simplest possible material boundary is presented. If a particle moves into a lattice node with material number $2, 4$ or $5$ its velocity is set to $0$ and it is neglected during further computations, its state of activity is set to false. This \class{MaterialBoundary3D} is instantiated in line~16. In lines 18 and 19 the force and boundary condition are added to and stored in the respective lists in the \class{SuperParticleSystem3D}.

The actual number crunching is then performed in line~25 which is positioned in the main loop of the program. The \class{supParticleSystem.simulate(T timeStep);}  function integrates the particle trajectories by \class{timeStep}. Therefore all stored particle forces are computed and summed up. The particles are moved one step according to Newton's laws. Then all stored particle boundary conditions are applied. Parallelization of the particles is achieved automatically.

Results of this simulation are published in \textcite{henn:16}.

\begin{lstlisting}[language=myc++,mathescape=true, caption=Usage of class SuperParticleSystem3D, label=lst:sPartSystem]
  // SuperParticleSystems3D
  SuperParticleSystem3D<T,PARTICLE> supParticleSystem(superGeometry);
  // define which properties are to be written in output data
  SuperParticleSysVtuWriter<T,PARTICLE> supParticleWriter(supParticleSystem, "particles",
    SuperParticleSysVtuWriter<T,PARTICLE>::particleProperties::velocity |
    SuperParticleSysVtuWriter<T,PARTICLE>::particleProperties::mass |
    SuperParticleSysVtuWriter<T,PARTICLE>::particleProperties::radius |
    SuperParticleSysVtuWriter<T,PARTICLE>::particleProperties::active);

  SuperLatticeInterpPhysVelocity3D<T,DESCRIPTOR> getVel(sLattice, converter);

  auto stokesDragForce = make_shared<StokesDragForce3D<T,PARTICLE,DESCRIPTOR>> (getVel, converter);

  // material numbers where particles should be reflected
  std::set<int> boundMaterial = { 2, 4, 5};
  auto materialBoundary = make_shared<MaterialBoundary3D<T, PARTICLE>> (superGeometry, boundMaterial);

  supParticleSystem.addForce(stokesDragForce);
  supParticleSystem.addBoundary(materialBoundary);
  supParticleSystem.setOverlap(2. * converter.getPhysDeltaX());

\* ... *\

main loop {
   supParticleSystem.simulate(converter.getPhysDeltaT());
}
\end{lstlisting}

\subsubsection{Interpolation of Fluid Velocity}
\label{subsec:interpolation}
As the particle position $\bm{X}: I \to \Omega$ moves in the continuous domain $\Omega$ and information on the fluid velocity can only be computed on lattice nodes $\bm{x}_i \in \Omega_h$, interpolation of the fluid velocity is necessary every time fluid-particle forces are computed. Let $\bm{u}^F_i = \bm{u}^F(\bm{x}_i)$ be the computed solution of the Navier--Stokes Equation at lattice nodes $\bm{x}_i$. Let $p\in P_n$ be the interpolating polynomial of order $n$ with $p(\bm{x}_i) = \bm{u}^F_i$ and $\overline{\left(\bm{x}_0, \ldots \bm{x}_n\right)}$ the smallest interval containing all points in the brackets. Furthermore, let $C^{n}\left[\bm{a}, \bm{b}\right]$ be the vector space of continuous functions that have continuous first $n$ derivatives in $\left[\bm{a}, \bm{b}\right]$. Then the interpolation error of the polynomial interpolation is stated by the following theorem.
\begin{theorem}[Interpolation error]
Let $\bm{u}\in C^{n+1}\left[\bm{a}, \bm{b}\right]$, $\bm{a}, \bm{b} \in \Omega$. Then for every $\bm{x} \in \left[\bm{a}, \bm{b}\right]$ there exists one $\widehat{\bm{x}}\in \overline{\left(\bm{x}_0, \ldots \bm{x}_n, \bm{x}\right)}$, such that
       \begin{equation}
         \bm{u}^F(\bm{x}) - p_n(\bm{x}) = \frac{d_{\bm{x}}^{n+1}\bm{u}^{F}(\widehat{\bm{x}})}{(n+1)!} \prod^n_{j=0}(\bm{x}-\bm{x}_j)
       \end{equation}
       holds.
     \end{theorem}
    \begin{proof}
      See \textcite[Satz 2.3]{rannacher:06}.
    \end{proof}
    Using linear ($n=1$) interpolation for the fluid velocity between two neighboring lattice nodes $\bm{a}=\bm{x}_0\in\Omega_h, \bm{b}=\bm{x}_1\in\Omega_h, \|\bm{x}_1-\bm{x}_0\|_2 = h$ clearly the following holds
    \begin{align}
      \label{eq:linInter}
      f( \bm{x} ) - p_1( \bm{x} ) &= \frac12 d_{\bm{x}}^2 \bm{u}^{F}( \widehat{\bm{x}} ) (\bm{x}-\bm{x}_0)(\bm{x}-\bm{x}_1) \\
      &\leq \frac12 d_{\bm{x}}^2 \bm{u}^{F}(\widehat{\bm{x}}) h^2~
    \end{align}
    and the approximation error of the linear interpolation is of order $\bigO{h^2}$. In the following we give reason why this order of interpolation is sufficient.

Lets assume there exists an ideal error law of the form
\[
\| \bm{u}_i^F - \bm{u}_i^{F^*}\|_{L^2(\Omega_h)} = c h^{\alpha}~,
\]
for the discrete solution $\bm{u}^F_h$ obtained by an LBM with lattice spacing $h$ and the analytic solution $\bm{u}^{F^*}$. Then $\alpha \in \Rplus$ is the to be determined order of convergence. We further define the relative error
\[
\Err_h = \frac{\|\bm{u}^F_h - \bm{u}^{F^*}\|_{L^2(\Omega_h)}}{\|\bm{u}^{F^*}\|_{L^2(\Omega_h)}}~.
\]
The ratio of the error laws of two distinct lattice spacings $h_i$ and $h_j$ forms the EOC as
\begin{equation}
\label{eq:eoc}
\EOC_{i,j} = \frac{\ln(\Err_{h_i} / \Err_{h_j})}{\ln(h_i / h_j)}~.
\end{equation}
With this \textcite[Chapter 2.3]{krause:10b} determines an of $\EOC\approx 2$ for the discrete solution towards the analytic solution of a stationary flow in the unit cube governed by the incompressible NSE. Therefore the order of converge of the fluid velocity obtained by an LBM can be assumed to be $\bigO{h^2}$. This conclusion is backed up by the theoretical results obtained by \cite{caiazzo:13}.
This leads to the assumption that, each interpolation scheme of higher order than $2$ would not be exhausted as the error of the incoming data is too large.

\begin{figure}
  \centering
    \includegraphics[scale=1]{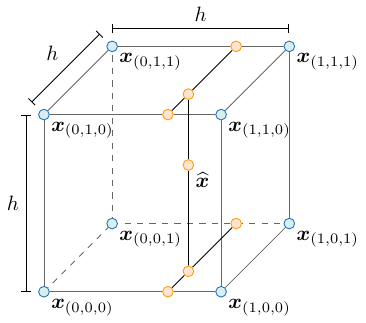}
  \caption{Trilinear interpolation.}
  \label{fig:interpolation}
\end{figure}
The interpolation is implemented as a trilinear interpolation using the eight nodes surrounding the particle. Let the point of interpolation $\widehat{\bm{x}} \in[x_{(0,0,0)}, x_{(1,1,1)}]$ be in the cube spanned by the lattice nodes $\bm{x}_{(0,0,0)}$ and $\bm{x}_{(1,1,1)}$, see Figure~\ref{fig:interpolation} for an illustration. We will denote by
\begin{align}
  \bm{d} = (d_0,d_1,d_2)^T &= \widehat{\bm{x}}-\bm{x}_{(0,0,0)} ~,
\end{align}
the distance of the particle to the next smaller lattice node. The fluid velocities at the eight corners are named accordingly $\bm{u}_{(i,j,k)},\ i,j,k\in \{0,1\}$. The trilinear interpolation is executed by three consecutive linear interpolations in the three different space directions. First we interpolate along the $x$-axis
\begin{align}
  u_{(d,0,0)} &= u_{(0,0,0)}(h-d_0) + u_{(1,0,0)}d_0 ~,\\
  u_{(d,1,0)} &= u_{(0,1,0)}(h-d_0) + u_{(1,1,0)}d_0 ~,\\
  u_{(d,0,1)} &= u_{(0,0,1)}(h-d_0) + u_{(1,0,1)}d_0 ~,\\
  u_{(d,1,1)} &= u_{(0,1,1)}(h-d_0) + u_{(1,1,1)}d_0 ~,
\end{align}
followed by interpolation along the $y$-axis
\begin{align}
  u_{(d,d,0)} &= u_{(d,0,0)}(h-d_1) + u_{(d,1,0)}d_1 ~, \\
  u_{(d,d,1)} &= u_{(d,0,1)}(h-d_1) + u_{(d,1,1)}d_1 ~,
\end{align}
and finally in direction of the $z$-axis
\begin{align}
 \bm{u}(\widehat{\bm{x}}) = u_{(d,d,d)}  &= u_{(d,d,0)}(h-d_2) + u_{(d,d,1)}d_2 ~.
\end{align}

\subsubsection{Class \lstinline{SuperParticleSystem3D}}
\begin{sloppypar}
The implementation of the particle phase follows an hierarchical ansatz, similar to the \class{Cell} $\to$ \class{BlockLattice3D} $\to$ \class{SuperLattice} ansatz used for the implementation of the LBM. The equivalent classes in the context of Lagrangian particles are \class{Particle3D} $\to$ \class{ParticleSystem3D} $\to$ \class{SuperParticleSystem3D}. The class \class{Particle3D} allocates memory for the variables of one single particle, such as its position, velocity, mass, radius and the force acting on it. It also provides the function \class{bool getActive()}, which returns the active state of the particle. Active particles' positions are updated during the simulation, in contrast to non-active particles, which are only used for particle-particle interaction. The class \class{Particle3D} is intended to be inherited from, in order to provide additional properties, such as electric or magnetic charge. The particles in the domain of a specific \class{BlockLattice3D} are combined in the class \class{ParticleSystem3D}. Finally the class \class{SuperParticleSystem3D} combines all \class{ParticleSystem3D}s, and handles the transfer of particles between them.

The concept of the class \class{SuperParticleSystem3D} is to provide an easily adaptable framework for simulation of a large number of particles arranged in and interacting with a fluid. In this context \emph{easily adaptable} means that simulated forces and boundary conditions are implemented in a modular manner, such that they are easily exchangeable. Development of new forces and boundary conditions can be readily done by inheritance of provided base classes. Particle-particle interaction can be activated if necessary and deactivated to decrease simulation time. The contact detection algorithm is interchangeable. This section introduces the \class{SuperParticleSystem3D} and the mentioned properties in more detail.

The class \class{SuperParticleSystem3D} is initialized by a call to the constructor simultaneously on all PUs. \\

\noindent\begin{minipage}{\linewidth}
\begin{lstlisting}[style=intext]
SuperParticleSystem3D(CuboidGeometry3D<T>& cuboidGeometry, LoadBalancer<T>& loadBalancer, SuperGeometry<T,3>& superGeometry);
\end{lstlisting}
\end{minipage}
 During the construction each PU instantiates one \class{ParticleSystem3D} for each local cuboid. Subsequently for each \class{ParticleSystem3D} a list of the ranks of PUs holding neighboring cuboids is created.

Particles can be added to the \class{SuperParticleSystem3D} by a call to one of the \class{addParticle()} functions.\\

\noindent\begin{minipage}{\linewidth}
\begin{lstlisting}[style=intext]
/// Add a Particle to SuperParticleSystem
void addParticle(PARTICLETYPE<T> &p);
/// Add a number of identical Particles equally distributed in a given IndicatorF3D
void addParticle(IndicatorF3D<T>& ind, T mas, T rad, int no=1, std::vector<T> vel={0.,0.,0.});
/// Add a number of identical Particles equally distributed in a given Material Number
void addParticle(std::set<int>  material, T mas, T rad, int no=1, std::vector<T> vel={0.,0.,0.});
/// Add Particles form a File. Save using saveToFile(std::string name)
void addParticlesFromFile(std::string name, T mass, T radius);
\end{lstlisting}
\end{minipage}

Currently there are four implementations of this class. The first adds single predefined particles, the second and third add a given number of equally distributed particles of the same mass and radius in an area that can be defined by either a set of material numbers or an indicator function. The initial particle velocity can be set optionally. Finally particles can be added from an external file, containing their positions. In all cases the assignment to the correct \class{ParticleSystem3D} is carried out internally.

Particle forces and boundaries are implemented by the base classes \class{Force3D} and \class{Boundary3D}. \\

\noindent\begin{minipage}{\linewidth}
\begin{lstlisting}[style=intext]
template<typename T, template<typename U> class PARTICLETYPE>
class Force3D {
public:
  Force3D();
  virtual void applyForce(typename std::deque<PARTICLETYPE<T> >::iterator p, int pInt, ParticleSystem3D<T, PARTICLETYPE>& psSys)=0;
}
\end{lstlisting}
\end{minipage}

Both classes are intended to be derived from in order to implement force and boundary specializations. 
The key function in both classes are \class{applyForce()} and \class{applyBoundary()}, which are called during each timestep of the main LBM loop. \class{Force3D} and \class{Boundary3D} specialisations are added to the \class{SuperParticleSystem3D} by passing a pointer to a class instantiation via a call to the respective function.\\

\noindent\begin{minipage}{\linewidth}
\begin{lstlisting}[style=intext]
/// Add a force to system
void addForce(std::shared_ptr<Force3D<T, PARTICLETYPE> > f);
/// Add a boundary to system
void addBoundary(std::shared_ptr<Boundary3D<T, PARTICLETYPE> > b);
\end{lstlisting}
\end{minipage}

Both functions add the passed pointer to a list of forces and boundaries, which will be looped over during the simulation step. If necessary a contact detection algorithm can be added. 
\\

\noindent\begin{minipage}{\linewidth}
\begin{lstlisting}[style=intext]
/// Set contact detection algorithm for particle-particle contact.
void setContactDetection(ContactDetection<T, PARTICLETYPE>& contactDetection);
\end{lstlisting}
\end{minipage}

A force based on contact between two particles is the contact force like described in the theory of Hertz and others and is named here as
\class{HertzMindlinDeresiewicz3D}.
\\

\noindent\begin{minipage}{\linewidth}
\begin{lstlisting}[style=intext]
  auto hertz = make_shared < HertzMindlinDeresiewicz3D<T, PARTICLE, DESCRIPTOR>
               > (0.0003e9, 0.0003e9, 0.499, 0.499);
  spSys.addForce(hertz);
\end{lstlisting}
\end{minipage}

Finally one timestep is computed by a call to the function \class{simulate()}.\\

\noindent\begin{minipage}{\linewidth}
\begin{lstlisting}[style=intext]
template<typename T, template<typename U> class PARTICLETYPE>
void SuperParticleSystem3D<T, PARTICLETYPE>::simulate(T dT)
{
  for (auto pS : _pSystems) {
    pS->_contactDetection->sort();
    pS->simulate(dT);
    pS->computeBoundary();
  }
  updateParticleDistribution();
}
\end{lstlisting}
\end{minipage}
This function contains a loop over the local \class{ParticleSystem3D}s calling the local sorting algorithm and the functions \class{ParticleSystem3D::simulate()} and \class{ParticleSystem3D::computeBoundary()}. The sorting algorithm determines potential contact between particles according to the set \class{ContactDetection}.
\noindent\begin{minipage}{\linewidth}
\begin{lstlisting}[style=intext]
  inline void simulate(T dT) {
    _pSys->computeForce();
    _pSys->explicitEuler(dT);
  }
\end{lstlisting}
\end{minipage}
The inline function \class{ParticleSystem3D::simulate()} first calls the local function \class{ParticleSystem3D::computeForce()}.

\noindent
\begin{lstlisting}[style=intext]
template<typename T, template<typename U> class PARTICLETYPE>
void ParticleSystem3D<T, PARTICLETYPE>::computeForce()
{
  typename std::deque<PARTICLETYPE<T> >::iterator p;
  int pInt = 0;
  for (p = _particles.begin(); p != _particles.end(); ++p, ++pInt) {
    if (p->getActive()) {
      p->resetForce();
      for (auto f : _forces) {
        f->applyForce(p, pInt, *this);
      }
    }
  }
}
\end{lstlisting}
This function consists of a loop over all particles stored by the calling \linebreak \class{ParticleSystem3D}. If the particle state is active, its force variable is reset to zero. Then the value computed by each previously added particle force is added to the particle's force variable. Finally, the particle velocity and position is updated by one step of an integration method.

Returning to the function \class{SuperParticleSystem3D::simulate(T dT)} the next command in the loop is a call of the function \class{ParticleSystem3D::computeBoundary()}, which has the same structure as the \class{ParticleSystem3D::computeForce()}. After executing the loop, the function \class{updateParticleDistribution()} is called, which redistributes the particles over the \class{ParticleSystem3D}s according to their updated position. A detailed description of this function is provided at the end of the next section.
\end{sloppypar}

\subsubsection{Implementation of the \emph{Communication Optimal Strategy}}
\label{sec:implParPar}
\begin{sloppypar}
The \emph{communication optimal strategy} is implemented in the function \class{SuperParticleSystem3D::updateParticleDistribution()} already mentioned above. The function has to be called after every update of the particle positions, in order to check if the particle remained in its current cuboid, as otherwise segmentation faults may occur during the computation of particle forces. The transfer  is implemented using nonblocking operations of the MPI library. \\

\noindent
\begin{lstlisting}[style=intext]
template<typename T, template<typename U> class PARTICLETYPE>
void SuperParticleSystem3D<T, PARTICLETYPE>::updateParticleDistribution()
{
  /* Find particles on wrong cuboid, store in relocate and delete */
  //maps particles to their new rank
  _relocate.clear();
  for (unsigned int pS = 0; pS < _pSystems.size(); ++pS) {
    auto par = _pSystems[pS]->_particles.begin();
    while (par != _pSystems[pS]->_particles.end()) {
      //Check if particle is still in his cuboid
      if (checkCuboid(*par, 0)) {
        par++
      }
      //If not --> find new cuboid
      else {
        findCuboid(*par, 0);
        _relocate.insert(
          std::make_pair(this->_loadBalancer.rank(par->getCuboid()), (*par)));
        par = _pSystems[pS]->_particles.erase(par);
      }
    }
  }
\end{lstlisting}
The function begins with with two nested loops. The outer loop is over all local \class{ParticleSystem3D}s, the inner loop over the \class{Particle3D}s of the current \class{ParticleSystem3D}. Each particle is checked if it remained in its cuboid during the last update, by the function \class{checkCuboid(*par, 0)}. The first parameter of \class{checkCuboid(*par, 0)} is the particle to be tested and the second parameter is an optional spatial extension of the cuboid. If the function returns \class{true} the counter is incremented and the next particle is tested. If the function returns \class{false} the particle together with the rank of its new cuboid are copied to the \class{std::multimap<int, PARTICLETYPE<T> > _relocate} for future treatment and removed from the \class{std::deque<PARTICLETYPE<T> > _particles} of particles. \\

\noindent\begin{minipage}{\linewidth}
\begin{lstlisting}[style=intext]
  /* Communicate number of Particles per cuboid*/
  singleton::MpiNonBlockingHelper mpiNbHelper;

  /* Serialise particles */
  _send_buffer.clear();
  T buffer[PARTICLETYPE<T>::serialPartSize];
  for (auto rN : _relocate) {
    rN.second.serialize(buffer);
    _send_buffer[rN.first].insert(_send_buffer[rN.first].end(), buffer, buffer+PARTICLETYPE<T>::serialPartSize);
  }
\end{lstlisting}
\end{minipage}

The function continues by instantiating the class \class{singleton::MpiNonBlockingHelper}, which handles memory for \class{MPI_Request} and \class{MPI_Status} messages. Then the particles buffered in \class{_relocate} are serialized. Meaning their data is written consecutively in memory and stored in a buffer \class{std::map<int, std::vector<double> > _send_buffer} in preparation for the transfer.

\noindent\begin{minipage}{\linewidth}
\begin{lstlisting}[style=intext]
  /*Send Particles */
  int noSends = 0;
  for (auto rN : _rankNeighbours) {
    if (_send_buffer[rN].size() > 0) {
      ++noSends;
    }
  }
  mpiNbHelper.allocate(noSends);
  for (auto rN : _rankNeighbours) {
    if (_send_buffer[rN].size() > 0) {
      singleton::mpi().iSend<double>(&_send_buffer[rN][0], _relocate.count(rN)*PARTICLETYPE<T>::serialPartSize, rN, &mpiNbHelper.get_mpiRequest()[k++], 1);
    }
  }
  singleton::mpi().barrier();
\end{lstlisting}
\end{minipage}

To find the number of send operations a loop over the ranks of neighboring cuboids is carried out, increasing the variable \class{count} each time data for a specific rank is available. Then the appropriate number of \class{MPI_Requests} is allocated. Finally the data is sent to the respective PUs via a nonblocking \class{MPI_Isend()} and all PUs wait until the send process is finished on each PU.

\noindent
\begin{lstlisting}[style=intext]
  /*Receive and add particles*/
  int flag = 0;
  MPI_Iprobe(MPI_ANY_SOURCE, 1, MPI_COMM_WORLD, &flag, MPI_STATUS_IGNORE);
  if (flag) {
    for (auto rN : _rankNeighbours) {
      MPI_Status status;
      int flag = 0;
      MPI_Iprobe(rN, 1, MPI_COMM_WORLD, &flag, &status);
      if (flag) {
        int amount = 0;
        MPI_Get_count(&status, MPI_DOUBLE, &number_amount);
        T recv_buffer[amount];
        singleton::mpi().receive(recv_buffer, amount, rN, 1);
        for (int iPar=0; iPar<amount; iPar+=PARTICLETYPE<T>::serialPartSize) {
          PARTICLETYPE<T> p;
          p.unserialize(&recv_buffer[iPar]);
          if (singleton::mpi().getRank() == this->_loadBalancer.rank(p.getCuboid())) {
            _pSystems[this->_loadBalancer.loc(p.getCuboid())]->addParticle(p);
          }
        }
      }
    }
  }
  if (noSends > 0) {
    singleton::mpi().waitAll(mpiNbHelper);
  }
}
\end{lstlisting}
On the receiving side the nonblocking routine \class{MPI_Iprobe()} checks whether an incoming transmission is available. The constant \class{MPI_ANY_SOURCE} indicates that messages from all ranks are accepted. If a message is awaiting reception the flag \class{flag} is set to a nonzero value and the following switch will be true. This query is not necessary, but the following loop can be entirely skipped if no particles are transferred, which is expected to be the case most of the time.

The subsequent loop tests for each single neighboring rank if a message awaits reception. If \class{true} the number of send \class{MPI_Double}s is read from the \class{status} variable via an \class{MPI_Get_count()}. The appropriate memory is allocated and the message is received by wrapped call to \class{MPI_Recv()}, and written consecutively. Then new \class{Particle3D}s are instantiated, initialized with the received data and assigned to the respective \class{ParticleSystem3D} on the updated PU. Finally, a call to \class{MPI_Waitall()} makes sure, that all \class{MPI_Isend()}s have been processed by the recipients.
\end{sloppypar}

\begin{figure}[ht]
\centering
  \includegraphics[scale=1]{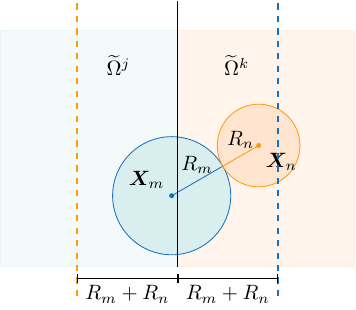}
\caption[Particle overlap.]{Overlap of the particle domains. Particles within a distance to of the sum of the two largest radii to a neighbor cuboid have to be transferred to this specific neighbor cuboid.}
\label{fig:particleOverlap}
\end{figure}
\subsubsection{Shadow Particles}
If particle collisions are considered, it may happen that particles $P_m$ with center $\bm{X}_m\in \widetilde{\Omega}^j$ collide with particles $P_n$ with center $\bm{X}_n\in \widetilde{\Omega}^k$ in a different cuboid, as illustrated in Figure~\ref{fig:particleOverlap}. Therefore $P_n$ has to be known on $\bm{X}_m\in \widetilde{\Omega}^j$ and so-called shadow particles are introduced. Shadow particles are static particles, whose positions and velocities are not explicitly computed during the update step. Particle collision across cuboid boundaries can only occur if the distance $d=\|\bm{X}_n-\bm{X}_m\|_2$ between the participating particles is less then the sum of the two largest radii of all particles in the system. Hence the width of the particle overlap has to be at least the sum of the two largest particle radii and all particles within this overlap have to be  transferred to the neighbor cuboid after each update of the particle position by an additional communication step similar to the one introduced above.


\chapter{Initial and Boundary Conditions}  

Each example located in \path{examples/} is typically structured along several function definitions and a \class{main} part. 
The function \class{setBoundaryValues} usually sets initial values for density and velocity (if not yet done in \class{prepareLattice}). 
Boundary values can be refreshed at certain time steps. 
In some applications, this function may be missing or empty, if there is no temporal change in the boundary values.
An exemplary implementation can be found in \path{examples/laminar/poiseuille2d}, where a smooth start-up is used at the velocity inflow boundary as follows: 
\begin{lstlisting}[language=myc++]
// No of time steps for smooth start-up
int iTmaxStart = converter.getLatticeTime( maxPhysT*0.4 );
int iTupdate = 5;

if (iT%iTupdate == 0 && iT <= iTmaxStart) {

    // Smooth start curve, polynomial
    PolynomialStartScale<T,T> StartScale(iTmaxStart, T(1));
    // Creates and sets the Poiseuille // inflow profile using functors
    T iTvec[1] = {T( iT )};
    T frac[1] = {};
    StartScale( frac,iTvec ); 
    T maxVelocity = converter.getCharLatticeVelocity()*3./2.*frac;
    T distance2Wall = L/2.;
    Poiseuille2D<T>
        poiseuilleU(superGeometry, 3,
        maxVelocity, distance2Wall); 
    sLattice.defineU(superGeometry, 3,
        poiseuilleU);
}
\end{lstlisting}

\section{Define Boundary Method}
\label{sec:defineBoundaryMethod}

There are two different types of boundaries, namely wet-node boundaries (on the nodes) and link-wise boundaries (in between the nodes).
Examples of wet-node boundaries are \class{LocalVelocity}, \class{LocalPressure}, \class{InterpolatedPressure}, \class{InterpolatedVelocity}, 
\class{Zou-HePressure}, \class{ZouHeVelocity} and \class{AdvectionDiffusion}.
Examples of link-wise boundaries are \class{Bouzidi}, \class{BouzidiVelocity}, \class{BouzidiZeroVelocity} (+\class{Interpolation}) and \class{YuPostProcessor} (+\class{Interpolation}).
The boundary declarations are usually designed as in the following code snippets: 
\begin{lstlisting}[language=myc++]
setBounceBackBoundary(superLattice, superGeometry, materialNumber);
setLocalVelocityBoundary<T, DESCRIPTOR>(superLattice, omega, superGeometry, materialNumber);
setBouzidiBoundary(superLattice, superGeometry, materialNumber, indicator);
\end{lstlisting}

OpenLB offers implementations of several methods for approximating macroscopic boundary conditions. 
Some of the frequently used boundary methods are listed below and described in terms of locality of operations and general applicability. 
This list does not raise a claim to completeness. 
OpenLB includes the following boundary methods, which are callable with \class{set...Boundary} and the respective arguments. 
\begin{enumerate}
    \item \textbf{\class{LocalPressure}, \class{LocalVelocity}, \class{LocalConvection}}, etc.:
    \begin{itemize}
        \setlength\itemsep{-0.5em}
        \item local
        \item wet-node
        \item \textit{Application:} boundary for fluid flows
        \item regularized boundary, re-computes all $f_{i}$ from local momenta reconstructing off-equilibrium parts
        \item stable in most regimes
        \item implemented according to Latt and Chopard~\cite{latt2008straight}
    \end{itemize}
    \item \textbf{\class{InterpolatedPressure}, \class{Interpolated}}, etc.:
    \begin{itemize}
    \setlength\itemsep{-0.5em}
        \item non-local
        \item wet-node
        \item \textit{Application:} boundary for fluid flows (stable for higher Reynolds numbers)
        \item re-computes all $f_{i}$ using a finite-difference scheme over adjacent cells for the velocity gradient
        \item implemented according to Skordos \cite{skordos:93}
    \end{itemize}
    \item \textbf{\class{setZouPressure}, \class{setZouHeVelocity}}, etc.:
    \begin{itemize}
    \setlength\itemsep{-0.5em}
        \item local
        \item wet-node
        \item \textit{Application:} boundary for fluid flows (low Reynolds number)
        \item computes missing $f_{i}$ by applying symmetry conditions on the off-equilibrium part, enforcing velocity and pressure on the equilibrium
        \item highly accurate, but less stable
        \item implemented according to Zou and He~\cite{zou-he:97}
    \end{itemize}
    \item \textbf{\class{Bouzidi}, \class{BouzidiVelocity}}, etc.:
    \begin{itemize}
        \setlength\itemsep{-0.5em}
        \item non-local
        \item link-wise
        \item \textit{Application:} boundary for fluid flows (for curved-boundaries)
        \item computes missing $f_{i}$ using a bounce-back rule which takes the distance between node and boundary into account
        \item second order accuracy in space
        \item implemented according to Bouzidi \textit{et al.}~\cite{bouzidi:01}
    \end{itemize}
    \item \textbf{\class{AdvectionDiffusionTemperature}}, etc.:
    \begin{itemize}
    \setlength\itemsep{-0.5em}
        \item local and non-local
        \item wet-node
        \item \textit{Application:} boundary for advection-diffusion problems (temperature, particle, ..)
        \item various boundary conditions for advection diffusion problems
        \item type of implementation differs for each condition
        \item for details see \eg Trunk \textit{et al.}~\cite{trunk:16}
    \end{itemize}
\end{enumerate}
The typical macroscopic effects that can be obtained with the wet-node-method are for example velocity and pressure Dirichlet boundary conditions and for the link-wise methods for example Dirichlet boundary conditions for a specific velocity, zero-velocity, convection and slip walls. 
With the \class{advectionDiffusion}-methods macroscopic boundary conditions in terms of Dirichlet boundaries with respect to temperature, convection, zero-distribution, or the external-field, are recoverable.

\subsection{Wet-node Method}
With the wet-node approach, for example Dirichlet-type boundaries for the velocity can be obtained.
On a macroscopic level, this is used for example at inflow boundaries where the values for the inflow velocity are given. 
\begin{lstlisting}[language=myc++]
setLocalVelocityBoundary<T, DESCRIPTOR> (superLattice, omega, superGeometry, materialNumber);
\end{lstlisting}
Moreover, the wet-node approach is applied for macroscopic boundaries for the pressure. 
This is used for example at outflow boundaries and fixes the values for the pressure in terms of a Dirichlet condition.

\begin{lstlisting}[language=myc++]
setLocalPressureBoundary<T, DESCRIPTOR> (superLattice, omega, superGeometry, materialNumber);
\end{lstlisting}

\subsection{Link-wise Method}
The link-wise method is applicable to recover macroscopic convection, which is used for outflow boundaries and approximates 
\begin{align}
    \frac{\partial u}{\partial t}+ u_{\text{average}}\frac{\partial u}{\partial n} = 0 ~, 
\end{align}
similar to a Neumann-type boundary for the velocity.
\begin{lstlisting}[language=myc++]
setInterpolatedConvectionBoundary<T, DESCRIPTOR> (superLattice, omega, superGeometry, materialNumber, averageVelocity);
\end{lstlisting}
Additionally, a slip boundary can be constructed that is used for solid boundaries and reflects outgoing $f_{i}$. 
The latter has the effect of zero velocity normal to the boundary and free flow tangential to the boundary.
\begin{lstlisting}[language=myc++]
setSlipBoundary<T, DESCRIPTOR> (superLattice, superGeometry, materialNumber);
\end{lstlisting}
Further, a macroscopic velocity used for curved boundaries with fixed values can be superimposed. 
It is realized as a Dirichlet boundary for the velocity and considers the boundary shape via an indicator functor. 
There are multiple \class{bulkMaterials} possible.
\begin{lstlisting}[language=myc++]
setBouzidiBoundary<T, DESCRIPTOR,BouzidiVelocityPostProcessor> (superLattice, superGeometry, material , indicator, bulkMaterialsList);
AnalyticalConstF3D<T> u;
setBouzidiVelocity(superLattice, superGeometry, material, u);
\end{lstlisting}
As a special case, the zero velocity boundary for curved boundaries which fixes the velocity to zero is callable with simplified syntax. 
It is a Dirichlet-type boundary for velocity and considers the shape by indicator. 
Here are also multiple \class{bulkMaterials} possible.
\begin{lstlisting}[language=myc++]
setBouzidiBoundary<T, DESCRIPTOR>(superLattice, superGeometry, material, indicator, bulkMaterialsList);
\end{lstlisting}

\subsection{AdvectionDiffusionBoundary}

An exemplary macroscopic effect of the \class{advectionDiffusion} boundary methods is the Dirichlet boundary condition for the temperature, which is used for inflow or wall boundaries and fixes the values for the temperature. 
\begin{lstlisting}[language=myc++]
setAdvectionDiffusionTemperatureBoundary<T, DESCRIPTOR> (superLattice, omega, superGeometry, materialNumber);
\end{lstlisting}
In addition, convection can be modeled, which is used for outflow boundaries and interpolates incoming $f_{i}$ from neighbors. 
It is similar to a Neumann-type temperature boundary.
\begin{lstlisting}[language=myc++]
setAdvectionDiffusionConvectionBoundary<T, DESCRIPTOR> (superLattice, superGeometry, materialNumber);
\end{lstlisting}
Further, the zero-distribution boundary is used for "sticky" boundaries and sets the incoming $f_{i}$ to zero, such that particles touching the boundary are \textit{trapped}. 
\begin{lstlisting}[language=myc++]
setZeroDistributionBoundary<T, DESCRIPTOR> (superLattice, superGeometry, materialNumber);
\end{lstlisting}
At last also macroscopic external fields used for example in particle simulations can be supplied with boundary conditions. 
The following boundary method provides data on the boundary for particle calculations (velocity gradient).
\begin{lstlisting}[language=myc++]
setExtFieldBoundary<T, DESCRIPTOR, FIELD >(superLattice, superGeometry, materialNumber);
\end{lstlisting}

\subsection{Additional Options}
An additional option to supply boundary conditions in OpenLB is to choose LBM collision dynamics also at the boundary. 
One can add a specific dynamics object to a boundary material number, e.g.: \class{BulkDynamics} which executes the collision step as in a fluid node (\eg with \class{BGKdynamics}).
\begin{lstlisting}[language=myc++]
    using BulkDynamics = BGKdynamics<T,DESCRIPTOR>;
    [...]
    setLocalVelocityBoundary<T,DESCRIPTOR, BulkDynamics>(superLattice,omega, superGeometry,materialNumber);
\end{lstlisting}
Another option is to choose an implementation via boundary dynamics.
One can add boundary methods as dynamics, \eg \class{BounceBack} (zero velocity), \class{BounceBackVelocity} (prescribes a nonzero velocity).
\begin{lstlisting}[language=myc++]
    superLattice.defineDynamics<BounceBack> (superGeometry, materialNumber);
    superLattice.defineDynamics <BounceBackVelocity>(superGeometry, materialNumber);
\end{lstlisting}

\section{Define Initial Conditions}
\label{sec:defineInitialConditions}

\begin{lstlisting}[language=myc++]
    sLattice.defineRhoU(superGeometry, materialNumber, analyticalFunctor, analyticalFunctor);
\end{lstlisting}
For each each material number, the density (usually $\rho =1$), the velocity (usually $u=0$), and the distribution functions should be initialized. 
The functions expect lattice values, therefore physical values have to be converted, e.g.: via the function \class{getLatticeDensity(density)} of the \class{UnitConverter}.
Note that instead of the \class{materialNumber} argument, any discrete indicator function can be used. 
Exemplary initializations are given below. 
\begin{lstlisting}[language=myc++]
// Initial conditions
AnalyticalConst2D<T,T> rhoF(1);
std::vector<T> velocity(dim,T(0));
AnalyticalConst2D<T,T> uF(velocity);

// Initialize density 
superLattice.defineRho(superGeometry, materialNumber, rhoF);

// Initialize velocity
superLattice.defineU(superGeometry, materialNumber, uF);

// Initialize distribution functions
// to local equilibrium
superLattice.iniEquilibrium(superGeometry, materialNumber,
       rhoF, uF);
\end{lstlisting}

\section{Define Boundary Values}

\begin{lstlisting}[language=myc++]
sLattice.defineU(superGeometry, materialNumber, analyticalFunctor);
\end{lstlisting}
Just like the initial conditions, boundary values can be set using \class{defineRho(...)} and \class{defineU(...)}. 
For a smooth start-up, the values can be scaled, \eg according to a sinus-scale or polynomial-scale for a given startup-time.
\begin{lstlisting}[language=myc++]
    // Smooth start curve, sinus
SinusStartScale<T,int>
StartScale(numerStartTimeSteps, maxValue);

//Smooth start curve,polynomial PolynomialStartScale<T,T>
StartScale(numerStartTimeSteps, maxValue);
  
// compute scale-factor "frac"
T iTvec[1] = {T(timestep)};
T frac[1] = {};
StartScale( frac, iTvec );
\end{lstlisting}
To apply a flow profile, one first has to update values in every \(n\)th time step, then initialize a functor and then set values using \class{defineRho(...)} and \class{defineU(...)}. 
These are the same functions as for the initial conditions. 
However, the time point when to call them is crucial.
\begin{lstlisting}[language=myc++]
if (timestep%updatePeriod==0 &&
       timestep <= numberStartupTimesteps) {
  Poiseuille2D<T> poiseuilleU(superGeometry,
         materialNumber, maxVelocity*frac[0],
         distance2Wall);
sLattice.defineU(superGeometry, materialNumber, poiseuilleU);
}
\end{lstlisting}
Further examples of 3D functors for this purpose are: 
rotating functors (linear for a rotating velocity field and quadratic for a rotating pressure field), 
Circle-Poiseuille, 
Elliptic-Poiseuille, and 
Rectangular-Poiseuille. 
The latter are only for on axis boundaries and can be constructed from points spanning a plane or a material number. 
These functors are summarized with arguments as follows: 
\begin{lstlisting}[language=myc++]
RotatingLinear3D(axisPoint,
  axisDirection, angularVelocity);
  
RotatingQuadratic1D(axisPoint, axisDirection, angularVelocity);

CirclePoiseuille3D(axisPoint,
  axisDirection, maxVelocity, radius);
  
EllipticPoiseuille3D(center, semiPrincipalAxis1, semiPrincipalAxis2, maxVelocity);

RectanglePoiseuille3D(point1, point2,
  point3, maxVelocity);
  
RectanglePoiseuille3D(superGeometry, materialNumber, maxVelocity, offsetX, offsetY, offsetZ);
\end{lstlisting}


\chapter{Input and Output} 
\label{cha:io}

During development or during actual simulations, it might be necessary to parametrize your program.
For this case, OpenLB provides an XML parser, which can read files produced by OpenGPI~\cite{opengpi-web}, thereby providing a user-friendly GUI. 
Details on the XML format and functions are given in Section~\ref{sec:xmlparams}.

The simulation data is stored in the VTK data format and can be further processed with ParaView.
For output tasks that are performed only once during the simulation, it is recommended to write the data sequentially.
Commonly, the geometry or cuboid information is one of these tasks.
In contrast to the parallel version, it is easier to use and does not produce unnecessary data overhead.
However, if the output is performed regularly in a parallel simulation, the performance may slow down using the sequential output method.
Therefore, OpenLB has implemented a parallel data output functionality.
At the lowest scope, every thread writes \texttt{.vti} files that contain the data.
OpenLB writes a \texttt{.vti} file for every cuboid, to provided parallel data processing.
Those \texttt{.vti} files are summarized and put together by the next hierarchy, the \texttt{.vtm} file.
A \texttt{.vtm} file corresponds to the entire domain with respect to a certain time step.
At the end, the different time steps are packed to a \texttt{.pvd} file, that is a collection of \texttt{.vtm} according to time steps.

The technical aspects are presented in Section~\ref{sec:vti_pvd-format}, whereas the usage is demonstrated with an example in Section~\ref{sec:outPut-vtk}.

\section{Output Data Structure}\label{sec:vti_pvd-format}
OpenLB simulation data is stored in file system according to the VTK data format~\cite{vtkFormat-web}.
This format has XML structure and the data therein is written as binary \texttt{Base64} code.
Additionally, the simulation data is compressed by zlib, which allows to reduce data tremendously.

On the top level, the parallel output hierarchy contains a \texttt{.pvd} file, which consists of links to \texttt{.vtm} files.
The \texttt{.vtm} files summarize the cuboids represented by \texttt{.vti} files.

\begin{lstlisting}[language=XML, caption={Example of a \texttt{.pvd} file that points for every time step to the corresponding \texttt{.vtm} file. Every time step is associated to a \texttt{.vtm} file.}]
<?xml version="1.0"?>
<VTKFile type="Collection" version="0.1" byte_order="LittleEndian">
<Collection>
<DataSet timestep="81920" group="" part="" file="data/VTM_iT0081920.vtm"/>
<DataSet timestep="163840" group="" part="" file="data/VTM_iT0163840.vtm"/>
<DataSet timestep="245760" group="" part="" file="data/VTM_iT0245760.vtm"/>
<DataSet timestep="327680" group="" part="" file="data/VTM_iT0327680.vtm"/>
<DataSet timestep="409600" group="" part="" file="data/VTM_iT0409600.vtm"/>
</Collection>
</VTKFile>
\end{lstlisting}

\begin{lstlisting}[language=XML, caption={Example of a \texttt{.vtm} file that points to \texttt{.vti} files that hold data of a cuboids. Every cuboid writes its data to a \texttt{.vti} file, which are assembles by a \texttt{.vtm} file.}]
<?xml version="1.0"?>
<VTKFile type="vtkMultiBlockDataSet" version="1.0" byte_order="LittleEndian">
<vtkMultiBlockDataSet>
<Block index="0" >
<DataSet index= "0" file="VTM_iT0081920iC00000.vti"></DataSet>
</Block>
<Block index="1" >
<DataSet index= "0" file="VTM_iT0081920iC00001.vti"></DataSet>
</Block>
<Block index="2" >
<DataSet index= "0" file="VTM_iT0081920iC00002.vti"></DataSet>
</Block>
<Block index="3" >
<DataSet index= "0" file="VTM_iT0081920iC00003.vti"></DataSet>
</Block>
</vtkMultiBlockDataSet>
</VTKFile>
\end{lstlisting}

There is also a \texttt{BlockVTKwriter} that writes data sequentially.
More details can be found in the source code and its documentation.

\section{Write Simulation Data to VTK File Format}\label{sec:outPut-vtk}
VTK data files can be visualized and postprocessed with the free software ParaView~\cite{paraview-web}, which offers a graphical interface with extensive functionality. 
The following listing shows, on the one hand, how to write VTK files sequential for a geometry and cuboid functors.
On the other hand, the usage of the parallel write-routine for velocity and pressure functors is shown.
\begin{lstlisting}[language=myc++, caption={An exemplary code to write simulation data to file system.}]
// create VTK writer object
SuperVTMwriter3D<T> vtmWriter("FileNameGoesHere");
// write only the first iteration step
if (iT==0) {
  SuperLatticeGeometry3D<T,DESCRIPTOR> geometry(sLattice, superGeometry);
  SuperLatticeCuboid3D<T,DESCRIPTOR> cuboid(sLattice);
  // writes the geometry and cuboids to file system, sequentially
  vtmWriter.write(geometry);
  vtmWriter.write(cuboid);
  // mandatory to call the following write()-method
  vtmWriter.createMasterFile();
}
// write every 2 sec (physical time scale)
if (iT%converter.getLatticeTime(2.)==0) {
  // create functors that process data from SuperLattice
  SuperLatticePhysVelocity3D<T,DESCRIPTOR> velocity(sLattice,
                                                    converter);
  SuperLatticePhysPressure3D<T,DESCRIPTOR> pressure(sLattice,
                                                    converter);
  vtmWriter.addFunctor( velocity );
  vtmWriter.addFunctor( pressure );
  // writes the added functors to file system, parallel
  vtmWriter.write(iT);
}
\end{lstlisting}
Note that the function call \texttt{creatMasterFile()} in \texttt{iT == 0} is essential to write parallel VTK data.

\section{CSV Writer} \label{sec:csvWriter}
For some data analysis a CSV format of the data is necessary. In this case it is possible to use the CSV Writer to create these data files.
The following lines show an application of the CSV Writer in the example advectionDiffusion1d (\ref{example:advectionDiffusion1d}). If one only wants to write in one data file, the filename can be given to the constructor of the CSV Writer. However the \texttt{plotFileName} parameter provides the possibility to set a new datafile with every call of this function. The \texttt{precision} parameter refers to the precision of the output data.

\begin{lstlisting}[language=myc++, caption={Exemplary application of the CSV Writer}]
CSV<T> csv();
csv.writeDataFile(N, simulationAverage, "averageSimL2RelErr");
\end{lstlisting}

\section{Write Images Instantaneously}\label{sec:writing_image}
OpenLB is able to output image data directly.
This is helpful to get a brief overview of how the simulation is going on without using external visualization tools.
Note that only $1D$ data or equivalent scalar-valued data can be represented by images.
Hence, for vector-valued data, \eg velocity, it is important to take an appropriate norm.
This step transforms the vector into a scalar and the data becomes one dimensional as required.

For $2D$ applications it is straight forward to generate images, since every point of the computational grid represents a pixel.
However, for $3D$ applications this assignment fails.
OpenLB allows one to reduce the $3D$ grid to a $2D$ plane by parametrizing a hyperplane in $3D$ space.
The resulting $2D$ block lattice represents the image by assigning lattice points to pixels.

An example of how to take a norm and how to reduce a plane is shown below.

\begin{lstlisting}[language=myc++, caption={An exemplary code reducing a plane in $3D$}, label=lst:planeReduction3d]
// get the pointwise l2 norm of velocity
SuperEuklidNorm3D<T> normVel( velocity );
// reduce a hyperplane parametrized by normal (0,0,1) and centered in the mother geometry from the 3D data
BlockReduction3D2D<T> planeReduction( normVel, {0, 0, 1} );
\end{lstlisting}

Note that internally the hyperplane is parametrized using the \class{Hyperplane3D} class. This example uses one of the helper constructors of \class{BlockReduction3D2D} to hide this detail for the common use case of parametrizing a hyperplane by a normal vector. There are further such helper constructors available if one wishes to for example define a hyperplane by two span vectors and its origin. However for full control over the hyperplane a \class{Hyperplane3D} instance may also be created by hand.

\begin{lstlisting}[language=myc++, caption={Exemplary code to write images of an explicitly instantiated $3D$ hyperplane with}, label=lst:instantisationHyperplane3D1]
SuperEuklidNorm3D<T> normVel( velocity );
BlockReduction3D2D<T> planeReduction(
  normVel,
  // explicitly construct a 3D hyperplane
  Hyperplane3D<T>()
  .centeredIn(superGeometry.getCuboidGeometry().getMotherCuboid())
  .spannedBy({1, 0, 0}, {0, 1, 0}));
  BlockGifWriter<T> gifWriter;
  gifWriter.write(planeReduction, iT, "vel" );
\end{lstlisting}

Both of these exemplary codes reduce a $3D$ hyperplane to a $2D$ lattice with $600$ points on its longest side. It is possible to change this resolution either by providing it as a constructor argument to \texttt{BlockReduction3D2D} or by explicitly instantiating a \texttt{HyperplaneLattice3D}.

\begin{lstlisting}[language=myc++, caption={Exemplary code using an explicitly instantiated $3D$ hyperplane lattice}, label=lst:instantisationHyperplane3D2]
SuperEuklidNorm3D<T> normVel( velocity );
HyperplaneLattice3D<T> gifLattice(
  superGeometry.getCuboidGeometry(),
  Hyperplane3D<T>()
  .centeredIn(superGeometry.getCuboidGeometry().getMotherCuboid())
  .normalTo({0, -1, 0}),
  // resolution (floating point values are used as grid spacing instead)
  1000);
\end{lstlisting}

In $2D$ the reduction of velocity data to a block can be achieved as follows.

\begin{lstlisting}[language=myc++, caption={Exemplary code reducing data in $2D$}]
SuperEuklidNorm2D<T,DESCRIPTOR> normVel( velocity );
BlockReduction2D2D<T> planeReduction( normVel );
\end{lstlisting}

The resolution of $600$ points on the longest side of the object is set as default, but can be altered similarly to the Listings~\ref{lst:planeReduction3d}, \ref{lst:instantisationHyperplane3D1} and \ref{lst:instantisationHyperplane3D2}. There are two options of generating images of the processed values in $2D$ and $3D$.

\subsection{GifWriter}
In this example the constructor \texttt{gifWriter} generates automatically scaled images of the \texttt{PPM} data type which are scaled according to the minimum and maximum value of the desired value of the time step.

\begin{lstlisting}[language=myc++, caption={Exemplary code using \texttt{gifWriter} to create \texttt{PPM} files}]
BlockReduction3D2D<T> planeReduction(normVel, gifLattice);
BlockGifWriter<T> gifWriter;
//gifWriter.write(planeReduction, 0, 0.7, iT, "vel"); //static scale
gifWriter.write( planeReduction, iT, "vel" ); // scaled
\end{lstlisting}

With \texttt{imagemagick}'s command \texttt{convert} the \texttt{PPM} files generated by \texttt{gifWriter} can be combined to an animated \texttt{GIF} file as follows.
\begin{lstlisting}[language=bash]
convert tmp/imageData/*.ppm animation.gif
\end{lstlisting}
To reduce the \texttt{GIF}'s file size you can use the options \texttt{fuzz} and \texttt{OptimizeFrame}, for example:
\begin{lstlisting}[language=bash]
convert -fuzz 3% -layers OptimizeFrame tmp/imageData/*.ppm animation.gif
\end{lstlisting}
Even smaller files are possible with \texttt{ffmpeg} and conversion to \texttt{MP4} video file. This could be done using a command like the following.
\begin{lstlisting}[language=bash]
ffmpeg -pattern_type glob -i 'tmp2/imageData/*.ppm' animation.mp4
\end{lstlisting}

\subsection{Heatmap}
Whereas the the \texttt{gifWriter} creates only automatically scaled \texttt{PPM} images, the functor \texttt{heatmap} has more options to adjust the \texttt{JPEG} files. For this purpose the variable \texttt{plotParam} can be created and the desired modifications, \eg minimum and maximum values of the scale, can be passed on to the optional variable.

\begin{lstlisting}[language=myc++, caption={Exemplary code using the functor \texttt{heatmap} with modified parameters}, label=lst:heatmap]
SuperEuklidNorm3D<T> normVel( velocity );
BlockReduction3D2D<T> planeReduction( normVel, {0, 0, 1} );
// write output as JPEG and changing properties
heatmap::plotParam<T> jpeg_Param;
jpeg_Param.contourlevel = 5;    //setting the number of contur lines
jpeg_Param.colour = "rainbow";  //colour combination "grey", "pm3d", "blackbody" and "rainbow" can be chosen
heatmap::write(planeReduction, iT, jpeg_Param);
\end{lstlisting}

The exemplary code in listing~\ref{lst:heatmap} shows how to change the color set and number of contour lines in the generated images. All possible adjustments are listed and used in the example venturi3D (see Section~\ref{sec:venturi3d}).

\section{Gnuplot Interface}\label{sec:gnuplot}
Often, for the analysis of simulations a plot of the data is required. OpenLB offers an interface which uses \texttt{Gnuplot} to create plots.
Furthermore, it is possible to see the particular data that was used for the plots in realtime and to use comparison data, which is directly
used in the plot.
An example for the usage from \texttt{examples/cylinder2d} is shown below.
\begin{lstlisting}[language=myc++, caption={An exemplary code to plot simulation data.}]
// Gnuplot constructor (must be static!)
// for real-time plotting: gplot("name", true) // experimental!
static Gnuplot<T> gplot( "drag" );

...

// set data for gnuplot: input={xValue, yValue(s),
// names (optional), position of key (optional)}
gplot.setData( converter.getPhysTime( iT ), {_drag[0], 5.58},
 {"drag(openLB)", "drag(schaeferTurek)"}, "bottom right" );

// writes a png (or optional pdf) in one file for every timestep,
// if the png file is opened by an imageviewer it can be used as a "liveplot"
// optional for pdf output, use: gplot.writePDF()
gplot.writePNG();
}
\end{lstlisting}
The data \texttt{drag[0]} is calculated in the example and compared with the value $ 5.58 $. This is then plotted as shown in Figure~\ref{fig:gnuplotDrag}.
\begin{figure}[ht]
  \centering
  \includegraphics[width=0.7\textwidth]{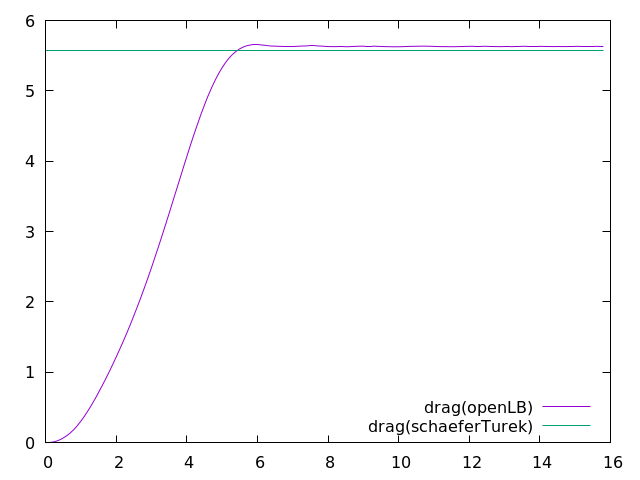}
  \caption{Gnuplot output of drag calculation in cylinder2d.}
  \label{fig:gnuplotDrag}
\end{figure}

In order to have plots for different times, the following usage is recommended.
\begin{lstlisting}[language=myc++, caption={Creating plots for different time steps.}]
  ...

  // every (iT%vtkIter) write an png of the plot
  if ( iT%( vtkIter ) == 0 ) {
  // writes pngs: input={name of the files (optional),
  // x range for the plot (optional)}
  gplot.writePNG( iT, maxPhysT );
\end{lstlisting}

\subsection{Regression with Gnuplot}\label{subsec:RegressionWithGnuplot}
Moreover, Gnuplot can be used to create a linear regression of datasets. For instance, the analysis of the
experimental order of convergence in a simulation can be executed as in the example \path{poiseuille2dEOC}. \\
The possible options are: Linear regression to the given data whereas it is possible to use a loglog-scaling
(loglogINVERTED for inverting the x-axis). The implementation is done via the constructor of plot in the \texttt{.cpp}
file itself as seen below.

\begin{lstlisting}[language=myc++]
static Gnuplot<T> gplot( "eoc", Gnuplot<T>::LOGLOG, Gnuplot<T>::LINREG);
\end{lstlisting}
The possible options for the scaling are: LINEAR (using the data as given), \texttt{LOGLOG} (using log of the x- and
y-dataset) and \texttt{LOGLOGINVERTED} (using log of y-dataset and 1/log of x-dataset). For the regression type one can
choose \texttt{LINREG} (linear Regression) and \texttt{OFF} (no regression).\\
\begin{figure}[ht]
	\centering
	\includegraphics[width=0.9\textwidth]{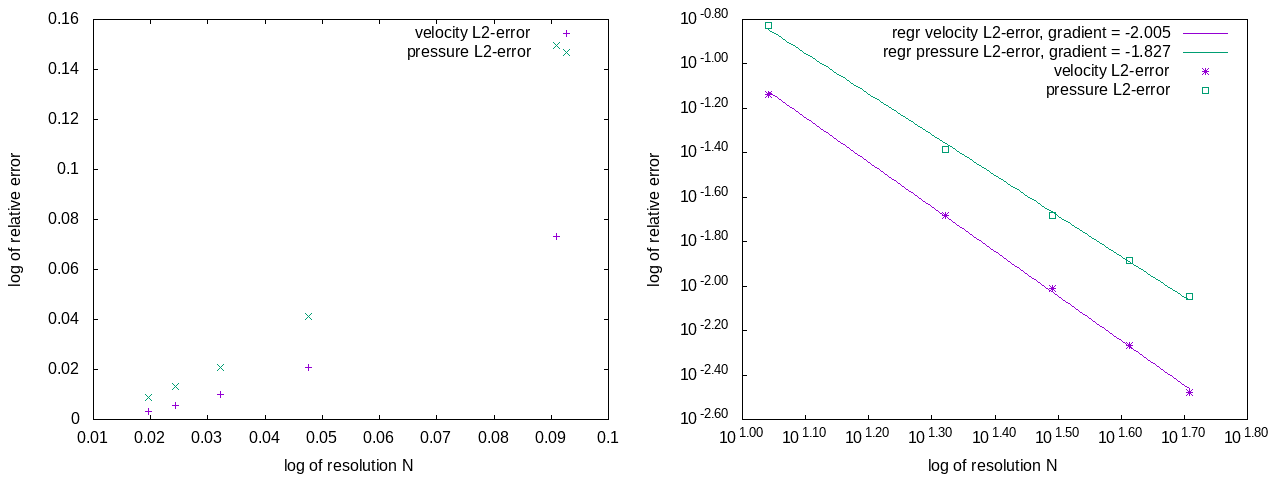}
	\caption{Example of using regression to analyze polynomial errors (left: old, right: new implementation) }
	\label{fig:gnuplot_eoc_example}
\end{figure}

\section{Console Output}\label{sec:consoleOutput}
In OpenLB, there is an extension of the default \class{ostream}, which handles parallel output and prefixes each line with the name of the class that produced the output.
Listing~\ref{lst:consoleOutput} is the output of the \texttt{bstep2d} example.

It is easy to determine which part of OpenLB has produced a specific message.
This can be very helpful in the debugging process, as well as for quickly postprocessing console output or filtering out important information without any need to go into the code.
Together with OpenLB's semi-CSV style output standard, it is possible to easily visualize any data imaginable with diagrams, such as convergence rates, data errors, or simple average mass density.
\begin{lstlisting}[language=myc++]
void MyClass::print() {
OstreamManager clout(std::cout, "MyClass");
...
clout << "step=" << step << "; avRho=" << avRho
      << "; maxU=" << maxU << std::endl;
}
\end{lstlisting}

Using the \class{OstreamManager} is easy and consists of two parts. First, an instance of the class \class{OstreamManager} is needed. The one created here in line~2 is called \texttt{clout} like all the oth\-er instances in OpenLB. This word consists of the two words class and output Moreover, it is quite similar to standard \texttt{cout}. The constructor receives two arguments: one describing the ostream to use, the other one setting the prefix-text. In line~4 the usage of an instance of the \class{OstreamManager} is shown. There is not much difference in usage between a default \texttt{std::cout} and an instance of OpenLB's \class{OstreamManager}. The only thing to consider is that a normal \verb "\n"  won't have the expected effect, so use \texttt{std::endl} instead.

In classes with many output producing functions however, you wouldn't like to instantiate \class{OstreamManager} for every single function, so a central instantiation is preferred. This is done by adding a \texttt{mutable OstreamManager} object as a private class member and initializing it in the initialization list of each defined constructor. An example implementation of this method can be found in \path{src/utilities/timer.hh}.

Another great benefit of \class{OstreamManager} is the reduction of output in parallel. Running a program using \texttt{cout} on multiple cores normally means getting one line of output for each process. \class{OstreamManager} will avoid this by default and display only the output of the first processor. If this behavior is unwanted in a specific case, it can be turned off for an instance named \texttt{clout} by \texttt{clout.setMultiOutput(true)}.

Further scenarios that are not yet implemented in OpenLB can make use of different streams like the ostream \texttt{std::cerr} for separate error output, file streams, or something completely different. In doing so, every stream needs its own instance.

\newpage
\begin{landscape}
\begin{lstlisting}[language=bash,breaklines= true,caption={Terminal output of example bstep2d.},label={lst:consoleOutput}]
$ ./bstep2d
...
[prepareGeometry] Prepare Geometry ...
[SuperGeometry2D] cleaned 0 outer boundary voxel(s)
[SuperGeometry2D] cleaned 0 inner boundary voxel(s)
[SuperGeometry2D] the model is correct!
[SuperGeometryStatistics2D] materialNumber=0; count=13846; minPhysR=(0,0); maxPhysR=(5,0.75)
[SuperGeometryStatistics2D] materialNumber=1; count=92865; minPhysR=(0.0166667,0.0166667); maxPhysR=(19.9833,1.48333)
[SuperGeometryStatistics2D] materialNumber=2; count=2448; minPhysR=(0,0); maxPhysR=(20,1.5)
[SuperGeometryStatistics2D] materialNumber=3; count=43; minPhysR=(0,0.783333); maxPhysR=(0,1.48333)
[SuperGeometryStatistics2D] materialNumber=4; count=89; minPhysR=(20,0.0166667); maxPhysR=(20,1.48333)
[prepareGeometry] Prepare Geometry ... OK
[prepareLattice] Prepare Lattice ...
[prepareLattice] Prepare Lattice ... OK
[main] starting simulation...
[SuperPlaneIntegralFluxVelocity2D] regionSize[m]=1.46667; flowRate[m^2/s]=0; meanVelocity[m/s]=0
[SuperPlaneIntegralFluxPressure2D] regionSize[m]=1.46667; force[N]=0; meanPressure[Pa]=0
[Timer] step=0; percent=0; passedTime=0.846; remTime=101519; MLUPs=0
[LatticeStatistics] step=0; t=0; uMax=1.49167e-154; avEnergy=0; avRho=1
[SuperPlaneIntegralFluxVelocity2D] regionSize[m]=1.46667; flowRate[m^2/s]=0; meanVelocity[m/s]=0
[SuperPlaneIntegralFluxPressure2D] regionSize[m]=1.46667; force[N]=0; meanPressure[Pa]=0
[Timer] step=300; percent=0.25; passedTime=2.503; remTime=998.697; MLUPs=17.2699
[LatticeStatistics] step=300; t=0.1; uMax=5.75006e-07; avEnergy=8.66459e-16; avRho=1
...
\end{lstlisting}
\end{landscape}

\section{Read and Write STL Files}\label{sec:stlfiles}
OpenLB offers the possibility to read and write geometry data in the Standard Triangulation Language, STL for short. The OpenLB class \class{STLreader} provides the desired functionality. In the case that the STL file you want to read is too large, you can use ParaView's filter "Decimate" to reduce the number of facets.

The constructor of the class \class{STLreader} takes two necessary and three optional arguments.
\begin{lstlisting}[language=myc++]
STLreader(const std::string fName, T voxelSize, T stlSize=1,
          unsigned short int method = 2, bool verbose = false);
\end{lstlisting}
\begin{itemize}
\item \texttt{fName}: The filename of the STL file to be read.
\item \texttt{voxelSize}: The intended spatial step size for the simulation in SI units (m).
\item \texttt{stlSize}: Conversion factor if the STL file is not given in SI units. E.g.:\ For an STL file in cm, this factor is $\texttt{stlSize} = 0.01$.
\item \texttt{method}: Switch between methods for determining inside and outside of geometry.
  \begin{itemize}
  \item    default: fast, less stable
  \item    1: slow, more stable (for untight STLs)
  \end{itemize}
\item \texttt{verbose}: Switch to get more output.
\end{itemize}

\textbf{Functionality}: The STL file is read and stored in the class \class{STLmesh}. A class \class{Octree} is instantiated of side-length $\text{rad} = 2^{j-1} \cdot \texttt{voxelSize}, j\in \mathbb{N}$ with $j$ such that a cube with diameter $2 \text{rad}$ covers the entire \texttt{STL}. Intersections of triangles and the nodes of the \class{Octree} are computed and an index of the respective triangles is stored in each node. A node is a leaf if either $\text{rad} = \text{voxelSize}$ or if it does not contain any triangles. \\
In a second step, it is determined whether a leaf is inside the STL geometry by one of the following methods:
\begin{itemize}
\item (Default) One ray in Z-direction is defined for each voxel in XY-layer. All nodes are indicated on the fly (faster, less stable).
\item Define three rays (X-, Y-, Z-direction) for each leaf and count intersections with STL for each ray. Odd number of intersection means inside. The final state is decided by a majority vote (slower, more stable).
\end{itemize}

\section{XML Parameter Files}\label{sec:xmlparams}

In OpenLB essential simulation parameters can be placed in a XML file.
This is a useful feature, since once a program is compiled, the parameters can be changed through the XML file and recompilation is redundant.
As a consequence whenever parameter fitting or general simulations are wanted, this approach can help you since only editing the XML file is necessary.
The parsing is implemented in the header file \path{io/xmlReader.h}.

The general format for the XML files is:
\begin{lstlisting}[language=XML]
<Param>
  <Output>
    <Log>
      <VerboseLog> true </VerboseLog>
    </Log>
  </Output>
  <VisualizationImages>
    <Filename> image </Filename
  </VisualizationImages>
</Param>
\end{lstlisting}
All parameters need to be wrapped in a \texttt{<Param>} tag. To open a config file, you just pass a string with the file name to the class constructor of \class{XMLreader}.
\begin{lstlisting}[language=myc++]
std::string fName("demo.xml");
XMLreader config(fName);

bool _verboseLog;
std::string imagename;
XMLreader outputConfig = config["Output"];

config.readOrWarn<bool>("Output", "Log","VerboseLog",_verboseLog);
outputConfig.readOrWarn<bool>("Log", "VerboseLog","",_verboseLog);
config.readOrWarn<std::string>("VisualizationImages", "Filename", "", imagename);
\end{lstlisting}
First, an \class{XMLreader} object \texttt{config} is created. There are multiple ways to access the configuration data.
To select the tag you would like to read, you just use an associative array like syntax as shown above.

To get a specific value out of an XML parameter file, there are multiple methods.
One is to pass a predefined variable to the method \texttt{readOrWarn}, which reads the respective value and prints a warning in case the data type is not matching or the value cannot be found.
For large subtrees with lots of parameters, you can also create a subobject.
For this, you just have to reassign your selected subtree to a new \class{XMLreader}-object as is done above for \texttt{Output}.

\section{Visualization with ParaView}
\label{sec:paraview}

As already mentioned, there are several data formats that can be used in ParaView. Use `File -- Open' and choose the set of data you want to use. In regards to OpenLB it is enough to open the file with the ending \texttt{.pvd}, since it contains a reference to the \texttt{.vti} files.
The chosen files should now be part of the `Pipeline Browser', which should be on the left hand side (if any of the panels are missing you can add them in the `View' menu on the top). Click on `Apply' in the `Properties' panel (usually located below the `Pipeline Browser') after opening.

Your data should now be visible in the center window. From within the `Properties' or in one of the top tool bars, you can change the `Coloring' properties, which selects what shall be displayed (\eg physical velocity, phys pressure), which part of this choice shall be displayed (\eg magnitude, x-value) and the way it is colored.

Make sure that `3D' is part of the tool bar directly above the window where you can see your objects. If you cannot find it click on `2D' which should be written instead and change it to `3D' by doing this. The commands for moving your whole set of visible objects and thus changing the perspective are the following:
\begin{itemize}
  \item Using the mouse wheel, you can zoom in and out.
  \item Using the right mouse button or `Ctrl + left mouse button', you can move the object to the background or the foreground. In comparison to zooming in and out, this changes the level of the 3D-effect.
  \item Using the left mouse button allows you to turn the object.
  \item Clicking the mouse wheel allows you to move the object center.
\end{itemize}
Of course you can also stick to `2D', although in this case the mouse commands might change a bit.

You can visualize the temporal development of your simulation using the `Play' button and the related buttons directly next to it. If you want to go to a certain time step, use the input field `Time', which is also located here.

To manipulate your data in ParaView, numerous so-called `Filters' are provided in the `Filters' menu in the top bar.

\section{Clip}
With this filter, you can cut off parts of your objects, for example, to make it possible to look inside the geometry. There are several tool options to determine which part is cut off. You can choose between plane, box and sphere.

If the ``wrong'' side is cut off, check `inside out' to make the other side visible.

\subsubsection{Contour}
Using `Contour' you can show lines or planes of certain data values, which you can set.

\section{Glyph}
If you have a point data set, you can represent it as spheres using the filter `Glyph' and choosing `Sphere' as setting for `Glyph Type'. Using the resolution settings, you can smooth the surface to make the sphere look more rounded.

There are alternative ways to represent the data. As an example, arrows can be used to show the direction of a velocity. Check `Glyph Type' for further possibilities.

\section{Stream Tracer}
The Stream Tracer filter is a powerful tool that allows users to visualize the flow as streamlines, making for example turbulent flow more visible. By placing the seed of the stream tracer next to the area of interest, users can apply this filter. There are two possible forms of the seed: it can be set as a pointcloud or a line.

To use the Stream Tracer filter, simply select the seed type that best fits your needs and place it in the desired location. In the case of a pointcloud seed, the flow that goes through the sphere is transformed, while a line seed transforms the flow along the line. An example of the OpenLB example \texttt{nozzle3d} visualized with streamlines using a pointcloud seed can be seen in Figure~\ref{fig:streamlines}.

\begin{figure}[ht]
    \centering
    \includegraphics[width =\textwidth]{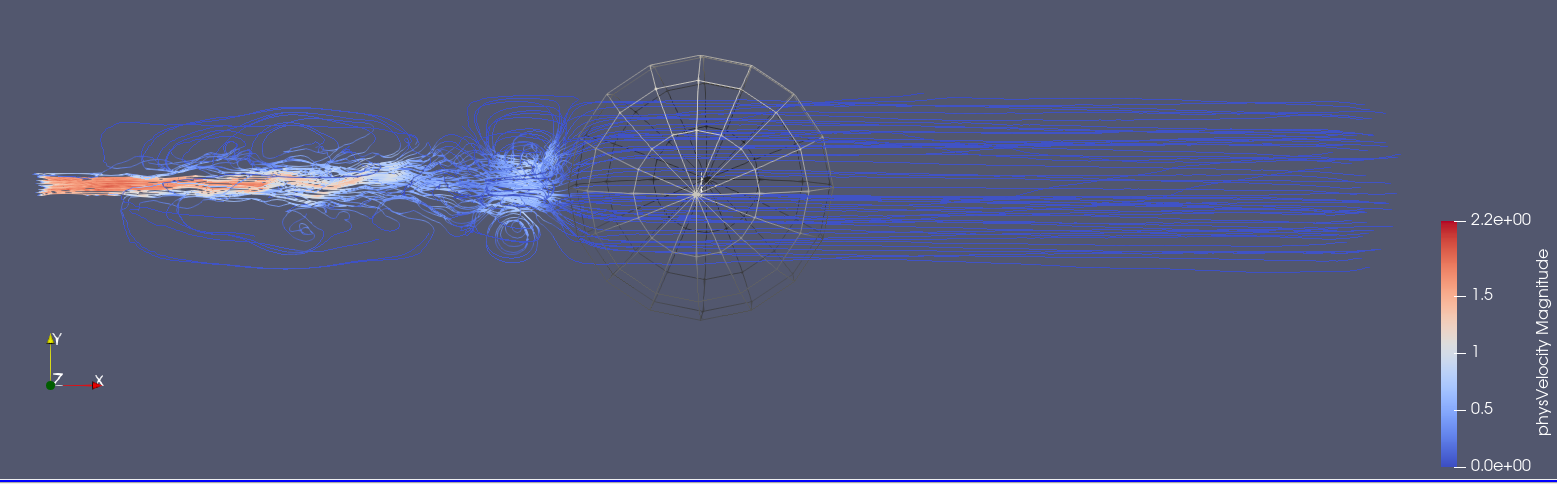}
    \caption{OpenLB example \texttt{nozzle3d} visualized with streamlines. In this case as seed, a pointcloud was used. }
    \label{fig:streamlines}
\end{figure}

\section{Resample To Image}
Another filter for visualization purposes is \texttt{Resample To Image}. It applies a volumetric raymarching algorithm on the data. The Figure~\ref{fig:resample_nozzel3d} shows the result when this filter is applied.

Another useful filter for visualization purposes is the \texttt{Resample To Image} filter. This filter applies a volumetric raymarching algorithm to the data, allowing users to visualize the data in a new way.
To use this filter, simply select the \texttt{Resample To Image} option and apply it to your data. The resulting visualization will be a 3D image that can be rotated and explored in real-time.
An example of the OpenLB example \texttt{nozzle3d} visualized using the \texttt{Resample To Image} filter can be seen in Figure~\ref{fig:resample_nozzel3d}. 

\begin{figure}[ht]
    \centering
    \includegraphics[width = \textwidth]{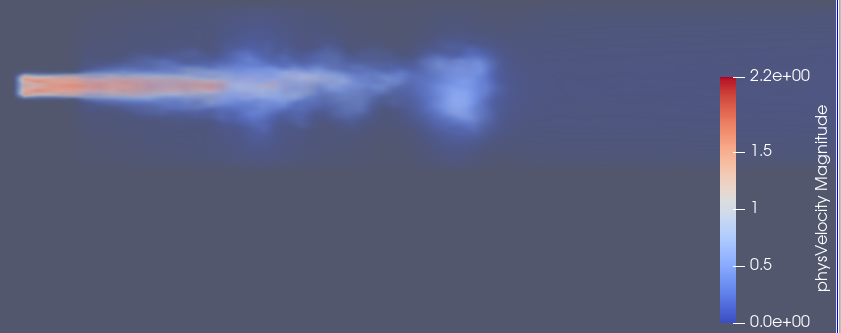}
    \caption{OpenLB example \texttt{nozzle3d} visualized with the help of the \texttt{Resample To Image} filter}
    \label{fig:resample_nozzel3d}
\end{figure}

This visualization provides a unique perspective on the data and can help users to better understand the underlying patterns and structures in the data.
It's important to note that the \texttt{Resample To Image} filter requires a significant amount of computational resources to generate the 3D image. As such, it may not be suitable for use on larger datasets or on less powerful hardware. Additionally, the parameters of the filter can be adjusted to fine-tune the resulting visualization, but this requires some knowledge of the filter and its underlying algorithm.

\subsection{Application of Functors} \label{sec:functor_aplication}
The concept of functors benefits from generality and therefore, they are used for many applications.

\subsubsection{Extract Simulation Data}
Velocity, pressure and other information can be extracted from the lattice using predefined functors, see Listing~\ref{lst:SuperLatticF}.
All they need to know is a \class{SuperLatticeXD} and an \class{UnitConverter} - if dimension or physical units are wanted.
\begin{lstlisting}[language=myc++,caption={Code example for calculating velocity and pressure using functors.},label={lst:SuperLatticF}]
// Create functors
SuperLatticePhysVelocity3D<T,DESCRIPTOR> velocity(sLattice, converter);
SuperLatticePhysPressure3D<T,DESCRIPTOR> pressure(sLattice, converter);
\end{lstlisting}

\subsubsection{Define Analytic Functions}
Often the inflow velocity has Poiseuille profile which is defined analytically, by means of a function.
OpenLB provides analytic functors to define \eg a Poiseuille velocity profile, random values, linear and constant values.
\begin{lstlisting}[language=myc++,caption={Define a Poiseuille velocity profile for inflow boundary condition.}]
Poiseuille2D<T> poiseuilleU(superGeometry, 3, maxVelocity, distance2Wall);
\end{lstlisting}

\subsubsection{Interpolation}
Another case for interpolation functors is the conversion of a given analytical functor, such as an analytical solution to a \class{SuperLattice} functor.
Afterwards, the difference can be easily calculated with the help of the functor arithmetic, see Listing~\ref{lst:computeErrorNorm}.
Finally, specific norms implemented as functors facilitate analysis of convergence.
\begin{lstlisting}[language=myc++,caption={Transition from an analytical functor to a lattice functor.}]
// define a analytic functor: R^3 -> R
AnalyticalConst3D<double,double> constAna(1.0);
// get analytic functor on the lattice: N^3 -> R
SuperLatticeFfromAnalyticalF3D<double,DESCRIPTOR> constLat(constAna,
                                                           lattice);
\end{lstlisting}
Application of this is shown in the example \path{poiseuille2d}, which is discussed in Section~\ref{sec:poiseuille2d and poiseuille3d}

\subsubsection{Arithmetic and Advanced Functor Usage}
\label{sec:advancedFunctorUsage}
Functors can be added, subtracted, etc. which is a very useful and elegant method to treat data.
Listing~\ref{lst:computeErrorNorm} shows how to compute the relative error over the whole three dimensional domain.
\begin{lstlisting}[language=myc++,caption={Computation of a relative error with respect to $L^2$-norm.},label={lst:computeErrorNorm}]
int input[1];
double normAnaSol[1], absErr[1], relErr[1];
// define analytical solution: R^3 -> R
// for snake of simplicity it is a constant function,
// however it may be any specialization of AnalyticalF3D
AnalyticalConst3D<double,double> dSol(1.0);
// get analytical solution on the lattice: N^3 -> R
SuperLatticeFfromAnalyticalF3D<double,DESCRIPTOR> dSolLattice(dSol, lattice);
// get density out of simulation data
SuperLatticeDensity3D<T,DESCRIPTOR> d(lattice);
// compute absolute error
SuperL2Norm3D<double> dL2Norm(dSolLattice - d, superGeometry, 1);
// compute norm of solution
SuperL2Norm3D<double> dSolL2Norm(dSolLattice, superGeometry, 1);
dL2Norm(absErr, input);         // access absolute error
dSolL2Norm(normAnaSol, input);  // access norm of the solution
relErr[0] = absErr[0] / normAnaSol[0];
clout << "denstity-L2-error(abs)=" << absErr[0] << ";"
      << "denstity-L2-error(rel)=" << relErr[0] << std::endl;
\end{lstlisting}
For more detail, see the source code of example~\ref{sec:poiseuille2d and poiseuille3d}.

Assemble geometry with geometric primitives of type \class{IndicatorFXD}.
\begin{lstlisting}[language=myc++,caption={Deploy functor arithmetic to build geometry data.}]
Vector<double,2> extendChannel(lx0, ly0);
Vector<double,2> originChannel;
IndicatorCuboid2D<double> channel(extendChannel, originChannel);
// setup step
Vector<double,2> extendStep(lx1, ly1);
Vector<double,2> originStep;
IndicatorCuboid2D<double> step(extendStep, originStep);
// remove step from channel
IndicatorIdentity2D<double> channelIdent(channel-step);
\end{lstlisting}

\subsubsection{Setting Boundary Value}
Boundary cells are marked by a certain material number in the \class{SuperGeometryXD}.
Using a functor, velocities can be set simultaneously on all cells of this material.
First, a vector that characterizes the maximum flow velocity and its directions is necessary.
Then, a special functor uses this vector to initialize a Poiseuille profile.
The direction can be extracted in the case of axis-parallel inflow regions automatically from the \class{SuperGeometryXD}.
In the last step, the \class{SuperLattice} initializes all cells of a certain material given by the \class{SuperLatticeXD} with the velocities computed by the functor.
\begin{lstlisting}[language=myc++,caption={Code example for setting a Poiseuille velocity profile and a constant pressure boundary in cylinder3d.}]
// Creates and sets the Poiseuille inflow profile using functors
double maxVel = converter.getCharLatticeVelocity();
CirclePoiseuille3D<double> poiseuilleU(superGeometry, 3, maxVel, distance2Wall);
sLattice.defineU(superGeometry, 3, poiseuilleU);
\end{lstlisting}

\subsubsection{Flux Functor}
The \emph{flux} of a quantity is defined as the rate at which this quantity passes through a fixed boundary per unit time.
As a \emph{mathematical concept}, flux is represented by the surface integral of a vector field
\begin{align}
\Phi = \int \bm{F} \cdot d\bm{A} ~,
\end{align}
where $ \bm{F} $ is a vector field, and $ d\bm{A} $ is an area element of the surface $ A $, in the direction of the surface normal $ \bm{n} $.

\emph{Flux functors} calculate the discrete flux
\begin{align}
\Phi_h = h^2\sum_{i}\bm{f}_{i}\cdot\bm{n} ~,
\end{align}
with $ h $ as the grid length of the surface and $ \bm{f}_{i} $ the vector of the quantity at grid point $ i $.

As the grid of the area has to be independent from the lattice, the value of $ \bm{f}_{i} $ will be interpolated from the surrounding lattice points.
In the general case this discrete value is calculated by \class{SuperPlaneIntegralF3D}. Note that the reduction of the relevant surface is performed by \class{BlockReduction3D2D} and that \class{SuperPlane\-IntegralF3D} adds only the multiplication by the area unit as well as the normal vector for multidimensional $\bm{f}_{i}$.
In turn specific flux functors such as \class{SuperPlaneIntegralFluxVelocity3D} only add functor instantiation and print methods.
So, for the \class{SuperPlaneIntegralF3D} functor a surface needs to be defined. OpenLB currently supports using subsets of hyperplanes as the surfaces on which to calculate a flux.

Such a \emph{hyperplane} can be defined by an origin and two span vectors, an origin and a normal vector or a $3D$ circle indicator. \class{BlockReduction3D2D} interpolates the full intersection of hyperplane and mother geometry. Optionally this maximal plane may be further restricted by arbitrary $2D$ indicators.

Note that \class{SuperPlaneIntegralF3D} as well as all specific flux functors provide a variety of constructors accepting various hyperplane parametrizations. For full control you may consider explicitly constructing a \class{Hyperplane3D} instance.

The \emph{discretization} of a hyperplane parametrization (given by \class{Hyperplane3D}) into a discrete lattice is performed by \class{HyperplaneLattice3D}.

\qquad\\
\textbf{Step 1}: Define the hyperplane by\\
a) \emph{origin and two span vectors}
\begin{lstlisting}[language=myc++]
Vector<T,3> origin;
Vector<T,3> u, v;
\end{lstlisting}
b) \emph{origin and normal vector}
\begin{lstlisting}[language=myc++]
Vector<T,3> origin;
Vector<T,3> normal;
\end{lstlisting}
c) \emph{normal vector (centered in mother cuboid)}
\begin{lstlisting}[language=myc++]
Vector<T,3> normal;
\end{lstlisting}
d) \emph{circle indicator}
\begin{lstlisting}[language=myc++]
IndicatorCircle3D<T> circleIndicator(center, normal, radius);
\end{lstlisting}
e) \emph{arbitrary hyperplane}
\begin{lstlisting}[language=myc++]
// example parametrization of a hyperplane centered in the mother cuboid and normal to the Z-axis
Hyperplane3D<T> hyperplane()
  .centeredIn(cuboidGeometry.getMotherCuboid())
  .normalTo({0, 0, 1});
\end{lstlisting}
\qquad\\ \textbf{Step 1.1} (optional): Define the hyperplane discretization by
\newline a) \emph{grid length}
\begin{lstlisting}[language=myc++]
T h = converter.getLatticeL();
HyperplaneLattice3D<T> hyperplaneLattice(
  cuboidGeometry,
  Hyperplane3D<T>().originAt(origin).spannedBy(u, v),
  h);
\end{lstlisting}
b) \emph{grid resolution}
\begin{lstlisting}[language=myc++]
HyperplaneLattice3D<T> hyperplaneLattice(
  cuboidGeometry,
  Hyperplane3D<T>().originAt(origin).spannedBy(u, v),
  600); // resolution
\end{lstlisting}
\qquad\\
\textbf{Step 1.2} (optional): Define the flux-relevant lattice points by
\newline
a) \emph{list of material numbers}
\begin{lstlisting}[language=myc++]
std::vector<int> materials = {1, 2, 3};
\end{lstlisting}
a) \emph{arbitrary indicator}
\begin{lstlisting}[language=myc++]
SuperIndicatorF3D<T> integrationIndicator...
\end{lstlisting}
\qquad\\
\textbf{Step 1.3} (optional): Restrict the discretized intersection of hyperplane and geometry by
\newline
a) \emph{2D circle indicator (relative to hyperplane origin)}
\begin{lstlisting}[language=myc++]
T radius = 1.0;
IndicatorCircle2D<T> subplaneIndicator({0,0}, radius);
\end{lstlisting}
a) \emph{arbitrary 2D indicator (relative to hyperplane origin)}
\begin{lstlisting}[language=myc++]
IndicatorF2D<T> subplaneIndicator...
\end{lstlisting}
\qquad\\
\textbf{Step 2}: Create a \class{SuperF3D} functor for
\newline
a) \emph{velocity flow}
\begin{lstlisting}[language=myc++]
SuperLatticePhysVelocity3D<T,DESCRIPTOR> f(sLattice, converter);
\end{lstlisting}
b) \emph{pressure}
\begin{lstlisting}[language=myc++]
SuperLatticePhysPressure3D<T,DESCRIPTOR> f(sLattice, converter);
\end{lstlisting}
c) \emph{any other \class{SuperF3D} functor}
\begin{lstlisting}[language=myc++]
SuperF3D<T> f...
\end{lstlisting}
\qquad\\
\textbf{Step 3}: Instantiate \texttt{SuperPlaneIntegralF3D} functor depending on how the hyperplane was defined and discretized.\\
a) \emph{using origin, two span vectors and materials list}
\begin{lstlisting}[language=myc++]
SuperPlaneIntegralF3D<T> fluxF(
  f, superGeometry, origin, u, v, materials);
\end{lstlisting}
b) \emph{using origin, normal vector and materials list}
\begin{lstlisting}[language=myc++]
SuperPlaneIntegralF3D<T> fluxF(
  f, superGeometry, origin, normal, materials);
\end{lstlisting}
c) \emph{using normal vector and materials list}
\begin{lstlisting}[language=myc++]
SuperPlaneIntegralF3D<T> fluxF(f, superGeometry, normal, materials);
\end{lstlisting}
d) \emph{using 3D circle indicator and materials list}
\begin{lstlisting}[language=myc++]
SuperPlaneIntegralF3D<T> fluxF(
  f, superGeometry, circleIndicator, materials);
\end{lstlisting}
e) \emph{using arbitrary hyperplane and integration point indicator}
\begin{lstlisting}[language=myc++]
SuperPlaneIntegralF3D<T> fluxF(
  f, superGeometry, hyperplane, integrationIndicator);
\end{lstlisting}
f) \emph{using arbitrary hyperplane, integration point indicator and subplane indicator}
\begin{lstlisting}[language=myc++]
SuperPlaneIntegralF3D<T> fluxF(f, superGeometry, hyperplane, integrationIndicator, subplaneIndicator);
\end{lstlisting}
g) \emph{using arbitrary hyperplane lattice, integration point indicator and subplane indicator}
\begin{lstlisting}[language=myc++]
SuperPlaneIntegralF3D<T> fluxF(f, superGeometry, hyperplaneLattice, integrationIndicator, subplaneIndicator);
\end{lstlisting}
\qquad\\
\textbf{Step 4}: Get results using \texttt{operator()}
\begin{lstlisting}[language=myc++]
int input[1]; // irrelevant
T output[5];
fluxF(output, input);
\end{lstlisting}
\begin{itemize}
\item \textit{output[0]}: flow rate or plane integral (if quantity has dimension 1)
\item \textit{output[1]}: size of the area
\item \textit{output[2..4]}: flow vector (ie. vector of summed quantities)
\end{itemize}

\qquad\\
In many cases the functor argument is either the velocity or the pressure functor.\\
Thus \textbf{Step 2} and \textbf{Step 3} may be combined using \class{SuperPlaneIntegralFluxVelocity3D} respectively \class{SuperPlaneIntegralFluxPressure3D}. Their constructors are mostly identical to the ones provided by \class{SuperPlaneIntegralF3D}. In fact the only difference is that the first functor argument is replaced by references to \class{SuperLattice} and \class{UnitConverter}.

\qquad\\
\textbf{Step 2.1)}: Combined steps for velocity flux
\begin{lstlisting}[language=myc++]
SuperPlaneIntegralFluxVelocity3D<T> vFlux(superLattice, converter, ...);
\end{lstlisting}
\qquad\\
\textbf{Step 2.2)}: Combined steps for pressure flux
\begin{lstlisting}[language=myc++]
SuperPlaneIntegralFluxPressure3D<T> pFlux(superLattice, converter, ...);
\end{lstlisting}
\qquad\\
\textbf{Step 3.1)}: Output region size, volumetric flow rate and mean velocity
\begin{lstlisting}[language=myc++]
vFlux.print(std::string regionName,
  std::string fluxSiScaleName, std::string meanSiScaleName);
\end{lstlisting}
\begin{itemize}
\item \texttt{fluxSiScaleName:} 'ml/s' or 'l/s' or '\ ' (default=$ m^3/s $)
\item \texttt{meanSiScaleName:} 'mm/s' or '\ ' (default=$ m/s $)
\end{itemize}
\qquad
\textbf{Step 3.2)}: Output region size, force and mean pressure
\begin{lstlisting}[language=myc++]
pFlux.print(std::string regionName,
  std::string fluxSiScaleName, std::string meanSiScaleName);
\end{lstlisting}
\begin{itemize}
\item \texttt{fluxSiScaleName:} 'MN' or 'kN' or '\ ' (default=$ N $)
\item \texttt{meanSiScaleName:} 'mmHg' or '\ ' (default=$ Pa $)
\end{itemize}

\subsubsection{Discrete Flux Functor}

If a hyperplane is axis-aligned, flux functors may optionally be used in \emph{discrete} mode. Passing \class{BlockDataReductionMode::Discrete} as the last argument to any plane integral or flux constructor instructs the internal \class{BlockReduction3D2D} instance to reduce the hyperplane by evaluating the underlying functor at the nearest lattice points instead of by interpolating physical positions.

Note that this imposes restrictions on the accepted hyperplane and its lattice:

\begin{itemize}
\item \class{Hyperplane3D} normal must be orthogonal to a pair of unit vectors
\item \class{HyperplaneLattice3D} spacing must equal the distance between lattice nodes
\end{itemize}

The restriction on the hyperplane lattice spacing is fulfilled implicitly when automatic lattice parametrization is used. 

\begin{lstlisting}[language=myc++]
// discrete flux usage in examples/aorta3d
SuperPlaneIntegralFluxVelocity3D<T> vFluxInflow( sLattice, converter, superGeometry, inflow, materials, BlockDataReductionMode::Discrete );
\end{lstlisting}

\subsubsection{Wall Shear Stress Functor}
The \emph{Wall Shear Stress} is defined as the parallel force per unit area exerted by a fluid on a wall.
In the context of macroscopic fluid mechanics the \emph{Wall Shear Stress} of a Newtonian fluid is given by
\begin{equation}\label{wss_analytical}
\tau_W = \mu \frac{\partial \bm{u}}{\partial y}\bigg|_{y=0} ~,
\end{equation}
where $\mu$ is the dynamic viscosity, $u$ is the velocity field and $y$ the coordinate perpendicular to the wall.
The \emph{Wall Shear Stress Functor} calculates the discrete Wall Shear Stress
\begin{equation}\label{wss.tractionsimplified}
\tau_W = \text{\boldmath$\sigma$} \cdot \bm{n} -((\text{\boldmath$\sigma$}\cdot \bm{n})\cdot \bm{n})\cdot \bm{n} ~,
\end{equation}
where $\text{\boldmath$\sigma$}$ is the Cauchy stress tensor and $\bm{n}$ the local unit normal vector of the surface.
Since the lattice stress tensor $\text{\boldmath$\Pi$}$ is not defined on boundary cells, it is read out from an adjacent fluid cell in a discrete velocity direction associated with each boundary cell. The unit normal vector is obtained by a given \class{IndicatorF3D} instance, which is slightly increased in size. See \path{examples/poiseuille3d} for usage details. Due to the staircase approximation of the boundary, the wall shear stress calculation is first order accurate.
\begin{figure}[ht]
	\centering
	\includegraphics[width=0.6\textwidth]{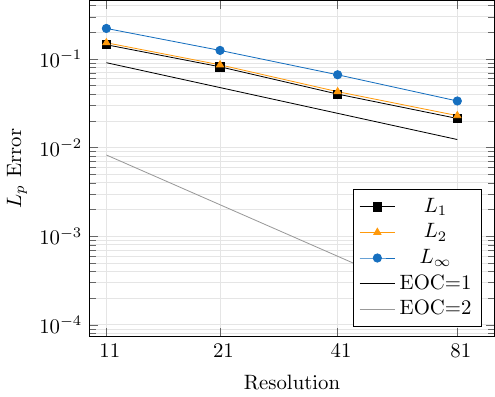}
	\caption{Relative error of Wall Shear Stress $L_p$ norm in \protect\path{poiseuille3d}.}
	\label{fig:wss_eoc}
\end{figure}

\subsection{Error Norm Functors}

While relative and absolute error norms may be calculated manually using functor arithmetic (see \ref{sec:advancedFunctorUsage}), they are also available as distinct functors. As such it is preferable to utilize \class{Super\-RelativeErrorLpNormXD} and \class{SuperAbsoluteErrorLpNormXD} if one uses the common definition of relative and absolute error norms.

Let \texttt{wantedF} be the simulated solution functor and \texttt{f} the analytical solution, then 
\begin{align}
\texttt{SuperRelativeErrorLpNormXD} &\text{ implements } \frac{\|\texttt{wantedF} - \texttt{f}\|_p}{\|\texttt{wantedF}\|_p} ~,  \\
\texttt{SuperAbsoluteErrorLpNormXD} &\text{ implements } \|\texttt{wantedF} - \texttt{f}\|_p ~.
\end{align}
An example of how to use these error norm functors in practice is given by the Poiseuille flow example as described in section~\ref{sec:poiseuille2d and poiseuille3d}.

\begin{lstlisting}[language=myc++,caption={L1 velocity error in \texttt{poiseuille2d}},label=lst:poiseuille2dVelError]
Poiseuille2D<T> uSol(axisPoint, axisDirection, maxVelocity, radius);
SuperLatticePhysVelocity2D<T,DESCRIPTOR> u(sLattice, converter);
auto indicatorF = superGeometry.getMaterialIndicator(1);

SuperAbsoluteErrorL1Norm2D<T> absVelErrorNormL1(u, uSol, indicatorF);
absVelErrorNormL1(result, tmp);
clout << "velocity-L1-error(abs)=" << result[0];
SuperRelativeErrorL1Norm2D<T> relVelErrorNormL1(u, uSol, indicatorF);
relVelErrorNormL1(result, tmp);
clout << "; velocity-L1-error(rel)=" << result[0] << std::endl;
\end{lstlisting}

Further implementation details are touched upon in section~\ref{sec:functorComposition}.

\subsubsection{Grid Refinement Metric Functors}

\class{SuperLatticeRefinementMetricKnudsen(2,3)D} implements an automatic block-level grid refinement criterion as described by \citeauthor{lagrava:15} in \citetitle{lagrava:15}~\cite{lagrava:15}. This criterion uses the quality of the cell-local Knudsen number approximation as measured by \class{SuperLatticeKnudsen*D} to judge the adequacy of the block resolution.


\chapter{Flow Control and Optimization} 

In almost every application, the optimization of a process regarding certain objectives is of general interest. 
In optimal flow control we are interested in optimizing the flow problem regarding a specific need, \ie to find optimal flow parameters which lead to the desired flow properties. 
OpenLB provides a framework to solve such optimization problems with iterative, gradient-based methods. Basically, a flow simulation is encapsulated into the optimization framework which performs several flow simulations. In each step the input parameters are varied in order to approach the optimal solution. 
This chapter explains first how these optimal solutions are found and second how this is realized within OpenLB.
Therefore, this section is structured as follows. Section~\ref{sec:opti-intro} gives a brief overview of the theoretical optimization concepts and Section~\ref{sec:opti-implementation} deals with their implementation in OpenLB. 
Note that, in this chapter, specific notation is introduced where necessary and might not conform to the rest of the user guide. 
For examples, tensors and scalars are mostly printed in similar fonts below. 

\section{Introduction}\label{sec:opti-intro}

This chapter is dedicated to the solution of optimization problems of the general form
\begin{equation}
  \textrm{Find } \alpha \textrm{, s.t. } J(f, \alpha) \textrm{ is minimized while } G(f, \alpha) = 0 ~.
\end{equation}
In this context, $\alpha \in \mathbb{R}^d$ is called the \emph{control variable}, $f$ is the state, $J$ is the \emph{objective functional} and $G$ is the side condition.
For simplicity, we assume in our code that the side condition $G(f,\alpha)=0$ yields a unique solution $f$ for any admissible control $\alpha$, s.t. we have $f=f(\alpha)$ (unrestricted optimization approach). 

Possible optimization problems are:
\begin{itemize}
  \item Finding the minimum of an analytic function $J$ w.r.t.\ an argument vector $\alpha$, \eg as a postprocessing step.
  \item Optimization of flow simulation parameters: e.g., the performance of a mixer could be optimized w.r.t.\ simulation parameters.
    In this case, $\alpha$ contains the free simulation parameters, $G$ are the governing fluid flow equations (mathematical model), $f$ is the solution of the governing equations for a given $\alpha$ (\eg the concentration field) and $J$ measures the mixing quality.
  \item Parameter identification: sometimes, physical simulation parameters are not known, but with the help of additional data (\eg NMR-data of the fluid velocity), we can reconstruct them:
    Here, we set $\alpha$ to be the unknown parameters, $G$ are the governing equations, $f$ is their solution and $J$ measures the difference between simulated and ''true'' data (\eg $J = \tfrac{1}{2} \Vert \mathbf{u}(f(\alpha))- \mathbf{u}_{NMR} \Vert^2$). Using optimization, we can then find the unknown parameter $\alpha$ with the least difference to the ''true'' data. 
\end{itemize}

The optimization framework provided in OpenLB is rather designed for large-scale applications like fluid flow simulations, where a time critical step (regarding computation time) is the function evaluation and not the optimization routine. 
For an application-oriented introduction, we also refer to the example \class{showcaseRosenbrock}.

\subsection{Iterative Approach -- Line Search Algorithm}\label{subsec:opti-linesearch}

The optimization methods in the OpenLB framework employ an iterative, gradient-based approach based on the line search concept. 
Starting with an initial guess for the control variables $\alpha_k$, the next set of control variables $\alpha_{k+1}$ are found by proceeding along a direction $d$ with a step size $s$. With each step the objective function should be minimized, \ie
\begin{align}
    \alpha_{k+1} = \alpha_k + d\,s, \quad with \quad J(\alpha_{k+1}) < J(\alpha_k) ~.
\end{align}

\noindent After each iteration step, either the objective function $J$ or its derivatives regarding the current control variables $\frac{\partial J}{\partial \alpha}$ are evaluated. If they are below a user-defined threshold, the algorithm exits. The open questions at this point are listed together with their corresponding subsections giving answers to those questions:
\begin{enumerate}
    \item How to choose the (descent) direction? (\ref{subsec:opti-descent})
    \item How to choose the step size? (\ref{subsec:opti-step})
    \item How can we compute the derivatives? (\ref{subsec:opti-derivative})
\end{enumerate}

\subsection{Descent Algorithms}\label{subsec:opti-descent}

The descent algorithm computes the descent direction $d$ for our next iteration step, \ie the direction where we likely approach the optimal solution. One possible approach is to use the local gradient of the objective function $\frac{\partial J}{\partial \alpha}$ as the descent direction. This method is referred to as the steepest descent approach. Since only first order derivatives are considered here, this method is relatively slow (it requires more iteration steps) compared to Newton or Quasi-Newton methods while its main advantage is its stability.
Quasi-Newton methods such as LBFGS achieve higher convergence orders since these additionally include the approximated second derivatives. In the OpenLB framework, the steepest descent and LBFGS algorithm are provided for computing the descent direction.

\subsection{Step Size Algorithms}\label{subsec:opti-step}

The choice of suitable step sizes is a nontrivial task; if the step sizes are too large we may overshoot and miss the optimum. On the other side, if the steps are too small we will need a lot of iterations until the optimal solution is reached, resulting in longer computation duration. To find the optimal step size, step conditions like Armijo, normal or strong Wolfe-Powell rules can be applied.

The Armijo rule prevents choosing too large step sizes by checking if we achieve a certain minimal decrease in our objective function. The objective function in the next iteration step $k+1$ can be approximated as
\begin{align}
    J(\alpha_{k+1}) \approx J(\alpha_k) + s_k \nabla J(\alpha_k) ^ T d ~,
\end{align}
where $s_k$ is the current step size. First, relatively large step sizes are chosen and the objective function is evaluated at $\alpha_{k+1}$. If the objective function is decreased sufficiently in comparison to the expected value for the objective function, the algorithm exits with the current step size; if not, a smaller step size is chosen. This condition can be written as
\begin{align}
    J(\alpha_{k+1}) = J(\alpha_k + s_k\,d) \leq J(\alpha_k) + \rho \,s_k \nabla J(\alpha_k) ^ T d \quad \text{with } \rho \in [0,1] ~.
\end{align}
These steps are repeated until this condition is fulfilled. To add further flexibility a scalar $\rho$ is introduced here allowing us to define what portion of the expected linear decrease we want to assure by adjusting the threshold. Note that the Armijo rule requires an evaluation of the objective function for each iteration step.

The Wolfe-Powell rule can be classified into the normal and strong formulation. Both contain the Armijo rule and consider an additional criterion which prevents us from choosing step sizes which are too small. For this, the directional derivative (scalar product of the proceeding direction and the derivative) on the old iteration is compared to the directional derivative on the new iteration point. Assuming that we are at the optimal solution, the directional derivative is zero. During our line search we are descending regarding the objective function which means that our directional derivative is always a negative value. If we move closer towards the optimum, the directional derivative must increase to get eventually to zero. Thus, we can check
\begin{align}
    \nabla J(\alpha_k + s_k\,d)^T d \geq \delta \nabla J(\alpha_k)^T d ~,
\end{align}
where $\delta$ is a tuning parameter. This is referred to as the normal Wolfe-Powell formulation. The strong formulation reads
\begin{align}
    |\nabla J(\alpha_k + s_k\,d)^T d| \leq - \delta \nabla J(\alpha_k)^T d ~.
\end{align}
The step size is varied until both the Armijo rule and the second condition of either the normal or strong Wolfe-Powell formulation are satisfied (or by exceeding the maximal number of attempts). Note that the Wolfe-Powell rule requires for each iteration step an evaluation of the objective function and its derivative at that position.

\subsection{Derivative Computation}\label{subsec:opti-derivative}

For almost all algorithms the evaluation of function gradients are necessary. For those there are also several approaches such as finite-difference schemes and algorithmic differencing. The following options are implemented in OpenLB:
\begin{itemize}
  \item Usage of forward/ central difference quotients: evaluate $J$ for neighboring values of $\alpha$ and compute the difference quotients.
    This is the simplest method and for a small number of control variables it is fast. However, it is the least accurate method.
  \item Forward automatic differentiation: evaluate $J$ for operator-overloaded variables of type \texttt{ADf<T,n>}, where \texttt{T} is the underlying arithmetic type and \texttt{n} is the number of control variables.
    This method is a little slower that difference quotients, but it usually returns derivatives at full machine precision.
  \item Adjoint LBM: adjoint Lattice-Boltzmann equations are used, cf.~\cite{krause:12a, klemens:18a}.
    This method is perfect for distributed control problems since the computational expense remains constant for any number of control variables.
    However, this method requires the (theoretical) calculation of adjoint LB-equations which is a problem-specific task.
\end{itemize}

\section{Implementation}\label{sec:opti-implementation}

The implementation is separated into two key steps: the \class{OptiCase} classes define how the objective functional $J$ and its gradient are computed as functions of $\alpha$.
The \class{Optimizer} classes define optimization methods such as steepest descent or step size algorithms (independent of the question, how the function and gradient evaluations are performed).

\subsection{Optimizer Classes: Optimization Methods}

The common methods steepest descent, LBFGS and Barzilai--Borwein are implemented as children classes of the \class{Optimizer} classes.
Their various free parameters (\eg an upper limit for the number of optimization steps) can be passed via the constructor or read from an \texttt{xml} file. 
In many situations, the control variables have to remain in a certain range in order to guarantee physical meaningfulness and a stable simulation. E.g.\ for porosity optimization, the porosity has to lie between $0$ and $1$ in order to be physically meaningful. This can be achieved via bounding the control by user-defined thresholds (typically used for parameter optimization) or by using a smooth projection map $p\colon \mathbb{R} \to [0, 1]$ (helpful for porosity optimization).
Alternatively, if one has more than one control variables, a vector with the upper and lower bounds for each variable can be used to bound the control.

\subsection{OptiCase Classes: Gradient Computation}

For the numerical evaluation of the functional gradient $\tfrac{dJ}{d\alpha}$, the following four options are implemented:
\begin{itemize}
  \item Forward/ central difference quotients.
    The method is implemented in the classes \class{OptiCaseFDQ} and \class{OptiCaseCDQ}, where the user (only) has to provide an expression for the evaluation of $J$ at construction.
  \item Forward automatic differentiation. Evaluate $J$ for operator-overloaded variables of type \texttt{ADf<T,n>}, where \texttt{T} is the underlying arithmetic type and \texttt{n} is the number of control variables.
    Therefore, the full source code has to be templatized w.r.t.\ the arithmetic type \texttt{T} and one has to take care that the arithmetic operations are differentiable (which is both satisfied by the OpenLB code basis).
    One then passes two instances of $J$ (one with type \texttt{T}, one with type \texttt{ADf<T,n>}) to the class \class{OptiCaseAD}, which then does the rest.
  \item Adjoint LBM.
    This method requires the (theoretical) calculation of adjoint LB-equations which is a problem-specific task.
    Because of that, the implementing class \class{OptiCaseDual} poses direct assumptions on the implementation of the LB simulation (e.g., it requires usage of the Solver framework and has so far only been implemented in the context of porosity and force optimization, cf.\ examples \class{DomainIdenfification3d} and \class{TestFlowOpti3d}).
\end{itemize}

\noindent For difference quotients and forward automatic differentiation, any functions of type \texttt{T} 
\texttt{(const} \texttt{std::vector\&)} (with or without LB simulation) are passed.
Flow simulations have to be encapsulated by a suitable wrapper that accepts the control variables, runs the simulation and computes the objective.
In the Solver app structure, the \texttt{getCallable} method fulfills this task.
An introduction into forward automatic differentiation, whose usage is not restricted to the optimization context, is provided in the example \class{showcaseADf}.

\subsection{Parameter Explanation and Reading from XML} 
In this short overview the relevant parameters for an app with optimization are listed with the respective names for using an \texttt{xml}-file for the input.
The different parameters of each part in the \texttt{xml}-file are explained via a table of the following form:
\begin{center}
\small
\begin{tabular}[ht]{ | m{1.5cm} | m{2.0cm}| m{4cm} | m{1.5cm} | m{3.5cm} | } 
  \hline
  \textbf{Parameter} &\textbf{Name}(type) & [\textbf{declaration} \& \textbf{definition}] (\&\& and other usages) & Default value & \textbf{Explanation:} (if available, all) \textbf{\textit{possibilities}}\\ 
  \hline
\end{tabular}
\end{center}

The explanation of each column is as follows:
\small
\begin{itemize}
    \item \textbf{Parameter:} Name of the parameter in the \texttt{xml}-file
    \item \textbf{Name} (type): Name of the parameter in the source code and its data type in brackets. Besides the common data types the abbreviations S and T are used for template parameters.
    
    \item $[$\textbf{declaration} \& \textbf{definition}$]$ (\&\& other usages): Location of declaration and definition (\&\& and sometimes some other important usages) of the parameter, e.g.:\\
    $[$\path{solver.h} \& \path{solver.hh}$]$ (\&\& \path{example.cpp})\\
    If the location of the declaration and definition is the same, only one location is indicated, \eg $[$\path{solver.hh}$]$\\
    If there is more than one location for the declaration and definition, it is indicated with an \textit{and}, \eg 
    $[$\path{solver.h} \& \path{solver.hh}$]$ and $[$\path{optiCaseDual.hh}$]$\\
    If there are different possibilities for the location of declaration and definition, it is indicated with curved brackets \textit{()} and an \textit{or}, \eg\\
    ($[$\path{optimizerSteepestDescent.h}$]$ or $[$\path{optimizerLBFGS.h}$]$ or $[$\path{optimizerBarzilaiBorwein.h}$]$)
    \item \textbf{Default value}: Is this parameter essential or optional?\\
    If a parameter is optional, it does not need to be defined. Then, the default value can be seen in this column.\\
    If a parameter is of such importance that without it the program has to exit, it is labeled as \textit{EXIT}.\\
    Some parameters are indicated with \textbf{unused} which means, that the parameter is read but not used afterwards.
    \item \textbf{Explanation: } (if available, all) \textbf{\textit{possibilities}}: Brief description and explanation of the parameter. Some parameters have different possibilities for their definition. In this case, all available possibilities are also offered in \textbf{bold} type letters, e.g.:\\
    \textit{\textbf{ad}} for \textit{OptiCaseAD} or\\
    \textit{\textbf{dual}} for \textit{OptiCaseDual} or\\
    \textit{\textbf{adTest}} for \textit{OptiCaseADTest}
\end{itemize}

The arrangement of the parameters in the \texttt{xml}-file has the following structure:\\
example:

\begin{lstlisting}[language=XML]
<Param>
  <Optimization>
    <MaxStepAttempts> 20 </MaxStepAttempts>
  </Optimization>
</Param>
\end{lstlisting}

\newpage

\begin{center}

\small
\begin{longtable}[ht!]{ | m{1.5cm} | m{2cm}| m{4cm} | m{1.5cm} | m{3.5cm} | }
  \hline
  \textbf{Parameter} &\textbf{Name}(type) & [\textbf{declaration} \& \textbf{definition}] (\&\& and other usages) & Default value & \textbf{Explanation:} (if available, all) \textbf{\textit{possibilities}}\\ 
  \hline \hline
  
  \textit{ControlType} & \_controlType & $[$\path{optiCaseDual.h} \& & & Defines the control type to optimize\\
  & (std::string) & \path{optiCaseDual.hh}$]$ & & \textit{\textbf{force}} or\\
  & & & & \textit{\textbf{porosity}}\\
  \hline
  
  \textit{Control-Material} & \_control-Material (int) & $[$\path{optiSolverParameters.h}$]$ (\&\& \path{optiCaseDual.h} and \path{optiCaseDual.hh}) & 0 & Defines the number of control material, here 6 for porosity optimization problems\\
  \hline
  
  \textit{Field-Dimension} & \_fieldDim (int) & $[$\path{optiSolverParameters.h}$]$ (\&\& \path{optiCaseDual.h} and \path{optiCaseDual.hh}) & 0 & Spatial dimension of controlled field\\
  \hline
  
  \textit{Dimension-Control} & \_dimCtrl (int) & ($[$\path{optiCaseDual.h} \& \path{optiCaseDual.hh}$]$) and ($[$\path{optimizer.h} \& \path{optimizer.hh}$]$) ($[$\path{optimizerSteepestDescent.h}$]$ or $[$\path{optimizerLBFGS.h}$]$ or $[$\path{optimizerBarzilaiBorwein.h}$]$) and $[$\path{controller.h}$]$ & EXIT & Upper limit for number of control variables, so far not read by xml\\
  \hline
  
  \textit{RegAlpha} & regAlpha (S) & $[$\path{optiCaseDual.h} \& \path{optiCaseDual.hh}$]$ & 0 & Weighting factor for regularization, so far unused\\
  \hline
  
  \textit{Lambda} & \_lambda (T) & ($[$\path{optimizerSteepestDescent.h}$]$ or $[$\path{optimizerLBFGS.h}$]$ or $[$\path{optimizerBarzilaiBorwein.h}$]$) (\&\& \path{optimizerLineSearch.h}) & 1 & Determines the initial step length lambda\\
  \hline
  
  \textit{MaxIter} & \_maxIt (int) & $[$\path{optimizer.h} \& (\path{optimizerSteepestDescent.h} or \path{optimizerLFBGS.h} or \path{optimizerBarzilaiBorwein.h})$]$ & 100 & Maximal number of iterations of the optimizer\\
  \hline

  \textit{MaxStep-Attempts} & \_maxStep-Attempts (int) & $[$\path{optimizer.h} \& (\path{optimizerSteepestDescent.h} or \path{optimizerLFBGS.h} or \path{optimizerBarzilaiBorwein.h})$]$ & 20 & Maximal number of attempts at each optimization step for finding a suitable step size\\
  \hline
  
  \textit{FailOnMax-Iter} & \_failOn-MaxIter (bool) & $[$\path{optimizer.h}$]$ (\&\& \path{optimizer.hh} or \path{optimizerLBFGS.h} or \path{optimizerLineSearch.h}) & 1 (true) & if true, the optimization fails when reaching \_maxIt and prints the warning: Optimization problem failed to converge within specified iteration limit of \_maxIt iterations with tolerance of \_eps\\
  \hline
  
  \textit{Tolerance} & \_eps (T) & $[$\path{optimizer.h}$]$ (\&\& \path{optimizer.hh}) & 1e-10 & Tolerance of the optimizer. Optimizer stops if the norm of the vector of derivatives of the object functional is smaller than the tolerance of the optimizer\\
  \hline
  
  \textit{L} & \_l (int) & $[$\path{optimizerLBFGS.h}$]$ & 20 & Maximal number of stored iteration steps for LBFGS algorithm.\\
  \hline
  
  \textit{Verbose} & \_verboseOn (bool) & $[$\path{optimizer.h} \& \path{optimizer.hh}$]$ (\&\& \path{optimizerLineSearch.h}) & 1 (true) & Print Warnings and further output in the terminal, if true\\
  \hline
  
  \textit{InputFile-Name} & fname (std::string) & ($[$\path{optimizerSteepstDescent.h}$]$ or $[$\path{optimizerLBFGS.h}$]$ or $[$\path{optimizerBarzilaiBorwein.h}$]$) (\&\& \path{optimizer.h} and \path{optimizerLineSearch.hh}) & "control.dat" & Name of the file that contains the initial guess for the control values\\
  \hline
  
  \textit{ControlTol-erance} & \_controlEps (T) & ($[$\path{optimizer.h} \& \path{optimizerSteepestDescent.h} and \path{optimizerLBFGS.h} and \path{optimizerBarzilaiBorwein.h}$]$ (\&\& \path{optimizer.hh}) & 0 & Optimization stops if the change of the control variables is less than this tolerance.\\
  \hline
  
  \textit{StepCondi-} & stepCondition & ($[$\path{optimizerSteepest}-  & "Armijo" & Defines the step condition:\\
  \textit{tion} & (std::string) & \path{Descent.h}$]$ or & (SD), & \textbf{None} or\\
  & & $[$\path{optimizerLBFGS.h}$]$ or & "Strong & \textbf{Smaller} or\\
  & & $[$\path{optimizerBarzilai}- & -Wolfe" & \textbf{Armijo} or\\
  & & \path{Borwein.h}$]$) & (LBFGS), & \textbf{Wolfe} or\\
  & & (\&\& \path{optimizerLineSearch.h}) & "None",& \textbf{StrongWolfe}\\
  & & & (Barzilai--Borwein) &\\
  \hline
  
  \textit{Vector-Bounds} & \_vectorBounds (bool) & $[$\path{optimizer.h} \& \path{optimizerLBFGS.h} or \path{optimizerBarzilaiBorwein.h}$]$ (\&\& \path{optimizer.h} and \path{optimizerLineSearch.h}) & 0 (false) & Determines, whether bounds on the control variables are applied component-wise\\
  \hline
  
  \textit{Projection} & \_projection-Name (std::string)& $[$\path{optiCaseDual.h} \& \path{optiCaseDual.hh}$]$ & & mapping method between the physical and computational control variables:\\
  &  & & & \textbf{Sigmoid} or\\
  & & & & \textbf{Rectifier} or \\
  & & & & \textbf{Softplus} or\\
  & & & & \textbf{Baron} or\\
  & & & & \textbf{Krause} or\\
  & & & & \textbf{Foerster} or\\
  & & & & \textbf{FoersterN} or\\
  & & & & \textbf{StasiusN} or\\
  \hline
  
\textit{Reference-Solution} & \_compute-Reference (bool) & $[$\path{optiCaseDual.h} \& \path{optiCaseDual.hh}$]$ & false & states if reference solution is available\\
\hline
  \textit{StartValue} & startValue (T) & (\&\& \path{optimizer.h} and \path{optimizer.hh}) & 0 & Determines the initial guess for the optimization algorithm\\
  \hline
  
  \textit{Upper-Bound} & \_upperBound (T) & ($[$\path{optimizer.h} \& \path{optimizerSteepestDescent.h} and \path{optimizerLBFGS.h} and \path{optimizerBarzilaiBorwein.h}$]$ (\&\& \path{optimizer.hh}) & false in all three optimizers & Sets an upper bound for the control values during the optimization process\\
  \hline
  
  \textit{Lower-Bound} & \_lowerBound (T) & ($[$\path{optimizer.h} \& \path{optimizerSteepestDescent.h} and \path{optimizerLBFGS.h} and \path{optimizerBarzilaiBorwein.h}$]$ (\&\& \path{optimizer.hh}) & false in all three optimizers & Sets a lower bound for the control values during the optimization process\\
  \hline
  
  \textit{Volume-Ratio} & volumeRatio (T) & $[$\path{optiCaseDual.hh}$]$ & \textbf{unused} (//) & Determines the volume ratio\\
  \hline 
  \pagebreak
  \hline
  \multicolumn{5}{|c|}{\textit{VisualizationGnuplot:}}\\
  \hline
  
  \textit{Visualized-Parameters} & gplotAnalysis-String (std::string) & $[$\path{optimizerSteepestDescent.h}$]$ or $[$\path{optimizerLBFGS.h}$]$ or $[$\path{optimizerBarzilaiBorwein.h}$]$ & & Lists the parameters that will be plotted by Gnuplot during the optimization process. Possibilities:\\
  & & & & \textbf{VALUE}\\
  & & & & \textbf{CONTROL}\\
  & & & & \textbf{DERIVATIVE}\\
  & & & & \textbf{ERROR}\\
  & & & & \textbf{NORM\_DERIVATIVE}\\
  \hline
\end{longtable}
\end{center}


\chapter{Examples} 
\label{sec:samples}

\section{Example Overview}
A list of the currently included examples is given together with related keywords in the following Tables~\ref{tab:examplesI} and \ref{tab:examplesII}.
\newgeometry{
    left=0.5cm, 
    right=0.5cm, 
    bottom=1cm, 
    top=1cm
    }

\begin{landscape}
\begin{footnotesize}
\renewcommand{\arraystretch}{1,1}

\begin{table}\footnotesize

\begin{tabular}{ |c|p{4.5cm}|c|c|c|c|c|c|c|c|c|c|c|c| }

\noalign{\hrule height 1pt}

\multirow{2}{*}{folder} & \multirow{2}{*}{example} & \multirow{2}{*}{turbulent} & \multirow{2}{*}{thermal} & multi & multi & \multirow{2}{*}{particles} & porous & transient  & \multirow{2}{*}{benchmark} & \multirow{2}{*}{showcase} & STL  & geometry  & check-\\

& & & & Component & Phase & & Media & flow & & & geometry & primitives & pointing \\

\noalign{\hrule height 1pt}

\multirow{2}{*}{\makecell{adsorption}} & \cellcolor{darkgray} \makecell[l]{adsorption3D}		  & \cellcolor{darkgray}&\cellcolor{darkgray} & \cellcolor{RoyalBluedark}& \cellcolor{RoyalBluedark}& \cellcolor{RoyalBluedark}& \cellcolor{darkgray}& \cellcolor{darkgray}&\cellcolor{RoyalBluedark} &\cellcolor{darkgray} & \cellcolor{darkgray}& \cellcolor{RoyalBluedark} & \cellcolor{darkgray}\\

	& \cellcolor{lightgray} \makecell[l]{microMixer3D}  	 	& \cellcolor{lightgray}&\cellcolor{lightgray} & \cellcolor{RoyalBlue}& \cellcolor{RoyalBlue}& \cellcolor{RoyalBlue}& \cellcolor{lightgray}& \cellcolor{lightgray}&\cellcolor{lightgray} &\cellcolor{RoyalBlue} & \cellcolor{RoyalBlue}& \cellcolor{RoyalBlue} & \cellcolor{lightgray}\\

\noalign{\hrule height 1pt}

\multirow{6}{*}{\makecell{advectionDiffusion \\ Reaction}} & \cellcolor{darkgray} \makecell[l]{advectionDiffusionReaction2d}		  & \cellcolor{darkgray}&\cellcolor{darkgray} & \cellcolor{darkgray}& \cellcolor{darkgray}& \cellcolor{darkgray}& \cellcolor{darkgray}& \cellcolor{darkgray}&\cellcolor{RoyalBluedark} &\cellcolor{darkgray} & \cellcolor{darkgray}& \cellcolor{RoyalBluedark} & \cellcolor{darkgray}\\

							& \cellcolor{lightgray} \makecell[l]{reactionFiniteDifferences2d}  	 	& \cellcolor{lightgray}&\cellcolor{lightgray} & \cellcolor{lightgray}& \cellcolor{lightgray}& \cellcolor{lightgray}& \cellcolor{lightgray}& \cellcolor{lightgray}&\cellcolor{RoyalBlue} &\cellcolor{lightgray} & \cellcolor{lightgray}& \cellcolor{RoyalBlue} & \cellcolor{lightgray}\\

& \cellcolor{darkgray} advectionDiffusion1d		  & \cellcolor{darkgray}&\cellcolor{darkgray} & \cellcolor{darkgray}& \cellcolor{darkgray}& \cellcolor{darkgray}& \cellcolor{darkgray}& \cellcolor{RoyalBluedark}&\cellcolor{RoyalBluedark} &\cellcolor{darkgray} & \cellcolor{darkgray}& \cellcolor{RoyalBluedark} & \cellcolor{darkgray}\\

							& \cellcolor{lightgray} advectionDiffusion2d  	 	& \cellcolor{lightgray}&\cellcolor{lightgray} & \cellcolor{lightgray}& \cellcolor{lightgray}& \cellcolor{lightgray}& \cellcolor{lightgray}& \cellcolor{RoyalBlue}&\cellcolor{RoyalBlue} &\cellcolor{lightgray} & \cellcolor{lightgray}& \cellcolor{RoyalBlue} & \cellcolor{lightgray}\\

							& \cellcolor{darkgray} advectionDiffusion3d  	 	& \cellcolor{darkgray}&\cellcolor{darkgray} & \cellcolor{darkgray}& \cellcolor{darkgray}& \cellcolor{darkgray}& \cellcolor{darkgray}& \cellcolor{RoyalBluedark}&\cellcolor{RoyalBluedark} &\cellcolor{darkgray} & \cellcolor{darkgray}& \cellcolor{RoyalBluedark} & \cellcolor{darkgray}\\

														& \cellcolor{lightgray} \makecell[l]{advectionDiffusionPipe2d}  	 	& \cellcolor{lightgray}&\cellcolor{lightgray} & \cellcolor{lightgray}& \cellcolor{lightgray}& \cellcolor{lightgray}& \cellcolor{lightgray}& \cellcolor{RoyalBlue}&\cellcolor{RoyalBlue} &\cellcolor{RoyalBlue} & \cellcolor{lightgray}& \cellcolor{RoyalBlue} & \cellcolor{lightgray}\\

\noalign{\hrule height 1pt}

\multirow{9}{*}{laminar}  &\cellcolor{darkgray} bstep2d								&\cellcolor{darkgray} &\cellcolor{darkgray} & \cellcolor{darkgray}& \cellcolor{darkgray}&\cellcolor{darkgray} & \cellcolor{darkgray} & \cellcolor{RoyalBluedark} & \cellcolor{RoyalBluedark} & \cellcolor{darkgray}& \cellcolor{darkgray}& \cellcolor{RoyalBluedark} &\cellcolor{RoyalBluedark}\\

& \cellcolor{lightgray}bstep3d										& \cellcolor{lightgray} & \cellcolor{lightgray} &\cellcolor{lightgray} & \cellcolor{lightgray}&\cellcolor{lightgray} &\cellcolor{lightgray} &\cellcolor{RoyalBlue} & \cellcolor{RoyalBlue} &\cellcolor{lightgray} &\cellcolor{lightgray} & \cellcolor{RoyalBlue} & \cellcolor{RoyalBlue}\\

& \cellcolor{darkgray}cavity2d	& \cellcolor{darkgray}& \cellcolor{darkgray}& \cellcolor{darkgray}& \cellcolor{darkgray}& \cellcolor{darkgray}& \cellcolor{darkgray}&\cellcolor{RoyalBluedark} & \cellcolor{RoyalBluedark} &\cellcolor{darkgray} & \cellcolor{darkgray}& \cellcolor{RoyalBluedark} & \cellcolor{darkgray} \\

& \cellcolor{lightgray}cavity3d										& \cellcolor{lightgray}& \cellcolor{lightgray}& \cellcolor{lightgray}& \cellcolor{lightgray}&\cellcolor{lightgray} & \cellcolor{lightgray}&\cellcolor{RoyalBlue} & \cellcolor{RoyalBlue} &\cellcolor{lightgray} & \cellcolor{lightgray}& \cellcolor{RoyalBlue} & \cellcolor{lightgray}\\

& \cellcolor{darkgray}cylinder2d										&\cellcolor{darkgray} & \cellcolor{darkgray}& \cellcolor{darkgray}& \cellcolor{darkgray}&\cellcolor{darkgray} & \cellcolor{darkgray}& \cellcolor{darkgray}&\cellcolor{RoyalBluedark} &\cellcolor{darkgray} & \cellcolor{darkgray}& \cellcolor{green} & \cellcolor{darkgray} \\

& \cellcolor{lightgray}cylinder3d										&\cellcolor{lightgray} & \cellcolor{lightgray}& \cellcolor{lightgray}& \cellcolor{lightgray}&\cellcolor{lightgray} &\cellcolor{lightgray} &\cellcolor{lightgray} &\cellcolor{RoyalBlue}&\cellcolor{lightgray} &\cellcolor{RoyalBlue} & \cellcolor{lightgray}& \cellcolor{lightgray} \\

& \cellcolor{darkgray}poiseuille2d
&\cellcolor{darkgray} &\cellcolor{darkgray} &\cellcolor{darkgray} & \cellcolor{darkgray}&\cellcolor{darkgray} & \cellcolor{darkgray}& \cellcolor{darkgray}&\cellcolor{RoyalBluedark} & \cellcolor{darkgray}& \cellcolor{darkgray}& \cellcolor{RoyalBluedark} & \cellcolor{darkgray}\\

& \cellcolor{lightgray}poiseuille2dEOC
&\cellcolor{lightgray} & \cellcolor{lightgray}&\cellcolor{lightgray} & \cellcolor{lightgray}& \cellcolor{lightgray}& \cellcolor{lightgray}&\cellcolor{lightgray} &\cellcolor{RoyalBlue} &\cellcolor{lightgray} &\cellcolor{lightgray} & \cellcolor{RoyalBlue} & \cellcolor{lightgray}\\

& \cellcolor{darkgray}poiseuille3d										&\cellcolor{darkgray} &\cellcolor{darkgray} & \cellcolor{darkgray}&\cellcolor{darkgray} &\cellcolor{darkgray} & \cellcolor{darkgray}& \cellcolor{darkgray}&\cellcolor{RoyalBluedark} &\cellcolor{darkgray} & \cellcolor{darkgray}& \cellcolor{RoyalBluedark} &\cellcolor{darkgray} \\

& \cellcolor{lightgray}powerLaw2d
&\cellcolor{lightgray} & \cellcolor{lightgray}&\cellcolor{lightgray} & \cellcolor{lightgray}& \cellcolor{lightgray}& \cellcolor{lightgray}&\cellcolor{lightgray} &\cellcolor{RoyalBlue} &\cellcolor{lightgray} &\cellcolor{lightgray} & \cellcolor{RoyalBlue} & \cellcolor{lightgray}\\

& \cellcolor{darkgray} testFlow3dSolver &\cellcolor{darkgray} & \cellcolor{darkgray}&\cellcolor{darkgray} & \cellcolor{darkgray}&\cellcolor{darkgray} &\cellcolor{darkgray} & \cellcolor{darkgray} & \cellcolor{RoyalBluedark} &\cellcolor{RoyalBluedark} & \cellcolor{darkgray} &\cellcolor{darkgray} & \cellcolor{darkgray}\\

\noalign{\hrule height 1pt}

\multirow{11}{*}{multiComponent} & \cellcolor{darkgray} binaryShearFlow2d				& \cellcolor{darkgray}& \cellcolor{darkgray}& \cellcolor{darkgray} & \cellcolor{RoyalBluedark} & \cellcolor{darkgray}& \cellcolor{darkgray}& \cellcolor{RoyalBluedark}& \cellcolor{RoyalBluedark} & \cellcolor{RoyalBluedark}& \cellcolor{darkgray}& \cellcolor{RoyalBluedark} & \cellcolor{darkgray}\\

 & \cellcolor{lightgray} contactAngle2d				& \cellcolor{lightgray}& \cellcolor{lightgray}& \cellcolor{lightgray} & \cellcolor{RoyalBlue} & \cellcolor{lightgray}& \cellcolor{lightgray}& \cellcolor{lightgray}& \cellcolor{RoyalBlue} & \cellcolor{lightgray}& \cellcolor{lightgray}& \cellcolor{RoyalBlue} & \cellcolor{lightgray}\\

& \cellcolor{darkgray} contactAngle3d		&\cellcolor{darkgray} &\cellcolor{darkgray} &\cellcolor{darkgray} &\cellcolor{RoyalBluedark} &\cellcolor{darkgray} & \cellcolor{darkgray}& \cellcolor{darkgray}& \cellcolor{RoyalBluedark}&\cellcolor{darkgray} & \cellcolor{darkgray}&\cellcolor{RoyalBluedark}& \cellcolor{darkgray}\\

& \cellcolor{lightgray} fourRollMill2d				& \cellcolor{lightgray}& \cellcolor{lightgray}& \cellcolor{lightgray} & \cellcolor{RoyalBlue} & \cellcolor{lightgray}& \cellcolor{lightgray}& \cellcolor{RoyalBlue}& \cellcolor{RoyalBlue} & \cellcolor{RoyalBlue}& \cellcolor{lightgray}& \cellcolor{RoyalBlue} & \cellcolor{lightgray}\\

& \cellcolor{darkgray} microFluidics2d								& \cellcolor{darkgray}& \cellcolor{darkgray}&\cellcolor{RoyalBluedark}&\cellcolor{darkgray} &\cellcolor{darkgray} & \cellcolor{darkgray}&\cellcolor{RoyalBluedark} & \cellcolor{darkgray}& \cellcolor{RoyalBluedark} &\cellcolor{darkgray} & \cellcolor{RoyalBluedark} & \cellcolor{darkgray}\\

& \cellcolor{lightgray} phaseSeperation2d								& \cellcolor{lightgray}& \cellcolor{lightgray}& \cellcolor{lightgray}&\cellcolor{RoyalBlue} & \cellcolor{lightgray}&\cellcolor{lightgray} & \cellcolor{RoyalBlue} &\cellcolor{lightgray} & \cellcolor{lightgray}&\cellcolor{lightgray} &\cellcolor{RoyalBlue} & \cellcolor{lightgray}\\

& \cellcolor{darkgray} phaseSeperation3d								&\cellcolor{darkgray} &\cellcolor{darkgray} &\cellcolor{darkgray} &\cellcolor{RoyalBluedark} & \cellcolor{darkgray}&\cellcolor{darkgray} & \cellcolor{RoyalBluedark} & \cellcolor{darkgray}& \cellcolor{darkgray} & \cellcolor{darkgray}&\cellcolor{RoyalBluedark} & \cellcolor{darkgray}\\

& \cellcolor{lightgray}rayleighTaylor2d								&\cellcolor{lightgray} & \cellcolor{lightgray}&\cellcolor{RoyalBlue}& \cellcolor{lightgray}&\cellcolor{lightgray} & \cellcolor{lightgray}&\cellcolor{RoyalBlue}&\cellcolor{RoyalBlue} &\cellcolor{lightgray} & \cellcolor{lightgray}& \cellcolor{RoyalBlue} &\cellcolor{lightgray} \\

& \cellcolor{darkgray} rayleighTaylor3d								&\cellcolor{darkgray} & \cellcolor{darkgray}&\cellcolor{RoyalBluedark} & \cellcolor{darkgray}& \cellcolor{darkgray}& \cellcolor{darkgray}& \cellcolor{RoyalBluedark} & \cellcolor{RoyalBluedark} & \cellcolor{darkgray}&\cellcolor{darkgray} &\cellcolor{RoyalBluedark} & \cellcolor{darkgray}\\

& \cellcolor{lightgray} youngLaplace2d									& \cellcolor{lightgray}& \cellcolor{lightgray}&\cellcolor{RoyalBlue} & \cellcolor{lightgray}& \cellcolor{lightgray}&\cellcolor{lightgray} & \cellcolor{lightgray}& \cellcolor{RoyalBlue} & \cellcolor{lightgray}& \cellcolor{lightgray}& \cellcolor{RoyalBlue} &\cellcolor{lightgray} \\

& \cellcolor{darkgray}youngLaplace3d									&\cellcolor{darkgray} &\cellcolor{darkgray} &\cellcolor{RoyalBluedark} &\cellcolor{darkgray} & \cellcolor{darkgray}& \cellcolor{darkgray}& \cellcolor{darkgray}& \cellcolor{RoyalBluedark} &\cellcolor{darkgray} & \cellcolor{darkgray}& \cellcolor{RoyalBluedark} & \cellcolor{darkgray} \\
\noalign{\hrule height 1pt}

\end{tabular}

\centering
\begin{tabular}{ p{4cm} p{0,6cm} p{0,6cm} l p{2cm} p{1,2cm} l}
& & & & & & \\
&\cellcolor{RoyalBlue}&\cellcolor{RoyalBluedark}& example includes relevant subject& &\cellcolor{green}&example includes relevant subject and is recommended for beginning \\
\end{tabular}
\caption{Currently included examples in OpenLB (continued in Table~\ref{tab:examplesII}).}
\label{tab:examplesI}
\end{table}

\pagebreak

\begin{table}\footnotesize
\begin{tabular}{ |c|p{4.5cm}|c|c|c|c|c|c|c|c|c|c|c|c| }

\noalign{\hrule height 1pt}

\multirow{2}{*}{folder} & \multirow{2}{*}{example} & \multirow{2}{*}{turbulent} & \multirow{2}{*}{thermal} & multi & multi & \multirow{2}{*}{particles} & porous & transient  & \multirow{2}{*}{benchmark} & \multirow{2}{*}{showcase} & STL  & geometry  & check-\\

& & & & Component & Phase & & Media & flow & & & geometry & primitives & pointing \\

\noalign{\hrule height 1pt}

		\multirow{4}{*}{optimization} & \cellcolor{lightgray} domainIdentification3d	&\cellcolor{lightgray} & \cellcolor{lightgray}&\cellcolor{lightgray} & \cellcolor{lightgray}&\cellcolor{lightgray} &\cellcolor{RoyalBlue} & \cellcolor{lightgray} & \cellcolor{RoyalBlue} &\cellcolor{RoyalBlue} & \cellcolor{lightgray} &\cellcolor{lightgray} & \cellcolor{lightgray}\\

& \cellcolor{darkgray} parameterIdentificationPoiseuille2d &\cellcolor{darkgray} &\cellcolor{darkgray} & \cellcolor{darkgray}& \cellcolor{darkgray}&\cellcolor{darkgray} & \cellcolor{darkgray}& \cellcolor{darkgray} &\cellcolor{darkgray} & \cellcolor{RoyalBluedark} & \cellcolor{darkgray}&\cellcolor{darkgray} & \cellcolor{darkgray}\\

& \cellcolor{lightgray} showcaseADf &\cellcolor{lightgray} &\cellcolor{lightgray} & \cellcolor{lightgray}& \cellcolor{lightgray}&\cellcolor{lightgray} & \cellcolor{lightgray}& \cellcolor{lightgray} &\cellcolor{lightgray} & \cellcolor{RoyalBlue} & \cellcolor{lightgray}&\cellcolor{lightgray} & \cellcolor{lightgray}\\

& \cellcolor{darkgray} showcaseRosenbrock & \cellcolor{darkgray}& \cellcolor{darkgray}&\cellcolor{darkgray} & \cellcolor{darkgray}&\cellcolor{darkgray} & \cellcolor{darkgray}& \cellcolor{darkgray} & \cellcolor{darkgray}& \cellcolor{RoyalBluedark} & \cellcolor{darkgray}&\cellcolor{darkgray} & \cellcolor{darkgray}\\

& \cellcolor{lightgray} testFlowOpti3d &\cellcolor{lightgray} &\cellcolor{lightgray} & \cellcolor{lightgray}& \cellcolor{lightgray}&\cellcolor{lightgray} & \cellcolor{lightgray}& \cellcolor{lightgray} &\cellcolor{RoyalBlue} & \cellcolor{RoyalBlue} & \cellcolor{lightgray}&\cellcolor{lightgray} & \cellcolor{lightgray}\\

\noalign{\hrule height 1pt}

		\multirow{4}{*}{particles} & \cellcolor{darkgray} bifurcation3d						&\cellcolor{darkgray} & \cellcolor{darkgray}&\cellcolor{darkgray} & \cellcolor{darkgray}&\cellcolor{RoyalBluedark} &\cellcolor{darkgray} & \cellcolor{RoyalBluedark} & \cellcolor{RoyalBluedark} &\cellcolor{RoyalBluedark} & \cellcolor{RoyalBluedark} &\cellcolor{darkgray} & \cellcolor{darkgray}\\

& \cellcolor{lightgray} dkt2d &\cellcolor{lightgray} &\cellcolor{lightgray} & \cellcolor{lightgray}& \cellcolor{lightgray}&\cellcolor{RoyalBlue} & \cellcolor{lightgray}& \cellcolor{RoyalBlue} &\cellcolor{RoyalBlue} & \cellcolor{RoyalBlue} & \cellcolor{lightgray}&\cellcolor{RoyalBlue} & \cellcolor{lightgray}\\

& \cellcolor{darkgray} magneticParticles3d							& \cellcolor{darkgray}& \cellcolor{darkgray}&\cellcolor{darkgray} & \cellcolor{darkgray}&\cellcolor{RoyalBluedark} & \cellcolor{darkgray}& \cellcolor{RoyalBluedark} & \cellcolor{darkgray}& \cellcolor{RoyalBluedark} & \cellcolor{darkgray}&\cellcolor{RoyalBluedark} & \cellcolor{darkgray}\\

& \cellcolor{lightgray} settlingCube3d								&\cellcolor{lightgray} &\cellcolor{lightgray} & \cellcolor{lightgray}& \cellcolor{lightgray}&\cellcolor{RoyalBlue} & \cellcolor{lightgray}& \cellcolor{RoyalBlue} &\cellcolor{lightgray} & \cellcolor{RoyalBlue} & \cellcolor{lightgray}&\cellcolor{RoyalBlue} & \cellcolor{lightgray}\\
\noalign{\hrule height 1pt}

\multirow{2}{*}{porousMedia} & \cellcolor{darkgray} porousPoiseuille2d				&\cellcolor{darkgray} & \cellcolor{darkgray}& \cellcolor{darkgray}& \cellcolor{darkgray}& \cellcolor{darkgray}&\cellcolor{RoyalBluedark} &\cellcolor{darkgray} & \cellcolor{RoyalBluedark} &\cellcolor{darkgray} & \cellcolor{darkgray}& \cellcolor{RoyalBluedark}& \cellcolor{darkgray}\\

& \cellcolor{lightgray} porousPoiseuille3d								&\cellcolor{lightgray} & \cellcolor{lightgray}& \cellcolor{lightgray}& \cellcolor{lightgray}& \cellcolor{lightgray}&\cellcolor{RoyalBlue} & \cellcolor{lightgray}& \cellcolor{RoyalBlue} & \cellcolor{lightgray}& \cellcolor{lightgray}& \cellcolor{RoyalBlue} &\cellcolor{lightgray} \\
\noalign{\hrule height 1pt}

\multirow{8}{*}{thermal} & \cellcolor{darkgray} galliumMelting2d		& \cellcolor{darkgray}&\cellcolor{RoyalBluedark} & \cellcolor{darkgray}& \cellcolor{RoyalBluedark}& \cellcolor{darkgray}& \cellcolor{darkgray}& \cellcolor{RoyalBluedark}&\cellcolor{RoyalBluedark} &\cellcolor{RoyalBluedark} & \cellcolor{darkgray}& \cellcolor{RoyalBluedark} & \cellcolor{darkgray}\\

& \cellcolor{lightgray} porousPlate2d						& \cellcolor{lightgray}&\cellcolor{RoyalBlue} & \cellcolor{lightgray}& \cellcolor{lightgray}& \cellcolor{lightgray}& \cellcolor{lightgray}& \cellcolor{lightgray}&\cellcolor{RoyalBlue} &\cellcolor{lightgray} & \cellcolor{lightgray}& \cellcolor{RoyalBlue} & \cellcolor{lightgray}\\

& \cellcolor{darkgray}porousPlate3d									& \cellcolor{darkgray}&\cellcolor{RoyalBluedark}&\cellcolor{darkgray} & \cellcolor{darkgray}& \cellcolor{darkgray}& \cellcolor{darkgray}&\cellcolor{darkgray} &\cellcolor{RoyalBluedark} &\cellcolor{darkgray} & \cellcolor{darkgray}& \cellcolor{RoyalBluedark} & \cellcolor{darkgray}\\

& \cellcolor{lightgray} porousPlate3dSolver						& \cellcolor{lightgray}&\cellcolor{RoyalBlue} & \cellcolor{lightgray}& \cellcolor{lightgray}& \cellcolor{lightgray}& \cellcolor{lightgray}& \cellcolor{lightgray}&\cellcolor{RoyalBlue} &\cellcolor{lightgray} & \cellcolor{lightgray}& \cellcolor{RoyalBlue} & \cellcolor{lightgray}\\

& \cellcolor{darkgray} rayleighBernard2d
&\cellcolor{RoyalBluedark} & \cellcolor{RoyalBluedark} &\cellcolor{darkgray} &\cellcolor{darkgray} &\cellcolor{darkgray} & \cellcolor{darkgray}&\cellcolor{RoyalBluedark} &\cellcolor{RoyalBluedark} & \cellcolor{RoyalBluedark} & \cellcolor{darkgray}&\cellcolor{RoyalBluedark}& \cellcolor{darkgray}\\

& \cellcolor{lightgray} rayleighBernard3d
&\cellcolor{RoyalBlue}&\cellcolor{RoyalBlue} & \cellcolor{lightgray}&\cellcolor{lightgray} &\cellcolor{lightgray} & \cellcolor{lightgray}& \cellcolor{RoyalBlue}&\cellcolor{RoyalBlue} &\cellcolor{RoyalBlue} & \cellcolor{lightgray}& \cellcolor{RoyalBlue} & \cellcolor{lightgray}\\

& \cellcolor{darkgray} squareCavity2d									&\cellcolor{RoyalBluedark} & \cellcolor{RoyalBluedark} &\cellcolor{darkgray} &\cellcolor{darkgray} &\cellcolor{darkgray} & \cellcolor{darkgray}&\cellcolor{RoyalBluedark} &\cellcolor{RoyalBluedark} & \cellcolor{darkgray} & \cellcolor{darkgray}&\cellcolor{RoyalBluedark}& \cellcolor{darkgray}\\

& \cellcolor{lightgray} squareCavity3d									&
\cellcolor{lightgray}&\cellcolor{RoyalBlue} &\cellcolor{lightgray} & \cellcolor{lightgray}& \cellcolor{lightgray}&\cellcolor{lightgray} & \cellcolor{lightgray}& \cellcolor{RoyalBlue} &\cellcolor{lightgray} & \cellcolor{lightgray}& \cellcolor{RoyalBlue}& \cellcolor{lightgray}\\

& \cellcolor{lightgray} stefanMelting2d									&
\cellcolor{darkgray}&\cellcolor{RoyalBluedark} &\cellcolor{darkgray} & \cellcolor{RoyalBluedark}& \cellcolor{darkgray}&\cellcolor{darkgray} & \cellcolor{RoyalBluedark}& \cellcolor{RoyalBluedark} &\cellcolor{darkgray} & \cellcolor{darkgray}& \cellcolor{RoyalBluedark}& \cellcolor{darkgray}\\
\noalign{\hrule height 1pt}

\multirow{4}{*}{turbulent} & \cellcolor{darkgray} aorta3d								&\cellcolor{RoyalBluedark} &\cellcolor{darkgray} & \cellcolor{darkgray}&\cellcolor{darkgray} & \cellcolor{darkgray}& \cellcolor{darkgray}& \cellcolor{RoyalBluedark}&\cellcolor{RoyalBluedark}&\cellcolor{RoyalBluedark}&\cellcolor{RoyalBluedark}&\cellcolor{darkgray} & \cellcolor{darkgray}\\

& \cellcolor{lightgray} channel3d										&\cellcolor{RoyalBlue}&\cellcolor{lightgray} & \cellcolor{lightgray}& \cellcolor{lightgray}&\cellcolor{lightgray} & \cellcolor{lightgray}&\cellcolor{RoyalBlue}&\cellcolor{RoyalBlue} &\cellcolor{RoyalBlue}& \cellcolor{lightgray}&\cellcolor{RoyalBlue}& \cellcolor{lightgray}\\

& \cellcolor{darkgray} nozzle3d										&\cellcolor{RoyalBluedark}&\cellcolor{darkgray} & \cellcolor{darkgray}& \cellcolor{darkgray}&\cellcolor{darkgray} & \cellcolor{darkgray}&\cellcolor{RoyalBluedark}&\cellcolor{darkgray} &\cellcolor{RoyalBluedark}& \cellcolor{darkgray}&\cellcolor{RoyalBluedark}& \cellcolor{darkgray}\\

& \cellcolor{lightgray}tgv3d											&\cellcolor{RoyalBlue}& \cellcolor{lightgray}&\cellcolor{lightgray} &\cellcolor{lightgray} & \cellcolor{lightgray}&\cellcolor{lightgray} &\cellcolor{RoyalBlue}&\cellcolor{RoyalBlue}& \cellcolor{lightgray}&\cellcolor{lightgray} &\cellcolor{RoyalBlue}& \cellcolor{lightgray}\\

& \cellcolor{darkgray}venturi3d										&\cellcolor{RoyalBluedark}& \cellcolor{darkgray}& \cellcolor{darkgray}& \cellcolor{darkgray}&\cellcolor{darkgray} &\cellcolor{darkgray} &\cellcolor{RoyalBluedark}&\cellcolor{darkgray} &\cellcolor{RoyalBluedark}& \cellcolor{darkgray}&\cellcolor{green}&\cellcolor{darkgray}\\
\noalign{\hrule height 1pt}

		\multirow{4}{*}{freeSurface} & \cellcolor{darkgray} breakingDam2d	&\cellcolor{darkgray} & \cellcolor{darkgray}&\cellcolor{RoyalBluedark} & \cellcolor{RoyalBluedark}&\cellcolor{darkgray} &\cellcolor{darkgray} & \cellcolor{RoyalBluedark} & \cellcolor{RoyalBluedark} &\cellcolor{darkgray} & \cellcolor{darkgray} &\cellcolor{darkgray} & \cellcolor{darkgray}\\

& \cellcolor{lightgray} breakingDam3d &\cellcolor{lightgray} &\cellcolor{lightgray} & \cellcolor{RoyalBlue}& \cellcolor{RoyalBlue}&\cellcolor{lightgray} & \cellcolor{lightgray}& \cellcolor{RoyalBlue} &\cellcolor{lightgray} & \cellcolor{lightgray} & \cellcolor{lightgray}&\cellcolor{lightgray} & \cellcolor{lightgray}\\

& \cellcolor{darkgray} deepFallingDrop2d & \cellcolor{darkgray}& \cellcolor{darkgray}&\cellcolor{RoyalBluedark} & \cellcolor{RoyalBluedark}&\cellcolor{darkgray} & \cellcolor{darkgray}& \cellcolor{RoyalBluedark} & \cellcolor{darkgray}& \cellcolor{darkgray} & \cellcolor{darkgray}&\cellcolor{darkgray} & \cellcolor{darkgray}\\

& \cellcolor{lightgray} fallingDrop2d &\cellcolor{lightgray} &\cellcolor{lightgray} & \cellcolor{RoyalBlue}& \cellcolor{RoyalBlue}&\cellcolor{lightgray} & \cellcolor{lightgray}& \cellcolor{RoyalBlue} &\cellcolor{lightgray} & \cellcolor{lightgray} & \cellcolor{lightgray}&\cellcolor{lightgray} & \cellcolor{lightgray}\\

& \cellcolor{darkgray} fallingDrop3d & \cellcolor{darkgray}& \cellcolor{darkgray}&\cellcolor{RoyalBluedark} & \cellcolor{RoyalBluedark}&\cellcolor{darkgray} & \cellcolor{darkgray}& \cellcolor{RoyalBluedark} & \cellcolor{darkgray}& \cellcolor{darkgray} & \cellcolor{darkgray}&\cellcolor{darkgray} & \cellcolor{darkgray}\\

& \cellcolor{lightgray} rayleighInstability3d &\cellcolor{lightgray} &\cellcolor{lightgray} & \cellcolor{RoyalBlue}& \cellcolor{RoyalBlue}&\cellcolor{lightgray} & \cellcolor{lightgray}& \cellcolor{RoyalBlue} &\cellcolor{lightgray} & \cellcolor{lightgray} & \cellcolor{lightgray}&\cellcolor{lightgray} & \cellcolor{lightgray}\\

\noalign{\hrule height 1pt}

\end{tabular}

\begin{tabular}{ p{4cm} p{0,6cm} p{0,6cm} l p{2cm} p{1,2cm} l}
& & & & & & \\
&\cellcolor{RoyalBlue}&\cellcolor{RoyalBluedark}& example includes relevant subject& &\cellcolor{green}&example includes relevant subject and is recommended for beginning \\

\end{tabular}
\caption{Currently included examples in OpenLB (continuation of Table~\ref{tab:examplesI}).}
\label{tab:examplesII}
\end{table}

\renewcommand{\arraystretch}{1}

\end{footnotesize}
\end{landscape}
\restoregeometry

All the demo codes can be compiled with or without MPI, with or without OpenMP, and executed in serial or parallel.

\section{adsorption}
\label{sec:ads}

\subsection{adsorption3D}
\label{subsec:ads3D}
This example shows the adsorption in a batch reactor using an Euler--Euler approach. 
The model is based on the linear driving force model and uses advection diffusion reaction lattices for particles, solute and particle loading.
Different isotherms and mass transfer models can be used. 
An analytical solution is implemented when using the linear isotherm and surface diffusion.

\subsection{microMixer3D}
\label{subsec:micromixer}
This example portrays the adsorption in a static mixing reactor using an Euler--Euler approach. 
Analogue to the example before, the model is based on the linear driving force model and uses advection diffusion reaction lattices for particles, solute and particle loading. 
Different isotherms and mass transfer models can be used.
 
\section{advectionDiffusionReaction}\label{sec:adr}

\subsection{advectionDiffusionReaction2d}\label{subsec:adre2d}
This example illustrates a steady-state chemical reaction in a plug flow reactor.
One can choose two types of reaction, \(A \longrightarrow C\) and \( A \longleftrightarrow C\).
The concentration and analytical solution along the centerline of the rectangle domain is given in \path{./tmp/N<resolution>/gnuplotData} as well as the error plot for the concentration along the centerline.
The default configuration executes three simulation runs and the average \(L^{2}\)-error over the centerline is computed for each resolution.
A plot of the resulting experimental order of convergence is provided in \path{./tmp/gnuplotData/}.

\subsection{reactionFiniteDifferences2d}\label{subsec:rfd2d}
Similarly to the previous example, a simplified domain with no fluid motion and homogeneous species concentrations is simulated, but here with finite differences. 
The chemical reaction \(\vert a \vert A \longrightarrow \vert b \vert B\) is approximated, where the reaction rate \(\nu = A/t_{0}\) is given, where \(t_{0}\) is a time conversion factor. 
The initial conditions are set to \(A(t=0)=1\) and \(B(t=0)=0\) such that an analytical solution is possible via
\begin{align}
A(t) & = \exp\left( -\frac{\left\vert a \right\vert t}{t_{0}} \right) ~, \\
B(t) & = \left\vert \frac{b}{a} \right\vert \left[ 1 -\exp\left(-\frac{\left\vert a\right\vert t}{t_{0}} \right) \right] ~.
\end{align}
By default, the executable produces a plot in \path{./tmp/gnuplotData/} which is given in Figure~\ref{fig:chemReaction2dFD} below.
\begin{figure}[ht!]
	\center
	\includegraphics[width=0.7\textwidth]{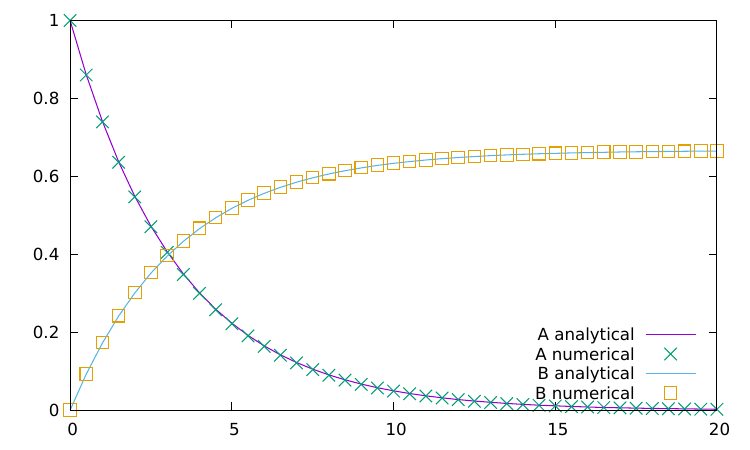}
	\caption{Solutions to concentration profiles in reactionFiniteDifferences2d.}
	\label{fig:chemReaction2dFD}
\end{figure}

\subsection{advectionDiffusion1d}  \label{example:advectionDiffusion1d}
The \path{advectionDiffusion1d} example showcases second order mesh convergence of LBM for scalar linear one-dimensional advection\textendash diffusion equations~\cite{simonis:20} of the form
\begin{align}
\partial_{t} \chi + \partial_{x} F \left( \chi \right) - \mu \partial_{xx} \chi = 0 ~,
\end{align}
where \(\chi : \mathcal{X} \times \mathcal{I} \to \mathbb{R}\) is the conservative variable dependent on space \(x \in \mathcal{X} \subseteq \mathbb{R}\) and time \(t \in \mathcal{I} \subseteq \mathbb{R}_{0}^{+}\), \(F \equiv u \chi \) is a linear function defined by the advection velocity \(u \in \mathbb{R}\), and \(\mu > 0\) denotes the diffusion coefficient. 
Hence, the LBM approximates the transport of the conservative variable \(\chi\) along a one-dimensional line with periodic boundary conditions on \(\mathcal{X} = \left[-1,1\right]\).
The initial pulse is defined by a sine profile, which is subsequently diffused and advected.
An analytical solution at point $x$ and time $t$ is given by
\begin{equation}
  \chi^{\star} (x, t) = \sin\left[ \pi(x - u t) \right] \exp \left( -\mu \pi^2 t   \right) ~.
\end{equation}
In practice the simulation uses a two-dimensional square domain which is evaluated along a centerline to obtain the desired one-dimensional result.
The domain is initialized with $\chi^{\star} (x,t= 0)$.

Diffusive scaling is applied which results in the input parameters listed in Table~\ref{tab:advectionDiffusion1dParameters}.
In the default setting, \path{advectionDiffusion1d} executes three simulation runs with increasing resolutions \(N=50, 100, 200\), respectively.
Each simulation recovers $\mu=1.5$ and a P\'{e}clet number of $Pe=40/3$.

\begin{table}[ht!]
\centering
\begin{tabular}{|p{2cm}|p{2cm}|p{2cm}|}
\hline
 \multicolumn{3}{|c|}{diffusive scaling \(\triangle t = \triangle x^{2}\) for \(Pe=40/3\) } \\ \hline
  \hline
  $N$ & $u_L$  & $\triangle x$ \\ \hline
  \hphantom{0}50 & $0.4$ & 0.04\\ \hline
             100 & $0.2$ & 0.02\\ \hline
             200 & $0.1$ & 0.01\\ \hline
\end{tabular}
\caption{
    Default simulation parameters of \protect\path{advectionDiffusion1d} with \(\mu=1.5\)
    }
\label{tab:advectionDiffusion1dParameters}
\end{table}
The output of each simulation run is stored in the \path{tmp/N<number>} directory.
At each simulation time step the average L2 relative error over the centerline is computed.
Said average is then stored within the respective resolution directory \path{./gnuplotData/data/averageL2RelError.dat}.
Additionally, the program averages the values in \path{averageL2RelError.dat} for each simulation run, which in turn is written to the global error file \path{tmp/gnuplotData/data/averageSimL2RelErr.dat}.
For post-processing, a python3 script can be executed via
\begin{lstlisting}[language=bash]
python3 advectionDiffusion1dPlot.py
\end{lstlisting}
The script requires the \path{matplotlib} python package which can be installed on any platform by issuing the following commands in a terminal:
\begin{lstlisting}[language=bash]
python3 -m pip install -U pip
python3 -m pip install -U matplotlib
\end{lstlisting}
The script generates basic error plots for every file with the file extension \path{.dat} in \path{./tmp}.
Finally, a global log-log error plot with reference curves is extracted from the data contained in \path{averageSimL2RelErr.dat}.

\subsection{advectionDiffusion2d}

The example \path{advectionDiffusion2d} acts as a mesh-convergence test for a solution to the scalar linear \textit{two-dimensional} advection\textendash diffusion equation
\begin{equation}
\partial_{t} \chi + \bm{\nabla}_{\bm{x}} \bm{F} \left( \chi \right) - \mu \bm{Delta}¸_{\bm{x}} \chi = 0 ~,
\end{equation}
where \(\chi : \mathcal{X} \times \mathcal{I} \to \mathbb{R}\) is the conservative variable dependent on space \(\bm{x} \in \mathcal{X} \subseteq \mathbb{R}^2\) and time \(t \in \mathcal{I} \subseteq \mathbb{R}_{0}^{+}\), \(\bm{F} \equiv \bm{u} \chi \) is a linear function defined by the advection velocity \(\bm{u} = \left(u_x, u_y\right)^{\mathrm{T}} \in \mathbb{R}^{2}\), and \(\mu > 0\) denotes the diffusion coefficient.
Similarly, the analytical solution is given for any point \(\bm{x} = \left(x, y\right)^{\mathrm{T}}\) and time \(t\) as
\begin{equation}
\chi^{\star}\left(x, y, t\right) = \mathrm{sin} \left[ \pi \left(x - u_x t\right) \right] \mathrm{sin} \left[ \pi \left( y - u_y t\right) \right] \mathrm{exp} \left( - 2 \mu \pi^{2} t \right) ~.
\end{equation}
The simulation is executed on a square \(\mathcal{X} = \left[-1, 1\right]^{2} \) which is periodically embedded in \(\mathbb{R}^2\).
An error norm over the domain measures the deviation from the analytical solution up to the time step at which the initial pulse is diffused below 10\%.
For the default setting (\(\mu = 0.05\) and \(Pe=100\)), the outputs of three subsequent simulation runs are stored in a subfolder structure in \path{./tmp} and directly post-processed for visualization.
A sequence of contour plots is generated with the highest computed resolution \(N=200\) and contained in \path{./tmp/N200/imageData}.
Note that via issuing the command
\begin{lstlisting}[language=bash]
python3 advectionDiffusion2dPlot.py
\end{lstlisting}
an error plot can be produced, which numerically validates the second order convergence in space.

\subsection{advectionDiffusion3d}\label{subsec:ade3d}
The example \path{advectionDiffusion3d} acts as a mesh-convergence test for a numerical solution to initial value problem
\begin{equation}
\begin{cases}
\partial_{t} \chi (\bm{x}, t) + \bm{\nabla}_{\bm{x}} \bm{F} \left( \chi(\bm{x}, t) \right) - \mu \\bm{\Delta}_{\bm{x}} \chi (\bm{x}, t) = 0  \quad & \text{in }\mathcal{X} \times \mathcal{I} ~, \\
\chi(\bm{x}, 0 ) \equiv \chi_{0} (\bm{x})  \quad & \text{in } \mathcal{X} ~,
\end{cases}
\end{equation}
where \(\chi : \mathcal{X} \times \mathcal{I} \to \mathbb{R}\) is the conservative variable dependent on space \(\bm{x} \in \mathcal{X} \subseteq \mathbb{R}^3\) and time \(t \in \mathcal{I} \subseteq \mathbb{R}_{0}^{+}\), \(\bm{F} \equiv \bm{u} \chi \) is a linear function defined by the advection velocity \(\bm{u} = \left(u_x, u_y, u_{z}\right)^{\mathrm{T}} \in \mathbb{R}^{3}\), and \(\mu > 0\) denotes the diffusion coefficient.
Note that the domain \(\mathcal{X} = [-1 , 1]^{3}\) is periodic.

The example implements a smooth initial profile \(\chi_{0}^{\mathrm{s}} (\bm{x})\) and an unsmooth version \(\chi_{0}^{\mathrm{u}} (\bm{x})\).
The former is a three-dimensional extrusion of the initial pulse in the \path{advectionDiffusion2d} example, such that the equation admits the analytical solution~\cite{dapelo:21,simonis:20,simonis:23}
\begin{equation}
\chi^{\star,\mathrm{s}}\left(x, y, z, t\right) = \mathrm{sin} \left[ \pi \left(x - u_x t\right) \right] \mathrm{sin} \left[ \pi \left( y - u_y t\right) \right] \mathrm{sin} \left[ \pi \left(z - u_z t\right) \right] \mathrm{exp} \left( - 3 \mu \pi^{2} t \right) ~.
\end{equation}: periodically connected in and outlets with forcing combined with different wall treatment. 
The latter comprises a Dirac delta at \(x_0\) as initial pulse which induces a Dirac comb as super-positioned analytical solution~\cite{dapelo:21}
\begin{equation}
\chi^{\star,\mathrm{u}}\left(x,y,z,t\right) =  \frac{1}{\sqrt{4\pi\mu t }} \sum\limits_{k\in \mathbb{Z}} \exp\left( \frac{-\left( x - x_0 - u_{x} t + 2 k\right)^{2}  }{4 \mu t }\right) + 1 ~. 
\end{equation}
For each case several error norms over the domain measure the deviations from the analytical solution.
For the default setting the outputs of three subsequent simulation runs are stored in a subfolder structure in \path{./tmp} and directly post-processed for visualization.
Via issuing the command
\begin{lstlisting}[language=bash]
python3 advectionDiffusion3dPlot.py
\end{lstlisting}
an error plot is produced, which numerically validates the: periodically connected in- and outlets with forcing combined with different wall treatment  second order convergence for both initializations, under the constraints on the grid P{\'e}clet number derived in~\cite{dapelo:21}.

\subsection{advectionDiffusionPipe3d}\label{subsec:adp3d}
This example implements a spreading Gaussian density package advecting within a square duct pipe with velocity \(\bm{U}\).
The precise description of the test-case can be found in~\cite{dapelo:21}.
Whereas the velocity is computed via approximating the incompressible Navier\textendash Stokes equations with a \(D3Q19\) BGK LBM, the advection\textendash diffusion equation for the density package is solved with finite differences (FD).
Four different FD schemes can be employed within the example.
The advantages of each of the schemes for a broad range of \(Pe\) are documented in~\cite{dapelo:21}.

\section{laminar}\label{sec:laminar}

\subsection{bstep2d and bstep3d}\label{sec:bstep2d and bstep3d}
This example implements the fluid flow over a backward facing step. Furthermore, it is shown how checkpointing is used to regularly save the state of the simulation.
The 2D geometry corresponds to Armaly \textit{et al.}~\cite{armaly:83}.

\subsection{cavity2d, cavity2dSolver and cavity3d}\label{sec:cavity2d and cavity3d}
This example illustrates a flow in a cuboid, lid-driven cavity. 
The 2D version also shows how to use the XML parameter files and has an example description file for OpenGPI. 
This example is available in two different versions for sequential and parallel use. 
The \path{2dSolver}-version illustrates the use of the solver class concept (cf.\ Section~\ref{sec:solver}) together with the XML parameter interface.

\subsection{cylinder2d and cylinder3d}
\label{sec:cylinder2dAndCylinder3d}
This example examines a steady flow past a cylinder placed in a channel. 
The cylinder is offset somewhat from the center of the flow to make the steady-state symmetrical flow unstable. 
At the inlet, a Poiseuille profile is imposed on the velocity, whereas the outlet implements a Dirichlet pressure condition set by \(p = 0\), inspired by~\cite{turek:96}. 
For high resolution, low \class{latticeU}, and enough time to converge, the results for pressure drop, drag and lift lie within the estimated intervals for the exact results. 
An unsteady flow with Karman vortex street can be created by changing the Reynolds number to \(Re=100\). 
The 3D version also shows the usage of the STL-reader. 
The model was created using the open source CAD tool FreeCAD~\cite{freecad-web}.

\subsection{poiseuille2d and poiseuille3d}\label{sec:poiseuille2d and poiseuille3d}
For basic tests of boundary conditions, a comparison with analytical solutions is the easiest and most accurate approach. 
One of the fundamental applications of fluid dynamics is that of laminar flow of a Newtonian fluid in a circular pipe. 
This is known as Poiseuille flow. 
The analytical solution is easily found and is therefore a common benchmark case (see Figure~\ref{fig:poiseuille3dCAD}). 
It is also one of the first examples in most fluid dynamics text books for the application of the principles of fluid dynamics. 
The extension of the Poiseuille flow in a round pipe from 2D to 3D is trivial, consequently it is also an ideal test case for curved boundaries in 3D as well.
\begin{figure}[ht]
	\center
	\includegraphics[width=0.7\textwidth]{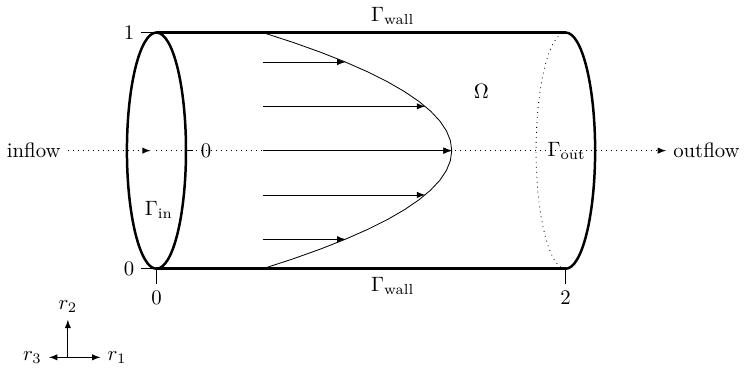}
	\caption{
		Geometry setup in example \protect\path{poiseuille3d} with boundary patches and velocity profile.}
	\label{fig:poiseuille3dCAD}
\end{figure}

\subsection{poiseuille2dEOC and poiseuille3dEOC}
The examples \path{poiseuille2dEOC} and \path{poiseuille3dEOC} are extensions of the \path{poiseuille2d} and \path{poiseuille3d} examples, respectively, focusing on the experimental order of convergence (EOC) of the simulation. 
The implementations therein run the Poiseuille flow simulations multiple times to subsequently analyze the different error norms of the simulations with an automated \class{gnuplot} output (cf.\ Section~\ref{subsec:RegressionWithGnuplot}). 
Exemplary, the \path{poiseuille3dEOC} example is simulated for multiple configurations: the periodic pipe flow with \textit{forcing} and several boundary methods at the walls (keywords \textit{interpolated}, \textit{bouzidi}, and \textit{bounce back} are used to specify the implementations in Section~\ref{sec:defineBoundaryMethod}), as well as the a bounded domain case where a \textit{velocity inlet} and a \textit{pressure outlet} are prescribed and combined with a \textit{bouzidi} boundary at the pipe wall. 
For all simulations, the Reynolds number of the flow problem is set to $Re = 10$ and the relaxation time to $\tau = 0.8$. 
The EOC is investigated by running four consecutive simulations with increasing grid resolution (number of lattice cells along the pipe diameter: $N_D = 21, 31, 41, 51;$). 
A residual below 1e-5 stops each simulation via the convergence criteria (see Section~\ref{subsec:convergenceCrit} and Section~\ref{sec:lessonConvergenceCheck}). 
The absolute error of the simulation from the analytic solution is evaluated on every lattice node whereby different norms for the error calculation are used. 
The $L_1$-norm is defined as 
\begin{align}
    E_{abs, L_1} = \sum^{N}_i \sum^D_d|\phi_{sim,i,d} - \phi_{ana,i,d}|\,\triangle x^D ~,
\end{align}
where $N$ is the number of lattice nodes, $D$ is the number of dimensions of the variable $\phi$, $\phi_{sim}$ is the solution of a flow variable obtained by the simulation, $\phi_{ana}$ is the corresponding analytic solution, and $\triangle x$ is the grid spacing. 
The $L_2$-norm is defined as
\begin{align}
    E_{abs, L_2} = \sqrt{\sum^{N}_i \sum^D_d|\phi_{sim,i,d} - \phi_{ana,i,d}|^2 \,\triangle x^D} ~.
\end{align}
Finally, the $L_\infty$-norm is defined as
\begin{align}
    E_{abs, L_\infty} = \max_{i,d}|\phi_{sim,i,d} - \phi_{ana,i,d}| ~.
\end{align}
The relative error according to the $L_1$-norm is defined as 
\begin{align}
    E_{rel, L_1} = \frac{E_{abs, L_1}}{\sum^{N}_i \sum^D_d|\phi_{ana,i,d}|\,\triangle x^D} ~,
\end{align}
the $L_2$-norm as
\begin{align}
    E_{rel, L_2} = \frac{E_{abs, L_2}}{\sqrt{\sum^{N}_i \sum^D_d|\phi_{ana,i,d}|^2\,\triangle x^D}} ~,
\end{align}
and the $L_\infty$-norm as
\begin{align}
    E_{rel, L_\infty} = \frac{E_{abs, L_\infty}}{\max_{i,d}|\phi_{ana,i,d}|} ~.
\end{align}
The results of the EOC tests are plotted in Figures~\ref{fig:eocPoiseuille3d_per_interpol},~\ref{fig:eocPoiseuille3d_per_bouzidi},~\ref{fig:eocPoiseuille3d_per_bounceback},~\ref{fig:eocPoiseuille3d_inout_bouzidi}, and~\ref{fig:eocPoiseuille3d_inout_bouzidi_pressure}. 
\begin{figure}[ht!]
\centering
\subcaptionbox{Absolute error}{
  \includegraphics[scale=1]{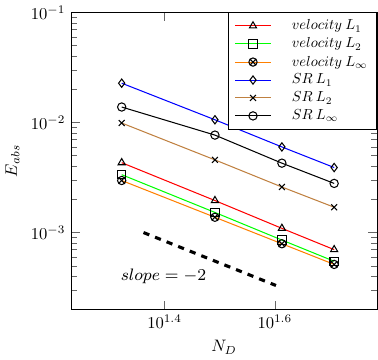}
}
\quad
\subcaptionbox{Relative error}{
  \includegraphics[scale=1]{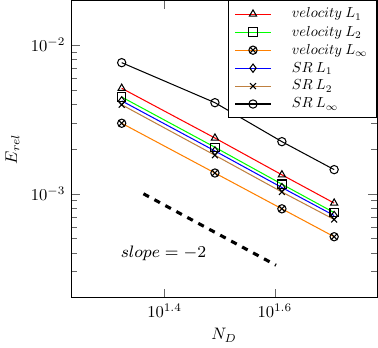}
}
\caption{EOC study for periodic pipe flow simulations (\protect\path{poiseuille3dEOC}) with \textit{forcing} and \textit{interpolated} boundary walls. 
The absolute and relative error of the velocity and the strain-rate (SR) are plotted.}
\label{fig:eocPoiseuille3d_per_interpol}
\end{figure}

\begin{figure}[ht!]
\centering
\subcaptionbox{Absolute error}{
      \includegraphics[scale=1]{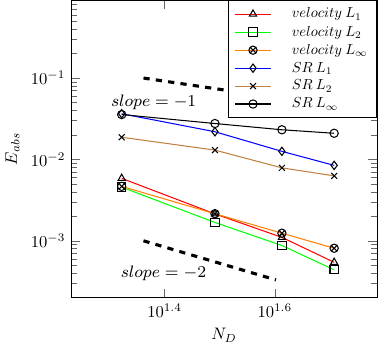}
}
\quad
\subcaptionbox{Relative error}{
    \includegraphics[scale=1]{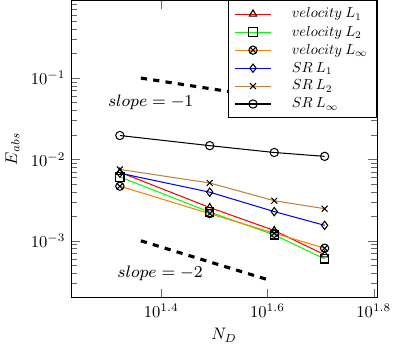}       
}
\caption{
EOC study for periodic pipe flow simulations (\protect\path{poiseuille3dEOC}) with \textit{forcing} and \textit{bouzidi} boundary walls. 
The absolute and relative error of the velocity and the strain-rate (SR) are plotted.}
\label{fig:eocPoiseuille3d_per_bouzidi}
\end{figure}

\begin{figure}[ht!]
\centering
\subcaptionbox{Absolute error}{
      \includegraphics[scale=1]{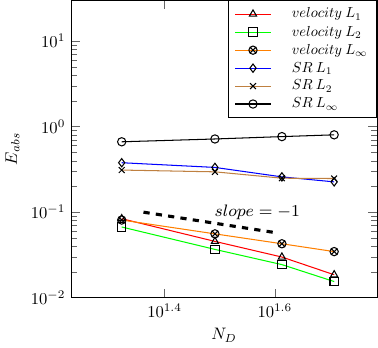}
}
\quad
\subcaptionbox{Relative error}{
   \includegraphics[scale=1]{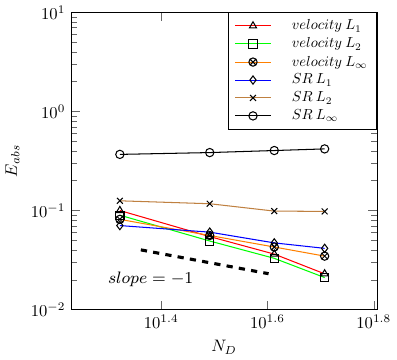}
}
\caption{
EOC study for periodic pipe flow simulations (\protect\path{poiseuille3dEOC}) with \textit{forcing} and \textit{bounce-back} boundary walls. 
The absolute and relative error of the velocity and the strain-rate (SR) are plotted.}
\label{fig:eocPoiseuille3d_per_bounceback}
\end{figure}

\begin{figure}[ht!]
\centering
\subcaptionbox{Absolute error}{
     \includegraphics[scale=1]{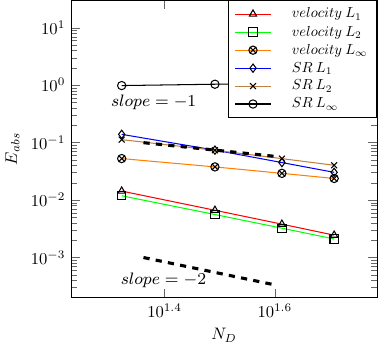}
}
\quad
\subcaptionbox{Relative error}{
    \includegraphics[scale=1]{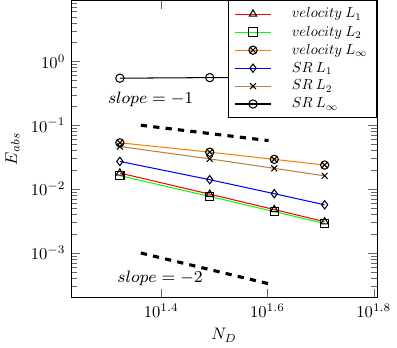}
}
\caption{
EOC study for pipe flow simulations (\protect\path{poiseuille3dEOC}) with \textit{velocity inlet}, \textit{pressure outlet} and \textit{bouzidi} boundary walls. 
The absolute and relative error of the velocity and the strain-rate (SR) are plotted.}
\label{fig:eocPoiseuille3d_inout_bouzidi}
\end{figure}

\begin{figure}[ht!]
\centering
\subcaptionbox{Absolute error}{
    \includegraphics[scale=1]{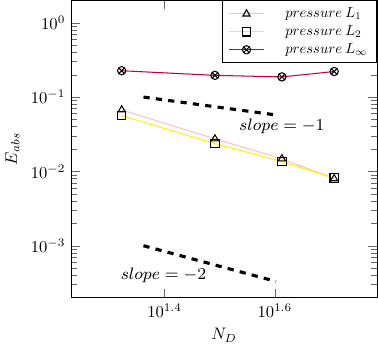}
}
\quad
\subcaptionbox{Relative error}{
      \includegraphics[scale=1]{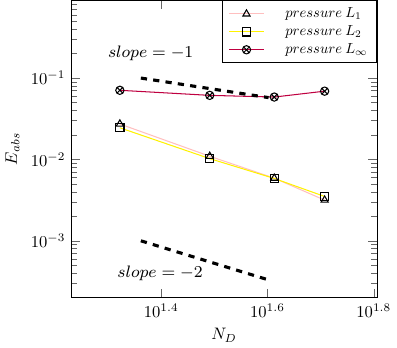}
}
\caption{
EOC study for pipe flow simulations (\protect\path{poiseuille3dEOC}) with \textit{velocity inlet}, \textit{pressure outlet} and \textit{bouzidi} boundary walls. 
The absolute and relative error of the pressure are plotted.
}
\label{fig:eocPoiseuille3d_inout_bouzidi_pressure}
\end{figure}

\subsection{powerLaw2d}\label{sec:powerLaw2d}
This example describe a steady non-Newtonian flow in a channel. 
At the inlet, a Poiseuille profile is imposed on the velocity, whereas the outlet implements a Dirichlet pressure condition set by \(p = 0\).

\subsection{testFlow3dSolver}\label{sec:testFlow3dSolver}
This app implements a Navier--Stokes flow with an analytical solution~\cite{krause:10b}. 
The \textit{standard} simulation as well as an EOC computation for various error norms are shown. 
The implementation makes use of the Solver framework, while simulation parameters are read from the corresponding parameter XML file.
Among others, the following options can be selected:
\begin{itemize}
  \item The flow domain can be either a cube or restricted to a sphere (node [Application][Domain])
  \item Different LBM-boundary conditions (all of them implement a Dirichlet fixed velocity, node [Application][BoundaryCondition])
\end{itemize}
Especially, the EOC computation can be used to benchmark various aspects of LBM implementation, \eg the boundary conditions and the collision steps. 
The result on an example run is shown in Figure~\ref{fig:tf_eoc}.
\begin{figure}[ht!]
	\center
  \includegraphics[width=0.8\textwidth]{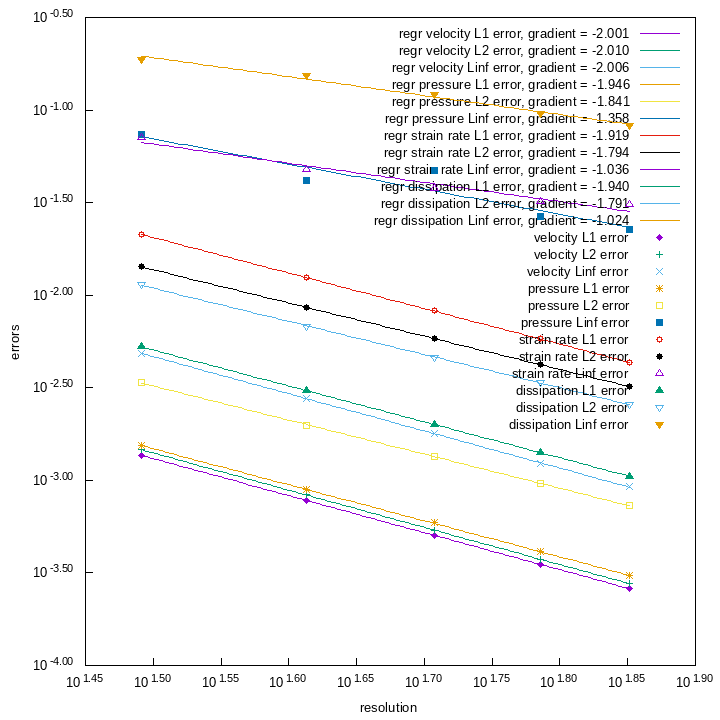}
	\caption{
	  EOC study for the example \protect\path{testFlow3dSolver} on a sphere, as it is returned via the \class{gnuplot} interface.}
  \label{fig:tf_eoc}
\end{figure}
The class structure is as follows: \class{TestFlowBase} implements the basic simulation details (geometry, boundary conditions etc.). 
Since it shall later be used for adjoint optimization, it inherits from \class{AdjointLbSolverBase}~. 
The specification \class{TestFlowSolver} then implements the standard simulation. 
The distinction between \class{TestFlowBase} and \class{TestFlowSolver} is solely because \class{TestFlowBase} will also be used for optimization (cf.\ examples in Section~\ref{sec:testFlowOpti3d}).

\section{optimization}\label{sec:ex_optimization}

\subsection{domainIdentification3d}\label{sec:domainId3d}
In this example, a domain identification problem is solved: a cubic obstacle in the middle of a cubic fluid flow domain has to be identified given only the surrounding fluid flow, cf.\ \cite{klemens:20}.
Therefore, an optimization problem is set up, where the porosity field $\alpha$ in the design domain has to be found such that the error $J$ between resulting velocity field $\mathbf{u}$ and original velocity field $\mathbf{u^\ast}$ is minimized. 
Hence, the objective is $J := \tfrac{1}{2} \int_{\Tilde{\Omega}} \vert \mathbf{u}- \mathbf{u^\ast} \vert ^2 \, dx$.
The example employs adjoint Lattice Boltzmann methods for gradient computation. Results are shown in Figure~\ref{fig:td_opti}.
\begin{figure}[ht!]
	\centering
	\begin{minipage}{0.5\textwidth}
     	\centering
		\includegraphics[width=0.96\textwidth]{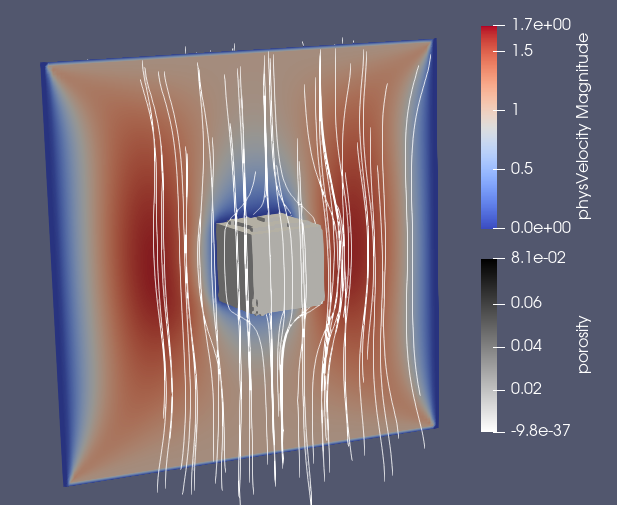}
	\end{minipage}%
	\begin{minipage}{0.5\textwidth}
	      \centering
		\includegraphics[width=0.96\textwidth]{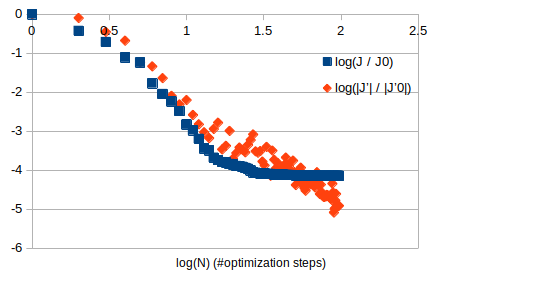}
	\end{minipage}
	\caption{Results of domain identification after 97 optimization steps. Left: identified object, surrounding velocity field and streamlines. Right: relative decrease of the objective functional and its derivative w.r.t.\ the design variables. We can see that the objective hardly changes after ca. 30 optimization steps.}
	\label{fig:td_opti}
\end{figure}

The implementation makes use of the Solver framework and the XML interface for parameter reading, cf.\ Section~\ref{sec:solver}. In addition to the typical parameters for simulation and optimization, the following geometric domains are read from the xml file:
\begin{itemize}
    \item $<$ObjectiveDomain$>$: The domain $\Tilde{\Omega}$, where the objective functional is computed.
    \item $<$SimulationObject$>$: The object that is to be identified by the example, it is used for the computation of the reference solution $\mathbf{u\ast}$.
    \item $<$DesignDomain$>$: the domain, where the porosity field is computed by the optimization algorithm.
\end{itemize}

The OpenLB software provides the method \textit{createIndicator} to create these domains or objects.\\
For example, a three-dimensional cuboid with length $1$ in the origin can be created with\\
\verb|createIndicatorF3D<T>|(\textit{reading from xml-file}) with the input from the xml-file:\\
\verb|<IndicatorCuboid3D extend = "1 1 1" origin = "0 0 0"/>|.\\
These domains are objects and depend on the specific optimization problem and have to be created or modified for each problem again.

The following optimization parameters are specific in the context of domain identification and set in the xml file:
\begin{itemize}
    \item ReferenceSolution: decide whether the reference velocity $\mathbf{u^\ast}$ is computed via a simulation
    \item StartValueType: select between Control, Permeability and Porosity
\end{itemize}

\subsection{poiseuille2Opti}\label{sec:parameterIdentificationPoiseuille2d}
This is a simple showcase for a two-dimensional fluid flow optimization/ parameter identification problem:
It is based on the simulation of a planar channel flow similar to the example \path{poiseuille2d}.
By an optimization loop, the inlet pressure is determined s.t. a given mass flow rate is achieved.
Hence, we solve the following optimization problem:
\begin{equation}
  \underset{p_0}{\mathrm{argmin}} \frac{1}{2} (m(\mathbf{u}(p_0), p(p_0)) - m^{\ast})^2 ~,
\end{equation}
where $(\mathbf{u}(p_0), p(p_0))$ is the solution of the Navier-Stokes equations corresponding to the inlet pressure $p_0$, $m$ is the mass flow rate corresponding to that solution, and $m^{\ast}$ is the wanted mass flow rate.

The optimization is performed with the LBFGS method; the derivative of the objective w.r.t.\ the argument (inlet pressure) is computed with automatic differentiation.

\subsection{showcaseADf}\label{sec:showcaseADf}
This app gives an introduction into the usage of forward automatic differentiation in order to compute derivatives of any numerical quantities in OpenLB.
It is written in the style of literate programming as a tutorial that should allow for being read sequentially.

\subsection{showcaseRosenbrock}\label{sec:showcaseRosenbrock}
This app gives an introduction into the usage of optimization functionalities in OpenLB in order to compute optima of functions or simulations.
It is written in the style of literate programming as a tutorial that should allow for being read sequentially.

\subsection{testFlowOpti3d}\label{sec:testFlowOpti3d}
This app is built on top of the example laminar/testFlow3dSolver and solves parameter identification problems with an optimization approach~\cite{krause:10b}.
It makes use of the Solver framework to create the following variants of application:
\begin{itemize}
  \item Computing sensitivity of flow quantities w.r.t.\ the force field with Difference Quotients (DQ) or Automatic Differentiation (AD)
  \item Optimization with AD: scale the force field s.t. the velocity or dissipation error w.r.t.\ a given flow is minimized. 3 control variables scale the components of the force field. The objective is $J = \tfrac{1}{2} \int_{\Omega} \vert \mathbf{u} - \mathbf{u^\ast} \vert ^2 \, dx$, where $\mathbf{u^\ast}$ is the reference solution field (for control $= (1, 1, 1)$) and $\mathbf{u}$ is the simulated solution (velocity or dissipation) with ''estimated'' control parameters.
  \item Optimization with adjoint LBM: identify the force field s.t. the velocity or dissipation error w.r.t.\ a given flow is minimized. The objective is as above, but the control is the (spatially) distributed field -- three components at each mesh point.
\end{itemize}

The basic simulation setup is implemented in the base class \class{TestFlowBase}. Further optimization-specific implementations (which are used in all optimization variants, \eg objective computation) are done in \class{TestFlowOptiBase}. For DQ and AD, \class{TestFlowSolverDirectOpti} implements the case-specific features, \eg that the vector of control variables acts as a scaling factor for the force field. For adjoint LBM, case-specific features are implemented in \class{TestFlowSolverOptiAdjoint}.
The structure of parameter classes corresponds to to that of the solver classes.

Simulation and optimization parameters are read from the corresponding parameter XML file. The following optimization parameters are specific for this example and set in the XML file:
\begin{itemize}
    \item ReferenceSolution: decide whether reference solution $\mathbf{u^\ast}$ is computed via a simulation
    \item TestFlowOptiMode: decide, whether velocity or dissipation is compared in the objective functional
    \item OptiReferenceMode: use analytical or discrete (simulated) solution as $\mathbf{u^\ast}$
    \item ControlType (force): type of control parameters in the context of adjoint optimization
    \item CuboidWiseControl: if true, then the domain is decomposed into an arbitrary number of cuboids and one control variable is set as a scaling factor for each cuboid (this is meaningful if the sensitivies are computed with finite difference quotients or automatic differentiation.
\end{itemize}

The results of example runs with AD and adjoint LBM are shown in Figure~\ref{fig:tf_opti_eoc}.

\begin{figure}[ht!]
	\centering
	\begin{minipage}{0.5\textwidth}
    	\centering
		\includegraphics[width=0.96\textwidth]{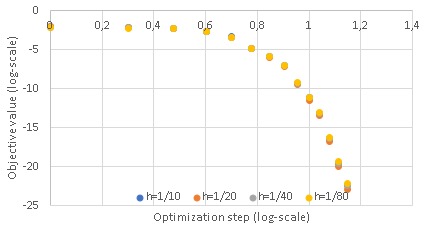}
	\end{minipage}%
	\begin{minipage}{0.5\textwidth}
    	\centering
		\includegraphics[width=0.96\textwidth]{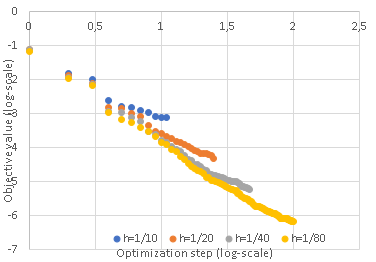}
	\end{minipage}
	\caption{Decrease of the objective for increasing optimization steps for different resolutions (left: AD, right: adjoint LBM). The trends differ since the number of control variables is fixed in the AD case while it increases with the resolution in the adjoint optimization case.}
	\label{fig:tf_opti_eoc}
\end{figure}

\section{multiComponent}\label{sec:multiComponent}
The examples in this folder demonstrate the use of specific LBM models (\eg Section~\ref{sec:freeEnergy}) for multiphase and multicomponent flows.

\subsection{binaryShearFlow2d} \label{sec:binaryShearFlow2d}
A circular domain of one fluid phase is immersed in a rectangle filled with another fluid phase.
The top and bottom walls are moving in opposite directions, such that the droplet shaped phase is exposed to shear flow and deforms accordingly.
The default parameter setting is taken from~\cite{komrakova:2014} and injected into the more general ternary free energy model from~\cite{semprebon:16}.
Both scenarios, breakup and steady state of the initial droplet, are implemented and visualized as \path{.vtk} output. 
Reference simulations are provided in~\cite{simonis:23a}.

\subsection{contactAngle2d and contactAngle3d}\label{sec:contactAngle2d and contactAngle3d}
In this example a semi-spherical droplet of fluid is initialized within a different fluid at a solid boundary. The contact angle is measured as the droplet comes to equilibrium. This is compared with the analytical angle predicted by the parameters set for the boundary (100 degrees for preset values). 
This example demonstrates how to use the solid wetting boundaries for the free-energy model with two fluid components. 

\subsection{fourRollMill2d}\label{sec:fourRollMill2d}
Here, a spherical domain filled with one fluid phase is immersed in a square filled with another phase of equal density and viscosity.
Four circle structures which represent roller sections are equidistantly distributed in the corners of the domain.
The bottom left and top right cylinders begin to spin in counterclockwise direction.
Whereas the top left and bottom right cylinders spin in clock-wise direction.
A velocity field of extensional type deforms the initial droplet accordingly.
Dependent on the non-dimensional parameter setting in the example header, the droplet reaches steady state or breaks up. Reference simulations are provided in~\cite{simonis:23a}.

\subsection{microFluidics2d}\label{microFluidics2d}
This example shows a microfluidic channel creating droplets of two fluid components. 
Poiseuille velocity profiles are imposed at the various channel inlets, while a constant density outlet is imposed at the end of the channel to allow the droplets to exit the simulation. 
This example demonstrates the use of three fluid components with the free energy model. It also shows the use of open boundary conditions, specifically velocity inlet and density outlet boundaries.

\subsection{phaseSeparation2d and phaseSeparation3d}\label{sec:phaseSeparation2d and phaseSeparation3d}
In these examples the simulation is initialized with a given density plus small, random variation over the domain. This condition is unstable and leads to liquid-vapor phase separation. Boundaries are assumed to be periodic. These examples show the usage of multiphase flow.

\subsection{rayleighTaylor2d and rayleighTaylor3d}\label{sec:rayleighTaylor2d and rayleighTaylor3d}
This example demonstrates Rayleigh--Taylor instability in 2D and 3D, generated by a heavy fluid penetrating a light one. The multicomponent fluid model by X. Shan and H. Chen is used~\cite{shan_chen:93}. These examples show the usage of multicomponent flow and periodic boundaries.

\subsection{youngLaplace2d and youngLaplace3d}\label{sec:youngLaplace2d and youngLaplace3d}
In this example the two-component free energy model is used in its simplest configuration to perform a Young--Laplace pressure test. A circular or spherical domain of a fluid with radius $R$ is immersed in another fluid. A diffusive interface forms and the pressure difference across the interface, $\triangle p$, is calculated and compared to that given by the Young--Laplace equation,
\begin{align}
\triangle p = \frac{\gamma}{R} &= \frac{\alpha}{6R} (\kappa_1 + \kappa_2) \qquad \text{for 2D} ~, \\
\triangle p = \frac{2\gamma}{R} &= \frac{\alpha}{3R} (\kappa_1 + \kappa_2) \qquad \text{for 3D} ~.
\end{align}
The parameters $\alpha$ and $\kappa_i$ are input parameters to the simulation which define the interfacial width and surface tension, $\gamma$, respectively.
The pressure difference is calculated between a point in the middle of the circular domain and a point furthest away from it in the computational domain.

\section{particles}\label{sec:particles}

\subsection{bifurcation3d}\label{sec:bifurcation3d}
The \path{bifurcation3d} example simulates particulate flow through an exemplary bifurcation of the human bronchial system. 
The geometry is a splitting pipe, with one inflow and two outflows. 
The fluid is transporting micrometer scale particles and the escape and capture rate is computed. 
There are two implementations of the problem. 
The first one is a Euler--Euler ansatz, meaning that the fluid phase as well as the particle phase are modeled as continua. 
The second is an Euler--Lagrange ansatz, where the particles are modeled as discrete objects.

\subsubsection{eulerEuler}
In this example the particles are viewed as a continuum and described by a advection--diffusion equation. 
This is done similar to the thermal examples, where the temperature is the considered quantity. 
For particles however, inertia has to be taken into account. 
This is achieved by applying the Stokes drag force to the velocity field. 
Since for this computations also the velocity of the previous time step is required, the new descriptor \class{ParticleAdvectionDiffusionD3Q7Descriptor} has to be used, that is capable of saving 2 velocity fields. 
Besides an extra lattice for the advection--diffusion equation, a \class{SuperExternal3D} structure is required to manage the communication for parallel execution.
\begin{lstlisting}[language=myc++]
SuperExternal3D<T,ADDESCRIPTOR,descriptors::VELOCITY> sExternal(
          superGeometry,
          sLatticeAD,
          sLatticeAD.getOverlap());

   ...

sExternal.communicate();
\end{lstlisting}
The function \texttt{communicate()} is called in the time loop and handles the communication analogue to the lattices.

Furthermore the new dynamics object \class{ParticleAdvectionDiffusionBGKdynamics} is required to access the saved velocity fields correctly and use them in an efficient way. 
For information on the coupling of the lattices we refer to the section on the advection--diffusion equation for particle flow problems~\ref{sec:particleADE}. 
In this example only the Stokes drag is applied by
\begin{lstlisting}[language=myc++]
advDiffDragForce3D<T, NSDESCRIPTOR> dragForce( converter,radius,partRho );
\end{lstlisting}
For the simulation of particles as a continuum, also new boundary conditions are required. 
Here \texttt{setZeroDistributionBoundary} represents an unidirectional outflow condition, that removes particle concentrations that cross a boundary. For the usual outflow at the bottom of the bifurcation a new \texttt{AdvectionDiffusionConvectionBoundary} for advection--diffusion lattices can be applied, that approximates a Neumann boundary condition, for further reference see~\cite{trunk:16}. Since non-local computations (gradient is required) are performed on the the external field, also a Neumann boundary condition is required that is here implemented as \texttt{setExtFieldBoundary}.

\subsubsection{eulerLagrange}
The main task of his example is to show the using of Lagrangian particles with OpenLB. 
As this example is used to show the application of the particle framework (\ref{sec:particles}), implementation specifics can be found there.

\subsection{dkt2d}\label{sec:dkt2d}
OpenLB provides an alternative approach to conventional resolved particle simulation methods, referred to
as the homogenized lattice Boltzmann method (HLBM). 
It was introduced in "Particle flow simulations with homogenized lattice Boltzmann methods" by Krause \textit{et al.}~\cite{krause:17} and extended for the simulation of 3D particles in Trunk \textit{et al.}~\cite{trunk:16}. 
It was eventually revisited in~\cite{trunk:21}.
In this approach the porous media model, introduced into LBM by Spaid and Phelan~\cite{spaid:97}, is extended by enabling the simulation of moving porous media. 
In order to avoid pressure fluctuations, the local porosity coefficient is used as a smoothing parameter.

The example \path{dkt2d} employs said approach for the sedimentation of two particles under gravity in a water-like fluid in 2D. 
The rectangular domain is limited by no-slip boundary conditions. This setup is usually referred to as a drafting–kissing–tumbling (DKT) phenomenon and is widely used as a reference setup for the simulation of particle dynamics submerged in a fluid. 
The benchmark case is described \eg in~\cite{wang2014drafting, feng2004immersed}.
For the calculation of forces a DNS approach is chosen which also leads to a back-coupling of the particle on the fluid, inducing a flow. 
The example demonstrates the usage of HLBM in the OpenLB framework as well as the utilization of the \path{gnuplot}-writer to print simulation results.

\subsection{magneticParticles3d}\label{sec:magneticParticles3d}
\textbf{Warning:} This example can currently only be run sequentially! 
High-gradient magnetic separation is a method to separate ferromagnetic particles from a suspension. 
The simulation shows the deposition of magnetic particles on a single magnetized wire and models the magnetic separation step of the complete process.

\subsection{settlingCube3d}\label{sec:settlingCube3d}

The case examines the settling of a cubical silica particle under gravity in a surrounding fluid. 
The rectangular domain is limited by no-slip boundary conditions. 
For the calculation of forces a DNS approach is chosen which also leads to a back-coupling of the particle on the fluid, inducing a flow.
The example demonstrates the usage of HLBM in the OpenLB framework as well as the utilization of the \path{gnuplot}-writer to print simulation results (Section~\ref{sec:dkt2d}).

\section{porousMedia}\label{sec:porousMedia}

\subsection{porousPoiseuille2d and porousPoiseuille3d}\label{sec:porousPoiseuille2d and porousPoiseuille3d}
This example simulates Poiseuille flow through porous media.
The implementation reproduces the benchmark \textit{Example A} by Guo and Zhao~\cite{guo2002lattice}. 
The theoretical maximum velocity is calculated as in \textit{Equation 21}, and the velocity profile as in \textit{Equation 23} of~\cite{guo2002lattice}. 
For a schematic simulation setup see Figure~\ref{fig:poiseuille2dCAD}.
\begin{figure}[ht!]
	\center
	\includegraphics[width=0.8\textwidth]{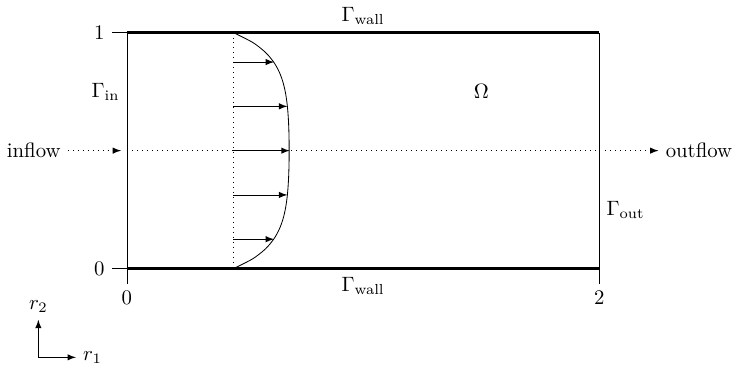}
	\caption{
		Geometry used in the example \protect\path{porousPoiseuille2d} with boundary patches and velocity profile.}
	\label{fig:poiseuille2dCAD}
\end{figure}

\section{reaction}\label{sec:reaction}

\subsection{advectionDiffusionReaction2d}
\label{example:advectionDiffusionReaction2d}
This example illustrates a forward or a reversible reaction of a substance into another one (A$\to$C or A$\textcolor{orange}{\leftrightarrow}$C).
The concentration of each substance is modeled with a one-dimensional advection diffusion equation. We consider a steady state in a plug flow reactor. 
This means that it is time independent and we have a constant velocity field $u$. The chemical reaction is modeled with a linear but coupled source term. 
The equations read
\begin{align}
\begin{cases}
u \partial_x c_A &= D \partial_x^2c_A-k_H c_A \textcolor{orange}{+k_R c_C} ~,\\
u \partial_x c_C &= D \partial_x^2c_C+k_H c_A \textcolor{orange}{-k_R c_C} ~,
\end{cases}
\end{align}
with the diffusion coefficient $D>0$, forth reaction rate coefficient $k_H>0$ and backwards reaction rate coefficient $k_R>0$ (= 0 in case of $A\to$ C).
This LBM models the transport of the species concentration along a one-dimensional line on $[0,10]$. 
In practice the simulation uses a two-dimensional rectangular domain which is evaluated along a centerline to obtain the desired one-dimensional result. 
The height of the domain depends on the resolution which holds the number of voxels for the height constant.

On the bottom and the top of the rectangular periodic boundaries and at the inlet and outlet the concentrations from the analytical solution are set. 
The solution is given by
\begin{align}
c_A(x)&= c_{A,0} e^{\lambda x}\textcolor{orange}{+\frac{k_R}{k_H+k_R}c_{A,0} \left(1-e^{\lambda x}\right)} ~, \\
c_C(x)&= c_{A,0}-c_A(x) ~,
\end{align}
with $\lambda= \frac{u-\sqrt{u^2+4\textcolor{orange}{(}k_H\textcolor{orange}{+k_R)} D}}{2D}$.\\
Diffusive scaling is applied and a physical diffusivity of $D=0.1$ and a flow rate of $u=0.5$ which leads to a P\'{e}clet number of $Pe=100$.

Every species has its own lattice and stored in a vector. In the \texttt{simulate} method we iterate over ever element of the vector \texttt{adlattices}.

In the default setting, \path{advectionDiffusionReaction2d} executes three simulation runs with increasing resolutions \(N=200, 250, 300\), respectively.
The output of each simulation run is stored in the \path{tmp/N<number>} directory. 
It contains a plot \path{centerConcentrations.pdf} of the concentrations and the analytical solution along the centerline and an error plot of the numerical and analytical solution along the centerline \path{Error Concentration.pdf}.
After each simulation has converged the average L2 relative Error over the centerline is computed.
Said average is then stored within \path{tmp/gnuplotData/data/averageL2RelError.dat}.
The order of convergence can be seen in the log-log error plot in \path{tmp/gnuplotData/concentration_eoc.png}. 
The EOC plot is only done for one species (A or C) and the species can be chosen in the return statement of the method \texttt{errorOverLine}.

One can select the reactionType \texttt{a2c} or \texttt{a2cAndBack} which automatically provides the data for modeling the reaction. 
It contains the number of reactions, the reaction rate coefficients \texttt{physReactionCoeff[numReactions]} ($k_H$, $k_R$), the number of species \texttt{numComponents} and their names, the stoichometric coefficients \texttt{stochCoeff} which are sorted according to the number of reactions and inside each reaction block according to the species number. Finally we assume that the reaction rate satisfies a power law depending on the concentration of the species. 
The exponent is given by the reaction order \texttt{reactionOrders} which is sorted in the same way as \texttt{stochCoeff}. 
In the example cases these exponents are always $1$. 
The chemical reaction itself is represented as a source term for each Advection Diffusion equation. This source term is calculated in the \class{ConcentrationAdvectionDiffusionCouplingGenerator} for every species which can handle arbitrary number of species and reactions and stored in the field \class{SOURCE}.

\section{thermal}\label{sec:thermal}

\subsection{galliumMelting2d}

The solution for the melting problem (solid-liquid phase change) coupled with natural convection is found using the lattice Boltzmann method after Huang and Wu~\cite{huangwu:2015}. The equilibrium distribution function for the temperature is modified in order to deal with the latent-heat source term. That way, iteration steps or solving a group of linear equations is avoided, which results in enhanced efficiency. The phase interface is located by the current total enthalpy, and its movement is considered by the immersed moving boundary scheme after Noble and Torczynski~\cite{nobletorczynski:1998}. This method was validated by comparison with experimental values (\eg Gau and Viskanta~\cite{gauviskanta:1986}).

\subsection{porousPlate2d, porousPlate3d and porousPlate3dSolver}

The porous plate problem is implemented as described in~\cite{guo2002coupled} and~\cite{peng2003simplified}. 
To test the coupled model's accuracy and to determine its EOC, we use a numerical simulation the porous plate problem including a temperature gradient and natural convection in a square cavity. 
The porous plate problem describes a channel flow, where the upper cool plate moves with a constant velocity, and through the bottom warm plate a constant normal flow is injected and withdrawn at the same rate from the upper plate.
At the left and right hand side of the domain a periodic boundary condition is applied and constant velocity and temperature boundary conditions are applied to the top and bottom plates according to Figure~\ref{fig:porousplateprob}.
\begin{figure}[ht!]
  \centering
    \includegraphics[scale=1]{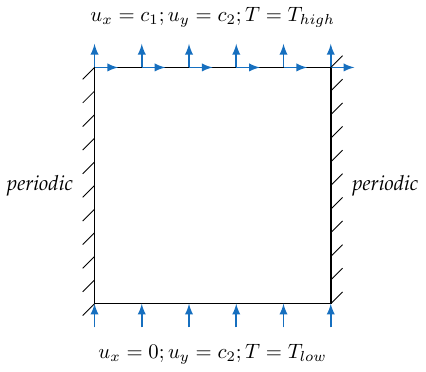}
  \caption{Schematic representation of the porous plate's simulation setup including the boundary conditions.}
  \label{fig:porousplateprob}
\end{figure}
An analytical solution for the given steady state problem is given for the velocity and temperature distributions by
\begin{align}
 u_x(y) & =u_{x,0}(\frac{e^{Re \cdot y/L}-1}{e^{Re}-1}) ~, \\
T(y) & =T_0+ \triangle T=(\frac{e^{Pr \cdot Re \cdot y/L} -1}{e^{Pr \cdot Re}-1}) ~.
\label{analytical solution}
\end{align}
Here $u_{x,0}$ is the upper plate's velocity, $Re = \frac{u_{y,0} L}{\nu}$ the Reynolds number depending on the injected velocity $u_{y,0}$, the fluid's viscosity $\nu$ and the channel length $L$. 
The temperature difference between the hot and cold plate is given by $\triangle T=T_h-T_c$.
First we implement a couple of simulations to scale the velocity and temperature profiles for a range of the Reynolds number and Prandtl number. 
The relative global error is computed via~\cite{guo2002coupled}
\begin{equation}
E= \frac{\sqrt{\sum \limits_{i} |T(x_i)-T_a(x_i)|^2}}{\sqrt{\sum \limits_{i} |T_a(x_i)|^2}} ~,
\end{equation}
where the summation is over the entire system, $T_a$ is the analytical solution \eqref{analytical solution}.
The \path{porousPlate3dSolver} example implements the same simulation, but additionally illustrates the application of the solver class concept.

\subsection{rayleighBenard2d and rayleighBenard3d}\label{sec:rayleighBenard2d and rayleighBenard3d}

The Rayleigh--B\'enard convection is a typical case of natural convection, where the lower boundary is heated and a regular pattern of convection cells is developed. This is a suitable test platform for thermal algorithms, since the driving force is a coupling between momentum and energy equations by means of a buoyancy force, which is function of the temperature, and the temperature varies spatially inside the domain.
This example demonstrates Rayleigh--B\'enard convection rolls in 2D and 3D, simulated with the thermal LB model by Guo \textit{et al.}~\cite{guo:02c}, between a hot plate at the bottom and a cold plate at the top.

\subsubsection{Setup}
The case considered has an aspect ratio ($AR$ = $Lx/Ly$) of \(2\), which enhances the appearance of unstable modes. The lower wall is heated with a constant temperature ($T = 1$), and the upper wall is isothermal and cold ($T = 0$). The vertical walls are set to be periodic. 

Among the example programs implemented in OpenLB, a demo code for the Rayleigh--Bénard convection in 2D and 3D is provided. This code is taken as a base for the development of most of the thermal applications. For the simulation of the Rayleigh--Bénard convection only one modification is made to the code regarding the initial conditions: to enhance the appearance of the convection cells, an instability in the domain is introduced. The available code initializes a small area near the lower boundary with a slightly higher temperature, introducing a perturbation in the system, whereas the rest of the domain is initialized with the cold temperature. In the modified code there is no local perturbation, but the initial temperature at the domain is dependent on the space coordinates. The domain is initialized with zero velocity and a temperature field by using a functor according to  
\begin{equation}
T(x,y,t=0)= T_{max} [(1-\frac{y}{L_y}) + 0.1cos(2\pi\frac{x}{L_x})] ~. 
\label{eq:T(x,y,t=0)}
\end{equation}
\begin{sloppypar}
The files created to help with the initialization of the temperature field are called \path{tempField.h} and \path{tempField2.h}. 
Listing~\ref{lst:listingTemperate} shows the corresponding usage. 
The first file computes the temperature at every point of the lattice, as a function of its macroscopic position, and then this value is applied on the lattice as the density (line~3). The second file calculates the equilibrium distribution functions for every node corresponding to the given temperature and zero velocity. Next, the populations are defined for the desired material number in line~4. 
\end{sloppypar}
\begin{lstlisting}[language=myc++,caption={Initialization of the temperature field}, label=lst:listingTemperate]
TemperatureField2D<T,T> Initial( converter );
TemperatureFieldPop2D<T,T> EqInitial( converter );
ADlattice.defineRho( superGeometry, 1, Initial );
ADlattice.definePopulations( superGeometry, 1, EqInitial );
\end{lstlisting}
In \eqref{eq:T(x,y,t=0)}, the \(y\)-dependent part of the equation matches the stationary solution of the problem, corresponding to a case where there is no fluid movement and the heat transfer only
occurs by conduction. The cosine term introduces a disturbance in the system, which enhances the appearance of the convection cells.

\subsubsection{Simulation Parameters}
Computations are run for a range of different Rayleigh ($3\cdot 10^3$, $6\cdot 10^5$) and Prandtl numbers (\(0.3\), \(1\)). The spatial resolution was fixed to \(100\) cells in the \(y\)-direction, and the time discretization was switched between $10^{-3}$ and $10^{-4}$, which give lattice velocities of \(0.1\) and \(0.01\) respectively. The convergence criterion is applied on the average energy,
and it is set to a precision of $10^{-5}$.

\subsection{squareCavity2d and squareCavity3d}
A common application for the validation of thermal models is the numerical simulation of the natural convection in a square cavity. 
For this configuration there is an extensive database in a wide range of Rayleigh numbers, which allows to verify the accuracy of the thermal model. 

\subsubsection{Setup}
The problem considered is shown schematically in Figure~\ref{fig:sDiagOTSimDom}. 
The horizontal walls of the cavity are adiabatic, while the vertical walls are kept isothermal, with the left wall at high temperature ($T_{hot}$ = 1) and the right wall at low temperature ($T_{cold}$ = 0).
\begin{figure}
\centering
  \includegraphics[scale=1]{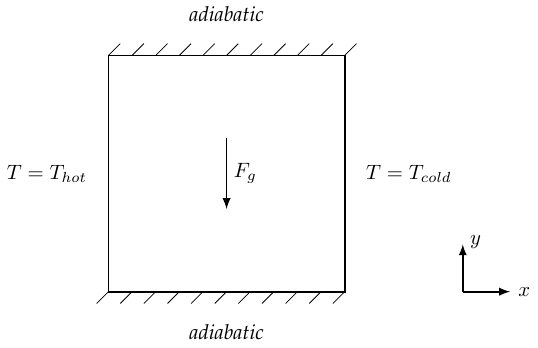}
\caption{Schematic diagram of the simulation domain for the example \protect\path{squareCavity2d}.}
\label{fig:sDiagOTSimDom}
\end{figure}

The dynamics chosen for the velocity field is \class{ForcedBGKdynamics}, and for the temperature field \class{AdvectionDiffusionBGKdynamics}.

\subsubsection{Simulation Parameters}
Taking air at \(293 \mathrm{K}\) as working fluid, the value of the Prandtl number is $Pr = 0.71$ and is kept constant. 
The Rayleigh number ranges from $10^3$ to $10^6$.
Different spatial resolutions are tested for each Rayleigh number, in order to study the grid convergence. 
The time-step size is adjusted so that the lattice velocity stays at the value \(0.02\). 
This ensures that the Mach number is kept at incompressible levels. The convergence criterion is set by a standard deviation of $10^{-6}$ in the kinetic energy.

\subsubsection{MRT}
The new implemented MRT model for thermal applications is first examined on the 2-dimensional cavity. 
The only setup differences to the BGK model are the lattice descriptors (\texttt{ForcedMRTD2Q9Descriptor} and \texttt{AdvectionDiffusionMRTD2Q5Descriptor}) and the dynamics objects selected, which are now specialized for the MRT dynamics (\texttt{ForcedMRTdynamics} and \texttt{AdvectionDiffusionMRTdynamics}).
This simulation is used as a test for different important aspects of the implementation. 
First, the formulation of the MRT model, particularly the values of the transformation matrix, the relaxation times and the sound speed of the lattice are based on~\cite{li2013boundary}, but it show variations over 10\% with respect to the BGK model. 
A second formulation \cite{liu2015double} is selected, which shows much closer results to the BGK model. 
No special treatment is required to make use of the available boundary conditions.
The number of iterations required to achieve the desired precision, that is, the number of time steps until the steady-state solution is reached, is found to be usually higher for the MRT simulations. 
Furthermore, the execution time is between 4 and 8 times longer when compared to
the BGK simulations.

\subsection{stefanMelting2d}

The solution for the melting problem (solid-liquid phase change) is computed using the LBM from Huang and Wu~\cite{huangwu:2015}. 
The equilibrium distribution function for the temperature is modified in order to deal with the latent-heat source term. 
That way, iteration steps or solving a group of linear equations is avoided, which results in enhanced efficiency. 
The phase interface is located by the current total enthalpy, and its movement is considered by the immersed moving boundary scheme after Noble and Torczynski~\cite{nobletorczynski:1998}. 
Huang and Wu validated this method by the problem of conduction-induced melting in a semi-infinite space, comparing its results to analytical solutions.

\section{turbulent}\label{sec:turbulent}

\subsection{aorta3d}\label{sec:aorta3d}
In this example, the fluid flow through a bifurcation is simulated. 
The geometry is obtained from a mesh in STL format. 
With Bouzidi boundary conditions, the curved boundary is adequately mapped and initialized entirely automatically. 
A Smagorinsky LES BGK model is used for the dynamics to stabilize the turbulent flow simulation for low resolutions. 
The output is the flux computed at the inflow and outflow region. 
The results have been validated through comparison with other results obtained with FEM and FVM.

\subsection{channel3d}\label{sec:channel3d}
This example features the application of wall functions in a bi-periodic, fully developed turbulent channel flow for friction Reynolds numbers of $Re_\tau=1000$ and $Re_\tau=2000$. 
For the published results and further reference see~\cite{haussmann:19}.

\subsection{nozzle3d}\label{sec:nozzle3d}
On the one hand this example describes building a cylindrical 3D geometry in OpenLB, on the other hand it examines turbulent flow in a nozzle injection tube using different turbulence models and Reynolds numbers.

For characterization different physical parameters have to be set. 
The resolution $N$ defines most physical parameters such as the velocity \class{charU}, the kinematic viscosity $\nu$ and two characteristic lengths \class{charL} and \class{latticeL}.
Physical length \class{charL} is used to characterize the geometry and the Reynolds number.
Lattice length \class{latticeL} defines the mesh size and is calculated as $\texttt{latticeL} = \texttt{charL}/N$.
More information about the parameter definitions are located in the file \path{units.h}.
\begin{table}[ht!]
\centering
\caption{Preset simulation parameters of the \protect\path{nozzle3d} example.}
\begin{tabular}{|l|l|}
  \hline
  parameter & value \\ \hline
  $\texttt{charL}$ & $1 \mathrm{m}$ \\ \hline
  $\texttt{latticeL}$  & $\frac{1}{3} \mathrm{m}$ \\ \hline
  $\texttt{charU}$ & $1 \frac{\mathrm{m}}{\mathrm{s}}$ \\ \hline
  $\nu$ & $0.00002 \frac{\mathrm{m}^2}{\mathrm{s}}$ \\ \hline
  $Re_{inlet}$  & $5000$ \\ \hline
  turbulence model  & Smagorinsky \\ \hline
\end{tabular}
\end{table}
Figure~\ref{fig:geometryNozzle3d} illustrates the geometry and the nozzle's size as a function of the characteristic length \class{charL}.
The nozzle consists of two circular cylinders.
The inflow (red) is located left in the \textit{inletCylinder}.
The outflow (green) is at the right end of the \textit{injectionTube}.
At the main inlet, either a block profile or a power $1/7$ profile is imposed as a Dirichlet velocity boundary condition, whereas at the outlet a Dirichlet pressure condition is set by $p=0$ (\ie $rho=1$).
\begin{figure}[ht!]
  \includegraphics[width=1.0\textwidth]{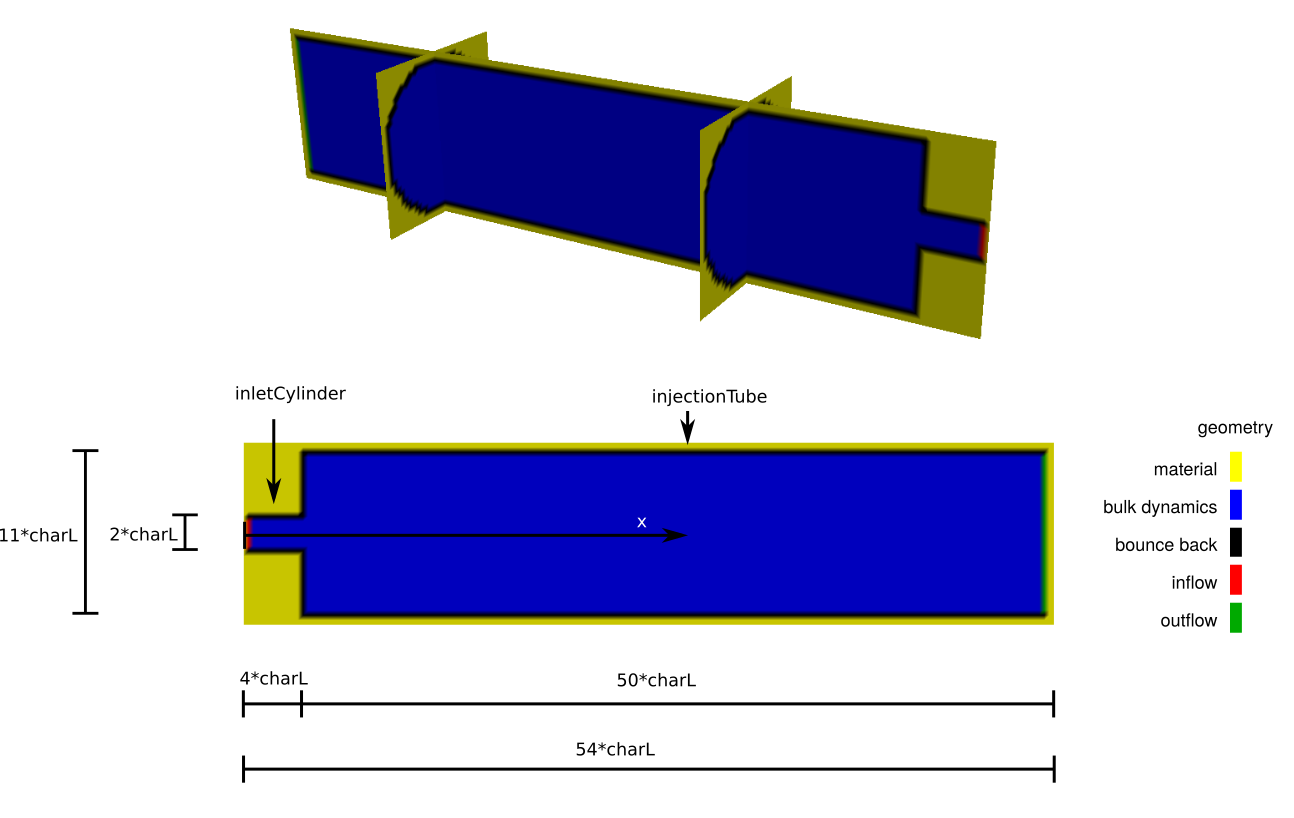}
  \caption{Cross section of a 3D geometry of \protect\path{nozzle3d} in dependency of characteristic length \protect\texttt{charL}.}
  \label{fig:geometryNozzle3d}
\end{figure}
Two vectors, origin and extend, describe the center and normal direction of the cylinder's circular start (origin) and end (extend) plane. 
The radius is defined in the function.

As mentioned before, this example simulates turbulent fluid flow.
The flow behavior in the inlet is characterized by the Reynolds number.
The following turbulence models are based on large eddy simulation (LES).
The idea behind LES is to simulate only eddies larger than a certain grid filter length, while smaller eddies are modeled.
Several models are currently implemented, e.g.:
\begin{itemize}
  \item The \textbf{Smagorinsky model} reduces the turbulence to a so called eddy viscosity. This viscosity depends on the Smagorinsky constant, which has to be defined. This model has certain disadvantages at the wall. 
  \item The \textbf{Shear-improved Smagorinsky model (SISM)} is based on the Smagorinsky model. Compared to the original model, the SISM works at the wall very well. Similarly, a model specific constant has to be defined.
\end{itemize}
The following code shows the model selection.
A model is selected, when the correlate line is uncommented.
Below, the model specific constants are defined.
In this case the Smagorinsky model is selected.
Smagorinsky constant is set to \(0.15\).
\noindent
\begin{lstlisting}[style=intext]
/// Choose your turbulent model of choice

#define Smagorinsky

...

#elif defined(Smagorinsky)
bulkDynamics = new SmagorinskyBGKdynamics<T, DESCRIPTOR>(converter.getOmega(), instances::getBulkMomenta<T, DESCRIPTOR>(),
0.04, converter.getLatticeL(), converter.physTime());

\end{lstlisting}

As an example, Figure~\ref{fig:nozzle3d_physVelocity} shows the results with preset parameters.
\begin{figure}[ht!]
  \includegraphics[width=1.0\textwidth]{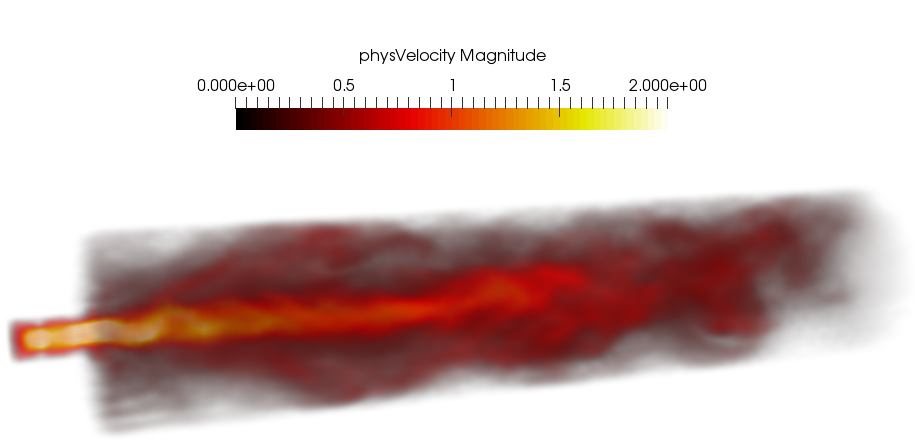}
  \caption{Physical velocity field after 200 seconds with preset parameters (Smagorinsky Model, $C_S = 0.15$, $\texttt{latticeL} = \frac{1}{3} \mathrm{m}$, $Re_{inlet}=5000$).}

  \label{fig:nozzle3d_physVelocity}
\end{figure}
The simulations strongly depends on the Smagorinsky constant's value, used in the turbulence model.
However, the constant is not a general calculable value and valid for one model.
It could be a function of the Reynolds number and/or another dimensionless parameter.
Thus, a physically useful value has to be found by trial and error, or chosen as an educated guess in the beginning.
Generally, if the constant's value is chosen too small, the simulation results will become unstable and/or unphysical.
If the value is too large, the model will introduce too much artificial viscosity and smooth the results.

\begin{figure}[ht!]
  \centering
  \includegraphics[width=0.8\textwidth]{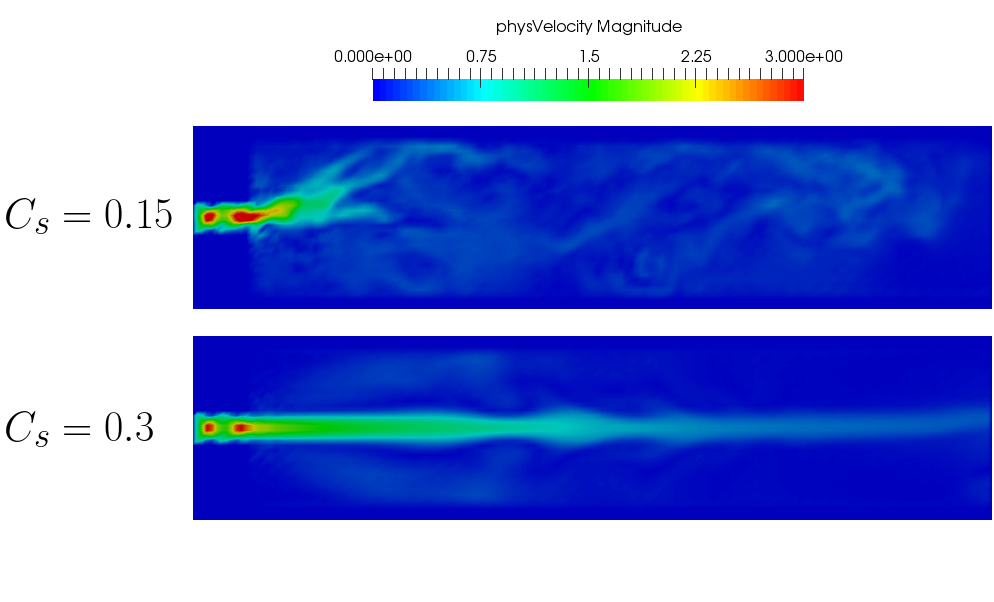}
  \caption{Physical velocity magnitude plotted on the cross section for two Smagorinsky constants. 
  Increasing the Smagorinsky constant visibly smoothens the results and straightens the appearing turbulence.}
\end{figure}

\subsection{tgv3d}\label{sec:tgv3d}
The Taylor--Green vortex (TGV) is one of the simplest configurations to investigate the generation of small scale structures and the resulting turbulence.
\begin{figure}[ht!]
  \centering
  \includegraphics[width=0.8\textwidth]{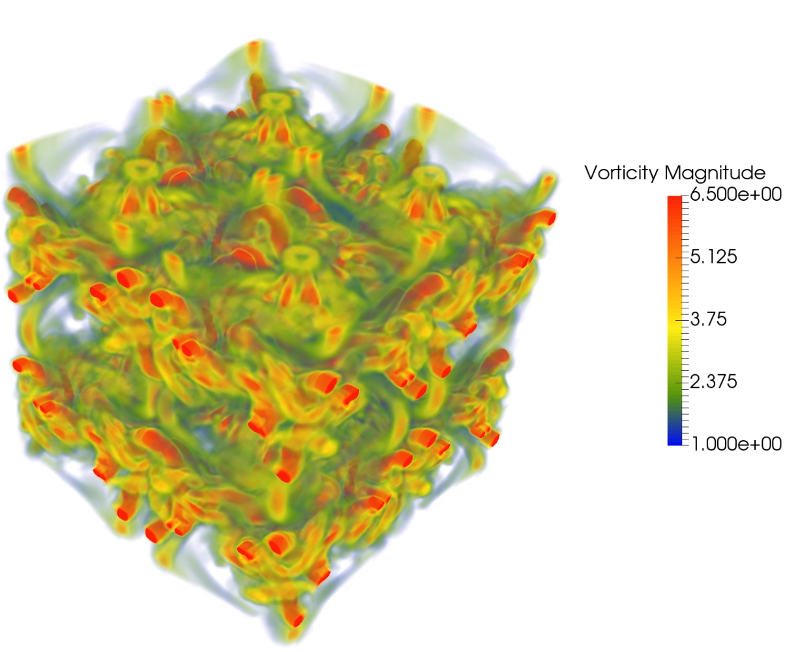}
  \caption{Isosurfaces of vorticity for the Taylor--Green vortex at $t=12 \mathrm{s}$.}
\end{figure}
The cubic domain $\Omega = \left(2\pi\right)^3$ with periodic boundaries and the single mode initialization contribute to the model's simplicity.
In consequence, the TGV is a common benchmark case for direct numerical simulation (DNS) as well as large eddy simulation (LES).
This example demonstrates the usage of different sub-grid models and visualizes their effects on global turbulence quantities.
The molecular dissipation rate, the eddy dissipation rate and the effective dissipation rate are calculated and plotted over the simulation time.
The results can be compared with a DNS solution published by Brachet \textit{et al.}~\cite{brachet:83}.

\subsection{venturi3d}\label{sec:venturi3d}
This example examines a steady flow in a Venturi tube. 
A Venturi tube is a cylindrical tube, which has a reduced cross-section in the middle part. 
At this constriction is an injection tube. 
As a result of the accelerating fluid in the constriction, the static pressure decreases and the injection tube's fluid is pumped in the main tube.
The overall geometry is built with adding together single bodies. 
Each body's geometry is defined by certain points (position vectors) in the coordinate system and their radius. 
A cone-shaped cylinder needs the center of the start an end circle as well as the radii. Following code builds the geometry and shows the semantics.

\noindent
\begin{lstlisting}[style=intext]
/// Definition of the geometry of the venturi

//Definition of the cross-sections' centers
Vector<T,3> C0(0,50,50);
Vector<T,3> C1(5,50,50);
Vector<T,3> C2(40,50,50);
Vector<T,3> C3(80,50,50);
Vector<T,3> C4(120,50,50);
Vector<T,3> C5(160,50,50);
Vector<T,3> C6(195,50,50);
Vector<T,3> C7(200,50,50);
Vector<T,3> C8(190,50,50);
Vector<T,3> C9(115,50,50);
Vector<T,3> C10(115,25,50);
Vector<T,3> C11(115,5,50);
Vector<T,3> C12(115,3,50);
Vector<T,3> C13(115,7,50);

//Definition of the radii
T radius1 = 10 ;  // radius of the tightest part
T radius2 = 20 ;  // radius of the widest part
T radius3 = 4 ;   // radius of the small exit

//Building the cylinders and cones
IndicatorCylinder3D<T> inflow(C0, C1, radius2);
IndicatorCylinder3D<T> cyl1(C1, C2, radius2);
IndicatorCone3D<T> co1(C2, C3, radius2, radius1);
IndicatorCylinder3D<T> cyl2(C3, C4, radius1);
IndicatorCone3D<T> co2(C4, C5, radius1, radius2);
IndicatorCylinder3D<T> cyl3(C5, C6, radius2);
IndicatorCylinder3D<T> outflow0(C7, C8, radius2);
IndicatorCylinder3D<T> cyl4(C9, C10, radius3);
IndicatorCone3D<T> co3(C10, C11, radius3, radius1);
IndicatorCylinder3D<T> outflow1(C12, C13, radius1);

//Addition of the cylinders to overall geometry
IndicatorIdentity3D<T> venturi(cyl1 + cyl2 + cyl3 + cyl4 + co1 + co2 + co3);
\end{lstlisting}
Figure~\ref{fig:venturi3d_geometry_2} visualizes the defined point positions and Figure~\ref{fig:venturi3d_geometry_1} shows the computational geometry.
\begin{figure}[ht!]
  \centering
  \includegraphics[width=0.7\textwidth]{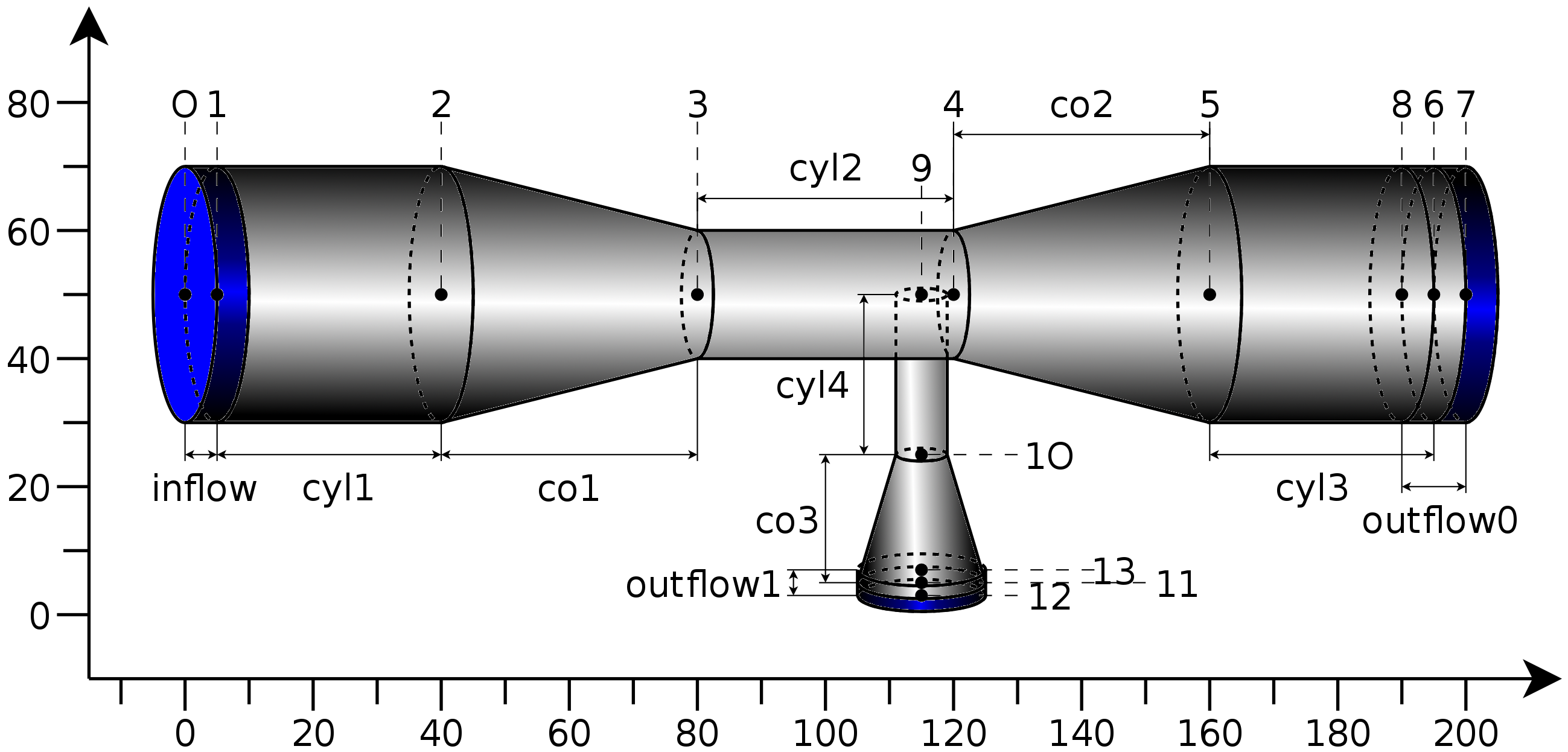}
  \caption{Schematic diagram visualizing the defined point positions for \protect\path{venturi3d}.}
  \label{fig:venturi3d_geometry_2}
\end{figure}
\begin{figure}[ht!]
 \centering
  \includegraphics[width=0.9\textwidth]{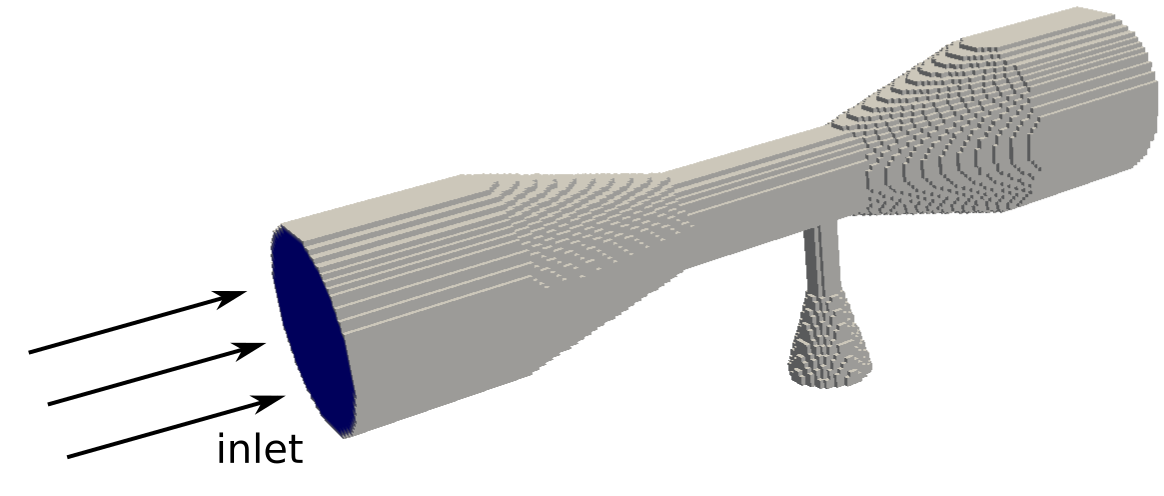}
  \caption{Built geometry as used in simulation of example \protect\path{venturi3d}.}
  \label{fig:venturi3d_geometry_1}
\end{figure}
\begin{figure}[ht!]
  \centering
  \includegraphics[width=0.7\textwidth]{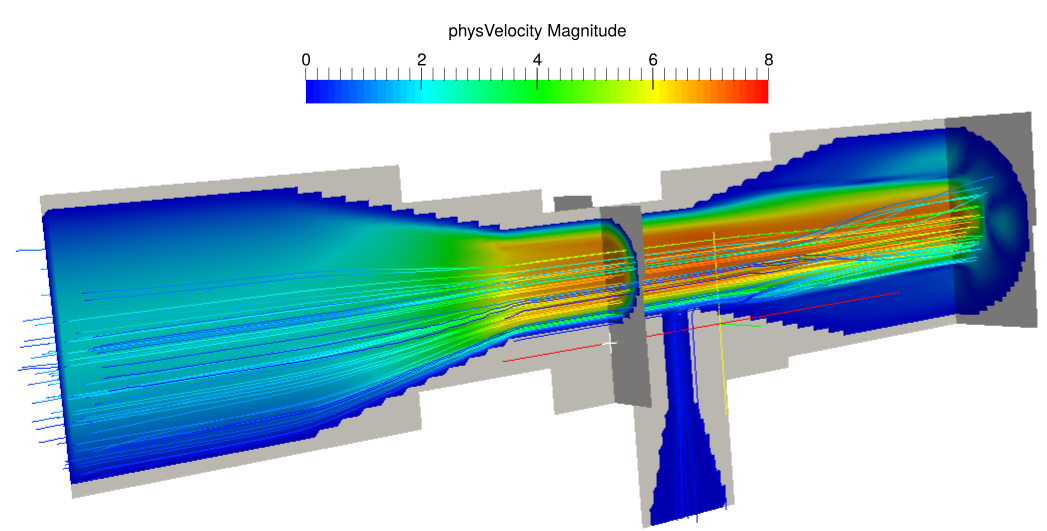}
  \caption{Simulation of the the example \protect\path{venturi3d} after \(200\) simulated time steps.}
  \label{fig:venturi3d_simulation_1}
\end{figure}
At the main inlet, a Poiseuille profile is imposed as a Dirichlet velocity boundary condition, whereas at the outlet and the minor inlet, a Dirichlet pressure condition is set by $p=0$ (\ie $\rho=1$).
Figure~\ref{fig:venturi3d_simulation_1} visualizes the computed velocity magnitude in the Venturi tube geometry.

\section{freeSurface}\label{sec:freeSurface}
The free surface approach~\cite{thurey:07} is a numerical simulation of two phases, where one of the phases does not have to be simulated thoroughly. 
In the provided examples of this category, these two phases are usually water and air, with the air phase being the one that is handled in a simplified manner. 

\subsection{breakingDam2d and breakingDam3d}\label{sec:breakingDam}
The \path{breakingDam2d} example is based on a physical experiment conducted by LaRocque \textit{et al.}~\cite{larocque:13}. 
An enclosed box contains an area of fluid in the lower left corner. 
With the start of the simulation, the fluid spreads throughout the box, with a visible wave forming after the fluid reaches the right-hand side wall.
\begin{figure}[ht]
  \includegraphics[width=1.0\textwidth]{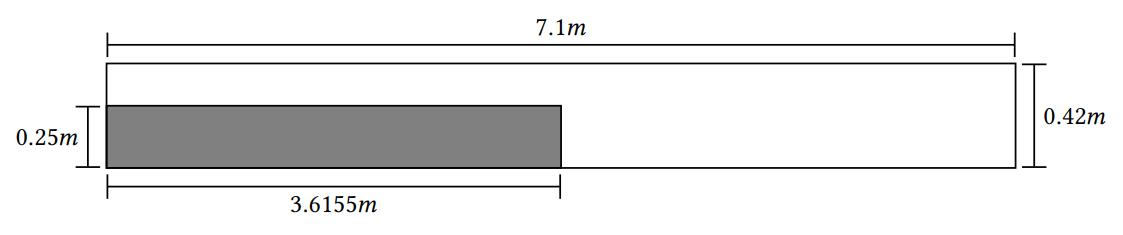}
  \caption{Initial breaking dam setup of \protect\path{breakingDam2d} with the fluid in gray.}
  \label{fig:breaking2d_setup}
\end{figure}

The example \path{breakingDam3d} extends the 2D example to 3D. 
A full simulation run of the \path{breakingDam3d} example, visualized in ParaView, can be found on the OpenLB YouTube page (\href{https://youtu.be/X8yeLCkUldQ}{\nolinkurl{https://youtu.be/X8yeLCkUldQ}}).

\subsection{fallingDrop2d and fallingDrop3d}\label{sec:fallingDrop}
This example type simulates a drop falling into a pool of the same liquid. 
Figure~\ref{fig:fallingDrop2dexample} shows an excerpt of the first couple of time steps. 
On the right half of these steps, the forming of a so-called crown can be seen.
\begin{figure}[ht]
  \includegraphics[width=1.0\textwidth]{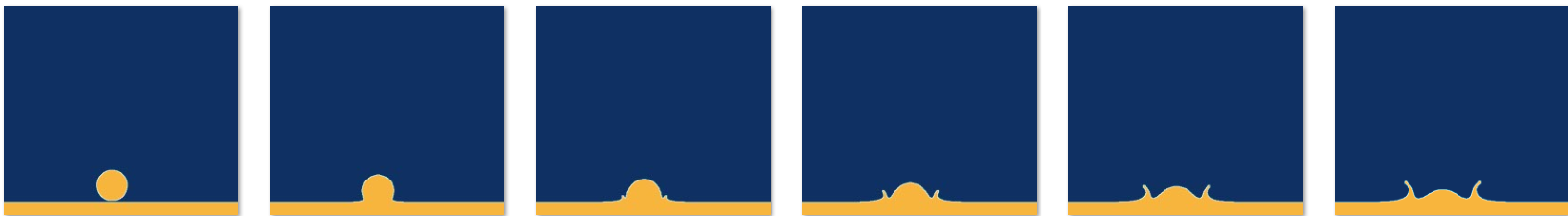}
  \caption{Setup and some exemplary steps of the example \protect\path{fallingDrop2d}.}
  \label{fig:fallingDrop2dexample}
\end{figure}

\subsection{deepFallingDrop2d}\label{sec:deepFallingDrop}
A variation of the \path{fallingDrop2d} example type with a deeper pool of liquid that a drop will fall into. This adapted pool depth allows changes to the droplet properties, such as size, density or velocity, to be more apparent in the simulation.

\subsection{rayleighInstability3d}\label{sec:freeSurfaceRayleigh}
This example covers Plateau-Rayleigh instabilities. 
Figure~\ref{fig:freeSurfaceRayleigh} shows the initial setup as well as multiple steps throughout a simulation of this example. 
Each of these steps after the initial setup was captured after 1800 additional simulation steps. 
The perturbation was made by setting the fill level on the border of the main cylinder according to a sine wave with a wave length set to $\delta = 9R$ with $R$ being the radius of the cylinder. 
This initial sine wave can be seen in the setup step: Red cells are full, blue cells are empty.
\begin{figure}[ht]
  \includegraphics[width=1.0\textwidth]{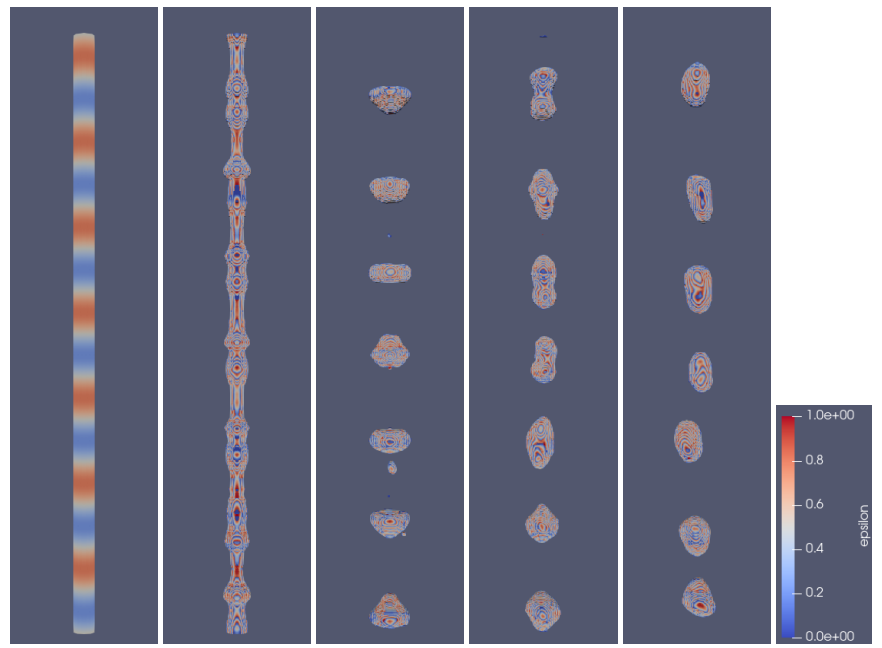}
  \caption{Initial setup and different steps of the Plateau-Rayleigh Instability in \protect\path{rayleighInstability3d}.}
  \label{fig:freeSurfaceRayleigh}
\end{figure}


\chapter{Building and Running} 

OpenLB is developed for high performance computing. 
As such, Linux-based systems are the first class target platform. 
The reason for this is that most HPC clusters and, in fact, all of the 500 fasted supercomputers in the world (cf.\ \href{https://www.top500.org/}{\nolinkurl{https://www.top500.org/}}) run some kind of Linux distributions.

\section{Install Dependencies}\label{sec:dependencies}
GNU Make in addition to a reasonably current C++ compiler supporting C++17 is all that is needed to compile and run non-parallelized OpenLB applications.
OpenLB is able to utilize vectorization (AVX2/AVX-512) on x86 CPUs and NVIDIA GPUs for block-local processing.
CPU targets may additionally utilize OpenMP for shared memory parallelization while any communication between individual processes is performed using MPI.
It has been successfully employed for simulations on computers ranging from low-end smartphones up to supercomputers.

The present release 1.6 has been explicitly tested in the following environments:
\begin{itemize}
\item NixOS 22.11 and unstable (Nix Flake provided)
\item Ubuntu 20.04, 22.04
\item Red Hat Enterprise Linux 8.x (HoreKa, BwUniCluster2)
\item Windows 10, 11 (WSL)
\item MacOS 13
\end{itemize}
as well as compilers:
\begin{itemize}
\item GCC 9 and later
\item Clang 13 and later
\item Intel C++ 2021.4 and later
\item NVIDIA CUDA 11.4 and later
\item NVIDIA HPC SDK 21.3 and later
\item MPI libraries OpenMPI 3.1, 4.1 (CUDA-awareness required for Multi-GPU); Intel MPI 2021.3.0 and later
\end{itemize}
Other CPU targets are also supported, \eg common Smartphone ARM CPUs and Apple M1/M2.

\subsection{Linux}\label{sec:linux}
It is recommended to work on a Linux-based machine. 
Please ensure to have the above-mentioned dependencies installed (Section~\ref{sec:dependencies}). 
Further description is provided below in Section~\ref{sec:windows}. 

\subsection{Mac}\label{sec:macos}
Configuring OpenLB on MacOS is explained for the release 1.5 using MacOS 11.6 in the technical report \href{https://www.openlb.net/wp-content/uploads/2022/06/olb-tr6.pdf}{TR6: Configuring OpenLB on MacOS}~\cite{tr6}. 

\subsection{Windows}\label{sec:windows}
The preferable approach is to use the Windows Subsystem for Linux (WSL) introduced in Windows 10. A
guide can be found in the technical report \href{https://www.openlb.net/tech-reports}{TR5: Installing OpenLB in Windows 10}~\cite{tr5}.

\section{Compiling OpenLB Programs}\label{sec:compiling}

OpenLB consists of generic, template-based code, which needs to be included in the code of application programs, and dependency libraries that are to be linked with the program.
The installation process is light and does not require an explicit precompilation and installation of libraries.
Instead, it is sufficient to unpack the source code into an arbitrary directory.
Compilation of libraries is handled on-demand by the \path{Makefile} of an application/example program.

To get familiar with OpenLB, new users are encouraged to have a look at programs in the \path{examples} directory.
Once inside one of the example directories, entering the command \texttt{make} will first produce libraries and then the end-user example program.
This close relationship between the production of libraries and end-user programs reflects the fact that using OpenLB presently translates to writing a C++ program using the OpenLB library functions.

The file \path{config.mk} in the root directory can be easily edited to modify the compilation process. Available options include the choice of the compiler (\texttt{GNU g++} is the default), optimization flags, a switch between normal/debug mode, between sequential/openmp-parallel/mpi-parallel programs, and between (un)vectorized CPU and GPU platforms.

Example configuration files for common build types and systems are included in the \path{config/} directory of the release tarball.

To compile your own OpenLB programs from an arbitrary directory, make a copy of a sample \path{Makefile} contained in a default example folder. Edit the \texttt{OLB\_ROOT} entry to indicate the location of the OpenLB source, and the \texttt{EXAMPLE} entry to explicit the name of your program, without file extension.

A minimal but perfectly sufficient development environment for OpenLB consists of a supported C++ compiler, a plain text editor of choice and the GNU core utilities including \texttt{make}.
This means that many Linux distributions either already include everything one needs to get started or at least include everything in their default package repositories for convenient installation.
One may of course also use more involved text editors with \eg debugger integration or even full integrated development environments.

For compiling OpenLB applications outside of the provided build system one only needs to make the \path{src} folder available for inclusion and define the compiler and linker flags as they are visible in the printout of the default build.
However, simply calling the appropriate \texttt{make} target from inside the text editor/IDE may be more convenient.

\subsection{Using NVIDIA GPUs in OpenLB}\label{sec:nvidiaGpuOlb}

The following section is a quick guide on how to install the CUDA functionality for Nvidia graphics cards on both Windows or Linux. The first two sections describe how to install CUDA on Windows via WSL or Linux, respectively. The third section discusses how to set up OpenMPI, a CUDA-aware MPI implementation, and the fourth and final section explains how to configure OpenLB to make use of the installed functionalities.

\subsubsection{CUDA on Windows with WSL}

As mentioned in the chapter about install dependencies (Section~\ref{sec:windows}), the preferred approach for OpenLB on Windows is to use the Windows Subsystem for Linux (WSL). The following was written with the assumption that OpenLB has been successfully set up on WSL with Ubuntu.

The following specifications are needed to get CUDA running via WSL:
\begin{itemize}
\item[•] Windows 10 version 21H2 or higher
\item[•] CUDA compatible Nvidia graphics card
\item[•] WSL 2 with a glibc-based distribution (\eg Ubuntu)
\end{itemize}

To find out which Windows version exactly you're using, open up the \texttt{run} dialog box in Windows and type in the command \texttt{winver}, which will display a pop-up window similar to the one below:
\begin{figure}[ht]
  \center
  \includegraphics[scale=0.75]{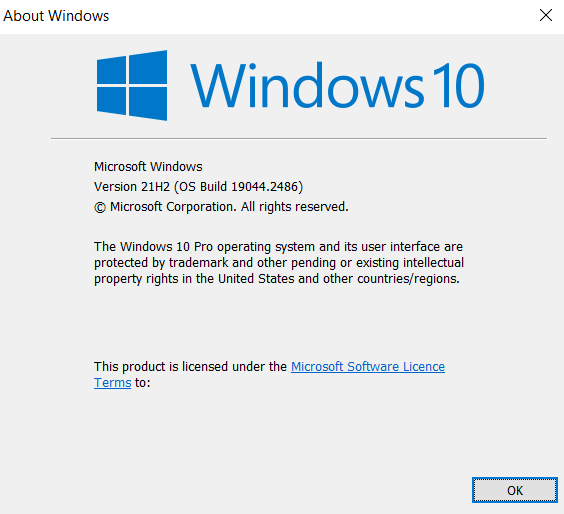}
  \caption{
  Pop-up window displaying the exact version and build of Windows.}
  \label{fig:aboutWindows}
\end{figure}

In order to find out what graphics card you have and whether it is compatible with CUDA, open the up the Windows \texttt{run} dialog and type in the command \texttt{dxdiag}, which will open the \texttt{DirectX Diagnostic Tool}. Under the tab \texttt{Render}, it will display the information regarding your graphics card. In the example picture of the \texttt{DirectX Diagnostic Tool} below, the graphics card is a \texttt{NVIDIA GeForce GTX 1650}.
\begin{figure}[ht]
  \center
  \includegraphics[scale=0.75]{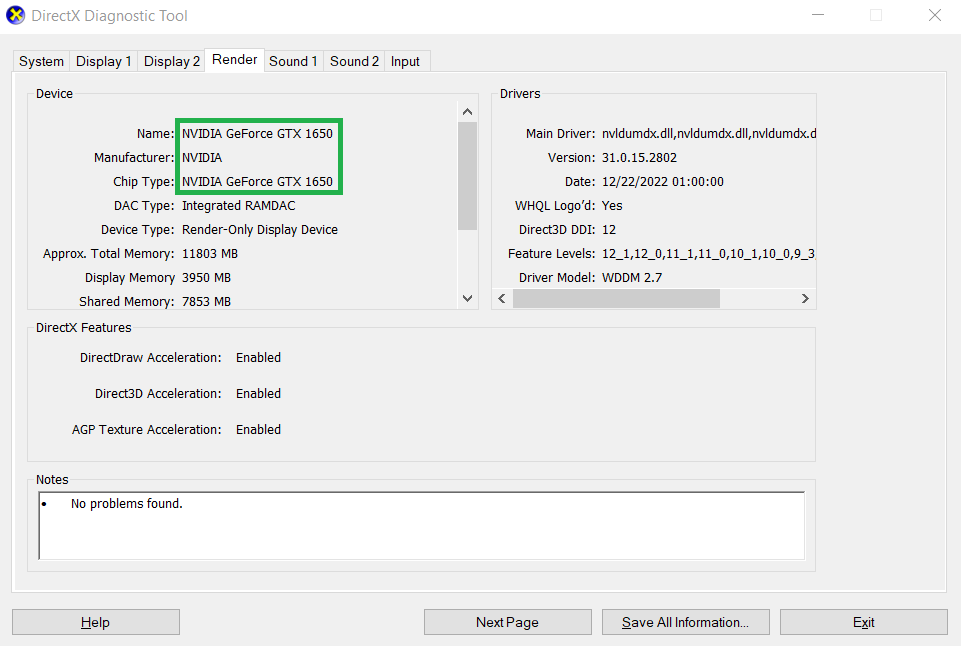}
  \caption{
  The \texttt{Render} tab of the \texttt{DirectX Diagnostic Tool}.}
  \label{fig:dxDiag}
\end{figure}
NVIDIA provides the information on which graphics card is compatible with CUDA on their website (\href{https://developer.nvidia.com/cuda-gpus}{\nolinkurl{https://developer.nvidia.com/cuda-gpus}}).

CUDA is only supported on version 2 of the Windows Subsystem for Linux (WSL). To confirm which version of WSL is installed, open the Windows PowerShell with administrator rights and type in the command
\begin{verbatim}
   wsl --list --verbose
\end{verbatim} 
This will display which Linux distribution and which version of WSL is currently installed. The output should look similar to the following:
\begin{verbatim}
PS C:\Windows> wsl --list --verbose
  NAME      STATE           VERSION
* Ubuntu    Stopped         1
\end{verbatim}
In this example the distribution that is installed is \texttt{Ubuntu} and the WSL version is 1. Upgrading to the necessary version 2 can be done by typing
\begin{verbatim}
   wsl --set-version Ubuntu 2
\end{verbatim}
into the \texttt{PowerShell} terminal. Note that when using a different distribution for WSL, the command has to be adjusted accordingly.

An error might occur claiming that a certain hard-link target does not
exist. This means that there is software installed on WSL that collides
with the update. The error message will provide the path of the non-existing hard-link, which will be a hint onto which package causes this error. In the WSL terminal, the command
\begin{verbatim}
   sudo apt list --installed
\end{verbatim}
will give an overview over all the installed packages. The conflicting package can then be removed with
\begin{verbatim}
   sudo apt-get remove [PACKAGE-NAME]
\end{verbatim}
Once the package has been removed, WSL can be upgraded. On a successful upgrade, we should receive a message that the conversion is complete and we can verify the version with the
\begin{verbatim}
   wsl --list --verbose
\end{verbatim} command. The conflicting package can then be reinstalled.

In order for WSL to have access to the GPU hardware, virtual GPU needs to be enabled on Windows. This can be done by installing an appropriate driver on Windows. It should not be necessary to install any device drivers on WSL itself. It is even highly suggested not to do so, since any installation of a driver on WSL itself might override the functionality provided by the driver that is installed onto Windows. As of the writing of this guide, the most recent NVIDIA drivers automatically support virtual GPU for WSL.
The newest driver can be directly downloaded from the NVIDIA website (\href{https://www.nvidia.com/download/index.aspx}{\nolinkurl{https://www.nvidia.com/download/index.aspx}}). The website offers drop down lists to specify what product type, device, operating system, etc. the driver is needed for. Once the most recent driver is installed, we can install the CUDA toolkit on WSL.

The following commands typed into the WSL terminal will install the Nvidia CUDA toolkit on WSL (Ubuntu):
\begin{verbatim}
sudo apt-key del 7fa2af80
		
wget https://developer.download.nvidia.com/compute/cuda/repos/
     wsl-ubuntu/x86_64/cuda-wsl-ubuntu.pin
		
sudo mv cuda-wsl-ubuntu.pin /etc/apt/preferences.d/
                            cuda-repository-pin-600
		
sudo apt-key adv --fetch-keys https://developer.download.nvidia.com/
                              compute/cuda/repos/wsl-ubuntu/
                              x86_64/3bf863cc.pub
		
sudo add-apt-repository 'deb https://developer.download.nvidia.com/
                             compute/cuda/repos/wsl-ubuntu/x86_64/ /'
		
sudo apt-get update
		
sudo apt-get -y install cuda
\end{verbatim} 
If the NVIDIA CUDA compiler is correctly installed, the command
\begin{verbatim}
   nvcc --version
\end{verbatim}
will reply with a message similar to the following:

\begin{lstlisting}[language=myc++,mathescape=true, caption=Version details of an installed Cuda compiler, label=lst:nvccVersion]
   nvcc: NVIDIA (R) Cuda compiler driver
   Copyright (c) 2005-2022 NVIDIA Corporation
   Built on Mon_Oct_24_19:12:58_PDT_2022
   Cuda compilation tools, release 12.0, V12.0.76
   Build cuda_12.0.r12.0/compiler.31968024_0
\end{lstlisting}

To check the versions of CUDA and the driver, the command
\begin{verbatim}
   nvidia-smi
\end{verbatim}
will respond with the NVIDIA System Management Interface, displaying various information about the installed GPUs (see Figure~\ref{fig:nvidiaSmi}). 
\begin{figure}[ht]
  \center
  \includegraphics[scale=0.75]{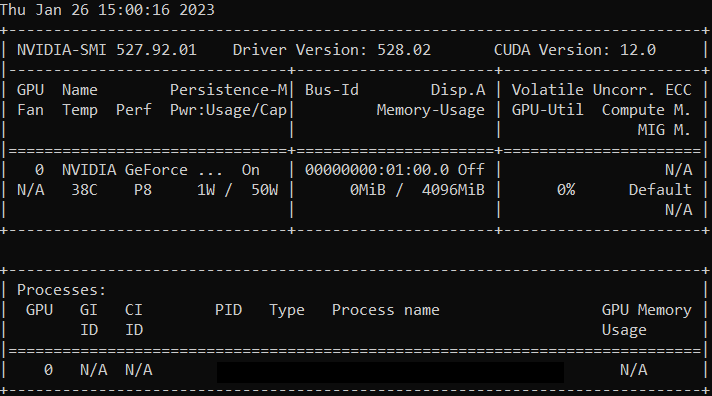}
  \caption{
  The NVIDIA System Management Interface}
  \label{fig:nvidiaSmi}
\end{figure}
The CUDA toolkit should now be properly installed and working. 

\subsubsection{CUDA on Linux}
Before installing the CUDA toolkit on Linux, typing the command
\begin{verbatim}
   lspci | grep -i nvidia
\end{verbatim}
can confirm that the GPU is CUDA-capable.

To install the CUDA toolkit on Linux, visit the the NVIDIA website and choose the fitting operating system, architecture, distribution, as well as the preferred installation type for your system (\href{https://developer.nvidia.com/cuda-toolkit}{\nolinkurl{https://developer.nvidia.com/cuda-toolkit}}). The website will then provide you with the correct commands with which you can install the CUDA toolkit on your Linux system.

After the installation of the toolkit, the environment variables need to be set:
\begin{verbatim}
   export PATH=/usr/local/cuda-12.0/bin${PATH:+:${PATH}}	
\end{verbatim}
If the installation was done with a run file, the \texttt{LD\_LIBRARY\_PATH} variable has to be set, as well. The following command sets this variable on a 64-bit system. The command for 32-bit systems is almost identical: \texttt{lib64} has to be exchanged for \texttt{lib}:
\begin{verbatim}
   export LD_LIBRARY_PATH=/usr/local/cuda-12.0/lib64\
                          ${LD_LIBRARY_PATH:+:${LD_LIBRARY_PATH}}
\end{verbatim}
If a different install path or version of the CUDA toolkit has been chosen during the installation process, both commands above have to be altered accordingly. To confirm that the installation has been successful, use the commands
\begin{verbatim}
   nvcc --version
\end{verbatim}
and
\begin{verbatim}
   nvidia-smi
\end{verbatim} 
If the CUDA toolkit has been installed correctly, an output similar to those shown in Listing~\ref{lst:nvccVersion} and Figure~\ref{fig:nvidiaSmi} respectively.

\subsubsection{OpenMPI}
To have the functionality of MPI in combination with CUDA, there are several CUDA-aware MPI implementations available. 
In this guide we will describe the installation of the open-source implementation OpenMPI in four steps: 
\begin{enumerate}
    \item Download the desired OpenMPI version from the website (\href{https://www.open-mpi.org/software/}{\nolinkurl{https://www.open-mpi.org/software/}}). 
    As of the writing of this section, the most current stable release version was \path{openmpi-4.1.5.tar.bz2}. 
    \item In your Linux (or WSL for Windows) terminal, move to the folder where the file was saved to and extract the downloaded package via the command
    \begin{verbatim}
        tar -jxf openmpi-4.1.5.tar.bz2
    \end{verbatim}
    \item Change into this new directory to configure, compile and install OpenMPI with the following commands:
    \begin{verbatim}
    ./configure --prefix=$HOME/opt/openmpi 
                --with-cuda=/usr/local/cuda-12.0/include
    make all
    make install
    \end{verbatim} 
    Note that the path following \texttt{--prefix=} is the path we wish to install OpenMPI in and the path following \texttt{--with-cuda=} is the location of the \path{include} folder of your CUDA installation. 
    These paths might be different depending on the users choices. 
    \item Change the environment variables with the following two commands in the Linux or WSL terminal:
    \begin{verbatim}
    echo "export PATH=\$PATH:\$HOME/opt/openmpi/bin" >> $HOME/.bashrc
    echo "export LD_LIBRARY_PATH=\$LD_LIBRARY_PATH:
                 \$HOME/opt/openmpi/lib" \ >> $HOME/.bashrc
    \end{verbatim}
    Once again the path for OpenMPI might be different, depending on where the software was installed. 
\end{enumerate}
To see whether the installation of OpenMPI was successful, we can enter the command
\begin{verbatim}
ompi_info --parsable -l 9 
          --all | grep mpi_built_with_cuda_support:value
\end{verbatim}
If the installation was done successfully, the terminal should respond with the output \texttt{true}.

\subsubsection{Utilizing CUDA in OpenLB}

The root directory contains a folder named \path{config}, in which several build config examples can be found. The \path{config.mk} makefile of the root directory can be replaced with the makefile that suits the current needs (\eg using only the GPU, using the GPU with MPI, using CPU with MPI, etc.). Each example makefile also includes instructions. Make a backup of the current \path{config.mk} in the root directory and replace it with a copy of the makefile \path{gpu\_only} found in the \path{config} folder. After renaming \path{gpu\_only} to \path{config.mk}, we open the file and check the value of \texttt{CUDA\_ARCH}: This value might have to be changed, depending on your graphics card and its architecture. The file \path{rules.mk} in the root directory contains a table that shows which architecture goes with which value:
\begin{verbatim}
   ## | CUDA Architecture | Version    |
   ## |-------------------+------------|
   ## | Fermi             | 20         |
   ## | Kepler            | 30, 35, 37 |
   ## | Maxwell           | 50, 52, 53 |
   ## | Pascal            | 60, 61, 62 |
   ## | Volta             | 70, 72     |
   ## | Turing            | 75         |
   ## | Ampere            | 80, 86, 87 |
\end{verbatim}
Another table on the internet (\href{https://en.wikipedia.org/wiki/CUDA}{\nolinkurl{https://en.wikipedia.org/wiki/CUDA}}) shows which graphics card corresponds to which architecture. This guide used the \texttt{GTX 1650} as an example for the graphics card. The following picture shows that the \texttt{GTX 1650} corresponds to the \texttt{Turing} architecture, so the value of \texttt{CUDA\_ARCH} has to be set to 75 in both \path{config.mk} and \path{rules.mk} files. 
\begin{figure}[ht]
  \center
  \includegraphics[scale=0.75]{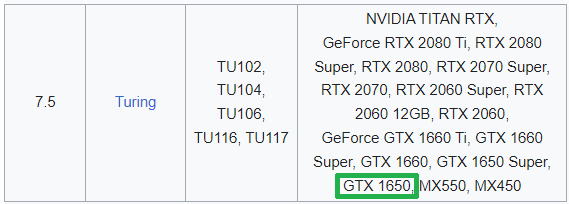}
  \caption{
  Table containing Nvidia GPUs with the \texttt{Turing} Microarchitecture.}
  \label{fig:turingGpu}
\end{figure}
After saving the changes of \texttt{CUDA\_ARCH} in both \path{config.mk} and \path{rules.mk}, the \path{config.mk} can be compiled via the command \texttt{make clean; make} in your WSL (or Linux) terminal. It is now possible to compile and execute one of the GPU-enabled OpenLB examples with CUDA support.


\chapter{Step by Step: Using OpenLB for Applications}

The general way of functioning in OpenLB follows a generic path.
The following structure is maintained throughout every OpenLB application example, to provide a common structure and guide beginners.
\begin{description}
\item[1st Step: Initialization]
The converter between physical and lattice units is set in this step.
It is also defined, where the simulation data is stored and which lattice type is used.
\item[2nd Step: Prepare geometry]
The geometry is acquired, either from another file (a \path{.stl} file) or from defining indicator functions.
Then, the mesh is created and initialized based on the given geometry.
This consists of classifying voxels with material numbers, according to the kind of voxels they are: an inner voxel containing fluid ruled by the fluid dynamics will have a different number than a voxel on the inflow with conditions on its velocity.
The function \class{prepareGeometry} is called for these tasks.
Further, the mesh is distributed over the threads to establish good scaling properties.
\item[3rd Step: Prepare lattice]
According to the material numbers of the geometry, the lattice dynamics are set here.
This step characterizes the collision model and boundary behavior.
The choices depend on whether a force is acting or not, the use of single relaxation time (BGK) or multiple relaxation times (MRT), the simulation dimension (it can also be a 2D model), whether compressible or incompressible fluid is considered, and the number of neighboring voxels chosen.
By the creation of a computing grid, the \class{SuperLattice}, the allocation of the required data is done as well.
\item[4th Step: Main loop with timer]
The timer is initialized and started, then a loop over all time steps \class{iT} starts the simulation, during which the functions \class{setBoundaryValues}, \class{collideAndStream} and \class{getResults} (the 5th, 6th, and 7th step, respectively) are called repeatedly until a maximum of iterations is reached, or the simulation has converged.
At the end, the timer is stopped and the summary is printed to the console.
\item[5th Step: Definition of initial and boundary conditions]
The first of the three important functions called during the loop, \class{setBoundaryValues}, sets the slowly increasing inflow boundary condition.
Since the boundary is time dependent, this happens in the main loop.
In some applications, the boundaries stay the same during the whole simulation and the function doesn't need to do anything after the very first iteration.
\item[6th Step: Collide and stream execution]
Another function \class{collideAndStream} is called each iteration step, to perform the collision and the streaming step.
If more than one lattice is used, the function is called for each lattice separately.
\item[7th Step: Computation and output of results]
At the end of each iteration step, the function \class{getResults} is called, which creates console output, \path{.ppm} files or \path{.vti} files of the results at certain time steps.
The ideal is to get the relevant simulation data with functors and thus facilitate the post processing significantly.
By passing the converter and the time step, the frequency of writing or displaying data can be chosen easily.
In many applications, the console output is required more often than the VTK data.
\end{description}

\section{Lesson~\arabic{section}: Getting Started - Sketch of Application}
\label{sec:lessonGettingStarted}

This section presents example \path{bstep2d} that can be found in the recent release of OpenLB.
This example simulates a flow over a backward-facing step and serves as an illustration of OpenLB and its features.
In order to execute the simulation and get some results, download and unpack OpenLB on (preferably) a Linux system, see Section~\ref{sec:linux}.
Then, generate a executable file by compiling the program through the command \texttt{make}.
Finally, launch the simulation by \texttt{./bstep2d} and observe the terminal output, see Section~\ref{sec:consoleOutput}.

A few lines are invariably the same for all OpenLB applications, see Listing~\ref{lst:code1_1}.
\begin{lstlisting}[language=myc++,caption={Framework of an OpenLB program. Fundamental properties of the simulation are defined here.}, label=lst:code1_1]
#include "olb2D.h"
#include "olb2D.hh"

using namespace olb;              // OpenLB namespaces
using namespace olb::descriptors; //

using T = FLOATING_POINT_TYPE;
using DESCRIPTOR = D2Q9<>;
\end{lstlisting}

\begin{description}
\item{Line~1:} The header file \path{olb2D.h} includes definitions for the whole 2D code present in the release.
In the same way, access to 3D code is obtained by including the file \path{olb3D.h}.
\item{Line~2:} Most OpenLB code depends on template parameters.
Therefore, it cannot be compiled in advance, and needs to be integrated ``as is'' into your programs via the file \path{olb2D.hh} or \path{olb3D.hh} respectively.
\item{Line~4:} All OpenLB code is contained in the namespace \texttt{olb}.
The descriptors have an own namespace and define the lattice arrangement, \eg $D2Q9$ or $D3Q19$.
\item{Line~7:} Choice of precision for floating point arithmetic.
The default type \texttt{FLOATING\_POINT\_TYPE} is defined in the \path{config.mk} file and usually equals \class{float} or \class{double}.
Any other floating point type can be used, including built-in types and user-defined types which are implemented through a C++ class.
\item{Line~8:} Choice of a lattice descriptor.
Lattice descriptors specify not only which lattice (see Figure~\ref{fig:discreteVelocitySets} for exemplary velocity sets) are employed, but is also are used to compute the size of various dependent fields such as force vectors.
\end{description}

The next code presents a brief overview about the structure of an OpenLB application, see Listing~\ref{lst:bstep2d}.
It aims rather to introduce and guidelines the beginners, than explain the classes and methods in depth.
Details on the shown functions can be found in the source code, this means in the \path{bstep2d.cpp} file, as well as in the following chapters.

\begin{lstlisting}[language=myc++,caption={A brief overview of a typically OpenLB application, \texttt{bstep2d}. Details on the specific functions can be found in the following chapters.}, label=lst:bstep2d]
SuperGeometry<T,2> prepareGeometry(LBconverter<T> const& converter)
{
  // create Cuboids and assign them to threads
  // create SuperGeometry
  // set material numbers
  return superGeometry;
}
void prepareLattice(...)
{
  // set dynamics for fluid and boundary lattices
  // set initial values, rho and u
}
void setBoundaryValues(...)
{
  // set Poiseuille velocity profile at inflow
  // increase inflow velocity slowly over time
}
void getResults(...)
{
  // write simulation data do vtk files and terminal
}

int main(int argc, char* argv[])
{
  // === 1st Step: Initialization ===
  olbInit( &argc, &argv );
  singleton::directories().setOutputDir( "./tmp/" );  // set output directory
  OstreamManager clout( std::cout, "main" );

  UnitConverterFromResolutionAndRelaxationTime<T, DESCRIPTOR> converter(
    (T)   N,                 // resolution
    (T)   relaxationTime,    // relaxation time
    (T)   charL,             // charPhysLength: reference length of simulation geometry
    (T)   1.,                // charPhysVelocity: maximal/highest expected velocity during simulation in __m / s__
    (T)   1./19230.76923,    // physViscosity: physical kinematic viscosity in __m^2 / s__
    (T)   1.                 // physDensity: physical density in __kg / m^3__
  );

  // Prints the converter log as console output
  converter.print();
  // Writes the converter log in a file
  converter.write("bstep2d");

  // === 2nd Step: Prepare Geometry ===
  // Instantiation of a superGeometry
  SuperGeometry<T,2> superGeometry( prepareGeometry(converter) );

  // === 3rd Step: Prepare Lattice ===
  SuperLattice<T,DESCRIPTOR> sLattice( superGeometry );
  BGKdynamics<T,DESCRIPTOR> bulkDynamics (
    converter.getLatticeRelaxationFrequency(),
    instances::getBulkMomenta<T,DESCRIPTOR>()
  );

  //prepare Lattice and set boundaryConditions
  prepareLattice( converter, sLattice, bulkDynamics, superGeometry );

  // instantiate reusable functors
  SuperPlaneIntegralFluxVelocity2D<T> velocityFlux( sLattice,
      converter,
      superGeometry,
      {lengthStep/2.,  heightInlet / 2.},
      {0.,  1.} );

  SuperPlaneIntegralFluxPressure2D<T> pressureFlux( sLattice,
      converter,
      superGeometry,
      {lengthStep/2.,  heightInlet / 2. },
      {0.,  1.} );

  // === 4th Step: Main Loop with Timer ===
  clout << "starting simulation..." << std::endl;
  Timer<T> timer( converter.getLatticeTime( maxPhysT ), superGeometry.getStatistics().getNvoxel() );
  timer.start();

  for ( std::size_t iT = 0; iT < converter.getLatticeTime( maxPhysT ); ++iT ) {
    // === 5th Step: Definition of Initial and Boundary Conditions ===
    setBoundaryValues( converter, sLattice, iT, superGeometry );
    // === 6th Step: Collide and Stream Execution ===
    sLattice.collideAndStream();
    // === 7th Step: Computation and Output of the Results ===
    getResults( sLattice, converter, iT, superGeometry, timer, velocityFlux, pressureFlux );
  }

  timer.stop();
  timer.printSummary();
}
\end{lstlisting}

\section{Lesson~\arabic{section}: Define and Use Boundary Conditions}
The current OpenLB release offers a wide range of boundary conditions for the implementation of pressure and velocity boundaries. They support boundaries that are aligned with the numerical grid, and also implement proper corner nodes in 2D and 3D, and edge nodes that connect two plane boundaries in 3D. The choice of a boundary condition is conceptually separated from the definition of the location of boundary nodes. It is therefore possible to modify the choice of the boundary condition by changing a single instruction in a program. An overview of the available boundary conditions is given by~\cite{krause:21}.

The new boundary condition system utilizes free floating functions and doesn't require a class structure. Consequently the following classes are obsolete in the current release: \\ 
\texttt{sOn/\-Off\-Lattice\-Boundary\-ConditionXD},\\ 
\texttt{On\-Lattice\-Boundary\-ConditionXD},\\ 
\texttt{(Off)\-Boundary\-Condition\-InstantiatorXD} and \\ 
\texttt{Regularized/\-Inter\-polation\-Boundary\-ManagerXD}. \\ 
Key Features of the new system are:
\begin{itemize}
\item Free floating design that allows for general functions like "setBoundary" to be used in multiple boundary Conditions.
\item Overall slimmer design with fewer function calls and fewer loops through the block domain.
\item MomentaVector and dynamicsVector is stored in \texttt{/src/\-core/\-block\-Lattice\-Structure\-3D.h}
\item Uncluttered function call design, which makes it easier to create new boundary conditions
\end{itemize}

The new boundaries are similarly named as in \class{addSlipBoundary} becomes \class{setSlipBoundary} in the new system.
For example, if you want to use a slip boundary condition:
\begin{description}
	\item[Define dynamics]  Keep in mind that \class{sLattice.defineDynamics} is the same for the old and new boundary system.
\item[Define your boundary type] In this case it is the \class{slipBoundary}. To set the boundary, call the fitting \class{setBoundaryCondition} function in the following manner inside your \class{prepareLattice} function:
\begin{itemize}
\item \texttt{setSlipBoundary<T,DESCRIPTOR>("superLattice",\\"superGeometry","MaterialNumber");}
\item The difference between the old and the new system is that every boundaryCondition needs to have the \class{SuperLattice} as an argument. The \class{latticeRelaxationFrequency} \class{omega} is usually called as the second argument.
\end{itemize}
\item[Define initial conditions]\
\begin{itemize}
\item on-lattice: \texttt{sLattice.defineRhoU(..)}
\item off-lattice: \texttt{sLattice.defineRho(..)},\ \texttt{sLattice.defineUBouzidi(..)}
\end{itemize}
\item[Define boundary values]\
\begin{itemize}
\item on-lattice: \texttt{sLattice.defineU(..)}
\item off-lattice: \texttt{sLattice.defineUBouzidi(..)}
\end{itemize}
\end{description}

With the help of this system, one can treat local and non-local boundary conditions the same way. Furthermore, they can be used both for sequential and parallel program execution, as it is shown in Lesson~10. The mechanism behind this is explained in Lesson~7. The bottom line is that both local and non-local boundary conditions instantiate a special dynamics object and assign it to boundary cells. Non-local boundaries additionally instantiate post-processing objects which take care of non-local aspects of the algorithm.

\section{Lesson~\arabic{section}: UnitConverter - Lattice and Physical Units}
\label{sec:LessonUnitConverter}

Fluid flow problems are usually given in a system of metric units. For example consider a cylinder of diameter \SI{3}{\centi\meter} in a fluid channel with average inflow velocity of \SI{4}{\m}. The fluid has a kinematic viscosity of \SI{0,001}{\square\meter\per\second}. The value of interest is the pressure difference measured in $Pa$ at the front and the back of the cylinder (with respect to the flow direction). However, the variables used in a LB simulation live in a system of lattice units, in which the distance between two lattice cells and the time interval between two iteration steps are chosen to be unity. Therefore, when setting up a simulation, a conversion directive has to be defined that takes care of scaling variables from physical units into lattice units and \textit{vice versa}.

In OpenLB, all these conversions are handled by a class called \class{UnitConverter}, see Listing~\ref{lst:converter}. An instance of the UnitConverter is generated with desired discretization parameters and reference values in SI units. It provides a set of conversion functions to enable a fast and easy way to convert between physical and lattice units. In addition, it gives information about the parameters of the fluid flow simulation, such as the Reynolds number or the relaxation parameter $\omega$.

Let's have a closer look at the input parameters: The reference values represent characteristic quantities of the fluid flow problem. In this example, it is suitable to choose the cylinder's diameter as characteristic length and the average inflow speed as characteristic velocity. Furthermore, two discretization parameters, namely the grid size $\triangle x$ in \si{\meter} and time step size $\triangle t$ in \si{\second} are provided to the converter. From these reference values and discretization parameters, all the conversion factors and the relaxation time $\tau$ are calculated.

Due to the fact, that there are stability bounds for the relaxation time and the maximum occurring lattice velocity, one does not usually chose $\triangle x$ and $\triangle t$, but sets stable and accurate values for any two out of resolution, relaxation time or characteristic (maximum) lattice velocity. To make that easily available for the user of OpenLB, there are different constructors for the \class{UnitConverter} class:\\
\class{UnitConverterFromRelaxationTimeAndLatticeVelocity},\\ \class{UnitConverterFromResolutionAndLatticeVelocity},\\ \class{UnitConverterFromResolutionAndRelaxationTime}.

Once the converter is initialized, its methods can be used to convert various quantities such as velocity, time, force or pressure.
The function for the latter helps us to evaluate the pressure drop in our example problem, as shown in the the following code snippet:
\begin{lstlisting}[language=myc++,caption={Use of UnitConverter in a 3D problem.},label={lst:converter}]
UnitConverterFromResolutionAndRelaxationTime<T, DESCRIPTOR> const converter(
    int {N},              // resolution: number of voxels per charPhysL
    (T)   0.53,           // latticeRelaxationTime: relaxation time, have to be greater than 0.5!
    (T)   0.1,            // charPhysLength: reference length of simulation geometry
    (T)   0.2,            // charPhysVelocity: maximal/highest expected velocity during simulation in __m / s__
    (T)   0.2*2.*0.05/Re, // physViscosity: physical kinematic viscosity in __m^2 / s__
    (T)   1.0             // physDensity: physical density in __kg / m^3__
  );
// Prints the converter log as console output
converter.print();
// Writes the converter log in a file
converter.write("converterLogFile");
// conversion from seconds to iteration steps and vice-versa
int iT = converter.getLatticeTime(maxPhysT);
T sec  = converter.getPhysTime(iT);
<...> simulation
<...> evaluation of latticeRho at the back and the front of the cylinder
T latticePressureFront = latticeRhoFront / descriptors::invCs2<T,DESCRIPTOR>();
T latticePressureBack = latticeRhoBack / descriptors::invCs2<T,DESCRIPTOR>();
T pressureDrop =   converter.getPhysPressure(latticePressureFront)
                 - converter.getPhysPressure(latticePressureBack);
\end{lstlisting}

\begin{description}
\item{Line~1--7:} Instantiate an \class{UnitConverter} object and specify discretization parameters as well as characteristic values.
\item{Line~9:} Write simulation parameters and conversion factors to terminal.
\item{Line~11:} Write simulation parameters and conversion factors in a logfile.
\item{Line~13 and 14:} The conversion from physical units (seconds) to discrete ones (time steps) is managed by the converter.
\item{Line~17--20:} The converter automatically calculates the pressure values from the local density.
\end{description}

\section{Lesson~\arabic{section}: Extract Data From a Simulation}
When the collision step is executed, the value of the density and the velocity are computed internally, in order to evaluate the equilibrium distribution.
Those macroscopic variables are however interesting for the OpenLB end-user as well, and it would be a shame to simply neglect their value after use.
These values are accessed through the method \texttt{getStatistics()} of a \class{BlockLattice}:

\begin{description}
\item[T lattice.getStatistics().getAverageRho()]
Returns average density evaluated during the previous collision step.
\item[T lattice.getStatistics().getAverageEnergy()]
Returns half the average velocity norm evaluated during the previous collision step.
\item[T lattice.getStatistics().getMaxU()]
Returns maximum value of the velocity norm evaluated during the previous collision step.
\end{description}

Often, the information provided by the statistics of a lattice in not sufficient, and more generally numerical result are required.
To do this, you can get data cell-by-cell from the \class{BlockLatticeXD} and \class{SuperLatticeXD} through \texttt{functors}, see Chapter~\ref{sec:functors}.
Functors act on the underlying lattice and process its data to relevant macroscopic units, \eg density, velocity, stress, flux, pressure and drag.
Functors provide an \texttt{operator()} that instead of access stored data, computes every time it is called the data.
Since OpenLB version 0.8, the concept of functors unfold not only for postprocessing, but also for boundary conditions and the generation of geometry, see Chapter~\ref{sec:functors}.
In Listing~\ref{lst:VTMwriter} it is shown, how to extract data out of a \class{SuperLattice} named \texttt{sLattice} and an \class{SuperGeometry3D} named \texttt{sGeometry}.
The data format is a legal \texttt{.vtk} file, that can be processed further with ParaView.
\begin{lstlisting}[language=myc++, caption={Extract simulation data to \texttt{vtk} file format.}, label=lst:VTMwriter]
  // generate the writer object
  SuperVTMwriter3D<T> vtmWriter("bstep3d");
  // write every 0.2 seconds
  if (iT==converter.getLatticeTime(0.2)) {
    // create functors
    SuperLatticeGeometry3D<T,DESCRIPTOR> geometry(sLattice, sGeometry);
    SuperLatticeCuboid3D<T,DESCRIPTOR> cuboid(sLattice);
    SuperLatticeRank3D<T,DESCRIPTOR> rank(sLattice);
    // write functors to file system, vtk formata
    vtmWriter.write(geometry);
    vtmWriter.write(cuboid);
    vtmWriter.write(rank);
  }
\end{lstlisting}

As before mentioned, OpenLB provides functors for a bunch of data, see Listing~\ref{lst:functorData}.
More details about writing simulation data can be found in Chapter~\ref{cha:io}.
\begin{lstlisting}[language=myc++,caption={Code example for creating velocity, pressure and geometry functors.}, label=lst:functorData]
// Create the functors by only passing lattice and converter
SuperLatticePhysVelocity3D<T,DESCRIPTOR> velocity(sLattice, converter);
SuperLatticePhysPressure3D<T,DESCRIPTOR> pressure(sLattice, converter);
// Create functor that corresponds to material numbers
SuperLatticeGeometry3D<T,DESCRIPTOR> geometry(sLattice, superGeometry);
\end{lstlisting}

The most straightforward and convenient way of visualizing simulation data is to produce a 2D snapshot of a scalar valued functor.
This is done through the \class{BlockReduction3D2D}, which puts a plane into arbitrary 3D functors.
Afterwards, this plane can be easily written to a image file.
OpenLB creates images of PPM format as shown in Listing~\ref{lst:GIFwriter}.

\begin{lstlisting}[language=myc++,caption=Create a PPM image out of a 3D velocity functor., label=lst:GIFwriter]
// velocity is an application: R^3 -> R^3
// an image in its very basic sense is an application: R^2 -> R

// transformation of data is presented below
// get velocity functor
SuperLatticePhysVelocity3D<T,DESCRIPTOR> velocity(sLattice, converter);
// get scalar valued functor by applying the point wise l2 norm
SuperEuklidNorm3D<T> normVel( velocity );
// put a plane with normal (0,0,1) in the 3 dimensional data
BlockReduction3D2D<T> planeReduction( normVel, {0, 0, 1} );
BlockGifWriter<T> gifWriter;
// write ppm image to file system
gifWriter.write( planeReduction, iT, "vel" );
\end{lstlisting}
This image writer provides \textit{in situ} visualization which, in contrast to the \class{VTKwriter}, produces smaller data sets that can be interpreted immediately without requiring other software.

\section{Lesson~\arabic{section}: Convergence Check}
\label{sec:lessonConvergenceCheck}

The class \texttt{ValueTracer} checks for time-convergence of a given scalar $\phi$. The convergence is reached when the standard deviation $\sigma$ of the monitored value  $\phi$ is smaller than a given residuum $\epsilon$ times the expected value $\bar\phi$.
\begin{equation}
\sigma(\phi)=\sqrt{\frac{1}{N+1}\sum \limits_{i=0}^{N}\left(\phi_i-\bar\phi\right)^2}< \epsilon\bar\phi
\end{equation}
The expected value  $\phi$ is the average over the last $N$ time steps with $\phi_i := \phi(t^* - i\triangle t)$ where \(t^*\) is the current time step and $\triangle t$ denotes the time step size.
\begin{equation}
\bar\phi=\frac{1}{N+1}\sum \limits_{i=0}^{N}\phi_i
\end{equation}
The value $N$ should be chosen as a problem specific time period. 
As an example $\texttt{charT} = \texttt{charL} / \texttt{charU}$ and $N = \texttt{converter.getLatticeTime}(\texttt{charT})$.
To initialize a \texttt{ValueTracer} object use:
\begin{lstlisting}[language=myc++,caption=Create a PPM image out of a 3D velocity functor.]
util::ValueTracer<T> converge( numberTimeSteps, residuum );
\end{lstlisting}
For example, to check for convergence with a residuum of $\epsilon = 10^{-5}$ every physical second:
\begin{lstlisting}[language=myc++,caption=Create a PPM image out of a 3D velocity functor.]
util::ValueTracer<T> converge( converter.getLatticeTime(1.0), 1e-5 );
\end{lstlisting}
It is required to pass the monitored value to the  \texttt{ValueTracer} object every time steps by:
\begin{lstlisting}[language=myc++,caption=Create a PPM image out of a 3D velocity functor.]
for (iT = 0; iT < maxIter; ++iT) {
  ...
  converge.takeValue( monitoredValue, isVerbose );
  ..
}
\end{lstlisting}
If you like to print average value and its standard derivation every number of time steps chosen during initialization set \texttt{isVerbose} to true otherwise choose false.
It is good idea to choose average energy as monitored value:
\begin{lstlisting}[language=myc++,caption=Create a PPM image out of a 3D velocity functor.]
converge.takeValue(SLattice.getStatistics().getAverageEnergy(),true);
\end{lstlisting}
Do something like the following in the time loop:
\begin{lstlisting}[language=myc++,caption=Create a PPM image out of a 3D velocity functor.]
if (converge.hasConverged()) {
  clout << "Simulation converged." << std::endl;
  break;
}
\end{lstlisting}

\section{Lesson~\arabic{section}: Use an External Force}
\label{sec:externalForces}

In simulations the dynamics of a fluid are often driven by some kind of externally imposed force field. 
In order to optimize memory access and to minimize cache-misses, the value of this force can be stored right alongside the cell's population values. 
This is achieved by specifying additional fields in the lattice descriptor (see Sections~\ref{ssec:Descriptor} and~\ref{sec:couplings}).

At this point we want to consider a time- and space-independent external force as a basic example. 
Listing~\ref{lst:externalForceExample} shows how such an external force can be defined for all cells of a certain material number.
\begin{lstlisting}[language=myc++,caption=Define a constant external force,label=lst:externalForceExample]
// Define constant force
AnalyticalConst2D<T,T> force(
  8.0 * converter.getLatticeViscosity()
      * converter.getCharLatticeVelocity()
      / ( Ly*Ly ), // x-component of the force
  0.0);            // y-component of the force

// Initialize force for materials 1 and 2
superLattice.defineField<FORCE>(
  superGeometry.getMaterialIndicator({1, 2}), force);
\end{lstlisting}
This code was adapted from \path{examples/laminar/poiseuille2d} where just such a constant force is used to drive the channel flow. 
Note that the underlying \class{D2Q9} descriptor's field list must contain a \class{FORCE} field in order for \texttt{defineField<FORCE>} to work.

The command \class{SuperLattice::defineField} provides just one of many template methods one can use to work with descriptor fields. 
For example, we can read and write a cell's force field as follows:
\begin{lstlisting}[language=myc++]
// Get a reference to the memory location of a cell's force vector
FieldPtr<T,DESCRIPTOR,FORCE> force = cell.template getFieldPointer<FORCE>();
// Read a cell's force vector as an OpenLB vector value
Vector<double,2> force = cell.template getField<FORCE>();
// Set a cell's force vector to zero
cell.template setField<FORCE>(Vector<double,2>(0.0, 0.0));
\end{lstlisting}
Note that these methods work the same way for any other field that might be declared by a cell's specific descriptor.

\section{Lesson~\arabic{section}: Understand Genericity in OpenLB}

OpenLB is a framework for the implementation of lattice Boltzmann algorithms. Although most of the code shipped with the distribution is about fluid dynamics, it is open to various types of physical models. Generally speaking, a model which makes use of OpenLB must be formulated in terms of the ``local collision followed by nearest-neighbor streaming'' philosophy. A current restriction to OpenLB is that the streaming step can only include nearest neighbors: there is no possibility to include larger neighborhoods within the modular framework of the library, \ie without tampering with OpenLB source code. Except for this restriction, one is completely free to define the topology of the neighborhood of cells, to implement an arbitrary local collision step, and to add non-local corrections for the implementation of, say, a boundary condition.

To reach this level of genericity, OpenLB distinguishes between non-modifiable core components, which you'll always use as they are, and modular extensions. As far as these extensions are concerned, you have the choice to use default implementations that are part of OpenLB or to write your own. As a scientific developer, concentrating on these, usually quite short, extensions means that you can concentrate on the physics of your model instead of technical implementation details. By respecting this concept of modularity, you can automatically take advantage of all structural additions to OpenLB. In the current release, the most important addition is parallelism: you can run your code in parallel without (or almost without) having to care about parallelism and MPI.

The most important non-modifiable components are the lattice and the cell.
You can configure their behavior, but you are not expected to write a new class which inherits from or replaces the lattice or the cell.
Lattices are offered in different flavors, most of which inherit from a common interface \class{BlockStructureXD}.
The most common lattice is the regular \class{BlockLatticeXD}, which is replaced by the \texttt{SuperLatticeXD} for parallel applications and for memory-saving applications when faced with irregular domain boundaries.
An alternative choice for parallelism and memory savings is the \class{CuboidStructureXD}, which does not inherit from \class{BlockStructureXD}, but instead allows for more general constructs.

The modular extensions are classes that customize the behavior of core-components. An important extension of this kind is the lattice descriptor. This specifies the number of particle populations contained in a cell, and defines the lattice constants and lattice velocities, which are used to specify the neighborhood relation between a cell and its nearest neighbors. The lattice descriptor can also be used to require additional allocation of memory on a cell for external scalars, such as a force field. The integration of a lattice descriptor in a lattice happens via a template mechanism of C++. This mechanism takes place statically, \ie before program execution, and avoids the potential efficiency loss of a dynamic, object-oriented approach. Furthermore, template specialization is used to optimize the OpenLB code specifically for some types of lattices. Because of the template-based approach, a lattice descriptor needs not inherit from some interface. Instead, you are free to simply implement a new class, inspired from the default descriptors in the files \texttt{core/lattice\-Descriptors.h} and \texttt{core/lattice\-Descriptor.hh}.

The dynamics executed by a cell are implemented through a mechanism of dynamic (run-time) genericity.
In this way, the dynamics can be different from one cell to another, and can change during program execution.
There are two mechanisms of this type in OpenLB, one to implement local dynamics, and one for non-local dynamics.
To implement local dynamics, one needs to write a new class which inherits the interface of the abstract class \class{Dynamics}.
The purpose of this class is to specify the nature of the collision step, as well as other important information (for example, how to compute the velocity moments on a cell).
For non-local dynamics, a so-called post-processor needs to be implemented and integrated into a \class{BlockLatticeXD} through a call to the method \texttt{add\-Post\-ProcessorXD}.
This terminology can be somewhat confusing, because the term ``post-processing'' is used in the CFD community in the context of data analysis at the end of a simulation.
In OpenLB, a post-processor is an operator which is applied to the lattice after each streaming step.
Thus, the time-evolution of an OpenLB lattice consists of three steps: (1) local collision, (2) nearest-neighbor streaming, and (3) non-local postprocessing.
Implementing the dynamics of a cell through a postprocessor is usually less efficient than when the mechanism of the \texttt{Dynamics} classes is used.
It is therefore important to respect the spirit of the lattice Boltzmann method and to express the collision as a local operation whenever possible.

\section{Lesson~\arabic{section}: Use Checkpointing for Long Duration Simulations}
All types of data in OpenLB can be stored in a file or loaded from a file. This includes the data of a \class{BlockLatticeXD} and the data of a \texttt{Scalar\-FieldXD} or a \texttt{Tensor\-FieldXD}. All these classes implement the interface \texttt{Serializable<T>}. This guarantees that they can transform their content into a data stream of type \texttt{T}, or read from such a stream. Serialization and unserialization of data is mainly used for file access, but it can be applied to different aims, such as copying data between two objects of different type. The data is stored in the ascii-based binary format \texttt{Base64}. Although \texttt{Base64}-encoded data requires $25\%$ more storage space than when a pure binary format is used, this approach was chosen in OpenLB to enhance compatibility of the code between platforms. Saving and loading data is invoked by calling the \texttt{save} and \texttt{load} method on the object to be serialized. These methods take the filename as an optional (but recommended) argument, as shown below:

\begin{lstlisting}[language=myc++,caption=Store and load the state of the simulation.]
int nx, ny;
<...> initialization of nx and ny
BlockLattice<T,DESCRIPTOR> lattice(nx, ny);
// load data from a previous simulation
lattice.load("simulation.checkpoint");
<...> run the simulation
// save data for security, to be able to take up
// the simulation at this point later
lattice.save("simulation.checkpoint");
\end{lstlisting}

Checkpointing is also illustrated in the example programs \texttt{bstep2D} and \texttt{bstep3D} (Section~\ref{sec:bstep2d and bstep3d}).

\section{Lesson~\arabic{section}: Run Your Programs on a Parallel Machine}\label{sec:lesson10}

OpenLB programs can be executed on a parallel machine with distributed memory, based on MPI.
To compile an OpenLB program for parallel execution, modify the file named \path{config.mk}, found in the OpenLB root directory, by changing the respective lines: \texttt{CXX := mpic++}, and \texttt{PARALLEL\_MODE := MPI}.
The modified lines are shown in Listing~\ref{lst:makefileMPI}. Execute \texttt{make clean} and \texttt{make cleanbuild} within the desired program directory to eliminate previously compiled libraries, and recompile the program by executing the \texttt{make} command.
To run the program in parallel, use the command \texttt{mpirun -np 2 ./cavity2d}.
Here \texttt{-np 2} specifies the number of processors to be used.

\begin{lstlisting}[language=myc++,caption=Edited config.mk for MPI-parallel programs., label=lst:makefileMPI]
CXX             := mpic++
...
PARALLEL_MODE   := MPI
\end{lstlisting}

\section{Lesson~\arabic{section}: Work with Indicators}

Many of the methods covered up until this point accepted geometry and material number arguments to define their working domain. This can lead to repetition and code that is harder to read than necessary. An alternative to \eg setting up bulk dynamics using raw material numbers is available in the form of indicator functors.

In fact most of the material number accepting operations we have covered so far use generic lattice indicators under the hood, specifically \class{SuperIndicatorMaterial3D}.

\begin{lstlisting}[language=myc++,caption=Indicator usage example]
superLattice.defineDynamics(superGeometry, 1, &bulkDynamics);
superLattice.defineDynamics(superGeometry, 3, &bulkDynamics);
superLattice.defineDynamics(superGeometry, 4, &bulkDynamics);
// is equivalent to:
SuperIndicatorMaterial3D<T> bulkIndicator({1, 3, 4});
superLattice.defineDynamics(bulkIndicator, &bulkDynamics);
\end{lstlisting}

This can be further abstracted using \class{SuperGeometry3D}'s indicator factory:

\begin{lstlisting}[language=myc++,caption=Indicator usage in bstep3d]
auto bulkIndicator = superGeometry.getMaterialIndicator({1, 3, 4});
superLattice.defineDynamics(bulkIndicator, &bulkDynamics);
\end{lstlisting}

The advantage of this pattern is that we explicitly named materials 1, 3 and 4 as bulk materials and can reuse the indicator whenever we operate on bulk cells:

\begin{lstlisting}[language=myc++,caption=Indicator reusage in bstep3d]
superLattice.defineRhoU(bulkIndicator, rho, u);
superLattice.iniEquilibrium(bulkIndicator, rho, u);
\end{lstlisting}

This way the bulk material domain is defined in a central place which will come in handy should we need to change them in the future.

Note that for one-off usage this can be written even more compactly:

\begin{lstlisting}[language=myc++,caption=Inline indicator usage]
superLattice.defineDynamics(
  superGeometry.getMaterialIndicator({1, 3, 4}), &bulkDynamics);
\end{lstlisting}

\begin{sloppypar}
This pattern of using indicators instead of raw material numbers is available for all material number accepting methods of \class{SuperGeometryXD}.
\end{sloppypar}

The methods themselves support arbitrary \class{SuperIndicatorFXD} instances and as such are not restricted to material indicators.

\section{Alternative Approach: Using a Solver Class}\label{sec:solver}
Quite a lot of program components are similar for each OpenLB application: \eg the collide and stream loop is part of every simulation. The concept of a solver class is meant to perform such steps automatically, s.t. the user only has to define those steps which are specific for his/her application. Moreover, a generic interface shall be given for other programs (\eg launch from python scripts or execution of optimization routines). For both purposes, this is work in progress and more improvements and functionalities are under development. In the following, the parts of an OpenLB program in solver style are explained. These steps are also illustrated by the examples cavity2dSolver and porousPlate3dSolver.

\subsection{Structure of an OpenLB Simulation in Solver Style}
\subsubsection{Parameter handling}
In order to allow flexible interfaces to other programs, all parameters which are needed for simulation and interface are stored publicly in structs. For different groups of parameters (\eg simulation/ output/ stationarity), different structs are used.

For \texttt{Simulation}, \texttt{Output} and \texttt{Stationarity}, basic versions containing the essential parameters are given by \texttt{SimulationBase}, \texttt{OutputBase}, \texttt{StationarityBase}, respectively. These can be supplemented by inheritance.

More parameter structs could be added for individualization. For instance, a \texttt{Results} struct could be used to save simulation results.

\subsubsection{List parameter structs and lattices}
A map of parameter structs with corresponding names is defined as a \texttt{meta::map}. Similarly, a map of lattice names and descriptors is defined. Some typical names are provided at \path{src/solver/names.h}; a list which is intended to be extended for individualization. The two maps are then given to the solver class as template parameters.

\subsubsection{Definition of a solver class}
Many standard routines for simulation are implemented in the existing class \class{LbSolver}. It is templatized w.r.t.\ maps of parameters and lattices and should therefore fit to most application cases. However, some steps (like the definition of the geometry) depend on the application and have to be defined for each application.
Therefore, an application-specific solver class is created as a child class of \class{LbSolver}. It has to implement the methods \class{prepareGeometry}, \class{prepareLattices}, \class{setInitialValues} and \class{setBoundaryValues}, similar to the classical app structure. Moreover, methods \class{getResults}, 
\class{computeResults}, \class{writeImages}, \class{writeVTK} and \class{writeGnuplot} can be defined if such output is desired. They are all called automatically during construction/ simulation.
The access to the parameter structs works with the tags defined above: \eg, \class{this->parameters} 
\class{(Simulation()).maxTime} gives the maximal simulation time (which is a member of the struct 
\class{SimulationBase}). Similarly, we find access to super geometry and super lattices via \class{this->geometry()} and \class{this->} 
\class{lattice(LatticeName())}, respectively.
An automatic check, whether the simulation became stationary, is executed if a parameter struct with tag \class{Stationarity} is available (and the corresponding struct inherits from \class{StationarityBase}).

\subsubsection{Main method}
First, instances of the parameter structs and the solver class are constructed. This can be done classically, using the constructors (cf.\ example \path{porousPlate3dSolver}), or, if XML-reading has been implemented for all parameter structs, with the create-from XML-interface  (cf.\ example \path{cavity2dSolver}).

Secondly, the \class{solve()} method of the solver instance is called in order to run the simulation.

\subsection{Set up an Application in Solver Style}
In order to set up your own OpenLB application in solver style, the following steps should be followed:
\begin{itemize}
	\item Select parameter structs. You can use existing ones or inherit from them/ define them completely new. The simulation parameters are expected to inherit from \class{SimulationParameters} and provide a unit converter. The output parameters should inherit from \class{OutputParameters}. You are free to add more parameter structs (\eg for simulation results) yourself.
	\item Define a solver class. It should inherit publicly from the \class{LbSolver} class and implement the missing virtual methods like \class{prepareGeometry}.
	\item Define the main method. Either construct instances of the parameter structs and the solver class or use the create-from-xml interface. Then call the \class{solve()} method and possibly perform postprocessing.
\end{itemize}

\subsection{Parameter Explanation and Reading from XML}

\subsubsection*{Introduction}
In this short overview the relevant parameters for an app in solver style are listed with the respective names for using an xml-file for the input. The documentation is divided into two different subsections, which are the different main parts of the xml-file:
\begin{enumerate}
    \item[] \ref{sec:Application}. Application
    \item[] \ref{sec:Output}. Output
\end{enumerate}

In each subsection, the different parameters are explained via a table of the following form:
\begin{center}
\small
\begin{tabular}[ht]{ | m{1.5cm} | m{2cm}| m{4cm} | m{1.5cm} | m{3.5cm} | } 
  \hline
  \textbf{Parameter} &\textbf{Name}(type) & \textbf{Class} (file) & Default value & \textbf{Explanation:} (if available, all) \textbf{\textit{possibilities}}\\ 
  \hline
\end{tabular}
\end{center}

The explanation of each column is as follows:
\small
\begin{itemize}
    \item \textbf{Parameter:} Name of the parameter in the xml-file
    \item \textbf{Name} (type): Name of the parameter in the source code and its data type in brackets. Besides the common data types the abbreviations S and T are used for template parameters.
    
    \item \textbf{Class} (file): name of the class in which the parameter is stored and name of the file in which the class is defined.
    \item Default value: Is this parameter essential or optional or unused?\\
    If a parameter is optional, it is not needed to be defined. Then, the default value can be seen in this column.\\
    If a parameter is so important that without it the program has to exit, it is labeled as \textit{EXIT}.\\
    Some parameters are indicated with \textbf{unused} which means, that the parameter is read but not used afterwards.
    \item \textbf{Explanation: } (if available, all) \textbf{\textit{possibilities}}: Brief description and explanation of the parameter. Some parameters have different possibilities for their definition. In this case, all available possibilities are also offered in \textbf{bold} type letters, e.g.:\\
    \textit{\textbf{ad}} for \textit{OptiCaseAD} or\\
    \textit{\textbf{dual}} for \textit{OptiCaseDual} or\\
    \textit{\textbf{adTest}} for \textit{OptiCaseADTest}
\end{itemize}

The arrangement of the parameters in the xml-file has the following structures:

\begin{lstlisting}[language=XML]
<Param>
  <name of the subsection in this documentation>
    <superordinate term of the parameter>
      <name of the parameter> value of the parameter </name of the parameter>
    </superordinate term of the parameter>
  </name of the subsection in this documentation>
</Param>
\end{lstlisting}
example:
\begin{lstlisting}[language=XML]
<Param>
  <Application>
    <Discretization>
      <Resolution> 128 </Resolution>
    </Discretization>
  </Application>
</Param>    
\end{lstlisting}

\newpage
\subsubsection{Application} \label{sec:Application}
In the following, all parameters for the general setup of the application are explained.
\begin{center}
\small
\begin{longtable}[ht!]{ | m{1.5cm} | m{2cm}| m{4cm} | m{1.5cm} | m{3.5cm} | }
  \hline
  \textbf{Parameter} &\textbf{Name}(type) & \textbf{Class} (file) & Default value & \textbf{Explanation:} (if available, all) \textbf{\textit{possibilities}}\\ 
  \hline \hline

  \textit{Name} & name (std::string) & \class{OutputGeneral} (\path{solverParameters.h}) & "unnamed" & Output name \textit{Name}.dat for information of UnitConverter in the \textit{tmp} folder\\
  \hline
  
  \textit{OlbDir} & olbDir (std::string) & \class{OutputGeneral} (\path{solverParameters.h}) & "../../../" & Defines the trail for the OpenLB directory\\
  \hline
  
  \textit{Pressure-Filter}& pressureFilter-On (bool) & \class{SimulationBase} (\path{solverParameters.h}) & & Weights for moments computing,\\
    &  &  & false &if (this-$>$ pressureFilterOn):\\
  & & & & \_lattice-$>$ stripeOffDensityOffset (\_lattice-$>$ getStatistics(). getAverageRho() - (T) 1) in lbSolver.hh\\
  \hline \hline

  \multicolumn{5}{|c|}{\textit{Discretization:} }\\
  \hline
  
  \textit{Resolution} & \_resolution (int) & \class{UnitConverter} (\path{unitConverter.hh}) & & Defines the resolution number of the simulation area\\
  \hline
  
  \textit{Lattice-Relaxation-Time} & \_lattice-Relaxation-Time (T) & \class{UnitConverter} (\path{unitConverter.hh}) & & Defines the lattice relaxation time and thereby the \textit{physicalDeltaT, \_conversionTime} respectively\\
  \hline \hline
  
  \multicolumn{5}{|c|}{\textit{PhysParameters:} }\\
  \hline
  
  \textit{CharPhys-Length} & \_charPhys-Length (T) & \class{UnitConverter} (\path{unitConverter.hh}) & & Defines the characteristic physical length\\
  \hline
  
  \textit{CharPhys-Velocity} & \_charPhys-Velocity (T) & \class{UnitConverter} (\path{unitConverter.hh}) &  & Defines the characteristic physical velocity\\
  \hline
  
  \textit{PhysDensity} & \_physDensity (T) & \class{UnitConverter} (\path{unitConverter.hh}) & & Defines the physical density\\
  \hline
  
  \textit{CharPhys-Pressure} & \_charPhys-Pressure (T) & \class{UnitConverter} (\path{unitConverter.hh}) & & Defines the characteristic physical pressure\\
  \hline
  
  \textit{Phys-Viscosity} & \_phys-Viscosity (T) & \class{UnitConverter} (\path{unitConverter.hh}) & & Defines the physical viscosity\\
  \hline
  
  \textit{PhysMax-Time} & maxTime (S) & \class{SimulationBase} (\path{solverParameters.h}) & & Defines the maximal simulation time in seconds\\
  \hline
  
  \textit{StartUp-Time} & startUpTime (S) & \class{SimulationBase} (\path{solverParameters.h})  & & Defines the start time until which time the simulation is started up. From then on, the \textit{convergence criterion} is checked (in \path{solver3D.hh}) and determines the boundaries for \textit{defineU} in primalMode.\\
  \hline
  
  \textit{Boundary-Value-Update-Time} & \_phys-Boundary-ValueUpdate-Time (S) & \class{SimulationBase} (\path{solverParameters.h}) & PhysMax-Time/100 & Defines the boundary value update time\\
  \hline \hline

  \multicolumn{5}{|c|}{\textit{mesh:}}\\
  \hline
  
  \textit{noCuboids-PerProcess} & noC (int) & \class{SimulationBase} (\path{solverParameters.h}) & 1, \textbf{unused} & Defines the number of cuboids per process\\
  \hline\hline

    \multicolumn{5}{|c|}{\textit{ConvergenceCheck:}}\\
    \hline
    \textit{Type} & convergence-Type[0] (std::string) & \class{Stationarity} (\path{solverParameters.h}) & "Max-Lattice-Velocity" &  Define which quantity is regarded when checking the convergence. Alternative values: \textbf{AverageEnergy}, \textbf{AverageRho}\\
    \hline
    \textit{Interval} & phys-Interval[0] (BaseType$<$T$>$) & \class{Stationarity} (\path{solverParameters.h}) & & Time interval in physical in which the values are checked for convergence\\
    \hline
    
    \textit{Residuum} & epsilon[0] (BaseType$<$T$>$) & \class{Stationarity} (\path{solverParameters.h}) & & sensitivity for the convergence check\\
    \hline
\end{longtable}
\end{center}

\newpage
\subsubsection{Output}
\label{sec:Output}
In the following, all parameters for the output produced by the application are explained.
\begin{center}
\small
\begin{longtable}[ht!]{ | m{1.5cm} | m{2cm}| m{4cm} | m{1.5cm} | m{3.5cm} | }
  \hline
  \textbf{Parameter} &\textbf{Name}(type) & \textbf{Class} (file) & Default value & \textbf{Explanation:} (if available, all) \textbf{\textit{possibilities}}\\ 
  \hline \hline
  
  \textit{Multi-Output} & Used as (bool) & \path{testDomain3d.cpp} & & Choose 0 or 1 for clout.setMultiOutput, output for all processors. Needs to be manually added to the olbInit call\\
  \hline
  
  \textit{OutputDir} & outputDir (std::string) & \class{OutputGeneral} (\path{solverParameters.h}) & "./temp/" & Choose your output directory\\
  \hline
  
  \textit{PrintLog-Converter} & printLog-Converter (bool) & \class{OutputGeneral} (\path{solverParameters.h}) & true & If true, call writeLogConverter(), \ie print the converter\\
  \hline \hline

  \multicolumn{5}{|c|}{\textit{Log:}}\\
  \hline
  
  \textit{SaveTime} & logT (BaseType$<$T$>$) & \class{OutputGeneral} (\path{solverParameters.h}) & EXIT & Choose physical time, write the data in the terminal each $x$ seconds until PhysMaxTime; printLog(iT) in \path{lbSolver.hh}\\
  \hline
  
  \textit{VerboseLog} & verbose (bool) & \class{OutputGeneral} (\path{solverParameters.h}) & 1 (true) & Show statements in the terminal, if true\\
  \hline \hline
  
  \multicolumn{5}{|c|}{\textit{VisualizationVTK}, \textit{VisualizationImages} and \textit{VisualizationGnuplot}:}\\
  \hline
  
  \textit{Output} & out (bool) & \class{OutputPlot} (\path{solverParameters.h})  & false & States whether Output or not\\
  \hline
  
  \textit{FileName} & filename (std::string) & \class{OutputPlot} (\path{solverParameters.h}) & "unnamed" & name of the outputfile\\
  \hline
  
  \textit{SaveTime} & saveTime (S) & \class{OutputPlot} (\path{solverParameters.h}) & EXIT & Choose physical time, write the data each $x$ seconds until PhysMaxTime\\
  \hline \hline
  
 \pagebreak
 \hline
  \multicolumn{5}{|c|}{\textit{Timer:}}\\
  \hline
  
  \textit{PrintMode} & timerPrint-Mode (int)& \class{OutputGeneral} (\path{solverParameters.h}) & 2 & mode of the display style passed to printStep() of Timer instance,\\
  & & & & \textbf{0} for single-line layout, usable for data extraction as csv,\\
  & & & & \textbf{1} for single-line layout, (not conform with output rules),\\
  & & & & \textbf{2} for double-line layout, (not conform with output rules),\\
  & & & & \textbf{3} for performance output only\\
  \hline
 
\end{longtable}
\end{center}

%
%
%
%
%
%
%
%
%


\appendix
\chapter{Appendix}

\section{Q\&A}
\label{sec:QandA}
In this Q\&A part, some potential questions concerning the code are answered:

\subsection*{What do I need the unit converter for?}
The unit converter (Listing~\ref{lst:UnitConverter}) is used in every simulation done with OpenLB.
In this class, the physical units, like length or mass, are converted to lattice units and vice versa
This step is necessary to get a result in the correct physical dimensions and units.
\begin{lstlisting}[language=myc++,mathescape=true, caption=UnitConverter, label=lst:UnitConverter]
	UnitConverterFromResolutionAndLatticeVelocity<T,DESCRIPTOR> converter(
	(int)   res,                  //resolution
	( T )   charLatticeVelocity,  //charLatticeVelocity
	( T )   charPhysLength,       //charPhysLength
	( T )   charPhysVelocity,     //charPhysVelocity
	( T )   physViscosity,        //physViscosity
	( T )   physDensity           //physDensity
	);
	converter.print();
\end{lstlisting}
For a closer look, also checkout the respective example in Section~\ref{sec:LessonUnitConverter}.

\subsection*{How can I write my own dynamics?}\label{faq:dynamics}
Dynamics are the classes that model the cell-specific computation of momenta, equilibria and collision operator.
Ideally, new dynamics are constructed using OpenLB's flexible \emph{dynamics tuple} system which allows the composition of dynamics as a tuple of momenta, equilibria, and collision operator in addition to a optional \emph{combination rule} for declaring \eg forcing schemes.
This is the approach used for most dynamics in OpenLB (also see Section~\ref{sec:dynamics}).
\begin{lstlisting}[language=myc++]
template <typename T, typename DESCRIPTOR>
using ForcedTRTdynamics = dynamics::Tuple<
  T, DESCRIPTOR,
  momenta::BulkTuple,
  equilibria::SecondOrder,
  collision::TRT,
  forcing::Guo
>;
\end{lstlisting}
Modifying this example to use the BGK collision operator without forcing is as simple as writing:
\begin{lstlisting}[language=myc++]
template <typename T, typename DESCRIPTOR>
using BGKdynamics = dynamics::Tuple<
  T, DESCRIPTOR,
  momenta::BulkTuple,
  equilibria::SecondOrder,
  collision::BGK
>;
\end{lstlisting}
In most cases that are not yet covered by the extensive library of dynamics tuples, it should be sufficient to write \eg a new collision operator to be plugged into this framework.
As a fallback, fully custom dynamics can be implemented using \class{dynamics::CustomCollision} following this basic scaffold:
\begin{lstlisting}[language=myc++]
template <typename T, typename DESCRIPTOR, typename MOMENTA=momenta::BulkTuple>
struct MyCustomDynamics final : public dynamics::CustomCollision<T,DESCRIPTOR,MOMENTA> {
  using MomentaF = typename MOMENTA::template type<DESCRIPTOR>;

  // Declare list of parameter fields (can be empty)
  using parameters = meta::list<...>;

  // Allow exchanging the momenta, used for example to construct boundary dynamics
  template <typename M>
  using exchange_momenta = MyCustomDynamics<T,DESCRIPTOR,DYNAMICS,M>;

  template <typename CELL, typename PARAMETERS, typename V=typename CELL::value_t>
  CellStatistic<V> apply(CELL& cell, PARAMETERS& parameters) any_platform {
    // Implement custom collision here
  };

  T computeEquilibrium(int iPop, T rho, const T u[DESCRIPTOR::d]) const override any_platform {
    // Implement custom equilibrium computation here
  };

  std::type_index id() override {
    return typeid(MyCustomDynamics);
  };

  AbstractParameters<T,DESCRIPTOR>& getParameters(BlockLattice<T,DESCRIPTOR>& block) override {
    return block.template getData<OperatorParameters<MyCustomDynamics>>();
  }

  // Return human readable name
  std::string getName() const override {
    return "MyCustomDynamics<" + MomentaF().getName() + ">";
  };

};
\end{lstlisting}

\subsection*{How can I write my own post processor?}\label{faq:postprocessor}
A non-local operator, also referred to as a \emph{post processor} in OpenLB, is any class that provides a scope, a priority and an \texttt{apply} method template (also see Section~\ref{sec:postprocessor}).
The scope declares how the \texttt{apply} method is to be called, the priority is used to sort the execution sequence of multiple operators assigned to the same \emph{stage} and the \texttt{apply} method template implements the actual instructions to be performed.
\begin{lstlisting}[language=myc++, caption=Simple post processor implementation, label=lst:simplePostProcessorImpl]
struct MyCustomPostProcessor {
  // One of OperatorScope::(PerCell,PerCellWithParameters,PerBlock)
  static constexpr OperatorScope scope = OperatorScope::PerCell;

  int getPriority() const {
    return 0;
  }

  template <typename CELL>
  void apply(CELL& cell) any_platform {
    // custom non-local code here
    // access neighbors via `cell.neighbor(c_i)`
  }
};
\end{lstlisting}
This new post processor can be assigned to cells of the lattice using the various overloads of \texttt{SuperLattice::addPostProcessor}:
\begin{lstlisting}[language=myc++, caption=Simple post processor assignment, label=lst:simplePostProcessorAssign]
// Assign MyCustomPostProcessor to all cells
sLattice.addPostProcessor(meta::id<MyCustomPostProcessor>{});
// Assign MyCustomPostProcessor to indicated cells
sLattice.addPostProcessor(indicatorF,
                          meta::id<MyCustomPostProcessor>{});
\end{lstlisting}
If the operator depends on non-cell-specific parameters, they can be declared by changing the \texttt{scope} to \texttt{OperatorScope::PerCellWithParameters} and modifying the \texttt{apply} template arguments.
\begin{lstlisting}[language=myc++, caption=Simple post processor implementation with parameters, label=lst:simplePostProcessorImplParams]
struct MyCustomPostProcessor {
  static constexpr OperatorScope scope = OperatorScope::PerCellWithParameters;

  using parameters = meta::list<OMEGA>;

  int getPriority() const {
    return 0;
  }

  template <typename CELL, typename PARAMETERS>
  void apply(CELL& cell, PARAMETERS& parameters) any_platform {
    // access parameter via `parameters.template get<OMEGA>()`
  }
};
\end{lstlisting}
The parameters are set in the same way as dynamics parameters.
\begin{lstlisting}[language=myc++]
// Set OMEGA parameter of all dynamics and post processors to 0.6
sLattice.setParameter<OMEGA>(0.6);
\end{lstlisting}

\subsection*{How can I write my own coupling operator?}\label{faq:coupling}
A coupling operator is any class that provides a scope and an \texttt{apply} method template. Different from single-lattice non-local operators we receive not one cell but a named tuple of them.
Consider for illustration a basic coupling between NSE and ADE lattices in Listing~\ref{lst:nseAdeCouplingImpl}.
\begin{lstlisting}[language=myc++, caption=Full implementation of basic NSE-ADE coupling, label=lst:nseAdeCouplingImpl]
struct NavierStokesAdvectionDiffusionCoupling {
  // Declare that we want cell-wise coupling with some global parameters
  static constexpr OperatorScope scope = OperatorScope::PerCellWithParameters;

  // Declare the two parameters custom to this coupling operator
  struct FORCE_PREFACTOR : public descriptors::FIELD_BASE<0,1> { };
  struct T0 : public descriptors::FIELD_BASE<1> { };

  // Declare which parameters are required
  using parameters = meta::list<FORCE_PREFACTOR,T0>;

  template <typename CELLS, typename PARAMETERS>
  void apply(CELLS& cells, PARAMETERS& parameters) any_platform
  {
    using V = typename CELLS::template value_t<names::NavierStokes>::value_t;
    using DESCRIPTOR = typename CELLS::template value_t<names::NavierStokes>::descriptor_t;

    // Get the cell of the NavierStokes lattice
    auto& cellNSE = cells.template get<names::NavierStokes>();
    // Get the cell of the Temperature lattice
    auto& cellADE = cells.template get<names::Temperature>();

    // Computation of the Bousinessq force
    auto force = cellNSE.template getFieldPointer<descriptors::FORCE>();
    auto forcePrefactor = parameters.template get<FORCE_PREFACTOR>();
    V temperatureDifference = cellADE.computeRho() - parameters.template get<T0>();
    for (unsigned iD = 0; iD < DESCRIPTOR::d; ++iD) {
      force[iD] = forcePrefactor[iD] * temperatureDifference;
    }
    // Velocity coupling
    V u[DESCRIPTOR::d] { };
    cellNSE.computeU(u);
    cellADE.template setField<descriptors::VELOCITY>(u);
  }

};
\end{lstlisting}
Coupling operators are instantiated using the \texttt{SuperLatticeCoupling} template by providing a list of names and assigned lattices. For the NSE-ADE coupling this will look similar to (see \eg \texttt{thermal/rayleighBenard(2,3)d} for a practical example):
\begin{lstlisting}[language=myc++]
SuperLatticeCoupling coupling(
  NavierStokesAdvectionDiffusionCoupling{},
  names::NavierStokes{}, NSlattice,
  names::Temperature{},  ADlattice);
coupling.setParameter<NavierStokesAdvectionDiffusionCoupling::T0>(...);
coupling.setParameter<NavierStokesAdvectionDiffusionCoupling::FORCE_PREFACTOR>(...);
\end{lstlisting}
As we can see the parameters are set for the specific coupling instance.

\subsection*{What are aliases used for?}
Aliases are used to keep the code simpler for users and allow them to get a faster overview over the code.
Aliases aren't used for all parts of the program.
Especially if the user likes to change the code for a special problem, alias-functions may not be available and therefore need to be created by the user himself or normal functions can be used.
To sum up, \textit{alias}-functions aren't necessary for the simulation to work, but allow to simplify the code for a better overview for the user.


\subsection*{Why do we use "constexpr" to define some functions?}
Constexpr allows you to get the output of the function at the time of compilation. The value returned is set to a constant expression and as a result the runtime can be reduced, due to the constant value. If the return value is part of an if-loop, it won't be checked more than once, as the result is already clear, after the first check. In the example below the rotation of a 2D-particle is calculated and therefore the output is set to a concrete value at compilation time, and this value stays constant.

\begin{lstlisting}[language=myc++,mathescape=true, caption=constexpr, label=lst:constexpr]
	static constexpr Vector<T,2> execute( Vector<T,2> input, Vector<T,4> rotationMatrix, Vector<T,2> rotationCenter = Vector<T,2>(0.,0.) )
	{
		Vector<T,2> dist = input - rotationCenter;
		return Vector<T,2>(
		dist[0]*rotationMatrix[0] +
		dist[1]*rotationMatrix[2],
		dist[0]*rotationMatrix[1] +
		dist[1]*rotationMatrix[3] );
	}
\end{lstlisting}

\begin{flushleft}
	\textbf{7. What is the difference between files ending with \textit{.h} and ending with \textit{.hh}?}
\end{flushleft}
The files ending with \path{.h} like \path{particleDynamics.h} are header-files, where classes and functions are only declared.
The specification (definition) happens in the \path{.hh} files. In the following two code-snippets of the same function, the differences in terms of specification can be compared.

\begin{lstlisting}[language=myc++,mathescape=true, caption=particleDynamics.h, label=lst:Dynamics_1]
	template<typename T, typename DESCRIPTOR>
	class VerletParticleDynamics : public ParticleDynamics<T,DESCRIPTOR> {
		public:
		/// Constructor
		VerletParticleDynamics( T timeStepSize );
		/// Procesisng step
		void process (Particle<T,DESCRIPTOR>& particle) override;
		private:
		T _timeStepSize;
	};
\end{lstlisting}

\begin{lstlisting}[language=myc++,mathescape=true, caption=particleDynamics.hh, label=lst:Dynamics_2]
	template<typename T, typename PARTICLETYPE>
	VerletParticleDynamics<T,PARTICLETYPE>::VerletParticleDynamics ( T timeStepSize )
	: _timeStepSize( timeStepSize )
	{
		this->getName() = "VerletParticleDynamics";
	}

	template<typename T, typename PARTICLETYPE>
	void VerletParticleDynamics<T,PARTICLETYPE>::process (
	Particle<T,PARTICLETYPE>& particle )
	{
		//Calculate acceleration
		auto acceleration = getAcceleration<T,PARTICLETYPE>( particle );
		//Check for angular components
		if constexpr ( providesAngle<PARTICLETYPE>() ) {
			//Calculate angular acceleration
			auto angularAcceleration = getAngAcceleration<T,PARTICLETYPE>( particle );
			//Verlet algorithm
			particles::dynamics::velocityVerletIntegration<T, PARTICLETYPE>(
			particle, _timeStepSize , acceleration, angularAcceleration );
			//Update rotation matrix
			updateRotationMatrix<T,PARTICLETYPE>( particle );
		} else {
			//Verlet algorithm without rotation
			particles::dynamics::velocityVerletIntegration<T, PARTICLETYPE>(
			particle, _timeStepSize , acceleration );
		}
	}
\end{lstlisting}

\section{List of Project Participants}
\label{sec: List of Project Participants}

Since 2006 the following persons have contributed source code to OpenLB:
\begin{description}
\item[Armani Arfaoui:] \textbf{core}: performance improvements for D3Q19 BGK collision operator
\item[Sam Avis (active):] \textbf{dynamics}: multicomponent free energy model
\item[Saada Badie:] \textbf{core}: performance improvements for D3Q19 BGK collision operator
\item[Lukas Baron:] \textbf{utilities}: (parallel) console output, time and performance measurement, \textbf{dynamics}: porous media model, \textbf{functors}: concept, div. functors implementation
\item[Fedor Bukreev (active):] 
\textbf{reaction}: adsorption and reaction models, \textbf{examples}: adsorption examples, \textbf{organization}: testing
\item[Vojtech Cvrcek:] \textbf{dynamics}: power law, \textbf{examples}: power law, updates, \textbf{functors}: 2D adaptation
\item[Davide Dapelo (active):] \textbf{core}: power-law unit converter, \textbf{dynamics}: Guo--Zhao porous, contributions on power-law, contributions on HLBM, \textbf{examples}: reactionFiniteDifferences2d, advectionDiffusion3d, advectionDiffusionPipe3d, \textbf{functors}: contributions on indicator and smooth indicator
\item[Tim Dornieden:] \textbf{functors}: smooth start scaling, \textbf{io}: vti writer
\item[Jonas Fietz:] \textbf{io}: configure file parsing based on XML, octree STL reader interface to CVMLCPP ($<$ release 0.9), \textbf{communication}: heuristic load balancer
\item[Benjamin F\"orster:] \textbf{core}: super data implementation \textbf{io}: new serializer and serializable implementation, vti writer, new vti reader, \textbf{functors}: new discrete indicator
\item[Max Gaedtke:] \textbf{core}: unit converter, \textbf{dynamics}: thermal, \textbf{examples}: thermal
\item[Simon Gro{\ss}mann:] \textbf{example}: poiseuille2dEOC, \textbf{io}: csv and gnuplot interface, \textbf{postprocessing}: eoc analysis
\item[Nicolas Hafen (active):] \textbf{particles}: core framework, surface resolved particles, coupling, dynamics, creator-functions, \textbf{dynamics}: moving porous media (HLBM), \textbf{examples}: surface resolved particle simulations, \textbf{particles}: particle framework refactoring, sub-grid scale refactoring
\item[Marc Haussmann:] \textbf{dynamics}: turbulence modeling, \textbf{examples}: tgv3d, \textbf{io}: gnuplot heatmap
\item[Thomas Henn:] \textbf{io}: voxelizer interface based on STL, \textbf{particles}: particulate flows
\item[Claudius Holeksa]: \textbf{postProcessor}: free surface, \textbf{example}: free surface
\item[Anna Husfeldt]: \textbf{functors}: signed distance surface framework
\item[Jonathan Jeppener-Haltenhoff:] \textbf{functors}: wall shear stress, \textbf{examples}: channel3d, poiseuille3d, \textbf{core}: contributions to define field, \textbf{documentation}: user guide
\item[Julius Jeßberger (active):] \textbf{core}: solver, template momenta concept, optimization, \textbf{examples}: poi\-seuil\-le2d, cavity2dSolver, porousPlate3dSolver, testFlow3dSolver, optimization examples, \textbf{postprocessing}: error analysis, \textbf{utilities}: algorithmic differentiation
\item[Fabian Klemens:] \textbf{functors}: flux, indicator-based functors \textbf{io}: gnuplot interface
\item[Jonas Kratzke:] \textbf{core}: unit converter, \textbf{io}: GUI interface based on description files and OpenGPI, \textbf{boundaries}: Bouzidi boundary condition
\item[Mathias J.\ Krause (active):] \textbf{core}: hybrid-parallelization approach, super structure,  \linebreak \textbf{communication}: OpenMP parallelization, cuboid data structure for MPI parallelization, load balancing, \textbf{general}: makefile environment for compilation, integration and maintenance of added components (since 2008), \textbf{boundaries}: Bouzidi boundary condition, convection, \textbf{geometry}: concept, parallelization, statistics, \textbf{io} new serializer and serializable concept, \textbf{functors}: concept, div. functors implementation, \textbf{examples}: venturi3d, aorta3d, optimization (-2020), \textbf{optimization:} (-2020), \textbf{organization}: integration and maintenance of added components (2008-2017), project management (2006-)
\item[Louis Kronberg:] \textbf{core}: ade unit converter, \textbf{dynamics}: KBC, entropic LB, Cumulant, \textbf{examples}: advectionDiffusion1d, advectionDiffusion2d, bstep2d
\item[Eliane Kummer:] \textbf{documentation:} user guide
\item[Adrian Kummerländer (active):] \textbf{core}: SIMD CPU support, CUDA GPU support, population and field data structure, propagation pattern, vector hierarchy, cell interface, field data interface, meta descriptors, automatic code generation, \textbf{dynamics}: new dynamics concept, dynamics tuple, momenta concept \textbf{communication}: block propagation, communication, \textbf{functors}: lp-norm, flux, reduction, lattice indicator, error norms, refinement quality criterion, composition, \textbf{boundaries}: new post processor concept, water-tightness testing and post-processor priority, \textbf{general}: CI maintenance, Nix environment
\item[Jonas Latt:] \textbf{core}: basic block structure, \textbf{communication}: basic parallel block lattice approach ($<$ release 0.9), \textbf{io}: vti writer, \textbf{general}: integration and maintenance of added components (2006-2008), \textbf{boundaries}: basic boundary structure, \textbf{dynamics}: basic dynamics structure, \textbf{examples}: numerous examples, which have been further developed in recent years, \textbf{organization}: integration and maintenance of added components (2006-2008), project management (2006-2008)
\item[Marie-Luise Maier:] \textbf{particles}: particulate flows, frame change
\item[Orestis Malaspinas:] \textbf{boundaries}: alternative boundary conditions (Inamuro, Zou/He, Nonlinear FD), \textbf{dynamics}: alternative LB models (Entropic LB, MRT)
\item[Jan E. Marquardt (active):] \textbf{particles}: surface resolved particles, coupling, creator-functions, \textbf{dynamics}: Euler-Euler particle dynamics, \textbf{functors}: signed distance surface framework, \textbf{utilities}: algorithmic differentiation
\item[Cyril Masquelier:] \textbf{functors}: indicator, smooth indicator
\item[Albert Mink:] \textbf{functors}: arithmetic, \textbf{io}: parallel VTK interface3, zLib compression for VTK data, GifWriter, \textbf{dynamics}: radiative transport, \textbf{boundary}: diffuse reflective boundary
\item[Markus Mohrhard:] \textbf{general}: makefile environment for parallel compilation, \textbf{organization}:  integration and maintenance of added components (2018-)
\item[Johanna Mödl:] \textbf{core}: convection diffusion reaction dynamics, \textbf{examples}: advectionDiffusionReaction2d
\item[Patrick Nathen:] \textbf{dynamics}: turbulence modeling (advanced subgrid-scale \linebreak models), \textbf{examples}: nozzle3d
\item[Aleksandra Pachalieva:] \textbf{dynamics}: thermal (MRT model), \textbf{examples}: thermal (MRT model)
\item[Maximilian Schecher (active)]: \textbf{postProcessor}: free surface, \textbf{example}: free surface
\item[Stephan Simonis (active):] \textbf{core}: ade unit converter \textbf{examples}: advectionDiffusion1d, advectionDiffusion2d, advectionDiffusion3d, advectionDiffusionPipe3d, binaryShearFlow2d, fourRollMill2d, \textbf{documentation}: user guide, \textbf{dynamics}: MRT, KBC, Cumulant, entropic LB, free energy model
\item[Lukas Springmann:] \textbf{particles}: user-guide, unit tests
\item[Bernd Stahl:] \textbf{communication}: 3D extension to MultiBlock structure for MPI parallelization ($<$ release 0.9), \textbf{core}: parallel version of (scalar or tensor-valued) data fields ($<$ release 0.9), \textbf{io}: VTK output of data ($<$ release 0.9)
\item[Dennis Teutscher (active)]: \textbf{visualisation}, \textbf{organization}: user guide
\item[Robin Trunk:] \textbf{dynamics}: parallel thermal, advection diffusion models, 3D HLBM, Euler-Euler particle, multicomponent free energy model
\item[Peter Weisbrod:] \textbf{dynamics}: parallel multi phase/component,   \textbf{examples}: structure and showcases, phaseSeparationXd
\item[Gilles Zahnd:] \textbf{functors}: rotating frame functors
\item[Asher Zarth:] \textbf{core}: vector implementation
\item[Simon Zimny:] \textbf{io}: pre-processing: automated setting of boundary conditions
\end{description}


\section{GNU Free Documentation License}\label{sec:license}
\centerline{\Large GNU Free Documentation License}

 \begin{center}

       Version 1.2, November 2002

 Copyright \copyright{} 2000,2001,2002  Free Software Foundation, Inc.
 
 \bigskip
 
     51 Franklin St, Fifth Floor, Boston, MA  02110-1301  USA
  
 \bigskip
 
 Everyone is permitted to copy and distribute verbatim copies
 of this license document, but changing it is not allowed.
\end{center}

\begin{center}
\large{\textbf{Preamble}}
\end{center}

The purpose of this License is to make a manual, textbook, or other
functional and useful document ``free'' in the sense of freedom: to
assure everyone the effective freedom to copy and redistribute it,
with or without modifying it, either commercially or noncommercially.
Secondarily, this License preserves for the author and publisher a way
to get credit for their work, while not being considered responsible
for modifications made by others.

This License is a kind of ``copyleft'', which means that derivative
works of the document must themselves be free in the same sense.  It
complements the GNU General Public License, which is a copyleft
license designed for free software.

We have designed this License in order to use it for manuals for free
software, because free software needs free documentation: a free
program should come with manuals providing the same freedoms that the
software does.  But this License is not limited to software manuals;
it can be used for any textual work, regardless of subject matter or
whether it is published as a printed book.  We recommend this License
principally for works whose purpose is instruction or reference.

\begin{center}
\large{\textbf{1. Applicability and definitions}}
\end{center}

This License applies to any manual or other work, in any medium, that
contains a notice placed by the copyright holder saying it can be
distributed under the terms of this License.  Such a notice grants a
world-wide, royalty-free license, unlimited in duration, to use that
work under the conditions stated herein.  The ``\textbf{Document}'', below,
refers to any such manual or work.  Any member of the public is a
licensee, and is addressed as ``\textbf{you}''.  You accept the license if you
copy, modify or distribute the work in a way requiring permission
under copyright law.

A ``\textbf{Modified Version}'' of the Document means any work containing the
Document or a portion of it, either copied verbatim, or with
modifications and/or translated into another language.

A ``\textbf{Secondary Section}'' is a named appendix or a front-matter section of
the Document that deals exclusively with the relationship of the
publishers or authors of the Document to the Document's overall subject
(or to related matters) and contains nothing that could fall directly
within that overall subject.  (Thus, if the Document is in part a
textbook of mathematics, a Secondary Section may not explain any
mathematics.)  The relationship could be a matter of historical
connection with the subject or with related matters, or of legal,
commercial, philosophical, ethical or political position regarding
them.

The ``\textbf{Invariant Sections}'' are certain Secondary Sections whose titles
are designated, as being those of Invariant Sections, in the notice
that says that the Document is released under this License.  If a
section does not fit the above definition of Secondary then it is not
allowed to be designated as Invariant.  The Document may contain zero
Invariant Sections.  If the Document does not identify any Invariant
Sections then there are none.

The ``\textbf{Cover Texts}'' are certain short passages of text that are listed,
as Front-Cover Texts or Back-Cover Texts, in the notice that says that
the Document is released under this License.  A Front-Cover Text may
be at most 5 words, and a Back-Cover Text may be at most 25 words.

A ``\textbf{Transparent}'' copy of the Document means a machine-readable copy,
represented in a format whose specification is available to the
general public, that is suitable for revising the document
straightforwardly with generic text editors or (for images composed of
pixels) generic paint programs or (for drawings) some widely available
drawing editor, and that is suitable for input to text formatters or
for automatic translation to a variety of formats suitable for input
to text formatters.  A copy made in an otherwise Transparent file
format whose markup, or absence of markup, has been arranged to thwart
or discourage subsequent modification by readers is not Transparent.
An image format is not Transparent if used for any substantial amount
of text.  A copy that is not ``Transparent'' is called ``\textbf{Opaque}''.

Examples of suitable formats for Transparent copies include plain
ASCII without markup, Texinfo input format, LaTeX input format, SGML
or XML using a publicly available DTD, and standard-conforming simple
HTML, PostScript or PDF designed for human modification.  Examples of
transparent image formats include PNG, XCF and JPG.  Opaque formats
include proprietary formats that can be read and edited only by
proprietary word processors, SGML or XML for which the DTD and/or
processing tools are not generally available, and the
machine-generated HTML, PostScript or PDF produced by some word
processors for output purposes only.

The ``\textbf{Title Page}'' means, for a printed book, the title page itself,
plus such following pages as are needed to hold, legibly, the material
this License requires to appear in the title page.  For works in
formats which do not have any title page as such, ``Title Page'' means
the text near the most prominent appearance of the work's title,
preceding the beginning of the body of the text.

A section ``\textbf{Entitled XYZ}'' means a named subunit of the Document whose
title either is precisely XYZ or contains XYZ in parentheses following
text that translates XYZ in another language.  (Here XYZ stands for a
specific section name mentioned below, such as ``\textbf{Acknowledgements}'',
``\textbf{Dedications}'', ``\textbf{Endorsements}'', or ``\textbf{History}''.)  
To ``\textbf{Preserve the Title}''
of such a section when you modify the Document means that it remains a
section ``Entitled XYZ'' according to this definition.

The Document may include Warranty Disclaimers next to the notice which
states that this License applies to the Document.  These Warranty
Disclaimers are considered to be included by reference in this
License, but only as regards disclaiming warranties: any other
implication that these Warranty Disclaimers may have is void and has
no effect on the meaning of this License.

\begin{center}
\large{\textbf{2. Verbatim Copying}}
\end{center}

You may copy and distribute the Document in any medium, either
commercially or noncommercially, provided that this License, the
copyright notices, and the license notice saying this License applies
to the Document are reproduced in all copies, and that you add no other
conditions whatsoever to those of this License.  You may not use
technical measures to obstruct or control the reading or further
copying of the copies you make or distribute.  However, you may accept
compensation in exchange for copies.  If you distribute a large enough
number of copies you must also follow the conditions in section~3.

You may also lend copies, under the same conditions stated above, and
you may publicly display copies.

\begin{center}
\large{\textbf{3. Copying in quantity}}
\end{center}

If you publish printed copies (or copies in media that commonly have
printed covers) of the Document, numbering more than 100, and the
Document's license notice requires Cover Texts, you must enclose the
copies in covers that carry, clearly and legibly, all these Cover
Texts: Front-Cover Texts on the front cover, and Back-Cover Texts on
the back cover.  Both covers must also clearly and legibly identify
you as the publisher of these copies.  The front cover must present
the full title with all words of the title equally prominent and
visible.  You may add other material on the covers in addition.
Copying with changes limited to the covers, as long as they preserve
the title of the Document and satisfy these conditions, can be treated
as verbatim copying in other respects.

If the required texts for either cover are too voluminous to fit
legibly, you should put the first ones listed (as many as fit
reasonably) on the actual cover, and continue the rest onto adjacent
pages.

If you publish or distribute Opaque copies of the Document numbering
more than 100, you must either include a machine-readable Transparent
copy along with each Opaque copy, or state in or with each Opaque copy
a computer-network location from which the general network-using
public has access to download using public-standard network protocols
a complete Transparent copy of the Document, free of added material.
If you use the latter option, you must take reasonably prudent steps,
when you begin distribution of Opaque copies in quantity, to ensure
that this Transparent copy will remain thus accessible at the stated
location until at least one year after the last time you distribute an
Opaque copy (directly or through your agents or retailers) of that
edition to the public.

It is requested, but not required, that you contact the authors of the
Document well before redistributing any large number of copies, to give
them a chance to provide you with an updated version of the Document.

\begin{center}
\large{\textbf{4. Modifications}}
\end{center}

You may copy and distribute a Modified Version of the Document under
the conditions of sections 2 and 3 above, provided that you release
the Modified Version under precisely this License, with the Modified
Version filling the role of the Document, thus licensing distribution
and modification of the Modified Version to whoever possesses a copy
of it.  In addition, you must do these things in the Modified Version:

\begin{itemize}
\item[A.] 
   Use in the Title Page (and on the covers, if any) a title distinct
   from that of the Document, and from those of previous versions
   (which should, if there were any, be listed in the History section
   of the Document).  You may use the same title as a previous version
   if the original publisher of that version gives permission.
   
\item[B.]
   List on the Title Page, as authors, one or more persons or entities
   responsible for authorship of the modifications in the Modified
   Version, together with at least five of the principal authors of the
   Document (all of its principal authors, if it has fewer than five),
   unless they release you from this requirement.
   
\item[C.]
   State on the Title page the name of the publisher of the
   Modified Version, as the publisher.
   
\item[D.]
   Preserve all the copyright notices of the Document.
   
\item[E.]
   Add an appropriate copyright notice for your modifications
   adjacent to the other copyright notices.
   
\item[F.]
   Include, immediately after the copyright notices, a license notice
   giving the public permission to use the Modified Version under the
   terms of this License, in the form shown in the Addendum below.
   
\item[G.]
   Preserve in that license notice the full lists of Invariant Sections
   and required Cover Texts given in the Document's license notice.
   
\item[H.]
   Include an unaltered copy of this License.
   
\item[I.]
   Preserve the section Entitled ``History'', Preserve its Title, and add
   to it an item stating at least the title, year, new authors, and
   publisher of the Modified Version as given on the Title Page.  If
   there is no section Entitled ``History'' in the Document, create one
   stating the title, year, authors, and publisher of the Document as
   given on its Title Page, then add an item describing the Modified
   Version as stated in the previous sentence.
   
\item[J.]
   Preserve the network location, if any, given in the Document for
   public access to a Transparent copy of the Document, and likewise
   the network locations given in the Document for previous versions
   it was based on.  These may be placed in the ``History'' section.
   You may omit a network location for a work that was published at
   least four years before the Document itself, or if the original
   publisher of the version it refers to gives permission.
   
\item[K.]
   For any section Entitled ``Acknowledgements'' or ``Dedications'',
   Preserve the Title of the section, and preserve in the section all
   the substance and tone of each of the contributor acknowledgements
   and/or dedications given therein.
   
\item[L.]
   Preserve all the Invariant Sections of the Document,
   unaltered in their text and in their titles.  Section numbers
   or the equivalent are not considered part of the section titles.
   
\item[M.]
   Delete any section Entitled ``Endorsements''.  Such a section
   may not be included in the Modified Version.
   
\item[N.]
   Do not retitle any existing section to be Entitled ``Endorsements''
   or to conflict in title with any Invariant Section.
   
\item[O.]
   Preserve any Warranty Disclaimers.
\end{itemize}

If the Modified Version includes new front-matter sections or
appendices that qualify as Secondary Sections and contain no material
copied from the Document, you may at your option designate some or all
of these sections as invariant.  To do this, add their titles to the
list of Invariant Sections in the Modified Version's license notice.
These titles must be distinct from any other section titles.

You may add a section Entitled ``Endorsements'', provided it contains
nothing but endorsements of your Modified Version by various
parties--for example, statements of peer review or that the text has
been approved by an organization as the authoritative definition of a
standard.

You may add a passage of up to five words as a Front-Cover Text, and a
passage of up to 25 words as a Back-Cover Text, to the end of the list
of Cover Texts in the Modified Version.  Only one passage of
Front-Cover Text and one of Back-Cover Text may be added by (or
through arrangements made by) any one entity.  If the Document already
includes a cover text for the same cover, previously added by you or
by arrangement made by the same entity you are acting on behalf of,
you may not add another; but you may replace the old one, on explicit
permission from the previous publisher that added the old one.

The author(s) and publisher(s) of the Document do not by this License
give permission to use their names for publicity for or to assert or
imply endorsement of any Modified Version.

\begin{center}
\large{\textbf{5. Combining documents}}
\end{center}

You may combine the Document with other documents released under this
License, under the terms defined in section~4 above for modified
versions, provided that you include in the combination all of the
Invariant Sections of all of the original documents, unmodified, and
list them all as Invariant Sections of your combined work in its
license notice, and that you preserve all their Warranty Disclaimers.

The combined work need only contain one copy of this License, and
multiple identical Invariant Sections may be replaced with a single
copy.  If there are multiple Invariant Sections with the same name but
different contents, make the title of each such section unique by
adding at the end of it, in parentheses, the name of the original
author or publisher of that section if known, or else a unique number.
Make the same adjustment to the section titles in the list of
Invariant Sections in the license notice of the combined work.

In the combination, you must combine any sections Entitled ``History''
in the various original documents, forming one section Entitled
``History''; likewise combine any sections Entitled ``Acknowledgements'',
and any sections Entitled ``Dedications''.  You must delete all sections
Entitled ``Endorsements''.

\begin{center}
\large{\textbf{6. Collections of documents}}
\end{center}

You may make a collection consisting of the Document and other documents
released under this License, and replace the individual copies of this
License in the various documents with a single copy that is included in
the collection, provided that you follow the rules of this License for
verbatim copying of each of the documents in all other respects.

You may extract a single document from such a collection, and distribute
it individually under this License, provided you insert a copy of this
License into the extracted document, and follow this License in all
other respects regarding verbatim copying of that document.

\begin{center}
\large{\textbf{7. Aggregation with independent works}}
\end{center}

A compilation of the Document or its derivatives with other separate
and independent documents or works, in or on a volume of a storage or
distribution medium, is called an ``aggregate'' if the copyright
resulting from the compilation is not used to limit the legal rights
of the compilation's users beyond what the individual works permit.
When the Document is included in an aggregate, this License does not
apply to the other works in the aggregate which are not themselves
derivative works of the Document.

If the Cover Text requirement of section~3 is applicable to these
copies of the Document, then if the Document is less than one half of
the entire aggregate, the Document's Cover Texts may be placed on
covers that bracket the Document within the aggregate, or the
electronic equivalent of covers if the Document is in electronic form.
Otherwise they must appear on printed covers that bracket the whole
aggregate.

\begin{center}
\large{\textbf{8. Translation}}
\end{center}

Translation is considered a kind of modification, so you may
distribute translations of the Document under the terms of section~4.
Replacing Invariant Sections with translations requires special
permission from their copyright holders, but you may include
translations of some or all Invariant Sections in addition to the
original versions of these Invariant Sections.  You may include a
translation of this License, and all the license notices in the
Document, and any Warranty Disclaimers, provided that you also include
the original English version of this License and the original versions
of those notices and disclaimers.  In case of a disagreement between
the translation and the original version of this License or a notice
or disclaimer, the original version will prevail.

If a section in the Document is Entitled ``Acknowledgements'',
``Dedications'', or ``History'', the requirement (section~4) to Preserve
its Title (section~1) will typically require changing the actual
title.

\begin{center}
\large{\textbf{9. Termination}}
\end{center}

You may not copy, modify, sublicense, or distribute the Document except
as expressly provided for under this License.  Any other attempt to
copy, modify, sublicense or distribute the Document is void, and will
automatically terminate your rights under this License.  However,
parties who have received copies, or rights, from you under this
License will not have their licenses terminated so long as such
parties remain in full compliance.

\begin{center}
\large{\textbf{10. Future revisions of this license}}
\end{center}

The Free Software Foundation may publish new, revised versions
of the GNU Free Documentation License from time to time.  Such new
versions will be similar in spirit to the present version, but may
differ in detail to address new problems or concerns. See
http://www.gnu.org/copyleft/.

Each version of the License is given a distinguishing version number.
If the Document specifies that a particular numbered version of this
License ``or any later version'' applies to it, you have the option of
following the terms and conditions either of that specified version or
of any later version that has been published (not as a draft) by the
Free Software Foundation.  If the Document does not specify a version
number of this License, you may choose any version ever published (not
as a draft) by the Free Software Foundation.


\chapter*{References} 
\addcontentsline{toc}{chapter}{References}

\begin{singlespace}

\printbibliography[
    heading=none, 
    keyword={a}, 
    title={References Involving OpenLB (Articles)}
    ]

\printbibliography[
    heading=none, 
    keyword={p}, 
    title={References Involving OpenLB (Proceedings)}
    ]

\printbibliography[
    heading=none, 
    keyword={s}, 
    title={References Involving OpenLB (Software)}
    ]

\printbibliography[
    heading=none, 
    keyword=external,
    title={Other References}   
    ]

\end{singlespace}

\end{document}